%% file: subexp-cycle-hitting.tex
\documentclass[11pt]{article}
\usepackage[margin=1in]{geometry}
\usepackage[utf8]{inputenc}
\usepackage[T1]{fontenc}
\usepackage[dvipsnames]{xcolor}
\usepackage{graphicx}
\usepackage{amsmath, amssymb, amstext, mathtools}
\usepackage{amsthm}
\usepackage{ifthen}
\usepackage{tikz}
\usepackage{pgfplots}
\usepackage{float}
\usepackage{hyperref}
\usepackage{enumitem}
\graphicspath{{images/}}
\usepackage[sort,compress]{cite}

\newcommand{\executeiffilenewer}[3]{%
\ifnum\pdfstrcmp{\pdffilemoddate{#1}}%
{\pdffilemoddate{#2}}>0%
{\immediate\write18{#3}}\fi%
} 

\newcommand{\svg}[2]{\def\svgwidth{#1}%
\executeiffilenewer{images/#2.svg}{images/#2.pdf}%
{inkscape -z -D --file=images/#2.svg %
--export-pdf=images/#2.pdf --export-latex}%
{\input{images/#2.pdf_tex}}}
\allowdisplaybreaks

\usetikzlibrary{patterns,arrows,decorations.pathreplacing,decorations.pathmorphing,calc}

\tikzstyle{smallvertex}=[draw,circle,fill=black,minimum size=3pt,inner sep=0pt]
\tikzstyle{vertex}=[draw,circle,fill=black,minimum size=5pt,inner sep=0pt]

\newtheorem{theorem}{Theorem}[section]
\newtheorem{lemma}[theorem]{Lemma}
\newtheorem{corollary}[theorem]{Corollary}

\newtheorem{observation}[theorem]{Observation}
\newtheorem{proposition}[theorem]{Proposition}
\newtheorem{conjecture}[theorem]{Conjecture}
\newtheorem{reduction rule}{Reduction Rule}[section]
\theoremstyle{definition}
\newtheorem{definition}[theorem]{Definition}
\theoremstyle{remark}

\newcommand{\calC}{\mathcal{C}}

\newcommand{\calH}{{\mathcal H}}

\newcommand{\calO}{\ensuremath{{\mathcal O}}}

\newcommand{\Otilde}{\widetilde{\mathcal{O}}}
\newcommand{\Oh}{\mathcal{O}}

\newcommand{\calP}{\mathcal{P}}

\newcommand{\calT}{\mathcal{T}}

\newcommand{\FPT}{\textsf{FPT}}

\newcommand{\WH}{\textsf{W[1]-Hard}}
\newcommand{\ETH}{\textsf{ETH}}

\newcommand{\yes}{\textsc{Yes}}
\newcommand{\no}{\textsc{No}}

\newcommand{\CO}{\mathcal O}
\newcommand{\CB}{\mathcal B}
\newcommand{\CC}{\mathcal C}
\newcommand{\CD}{\mathcal D}
\newcommand{\CF}{\mathcal F}

\newcommand{\CI}{\mathcal I}
\newcommand{\CP}{\mathcal P}

\newcommand{\CS}{\mathcal S}

\newcommand{\FT}{\mathfrak T}

\newcommand{\Fc}{\mathfrak c}

\newcommand{\NN}{\mathbb N}

\newcommand{\op}{{\sf op}}
\newcommand{\glo}{{\sf global}}
\newcommand{\loc}{{\sf local}}

\DeclareMathOperator{\tw}{tw}
\DeclareMathOperator{\dist}{dist}
\DeclareMathOperator{\adeg}{adeg}

\DeclareMathOperator{\ar}{ar}

\DeclareMathOperator{\ext}{\sf ext}
\DeclareMathOperator{\shr}{\sf shr}

\DeclareMathOperator{\cut}{cut}

\DeclareMathOperator{\val}{val}

\DeclareMathOperator{\desc}{Desc}

\DeclareMathOperator{\cost}{cost}

\numberwithin{equation}{section}
\numberwithin{figure}{section}

\newcounter{claimcounter}
\newenvironment{claim}[1][]{
  \renewcommand{\proof}{\smallskip\par\noindent\textit{Proof. }}
  \medskip\par\noindent%
  \ifthenelse{\equal{#1}{}}{%
    \setcounter{claimcounter}{0}\refstepcounter{claimcounter}\textit{Claim~\arabic{claimcounter}.}
  }{%
    \ifthenelse{\equal{#1}{resume}}{%
      \refstepcounter{claimcounter}\textit{Claim~\arabic{claimcounter}.}
    }{%
      \textit{Claim~#1.}
    }
  }
}{
  \par\medskip
}
\newcommand{\uend}{\hfill$\lrcorner$}

\newcommand{\defparproblem}[4]{
  \vspace{1mm}
\noindent\fbox{
  \begin{minipage}{0.96\textwidth}
  \begin{tabular*}{\textwidth}{@{\extracolsep{\fill}}lr} #1  & {\bf{Parameter:}} #3
\\ \end{tabular*}
  {\bf{Input:}} #2  \\
  {\bf{Question:}} #4
  \end{minipage}
  }
  \vspace{1mm}
}

\newcommand{\defpargarproblem}[5]{
  \vspace{1mm}
\noindent\fbox{
  \begin{minipage}{0.96\textwidth}
  \begin{tabular*}{\textwidth}{@{\extracolsep{\fill}}lr} #1  & {\bf{Parameter:}} #3
\\ \end{tabular*}
  {\bf{Input:}} #2  \\
  {\bf{Question:}} #4 \\
  {\bf{Guarantee:}} #5
  \end{minipage}
  }
  \vspace{1mm}
}

\newcommand{\defproblem}[3]{
  \vspace{1mm}
\noindent\fbox{
  \begin{minipage}{0.96\textwidth}
  \begin{tabular*}{\textwidth}{@{\extracolsep{\fill}}lr} #1 \\ \end{tabular*}
  {\bf{Input:}} #2  \\
  {\bf{Question:}} #3
  \end{minipage}
  }
  \vspace{1mm}
}

\title{A Framework for Parameterized Subexponential Algorithms for Generalized Cycle Hitting Problems on Planar Graphs}
\author{
D{\'{a}}niel Marx\\
CISPA Helmholtz Center for Information Security, SIC, Germany.\\
\texttt{marx@cispa.de}
\and
Pranabendu Misra\\
Chennai Mathematical Institute, India.\\
\texttt{pranabendu@cmi.ac.in}
\and
Daniel Neuen\\
CISPA Helmholtz Center for Information Security, SIC, Germany\\
\texttt{daniel.neuen@cispa.de}
\and
Prafullkumar Tale\\
CISPA Helmholtz Center for Information Security, SIC, Germany\\
\texttt{prafullkumar.tale@cispa.de}
}

\date{}

\begin{document}

\maketitle

\begin{abstract}
  Subexponential parameterized algorithms are known for a wide range of natural problems on planar graphs, but the techniques are usually highly problem specific.
  The goal of this paper is to introduce a framework for obtaining $n^{\CO(\sqrt{k})}$ time algorithms for a family of graph modification problems that includes problems that can be seen as generalized cycle hitting problems.

  Our starting point is the \textsc{Node Unique Label Cover} problem (that is, given a CSP instance where each constraint is a permutation of values on two variables, the task is to delete $k$ variables to make the instance satisfiable). We introduce a variant of the problem where $k$ vertices have to be deleted such that every 2-connected component of the remaining instance is satisfiable. Then we extend the problem with cardinality constraints that restrict the number of times a certain value can be used (globally or within a 2-connected component of the solution). We show that there is an $n^{\CO(\sqrt{k})}$ time algorithm on planar graphs for any problem that can be formulated this way, which includes a large number of well-studied problems, for example, \textsc{Odd Cycle Transversal}, 
  \textsc{Subset Feedback Vertex Set}, \textsc{Group Feedback Vertex Set}, \textsc{Subset Group Feedback Vertex Set},
  \textsc{Vertex Multiway Cut}, and \textsc{Component Order Connectivity}.
  
  For those problems that admit appropriate (quasi)polynomial kernels (that increase the parameter only linearly and preserve planarity), our results immediately imply $2^{\CO(\sqrt{k} \cdot \textup{polylog}(k))}\cdot n^{\CO(1)}$ time parameterized algorithms on planar graphs.
  In particular, we use or adapt known kernelization results to obtain $2^{\CO(\sqrt{k}\cdot \textup{polylog}(k))}\cdot n^{\CO(1)}$ time (randomized) algorithms for \textsc{Vertex Multiway Cut}, \textsc{Group Feedback Vertex Set}, and \textsc{Subset Feedback Vertex Set}.

  Our algorithms are designed with possible generalization to $H$-minor free graphs in mind.
  To obtain the same $n^{\CO(\sqrt{k})}$ time algorithms on $H$-minor free graphs, the only missing piece is the vertex version of a contraction decomposition theorem that we currently have only for planar graphs.
\end{abstract}

\thispagestyle{empty}
\newpage

\tableofcontents
\thispagestyle{empty}

\newpage
\pagenumbering{arabic}

\input{introduction}

\input{preliminaries}

\input{perm-csp}

\input{structure}

\input{alg-perm-csp}

\input{kernel}

\bibliographystyle{plain}
\bibliography{references}

\appendix

\input{lower-bound}

\end{document}

%% file: introduction.tex
\section{Introduction}

For most NP-hard graph problems, restriction to planar graphs does not
seem to impact complexity: they remain NP-hard. However, in many
cases, planarity can be exploited to obtain improved algorithms
compared to what is possible for general graphs. In terms of approximability, polynomial-time approximation schemes (PTAS) are known for the planar restrictions of many problems that are APX-hard (or worse) on general graphs \cite{DBLP:conf/stoc/Klein06,DBLP:conf/focs/Klein05,DBLP:conf/stoc/Cohen-AddadGK021,DBLP:conf/soda/Fox-EpsteinKS19,DBLP:conf/wads/BeckerKS19,DBLP:conf/soda/BateniFH19,DBLP:conf/esa/BeckerKS17,DBLP:conf/stoc/BateniDHM16,DBLP:conf/stoc/FoxKM15,DBLP:conf/soda/EisenstatKM12,DBLP:conf/soda/BateniHKM12,DBLP:journals/jacm/BateniHM11,DBLP:conf/soda/BateniCEHKM11,DBLP:journals/talg/BorradaileKM09}.
Planarity can decrease also the exponential time needed to compute exact solutions. The Exponential-Time Hypothesis (ETH) implies for many NP-hard problems that they cannot be solved in time $2^{o(n)}$ on $n$-vertex graphs. An $n$-vertex planar graph has treewidth $\Oh(\sqrt{n})$ and there is a long list of problems that can be solved in time $2^{\Otilde(t)}\cdot n^{\Oh(1)}$ on graphs with treewidth at most $t$ (the $\Otilde$-notation hides polylogarithmic factors). Putting together these two results gives $2^{\Otilde(\sqrt{n})}$ time algorithms on $n$-vertex planar graphs. This form of running time can be considered essentially optimal, as usually it can be shown that no $2^{o(\sqrt{n})}$ time algorithm exists, assuming the ETH. Therefore, the appearance of the square root in the exponent is the natural behavior of the running time for many problems.

This ``square root phenomenon'' has been observed also in the context
of parameterized problems: subexponential running times of the form
$2^{\Otilde(\sqrt{k})}\cdot n^{\Oh(1)}$ and
$f(k)\cdot n^{\Otilde(\sqrt{k})}$ are known for a wide range of
problems, with matching lower bounds ruling out
$2^{o(\sqrt{k})}\cdot n^{\Oh(1)}$ and $f(k)\cdot n^{o(\sqrt{k})}$ time
algorithms, assuming the ETH \cite{DBLP:conf/focs/MarxPP18,DBLP:journals/siamcomp/ChitnisFHM20,DBLP:conf/focs/FominLMPPS16,DBLP:conf/esa/MarxP15,DBLP:conf/soda/KleinM14,DBLP:conf/icalp/KleinM12,DBLP:conf/icalp/Marx12,DBLP:conf/fsttcs/LokshtanovSW12,DBLP:journals/algorithmica/Verdiere17,DBLP:conf/stoc/Nederlof20a}.
The lower bounds usually follow a well-understood methodology: either they can be obtained from known
planar NP-hardness proofs or they are based on grid-like W[1]-hardness
proofs (see \cite[Section 14.4.1]{CyganFKLMPPS15}). On the other hand,
there is no generic argument for obtaining subexponential
running times with the square root in the exponent. As treewidth can be much larger than the parameter $k$ (representing for example the size of the solution we are looking for), the fact that planar graphs have treewidth $\Oh(\sqrt{n})$ does not help us on its own. While there are
certain algorithmic ideas that were used multiple times, the
algorithms are highly problem-specific and they do not give a general
understanding of why planar parameterized problems should have
subexponential algorithms on planar graphs. In fact, there is an example of a natural
problem, \textsc{Steiner Tree}, which was shown not to have
$2^{o(k)}\cdot n^{\Oh(1)}$ time algorithms on planar graphs, assuming
the ETH \cite{DBLP:conf/focs/MarxPP18}. Thus we cannot take it for granted that every reasonable
parameterized problem has subexponential algorithms on planar graphs.

There is a set of basic problems for which we can obtain subexponential parameterized algorithms using an argument called {\em bidimensionality}\cite{DBLP:journals/siamcomp/FominLST20,DBLP:conf/soda/FominLS12,DBLP:conf/soda/FominLRS11,DBLP:conf/esa/FominGT09,DBLP:journals/cj/DemaineH08,DBLP:journals/combinatorica/DemaineH08,DBLP:journals/siamdm/DemaineFHT04}. Let us consider for example the \textsc{Feeback Vertex Set} problem, where given a graph $G$ and an integer $k$, the task is to find a set $Z$ of at most $k$ vertices that hits all the cycles, that is, $G-Z$ is a forest. It is known that if a planar graph $G$ has treewidth $t$, then $G$ contains an $\Omega(t)\times \Omega(t)$ grid minor, which shows that there are $\Omega(t^2)$ vertex-disjoint cycles, requiring at least that many vertices for the solution. Therefore, there is a constant $c$ such that if $G$ has treewidth at least $c\sqrt{k}$, then we can safely answer ``no''; otherwise, if treewidth is at most $c\sqrt{k}$, then 
a standard $2^{\Otilde(t)}\cdot n^{\Oh(1)}$ treewidth based algorithm achieves running time $2^{\Otilde(\sqrt{k})}\cdot n^{\Oh(1)}$. There are a few other problems that can be handled by similar arguments (for example, \textsc{Independent Set}, \textsc{Dominating Set}, \textsc{Longest Path}, \textsc{Longest Cycle}), but the technique is not very robust and does not extend to most generalizations of the problems. Two well-studied generalizations of \textsc{Feedback Vertex Set} are \textsc{Odd Cycle Transversal} (\textsc{OCT}, hit every cycle of odd length) and \textsc{Subset Feedback Vertex Set} (\textsc{SFVS}, hit every cycle that contains at least one vertex of the terminal set $T$). Now a large grid minor does not imply an immediate answer to the problem: it could be that every cycle in the grid minor is of even length or they do not contain any terminal vertices.

The main contribution of this paper is formulating a family of deletion-type problems and showing that each such problem can be solved in time $n^{\Oh(\sqrt{k})}$ on planar graphs. Our original motivation comes from cycle-hitting problems such as \textsc{OCT} and \textsc{SFVS}, but our framework is much wider than that: it also includes problems such as \textsc{Vertex Multiway Cut} and \textsc{Component Order Connectivity} (where the task is to delete $k$ vertices such that every component has size at most $t$). A key feature of our framework is the ability to reason about problems that are defined in terms of the 2-connected components of the graph $G-Z$ resulting after the deletions. For example, \textsc{SFVS} can be defined as saying that in $G-Z$ every 2-connected component is either of size at most 2 or terminal free. \textsc{Odd Subset Feedback Vertex Set}, the common generalization of \textsc{OCT} and \textsc{SFVS} can be defined as saying that every 2-connected component is either bipartite or terminal free. The most novel technical ideas of our proofs are related to handling these constraints on 2-connected components.

It may look underwhelming that we are targeting running time $n^{\Oh(\sqrt{k})}$ for the above-mentioned problems: after all, they are known to be fixed-parameter tractable, many of them with $2^{\Otilde(k)}\cdot n^{\Oh(1)}$ time algorithms, hence the natural expectation is that they can be solved in time $2^{\Otilde(\sqrt{k})}\cdot n^{\Oh(1)}$ on planar graphs. However, let us observe that a polynomial kernelization together with an $n^{\Oh(\sqrt{k})}$ time algorithm delivers precisely this running time. A {\em polynomial kernelization} of problem $\Pi$ is a polynomial-time algorithm that, given an instance $x$ of $\Pi$ with parameter $k$, produces an equivalent instance $x'$ of $\Pi$ with parameter $k'$ such that $|x'|=|k'|^{\Oh(1)}$ and $k'=k^{\Oh(1)}$. In other words, the kernelization compresses the problem to a small instance $x'$ whose size is polynomially bounded in the new parameter $k'$ (hence also in the original parameter $k$). We need two additional properties of the kernelization for our applications. First, we require that the parameter be increased only linearly, that is, $k'=\Oh(k)$. Second, we need that if $x$ is a planar instance, then $x'$ is also planar. In general, a kernelization for a problem defined on general graphs does not need to keep planarity during the compression; therefore, we need a kernelization for the {\em planar version} of the problem. If we perform such a kernelization and then we solve the resulting instance $(x',k')$ using our algorithm, then the running time is $|x'|^{\Otilde(\sqrt{k'})}=k^{\Otilde(\sqrt{k})}=2^{\Otilde(\sqrt{k})}$, plus the polynomial running time of the kernelization. Thus our $n^{\Oh(\sqrt{k})}$ time algorithm can be seen as an important step towards obtaining subexponential FPT algorithms. As we shall see, for some problems the required kernelization results already exist (or existing results can be adapted to preserve planarity), hence new subexponential FPT algorithms follow from our work.

In order to appreciate our technical contributions, it is useful to take the following perspective. We can classify the problems considered here into three groups, characterized by a distinctive set of tools required to obtain subexponential FPT algorithms. Below we briefly highlight the technical tools that become relevant, and then proceed with a more detailed explanation. The first technique is well known, for the second we need to put together known pieces and extend them, and we use a significantly novel approach for the third technique. 

\begin{itemize}
 \item \textbf{Technique 1: Bidimensionality.} As explained above, for some problems (e.g., \textsc{Feedback Vertex Set}), we can exploit the fact that the answer is trivial if there is a grid of size $\Omega(\sqrt{k})\times \Omega(\sqrt{k})$, hence we may assume that treewidth is $\Oh(\sqrt{k})$. Then a $2^{\Otilde(t)}\cdot n^{\Oh(1)}$ algorithm on graphs of treewidth $t$ gives an algorithm with running time $2^{\Otilde(\sqrt{k})}\cdot n^{\Oh(1)}$.

 \item \textbf{Technique 2: Kernelization, contraction decomposition, guessing.} A contraction decomposition theorem can be ue used to find a set $A$ that has intersection $\Oh(\sqrt{k})$ with the solution $Z$. After guessing this intersection $A\cap Z$, we can remove it from the graph. We know that each component of $A\setminus Z$ is disjoint from the solution and in some problems such as \textsc{OCT}, the component can be contracted and represented by a single vertex (after appropriate modifications).
  If the contraction decomposition theorem ensures that the contracted graph has treewidth $\Oh(\sqrt{k})$, then the procedure involves guessing $n^{\Oh(\sqrt{k})}$ possibilities, followed by a $2^{\Otilde(\sqrt{k})}\cdot n^{\Oh(1)}$ time treewidth algorithm. Together with a kernelization, this is sufficient for our purposes. For edge-deletion problems, the required contraction decomposition theorems are readily available. While the technique is simple and may be implicit in earlier work, our main technical contribution for this group of problems is precisely formulating the required contraction decomposition theorem for vertex-deletion problems and proving it for planar graphs.

 \item \textbf{Technique 3: Guessing the 2-connected components.} The technique of solving \textsc{OCT} sketched above essentially relied on the assumption that if a set of edges is disjoint from the solution, then they can be contracted and represented by a single vertex. For problems that are defined through the 2-connected components of the remaining graph $G-Z$, it is not clear how such a contraction could be performed, as it could significantly change the structure of 2-connected components. Nevertheless, our main contribution for this group of problems is a way of making the approach via contraction work. If a connected set $I$ of vertices is disjoint from the solution $Z$, then the vertices of $I$ are contained in some number of 2-connected components of $G-Z$. We show that we can identify polynomially many possibilities for the union $B$ of all (but one) of these components.
  Then we can contract $I$ to a vertex $v_I$ and implement the problem as a binary constraint satisfaction (CSP) instance on the contract graph. The variable $v_I$ stores the correct choice of this union $B$ and constraints ensure that this choice is consistent with the 2-connected structure of $G-Z$.
\end{itemize}

To formally state our main technical results, we need to introduce the framework of constraint satisfaction problems (CSPs). A {\em binary CSP} (or {\em 2-CSP}) instance consists of a set $X$ of variables, a domain $D$, and a set of constraints on two variables. A constraint $((x_1,x_2),R)$ requires that in a satisfying assignment $\alpha:X\to D$, we have $(\alpha(x_1),\alpha(x_2))\in R$. We say that a relation $R$ is a {\em permutation} if for every $x\in D$ there is at most one $y_1\in D$ with $(y_1,x)\in R$ and at most one $y_2\in D$ with $(x,y_2)\in R$. {\em Permutation CSP} is an instance where every relation is a permutation. Given a CSP instance the {\em constraint graph} or {\em primal graph} is the graph with vertex set $X$ where two vertices are adjacent if there is a constraint on them. With some abuse of notation, we often identify the variables and constraints with the vertices and edges of the constraint graph, respectively.

The most basic question about CSP is satisfiability. Observe that this is trivial for a permutation CSP, as the value of a variable uniquely determines the values of the neighboring variables. However, for unsatisfiable instances, maximizing the number of satisfied constraints or minimizing the number of unsatisfied constraints is still a non-trivial task. We can also define these optimization problems using the language of deletion problems, which is more natural in our applications. In the \textsc{(Perm) CSP Edge Deletion} problem, we need to make a (permutation) CSP instance satisfiable by removing at most $k$ constraints. This problem is often called the \textsc{Unique Label Cover} problem in the approximation algorithms literature \cite{DBLP:conf/coco/Khot10,DBLP:conf/coco/Khot02,DBLP:conf/esa/LokshtanovRS17,DBLP:journals/jacm/AroraBS15}. In \textsc{(Perm) CSP Vertex Deletion}, the instance needs to be made satisfiable by removing at most $k$ variables (together with all the constraints appearing on them); this problem is also called \textsc{Node Unique Label Cover}. Given a tree decomposition of width at most $t$ of the constraint graph, these CSP deletion problems can be solved in time $|D|^{\Oh(t)}\cdot n^{\Oh(1)}$ using standard dynamic programming techniques. This algorithm will be a basic tool in our results.

To express a wider range of problems, we augment the CSP deletion problems with size constraints. Let $G$ be the constraint graph and let $Z$ be the set of edges or vertices that we delete from the instance. Let $C$ be a connected component of the constraint graph after the removal of $Z$. A size constraint may restrict the size of such a component, or the number of certain types of vertices that appear in a component, or the number of times certain values can appear in the component, etc. We define size constraints in the following very general way. Let $w:X\times D\to \mathbb{N}$ be a function and let $q\in \mathbb{N}$. Then an assignment $\alpha:X\to D$ satisfies the $\le$-constraint $(w,q)$ on component $C$ if $\sum_{x\in C}w(x,\alpha(x))\le q$. We say that the constraint is $Q$-bounded if $q\le Q$. A $\ge$-constraint $(w,q)$ is defined similarly. A CSP deletion instance can be augmented with more than one size constraint, in particular, a $\le$- and a $\ge$- constraint together can express equality.

Our first main technical result considers permutation CSP deletion problems and formalizes Technique 2 described above.
\begin{theorem}\label{thm:intromain1}
 \textsc{Perm CSP Edge Deletion} and \textsc{Perm CSP Vertex Deletion} with a set $\mathcal{S}$ of $Q$-bounded $\le$- or $\ge$-constraints can be solved in time
 $(|X|+|D|+Q^{|\mathcal{S}|})^{\Oh(\sqrt{k})}$ if the constraint graph is planar.
\end{theorem}
By simple reductions, we can obtain the following corollaries as example applications:
\begin{corollary}\label{cor:intro1}
The following problems can be solved in time $n^{\Oh(\sqrt{k})}$ on planar graphs: 
\textsc{OCT}, \textsc{Group Feedback Vertex Set} (where the group is part of the input), \textsc{Vertex Multiway Cut with Undeletable Terminals}, \textsc{Vertex Multiway Cut with Deletable Terminals}, \textsc{Component Order Connectivity} (where $t$ is part of the input), and the edge-deletion versions of all these problems. 
\end{corollary}
As mentioned earlier, if appropriate kernels are available, then Corollary~\ref{cor:intro1} implies subexponential FPT algorithms on planar graphs. In particular, the kernelization needs to preserve planarity, which is not necessarily a goal in kernelization algorithms designed to work on general graphs. We use/adapt (quasi)polynomial (randomized) kernels from the literature to obtain subexponential FPT algorithms for the following problems (see Section~\ref{sec:kernels} for more details):

\begin{corollary}\label{cor:fptintro1}
The following problems can be solved in randomized time $2^{\Otilde(\sqrt{k})}\cdot n^{\Oh(1)}$ on planar graphs: \textsc{Odd Cycle Transversal}, \textsc{Edge Multiway Cut}, \textsc{Vertex Multiway Cut},  \textsc{Group Feedback Vertex Set} (for a fixed group), \textsc{Group Feedback Edge Set} (where the size of the group is $k^{\Oh(1)}$).
\end{corollary}
Note that a subexponential FPT algorithm is known for \textsc{OCT} on planar graphs, but it uses more problem specific arguments \cite{DBLP:conf/fsttcs/LokshtanovSW12}. In particular, it is not based on a polynomial kernel for \textsc{OCT}, which was not known when \cite{DBLP:conf/fsttcs/LokshtanovSW12} was published.

Let us turn now our attention towards Technique 3, which allows us to handle problems defined by properties of the 2-connected components of the graph. In order to model such problems, we introduce first the 2-connected version of the permutation CSP deletion problem. In \textsc{2Conn Perm CSP Edge/Vertex Deletion}, we need to delete at most $k$ constraints/variables such that the remaining instance is satisfiable {\em on every 2-connected component} of the constraint graph. Note that even if every 2-connected component is satisfiable, this does not mean that the whole instance is satisfiable: a cut vertex appearing in 2-connected components $C_1$ and $C_2$ may need to receive different values in satisfying assignments of the two components.
Similarly to \textsc{Perm CSP Edge/Vertex Deletion}, we augment the problem with size constraints, but now the constraints restrict the value of $\sum_{v\in C}w(v,\alpha(v))$ on each 2-connected component $C$ and for technical reasons we allow only $\le$-constraints. The main technical result is the following:
\begin{theorem}\label{thm:intromain2}
  \textsc{2Conn Perm CSP Edge Deletion} and \textsc{2Conn Perm CSP Vertex Deletion} with a set $\mathcal{S}$ of $Q$-bounded $\le$-constraints can be solved in time
$\Oh(|X|+|D|+Q^{|\mathcal{S}|})^{\Oh(\sqrt{k})}$ if the constraint graph is planar.
\end{theorem}
Again by simple reductions, we can obtain the following corollaries:
\begin{corollary}\label{cor:intro2}
The following problems can be solved in time $n^{\Oh(\sqrt{k})}$ on planar graphs: 
\textsc{Subset Feedback Vertex Set}, \textsc{Two Subset Feedback Vertex Set}, \textsc{Subset OCT}, \textsc{Subset Group Feedback Vertex Set} (where the group is part of the input),  \textsc{$2$Conn Component Order Connectivity} (where $t$ is part of the input), and the edge-deletion versions of all these problems. 
\end{corollary}
We can adapt the \textsc{Subset Feedback Vertex Set} kernel of Hols and Kratsch \cite{HolsK18} to preserve planarity, yielding a subexponential FPT algorithm when combined with Corollary~\ref{cor:intro2}. 
\begin{theorem}\label{thm:intosfvs-fpt}
\textsc{Subset Feedback Vertex Set} can be solved in randomized time  $2^{\Otilde(\sqrt{k})}\cdot n^{\Oh(1)}$ on planar graphs. 
\end{theorem}
Currently no  polynomial kernel is known for \textsc{Subset Odd Cycle Transversal}. Note that this problem is a common generalization of \textsc{OCT} and \textsc{Subset Feedback Vertex Set}. We could use such a (planarity preserving) kernel and Corollary \ref{cor:intro2} to obtain a subexponential FPT algorithm for \textsc{Subset Odd Cycle Transversal} on planar graphs, which would be a common generalization of the subexponential FPT algorithms for \textsc{OCT} and \textsc{Subset Feedback Vertex Set}.

We would like to emphasize that our framework could be applied to a much wider range of problems than those explicitly stated in Corollaries~\ref{cor:intro1} and \ref{cor:intro2}. For example, we obtain $n^{\Oh(\sqrt{k})}$ time algorithms on planar graphs for problems of the type ``delete $k$ vertices such that every 2-connected component is either (1) bipartite or (2) contains at most 3 red and at most 5 blue terminals.'' A useful feature of our framework is that if two properties can be expressed as the specified permutation CSP with size constraints, then the OR of the properties can be expressed as well (by a permutation CSP on the union of the two domains) and also the AND of the properties (by a permutation CSP on the product of the two domains). However, these results give subexponential FPT algorithms only if suitable polynomial kernels are also available. Our work decouples the question of kernels (which is an interesting question on its own right) from the question of obtaining subexponential running time in planar graphs.

\subsection{Technical overview}
In this section, we give an intuitive overview of the main technical ideas for subexponential FPT algorithms that will be developed in detail in later sections.

\paragraph{Bidimensionality.} Earlier we briefly recalled how bidimensionality can be used to obtain subexponential parameterized algorithms on planar graphs for some problems. As our results concern problems where bidimensionality is not applicable, we do not want to elaborate further on this technique. The following bidimensional problems are listed only to contrast them with later problems that do not have this property.
\begin{itemize}
\item \textsc{Feedback Vertex Set} (remove $k$ vertices to make the graph acyclic). As we have seen, the existence of a $c\sqrt{k}\times c\sqrt{k}$ grid implies that there is no solution of size $k$.
\item \textsc{Even Cycle Transversal} (remove $k$ vertices to destroy every even cycle). It can be shown that a $c\sqrt{k}\times c\sqrt{k}$ grid contains $k$ vertex-disjoint even cycles, hence implying that there is no solution of size $k$.
\item \textsc{$\ge \ell$-Cycle Transversal} (remove $k$ vertices to destroy every cycle of length at least $\ell$). Again, for fixed $\ell$, a $c_\ell\sqrt{k}\times c_\ell\sqrt{k}$ grid minor implies that there is no solution of size $k$,
\item \textsc{Component Order Connectivity} (remove $k$ vertices such that every component has size at most $t$). For fixed $t$, a $c_t\sqrt{k}\times c_t\sqrt{k}$ grid minor implies that there is no solution of size $k$.
\item \textsc{$2$Conn Component Order Connectivity} (remove $k$ vertices such that every 2-connected component has size at most $t$). For fixed $t$, a $c_t\sqrt{k}\times c_t\sqrt{k}$ grid minor implies that there is no solution of size $k$.
\end{itemize}
We remark that if $\ell$ is part of the input, then \textsc{$\ge \ell$-Cycle Transversal} is NP-hard even for $k=0$, as it contains the \textsc{Hamiltonian Cycle} problem.
For \textsc{($2$Conn) Component Order Connectivity}, if $t$ is part of the input, then it becomes W[1]-hard, but our techniques show that they can be solved in time $n^{\Oh(\sqrt{k})}$. As this lower bound is secondary to the main algorithmic message of the paper, the lower bounds for these problems (and their variants) are moved to Appendix~\ref{sec:lower-bounds}.

\paragraph{Edge-Deletion Problems.} The edge-deletion version of \textsc{OCT} (delete $k$ edges to make the graph bipartite, also called \textsc{Edge Bipartization}) will be a convenient example to explain the main ideas behind the $n^{\Oh(\sqrt{k})}$ time algorithm following from Technique 2. For the clean treatment of the problem, let us introduce the following two minor extensions:
\begin{itemize}
\item The input contains a set $U\subseteq E(G)$ of {\em undeletable edges} and the solution $Z\subseteq E(G)$ has to be disjoint from $U$.
  \item The edges are of two types, odd and even, and the task is to destroy every cycle whose total parity is odd.
\end{itemize}
Our first observation is that if we have an undeletable edge $uv$, then we can simplify the instance: let us remove $v$, and for every edge $vw$, let us introduce an edge $uw$ whose parity is the parity of $uv$ plus the parity of $vw$ (see Figure~\ref{fig:edgecontr}). It is not difficult to observe that this does not change the problem: for any closed walk going through $u$ in the new graph, there is a corresponding closed walk with the same parity in the original graph.
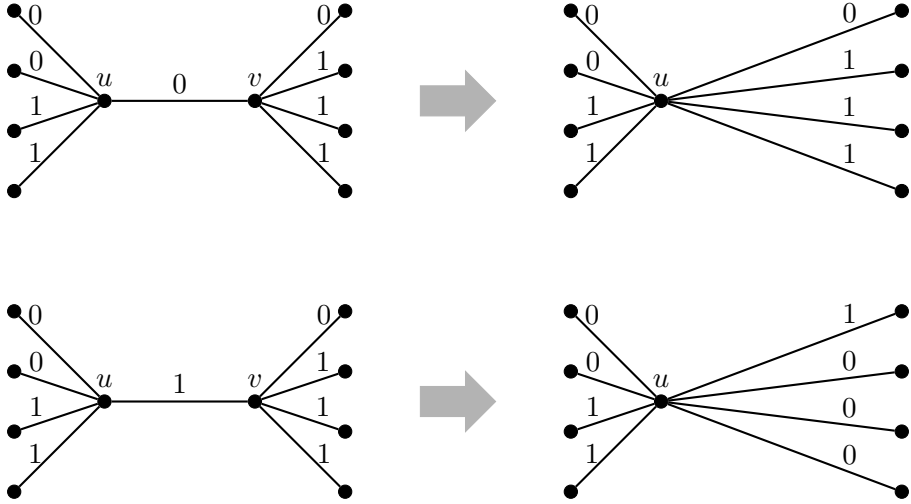
\begin{figure}
 \begin{center}
  \begin{tikzpicture}
   \node[vertex] (u1) at (0,-1.2) {};
   \node[vertex] (u2) at (0,-0.4) {};
   \node[vertex] (u3) at (0,0.4) {};
   \node[vertex] (u4) at (0,1.2) {};
   \node[vertex,label={[label distance=-1pt]90:{$u$}}] (v) at (1.2,0) {};
   \node[vertex,label={[label distance=-1pt]90:{$v$}}] (w) at (3.2,0) {};
   \node[vertex] (u5) at (4.4,-1.2) {};
   \node[vertex] (u6) at (4.4,-0.4) {};
   \node[vertex] (u7) at (4.4,0.4) {};
   \node[vertex] (u8) at (4.4,1.2) {};
   
   \path[draw,thick,-] (v) edge node[pos = 0.5, label={[label distance=-5pt]90:{$0$}}] {} (w);
   \path[draw,thick,-]
    (u1) edge node[pos = 0.2, label={[label distance=-5pt]90:{$1$}}] {} (v)
    (u2) edge node[pos = 0.2, label={[label distance=-5pt]90:{$1$}}] {} (v)
    (u3) edge node[pos = 0.2, label={[label distance=-5pt]90:{$0$}}] {} (v)
    (u4) edge node[pos = 0.2, label={[label distance=-5pt]90:{$0$}}] {} (v)
    (u5) edge node[pos = 0.2, label={[label distance=-5pt]90:{$1$}}] {} (w)
    (u6) edge node[pos = 0.2, label={[label distance=-5pt]90:{$1$}}] {} (w)
    (u7) edge node[pos = 0.2, label={[label distance=-5pt]90:{$1$}}] {} (w)
    (u8) edge node[pos = 0.2, label={[label distance=-5pt]90:{$0$}}] {} (w);
   
   \draw[fill=gray!60,gray!60] (5.4,-0.2) -- (6.0,-0.2) -- (6.0,-0.4) -- (6.4,0.0) -- (6.0,0.4) -- (6.0,0.2) -- (5.4,0.2) -- cycle;
   
   \node[vertex] (x1) at (7.4,-1.2) {};
   \node[vertex] (x2) at (7.4,-0.4) {};
   \node[vertex] (x3) at (7.4,0.4) {};
   \node[vertex] (x4) at (7.4,1.2) {};
   \node[vertex,label={[label distance=-1pt]90:{$u$}}] (y) at (8.6,0) {};
   \node[vertex] (x5) at (11.8,-1.2) {};
   \node[vertex] (x6) at (11.8,-0.4) {};
   \node[vertex] (x7) at (11.8,0.4) {};
   \node[vertex] (x8) at (11.8,1.2) {};
   
   \path[draw,thick,-]
    (x1) edge node[pos = 0.2, label={[label distance=-5pt]90:{$1$}}] {} (y)
    (x2) edge node[pos = 0.2, label={[label distance=-5pt]90:{$1$}}] {} (y)
    (x3) edge node[pos = 0.2, label={[label distance=-5pt]90:{$0$}}] {} (y)
    (x4) edge node[pos = 0.2, label={[label distance=-5pt]90:{$0$}}] {} (y)
    (x5) edge node[pos = 0.2, label={[label distance=-5pt]90:{$1$}}] {} (y)
    (x6) edge node[pos = 0.2, label={[label distance=-5pt]90:{$1$}}] {} (y)
    (x7) edge node[pos = 0.2, label={[label distance=-5pt]90:{$1$}}] {} (y)
    (x8) edge node[pos = 0.2, label={[label distance=-5pt]90:{$0$}}] {} (y);

   \node[vertex] (a1) at (0,-5.2) {};
   \node[vertex] (a2) at (0,-4.4) {};
   \node[vertex] (a3) at (0,-3.6) {};
   \node[vertex] (a4) at (0,-2.8) {};
   \node[vertex,label={[label distance=-1pt]90:{$u$}}] (b) at (1.2,-4) {};
   \node[vertex,label={[label distance=-1pt]90:{$v$}}] (c) at (3.2,-4) {};
   \node[vertex] (a5) at (4.4,-5.2) {};
   \node[vertex] (a6) at (4.4,-4.4) {};
   \node[vertex] (a7) at (4.4,-3.6) {};
   \node[vertex] (a8) at (4.4,-2.8) {};
   
   \path[draw,thick,-] (b) edge node[pos = 0.5, label={[label distance=-5pt]90:{$1$}}] {} (c);
   \path[draw,thick,-]
    (a1) edge node[pos = 0.2, label={[label distance=-5pt]90:{$1$}}] {} (b)
    (a2) edge node[pos = 0.2, label={[label distance=-5pt]90:{$1$}}] {} (b)
    (a3) edge node[pos = 0.2, label={[label distance=-5pt]90:{$0$}}] {} (b)
    (a4) edge node[pos = 0.2, label={[label distance=-5pt]90:{$0$}}] {} (b)
    (a5) edge node[pos = 0.2, label={[label distance=-5pt]90:{$1$}}] {} (c)
    (a6) edge node[pos = 0.2, label={[label distance=-5pt]90:{$1$}}] {} (c)
    (a7) edge node[pos = 0.2, label={[label distance=-5pt]90:{$1$}}] {} (c)
    (a8) edge node[pos = 0.2, label={[label distance=-5pt]90:{$0$}}] {} (c);
   
   \draw[fill=gray!60,gray!60] (5.4,-4.2) -- (6.0,-4.2) -- (6.0,-4.4) -- (6.4,-4.0) -- (6.0,-3.6) -- (6.0,-3.8) -- (5.4,-3.8) -- cycle;
   
   \node[vertex] (d1) at (7.4,-5.2) {};
   \node[vertex] (d2) at (7.4,-4.4) {};
   \node[vertex] (d3) at (7.4,-3.6) {};
   \node[vertex] (d4) at (7.4,-2.8) {};
   \node[vertex,label={[label distance=-1pt]90:{$u$}}] (e) at (8.6,-4) {};
   \node[vertex] (d5) at (11.8,-5.2) {};
   \node[vertex] (d6) at (11.8,-4.4) {};
   \node[vertex] (d7) at (11.8,-3.6) {};
   \node[vertex] (d8) at (11.8,-2.8) {};
   
   \path[draw,thick,-]
    (d1) edge node[pos = 0.2, label={[label distance=-5pt]90:{$1$}}] {} (e)
    (d2) edge node[pos = 0.2, label={[label distance=-5pt]90:{$1$}}] {} (e)
    (d3) edge node[pos = 0.2, label={[label distance=-5pt]90:{$0$}}] {} (e)
    (d4) edge node[pos = 0.2, label={[label distance=-5pt]90:{$0$}}] {} (e)
    (d5) edge node[pos = 0.2, label={[label distance=-5pt]90:{$0$}}] {} (e)
    (d6) edge node[pos = 0.2, label={[label distance=-5pt]90:{$0$}}] {} (e)
    (d7) edge node[pos = 0.2, label={[label distance=-5pt]90:{$0$}}] {} (e)
    (d8) edge node[pos = 0.2, label={[label distance=-5pt]90:{$1$}}] {} (e);
  \end{tikzpicture}
 \end{center}
 \caption{Contracting the undeletable edge $uv$ in the \textsc{Edge Bipartization} problem. The numbers on the edges show their parity.}
 \label{fig:edgecontr}
\end{figure}

The second key step is the use of the following contraction decomposition theorem:
\begin{theorem}[\cite{DBLP:conf/stoc/Klein06,DBLP:conf/focs/Klein05,DemaineHK11,DBLP:journals/combinatorica/DemaineHM10}]
 \label{thm:planarcontr}
 Let $G$ be a planar graph and $\ell \geq 1$.
 In polynomial-time, we can find a partition of the edge set $E(G) = E_1 \uplus E_2 \uplus \dots \uplus E_\ell$ such that, for every $i \in [\ell]$, it holds that
 \[\tw\big(G/E_i\big) =\Oh(\ell).\]
\end{theorem}
Let us invoke the algorithm of Theorem~\ref{thm:planarcontr} with
$\ell=\sqrt{k}$ and let $Z\subseteq E(G)$ be a hypothetical solution
of size at most $k$. Then for some $i\in[\ell]$, we have
$|Z\cap E_i|\le \sqrt{k}$. Let us guess this value of $i$ (we have $\sqrt{k}$
possibilities) and let us guess the set of edges $Z_i=Z\cap E_i$
(we have $n^{\CO(\sqrt{k})}$ possibilities). Now we can remove $Z_i$ from the
instance and mark every edge in $E_i\setminus Z$ as undeletable. Then, as
described above, we can modify the instance by contracting every edge
in $E_i\setminus Z$. We argue that the resulting graph
$G'=G/(E_i\setminus Z)$ has treewidth $\Oh(\sqrt{k})$. By
Theorem~\ref{thm:planarcontr}, the graph $G/E_i$ has treewidth
$\Oh(\sqrt{k})$ and $G/E_i$ can be obtained from $G/(E_i\setminus Z_i)$
by $|Z_i|$ further contractions. As contracting an edge can decrease
treewidth only by at most 1 and $|Z_i|\le \sqrt{k}$, it follows that
$\tw(G/(E_i\setminus Z))$ has treewidth $\Oh(\sqrt{k})$. Therefore, using
known $2^{\Oh(\tw(G))}\cdot n^{\Oh(1)}$ time algorithms, we can solve the
resulting equivalent instance on $G'$ in time $2^{\Oh(\sqrt{k})}\cdot n^{\Oh(1)}$.

    \paragraph{Vertex-Celetion Problems.} We can try to handle vertex-deletion problems such as \textsc{OCT} in a similar way, but we need to overcome some technical difficulties. Again, we can extend the problem with parities of the edges and a set $U\subseteq E(G)$ of undeletable edges, meaning that if $uv\in U$, then neither endpoint of $uv$ can be deleted. Then we would need a vertex-partition version of a contraction decomposition, where we contract each connected component of a vertex set $V_i$ into a single vertex (which is the same as saying that we contract every edge induced by $V_i$).
\begin{theorem}\label{thm:planarcontr-vert1}
 Let $G$ be a planar graph and $\ell \geq 1$.
In polynomial time, we can find a partition of the vertex set  $V(G) = V_1 \uplus V_2 \uplus \dots \uplus V_\ell$ such that, for every $i \in [\ell]$, it holds that
 \[\tw\big(G/E(G[V_i])\big)=\Oh(\ell).\]
\end{theorem}

Now we can proceed similarly as in the edge version: we guess an $i$ where $Z_i \coloneqq V_i \cap Z$ has size at most $\sqrt{k}$, guess $Z_i$, and mark every edge in $G[V_i\setminus Z_i]$ as undeletable. Then we can again contract these edges and arrive to an instance on $G'=G/E(G[V_i\setminus Z_i])$. We would need to connect the treewidth of $G'$ with the treewidth of $G/E(G[V_i])$, but this is not as obvious as in the edge-deletion case. We prove the following extension of Theorem~\ref{thm:planarcontr-vert1} that makes the connection apparent.
Intuitively, we can think of $V_i$ as collection of some kind of nested layers (see Figure~\ref{fig:segments}).
Accordingly, we call each component $I$ of $V_i\setminus Z$ a {\em segment.}
\begin{figure}
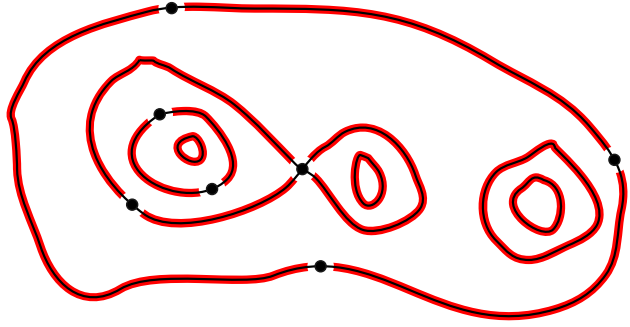

  \begin{center}
    \svg{0.5\linewidth}{segments}
  \end{center}
  \caption{The 7 vertices of $Z\cap V_i$ split $V_i$ into 12 segments.}\label{fig:segments}
\end{figure}

\begin{theorem}\label{thm:planarcontr-vert2}
 Let $G$ be a planar graph and $\ell \geq 1$.
In polynomial time, we can find a partition of the vertex set  $V(G) = V_1 \uplus V_2 \uplus \dots \uplus V_\ell$ such that, for every $i \in [\ell]$ and every $W\subseteq V_i$, it holds that
 \[\tw\big(G/E(G[V_i\setminus W])\big) =\Oh(\ell+|W|).\]
\end{theorem}
This extended stronger bound on treewidth allows the algorithm to go through for the vertex-deletion case without any difficulty. To make this formal, the following corollary combines Theorem~\ref{thm:planarcontr-vert2} with the guessing of $i$ and $V_i\cap Z$.
\begin{corollary}\label{cor:into-guess}
  Let $G$ be a planar graph and $k\ge 1$. In time $n^{\Oh(\sqrt{k})}$, we can find a sequence $V_1,\dots,V_h \subseteq V(G)$ of sets of vertices with $h=n^{\Oh(\sqrt{k})}$ such that for every $Z\subseteq V(G)$ with $|Z|\le k$, the following holds for at least one $1\le i \le h$:
  \begin{enumerate}
  \item $V_i\cap Z=\emptyset$, and
  \item $\tw(G/E(G[V_i]))=\Oh(\sqrt{k})$.
  \end{enumerate}  
\end{corollary}
That is, by trying every $V_i$, we are guaranteed to find one that can be safely contracted (as it is disjoint from the solution) and the contracted graph has treewidth $\Oh(\sqrt{k})$.

\paragraph{Permutation CSPs.} The arguments presented above can be generalized to every CSP deletion problem with permutation constraints, that is, for \textsc{Perm CSP Vertex Deletion}. There are two problem-specific points in the algorithm:
\begin{enumerate}
\item[(1)] We used a $2^{\Oh(t)}\cdot n^{\Oh(1)}$ time algorithm on graphs of treewidth $t$.
\item[(2)] We need a way of contracting the vertices in a component $C$ of undeletable edges to a single vertex.
\end{enumerate}
Condition (1) certainly holds for permutation CSP deletion problems over a fixed domain $D$: using standard dynamic programming techniques, we can find a minimum set of deletions that make the instance satisfiable. For (2), let us observe that if $C$ is a set of variables that are connected by an undeletable set of edges, then knowing the value of any variable $x\in C$ allows us to deduce the value of any other variable $y\in C$: the value of a variable uniquely determines the value of all its neighbors. Therefore, we can represent all the variables in $C$ by one of the variables, say $x$, and every constraint involving some $y\in C$ can be replaced by an equivalent constraint involving $x$. Therefore, we get $n^{\Oh(\sqrt{k})}$ time algorithms for \textsc{Perm CSP Edge/Vertex Deletion} on planar graphs the same way as we get such algorithms for \textsc{Edge Bipartization} and \textsc{OCT}.

\paragraph{Size Constraints.} Our goal is to extend the permutation CSP deletion problems with size constraints that restrict the number of appearances of certain values in {\em each connected component} of $G-Z$.
We need to check if (1) and (2) above still holds if we extend the problem with a set $\mathcal{S}$ of $\le$- or $\ge$-constraints, as defined earlier in the introduction. For (1), it is routine to extend a dynamic programming algorithm over a tree decomposition to enforce that a size constraint holds {\em globally} for $G-Z$. If there are $|\mathcal{S}|$ size constraints, and each size constraint bounds a value up to $Q$, then this adds a factor of $Q^{|\mathcal{S}|}$ overhead to the number of states of the dynamic programming table. However, we need to ensure that the size constraints hold for each component separately, and we cannot bound the number of components of $G-Z$. But from the viewpoint of dynamic programming, all that matters is that in a tree decomposition of width $t$, every bag can intersect $t+1$ of the components of $G-Z$. Therefore, it is sufficient to extend the dynamic programming with counters that keep track of each size constraint in at most $t+1$ components. Therefore, the $|\mathcal{S}|$ size constraints add only a factor of $Q^{\Oh(t|\mathcal{S}|)}$ to the running time.

For (2), let us first observe that every variable of $C$ is in the same component of $G-Z$. Therefore, knowing the value of one variable in $C$ allows us to tell how the variables of $C$ contribute to the size constraint. Therefore, when contracting $C$ to a single representative vertex, we can define the size constraint on that variable in a way that faithfully represents the contribution of $C$ to the size constraints.
    
\paragraph{2-Connected Components.} To explain Technique 3, let us consider for example the 
\textsc{Subset Feedback Vertex Set} problem, where we require that every 2-component of $G-Z$ be either of size at most 2, or terminal free.
Now if $I$ is a segment of undeletable edges, then contracting $I$ is highly problematic, as this could change the 2-connected components of $G-Z$ (see Figure~\ref{fig:2conn}). For example, in the figure the color of the rectangles represent the 2-connected components where they belong, but after the contraction of $I$ it is no longer possible to recover that rectangles of different color are supposed to be in different 2-connected components. Nevertheless, we manage to do the seemingly impossible:  given a set $V_i$ of Corollary~\ref{cor:into-guess}, we represent the problem as a CSP instance whose constraint graph is $G/E(G[V_i])$ and hence has treewidth $O(\sqrt{k})$.
\begin{figure}[t]
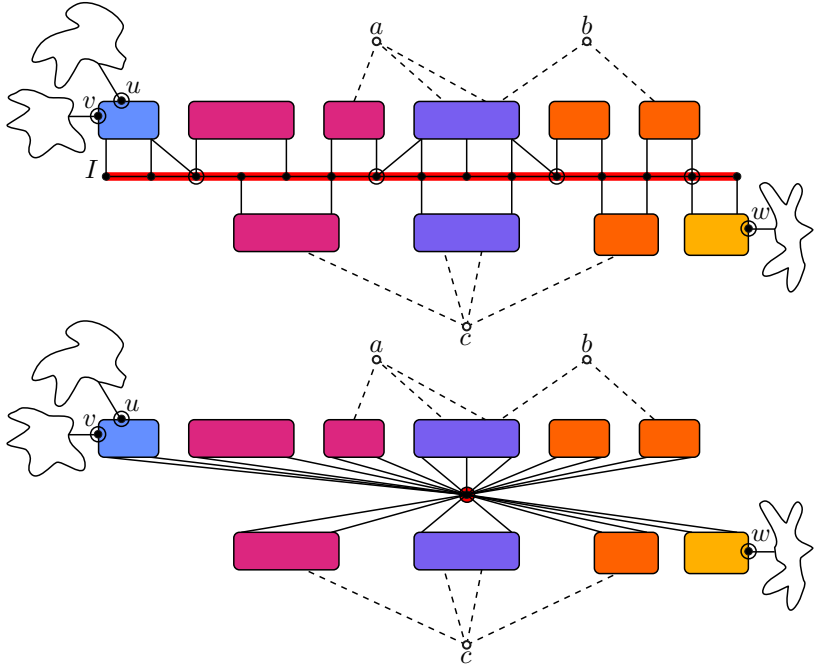

 \begin{center}
  {\small \svg{0.65\linewidth}{2conn}}
 \end{center}
 \caption{Top: After removing $Z=\{a,b,c\}$, segment $I$ is contained in 5 different 2-connected components of $G-Z$, as shown by the colors. Cut vertices are circled black dots; in particular, cut vertices $u,v,w$ connect the 2-connected components visited by $I$ to the rest of the graph. Bottom: after contracting $I$, it is no longer possible to recognize these 2-connected components.}
 \label{fig:2conn}
\end{figure}

A key idea in our handling of this issue is that guessing of a small number of vertices allows us to recover all (but one) of the 2-connected components visited by $I$ in $G-Z$. For example, in the figure, if we know that $a,b,c$ are precisely the vertices of $Z$ that are adjacent to these components, and $u,v,w$ are precisely the cut vertices where these components are connected to the rest of the graph, then this information is sufficient to recover all the 2-connected components visited by $I$ in $G-Z$. If we can argue that there is only a constant number of such vertices that are important to localizing the 2-connected components, then there is only polynomially many possibilities for the 2-connected components where $I$ appears in $G-Z$. Our main technical contribution is showing that (after some preprocessing and careful choice of $V_i$ in Theorem~\ref{thm:planarcontr-vert2}), we may assume that indeed there is only a constant number of such vertices for every segment $I$ of $V_i\setminus Z$.

Let us make these ideas more formal. Let us consider the decomposition of $G-Z$ into 2-connected components; we may assume that this is a rooted decomposition (see Figure~\ref{fig:2connbody}). Let $I$ be a set of vertices such that $G[I]$ is connected, and let us consider the 2-connected components that contain at least one vertex of $I$. There is a unique such component closest to the root, which we call the {\em root} $r(I)$ of $I$. The union of the vertices of every other 2-connected component intersected by $I$ is called the {\em body} $B(I)$ of $I$ (it is possible that $B(I)=\emptyset$). Note that a 2-connected component of $G-Z$ can be the root of multiple segments, can be both the root of some segment and the body of some other, but it can be in the body of only a single component: if a 2-connected component $C$ is in the body of $I$, then $I$ contains the vertex of $C$ joining it with the parent component. The following lemma can be seen as an extension of Corollary~\ref{cor:into-guess}: for each $V_i$, we also produce a collection $\mathcal{B}_i$ of sets that are possible candidates for the bodies of the components of $G[V_i]$.
\begin{figure}[t]
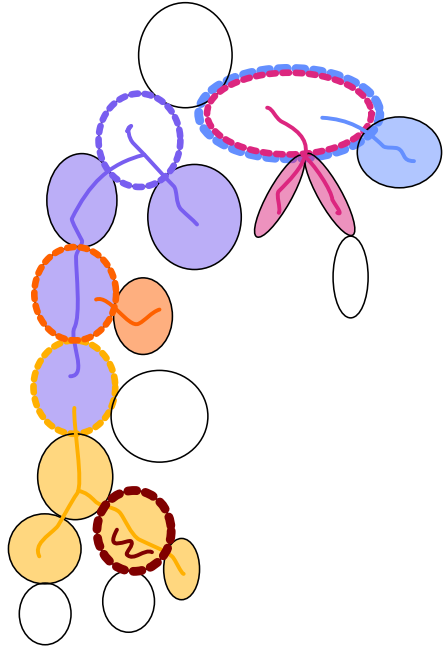

  \begin{center}
    {\small \svg{0.35\linewidth}{bodies}}
  \end{center}
  \caption{The rooted decomposition of the 2-connected components of $G-Z$ with 6 segments (shown by colored lines). The shading shows the bodies of the segments, the root components are shown by highlighted boundaries.}
  \label{fig:2connbody}
  \end{figure}

\begin{lemma}\label{lem:introbodyguess}
  Given a planar graph and an integer $k$, we can compute in time $n^{\Oh(\sqrt{k})}$ a sequence $(V_1,\mathcal{B}_1),\dots,(V_h,\mathcal{B}_h)$ with $h=n^{\Oh(\sqrt{k})}$, where $V_i\subseteq V(G)$ and $\mathcal{B}_i$ contains $n^{\Oh(1)}$ sets of vertices, such that the following holds.
  For every $Z\subseteq V(G)$ with $|Z|\le k$, there is at least one $1\le i \le h$ such that
  \begin{enumerate}
  \item $V_i\cap Z=\emptyset$,
  \item $\tw(G/E(G[V_i]))=\Oh(\sqrt{k})$, and
    \item for every component $I$ of $G[V_i]$, the body of $I$ in $G-Z$ is a set from $\mathcal{B}_i$.
  \end{enumerate}
\end{lemma}
We remark that it seems unavoidable that only the body of $I$ appears in Lemma~\ref{lem:introbodyguess} and not the union of every 2-connected component visited by $I$. The reason is that it may happen that $G-Z$ is 2-connected, in which case some collection $\mathcal{B}_i$ would need to contain precisely the set $V(G)\setminus Z$. But it is certainly easy to construct an instance that has $n^{O(k)}$ solutions $Z$ for which $G-Z$ is 2-connected. For such an instance, the output of Lemma~\ref{lem:introbodyguess} would need to contain all these $n^{O(k)}$ sets. The definition of body avoids this issue: if $G-Z$ is 2-connected, then the body of $I$ is empty.

\paragraph{Representing a Hypothetical Solution.}
With Lemma~\ref{lem:introbodyguess} at hand, we can translate a \textsc{2Conn Perm CSP Vertex Deletion} instance to a CSP deletion problem the following way. 
We invoke Lemma~\ref{lem:introbodyguess}, and try every $1\le i \le h$.
For a fixed $i$, we define a CSP instance whose constraint graph is $G'=\tw(G/E(G[V_i]))$.
We may assume that every set $B \in \mathcal{B}_i$ has the property that every 2-connected component induced by $B$ is satisfiable in the permutation CSP instance: if not, then this $B$ cannot be the body of any segment in the solution $G-Z$. Let us consider a hypothetical solution $G-Z$, the decomposition of $G-Z$ into 2-connected components, and satisfying assignments of the permutation CSP instances induced by each 2-connected component of $G-Z$. We describe how the variables can represent this solution.
If variable $v_I$ corresponds to the contraction of some component $I$ of $G[V_i]$, then it needs to store the following pieces of information:
\begin{itemize}
\item the value of one of the variables in $r(I)\cap I$ in a satisfying assignment of $r(I)$ (which also determines the value of every other variable in $r(I)\cap I$), and
  \item the body $B\in \mathcal{B}_i$ of $I$ in the hypothetical solution $Z$ that we are looking for.
  \end{itemize}
  Note that $v_I$ {\em does not} contain any information about the values of the variables of $I$ outside $r(I)$. However, this information is not necessary if the guess of the body $B$ is correct: by our assumption on $\mathcal{B}_i$, every 2-connected component induced by $B$ is satisfiable in the permutation CSP. It is the job of the variables corresponding to the vertices of $B\setminus I$ to enforce that the guess of $v_I$ about the body $B$ is correct. This is only possible if the information stored at a variable $v\in B\setminus I$ is sufficient to determine whether $v$ is in the same 2-connected component as a segment $I$, or more precisely, whether $v$ is in the body of $I$. Therefore, if $v\not\in V_i$, then $v$ stores the following pieces of information:
  \begin{itemize}
  \item The value of $v$ in the satisfying assignment of the 2-connected component where $v$ appears.\footnote{This is not well defined for cut vertices, which appear in more than one 2-connected component. Further tricks are needed to handle the value of cut vertices.}
  \item The segment $I$ whose body contains $v$ (if exists).
\item The body $B$ of this segment $I$ (if exists). 
  \item The cut vertex that joins the 2-connected component of $v$ to its parent component.
    \item The level of the 2-connected component of $v$ in the rooted decomposition into 2-connected components.
    \end{itemize}
We show that if the variables store all this information, then binary constraints can ensure that the stored information consistently describe the 2-connected structure of the instance and ensure the correctness of the hypothetical solution $Z$. Therefore, we can solve the \textsc{2Conn Perm CSP Vertex Deletion} instance by checking whether the constructed CSP instance can be made satisfiable by the removal of $k$ variables. As the primal graph of this instance is $G/E(G[V_i])$, Lemma~\ref{lem:introbodyguess} guarantees that it has treewidth $\Oh(\sqrt{k})$, leading to an $n^{\Oh(\sqrt{k})}$ time algorithm.

There is a technical detail, which will become important when introducing size constraints. We said that the values of the variables in the CSP instance describe a decomposition of $G-Z$ into 2-connected components. However, this decomposition may be coarser than the actual decomposition into 2-connected components. Fortunately, this does not cause any problem for \textsc{2Conn Perm CSP Vertex Deletion}: if the instance is satisfiable on every component of the coarser decomposition, then it follows that it is satisfiable on every subset of each such component.

\paragraph{Size Constraints for 2-Connected Components.} 
In the \textsc{$2$Conn Perm CSP Vertex Deletion} problem, we want to introduce size constraints that restrict the appearance of the values in each 2-connected component. Given that the CSP instance constructed above can identify the 2-connected components of $G-Z$, we can augment the dynamic programming algorithm for CSP deletion with size constraints on the 2-connected components. The technical issue mentioned in the previous paragraph does not matter for $\le$-constraints: if such constraints are satisfied for a superset of a 2-connected component, then they are satisfied for the component as well. However, because of this issue, we cannot introduce $\ge$-constraints \textsc{$2$Conn Perm CSP Vertex Deletion}. In particular, we do not have an $n^{\Oh(\sqrt{k})}$ time algorithm for the version of \textsc{$2$Conn Component Order Connectivity} where each 2-connected component has to be of size exactly $t$, with $t$ being part of the input. 

\paragraph{Kernelization.}
There is a somewhat counterintuitive phenomenon when considering kernelization algorithms for restricted graph classes. Even though problem \textsc{Planar $\Pi$} is a special case of problem $\Pi$, a polynomial kernelization for $\Pi$ {\em does not} imply a polynomial kernelization for \textsc{Planar $\Pi$}. The reason is that a kernelization algorithm for general graphs need not produce a planar output when the input is planar, thus it does not necessarily produce an instance of \textsc{Planar $\Pi$}.

There is a particular kernelization step, which we call {\em Mark and Torso}, that was used for many of the problems considered in the paper. Unfortunately, this step can ruin planarity (see Figure~\ref{fig:torso}). Given a set $X$ that is guaranteed to be disjoint from the solution, we put a clique on $N(X)$ and remove $X$ from the graph. Intuitively, we introduce all possible shortcuts for the paths that could go through $X$, which means that $X$ is no longer necessary in the graph. Clearly, adding a large clique ruins planarity. To obtain planarity-preserving kernels, we replace {\em Torso} with what we call {\em Contract to Undeletable}: we contract $X$ to a single vertex and mark it undeletable. This step also retains every path going through $X$, but leaves a somewhat larger graph. However, after some processing (removing vertices with the same neighborhood), we can prove a combinatorial bound showing that in planar graphs (and more generally, in $H$-minor-free graphs) the number of such undeletable vertices that we introduce is polynomial in the size of the rest of the graph.
\begin{figure}[t]
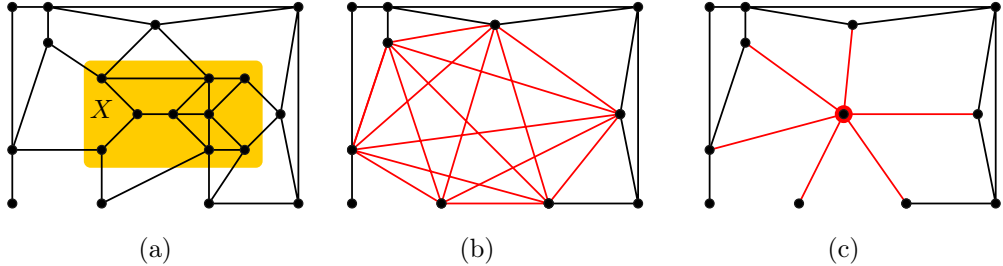

  \begin{center}
    {\small \svg{0.8\linewidth}{torso}}
  \end{center}
  \caption{(a) A set $X$ of vertices that is disjoint from the solution. (b) {\em Torso:} removing $X$ and adding a clique on $N(X)$. (c) {\em Contract to Undeletable:} set $X$ is contracted to a single undeletable vertex.}
  \label{fig:torso}
  \end{figure}

A minor technical issue is that this technique requires that we can mark  vertices as undeletable in the output instance, but some of the known kernelization algorithms do not work if we extend the input with undeletable vertices. This does not cause a problem for our algorithmic applications, but it means that some of the kernelization results should be stated as a compression into the problem extended with undeletable vertices (see Section~\ref{sec:kernels} for details).

There are additional adaptations that need to be done, and additional results that we use from the literature, which we only briefly highlight here.
\begin{itemize}
\item In the \textsc{Subset FVS} kernel of Hols and Kratsch \cite{HolsK18}, there is another step beyond {\em Mark and Torso} that does not preserve planarity. We need to carefully redo that part using different arguments that use and preserve planarity.
\item For \textsc{Group Feedback Vertex Set}, contracting $X$ to a single vertex and bounding the number of such vertices is more complicated, as we have to take into account the group elements on the edges.
\item The recent work of Wahlstr\"om \cite{Wahlstrom20} provides quasipolynomial kernels (which is suitable for our applications), but only for edge-deletion problems.
  \item The work of Jansen et al.~\cite{JansenPL19} provides deterministic kernels, but only for planar graphs and no generalization is known for $H$-minor-free graphs.
\end{itemize}

\paragraph{Extension to $H$-Minor-Free Graphs.} Graphs excluding a fixed graph $H$ as a minor are often considered to be a generalization of planar and bounded-genus graphs \cite{DemaineHK11,DBLP:journals/jacm/DemaineFHT05}. Indeed, many of the good algorithmic and combinatorial properties of planar graphs appear to be generalizable to $H$-minor-free graphs for every fixed $H$: the Graph Minors Structure Theorem of Roberston and Seymour \cite{DBLP:journals/jct/RobertsonS03a} provides a roadmap for proving such results.

While results for $H$-minor-free graphs are (understandably) more complicated to prove than for planar graphs, often the increase of complexity is hidden in ``black box'' results that can be used as convenient tools. Therefore, in addition to increasing the generality of the result, there is another potential motivation for the generalization to $H$-minor-free graphs: to develop cleaner and more robust arguments. Arguments about planar graphs typically rely on topological intuition, which can be misleading and formal proofs require very careful treatment of all possible situations. On the other hand, if an argument for $H$-minor-free graphs only uses some black box results, then it can be actually simpler than its planar counterpart and could reveal better the nature of the problem.

With this goal in mind, we developed our results in a way that would work also for $H$-minor-free graphs. The only point where we use specific geometric properties of planar graphs is in the proof of Theorem~\ref{thm:planarcontr-vert2} and indeed the only missing piece is the generalization of Theorem~\ref{thm:planarcontr-vert2} to $H$-minor-free graphs. For the edge partition version, such a generalization is known:

\begin{theorem}[Demaine et al.\ \cite{DemaineHK11}]
 \label{thm:Hcontr}
 For every fixed $H$, there is a $c_H>0$ such that the following holds. Let $G$ be an $H$-minor-free graph and $\ell \geq 1$.
 In polynomial time, we can find a partition of the edge set $E(G) = E_1 \uplus E_2 \uplus \dots \uplus E_\ell$ such that, for every $i \in [\ell]$, it holds that
 \[\tw\big(G/E_i\big) =c_H\cdot \ell.\]
\end{theorem}

This allows us to generalize for example the edge-deletion version of Theorem~\ref{thm:intromain1} to $H$-minor-free graphs. However, we need the vertex partition version of Theorem~\ref{thm:Hcontr} for two reasons: to handle vertex-deletion problems and to handle the 2-connected versions of the problems (where the cut vertices make it necessary to work with a vertex partition, even in the edge-deletion version).
We formulate as a conjecture that Theorem~\ref{thm:planarcontr-vert2} can be generalized in a similar way to $H$-minor-free graphs.

\begin{conjecture}
 \label{conj:contraction-decomposition-minor-vertex}
 For every fixed $H$, there is a $c_H>0$ such that the following holds.
 Let $G$ be an $H$-minor-free graph and $\ell \geq 1$.
 In polynomial time, we can find a partition of the vertex set  $V(G) = V_1 \uplus V_2 \uplus \dots \uplus V_\ell$ such that, for every $i \in [\ell]$ and every $W\subseteq V_i$, it holds that
 \[\tw\big(G/E(G[V_i\setminus W])\big) =c_H\cdot (\ell+|W|).\]
\end{conjecture}

It is likely that the techniques in the proof of Theorem~\ref{thm:Hcontr}, together with additional ideas from the proof of Theorem~\ref{thm:planarcontr-vert2} could lead to a proof of Conjecture~\ref{conj:contraction-decomposition-minor-vertex}. However, this would require going into the details of the proof of Theorem~\ref{thm:Hcontr}~\cite{DemaineHK11}, carefully checking and adapting every step to the vertex-partition case, which is beyond the scope of this paper. Therefore, in the rest of the paper, we point out whenever a result can be generalized to $H$-minor-free graphs assuming Conjecture~\ref{conj:contraction-decomposition-minor-vertex}. We also make it clear which of the kernelization results keep not only planarity, but $H$-minor-freeness as well.

\subsection{Organization}

The paper is organized as follows. Section~\ref{sec:preliminaries} introduces basic notions related to graph theory, parameterized algorithms, and CSPs.
Section~\ref{sec:perm-csp-results} formally defines the problems considered in the paper, states our results for the technical CSP problems, and briefly goes through how problems of interest can be reduced to these technical problems.
Additionally, we state the kernelization results known from the literature or adapted in this paper, and obtain subexponential FPT algorithms as corollaries.

Technical work is started in Section~\ref{sec:structure}, where we prove our results on contraction decomposition by partitioning the vertices (Theorem~\ref{thm:planarcontr-vert2}) and the guessing of the bodies of the segments (Lemma~\ref{lem:introbodyguess}).
Section~\ref{sec:alg-perm-csp} first presents an algorithm for CSP deletion with size constraints.
Then we use this tool to obtain the algorithms for \textsc{($2$Conn) Perm CSP Edge/Vertex Deletion}.
Section~\ref{sec:kernels} shows how known kernelization results can be adapted to make them planarity preserving (or even preserving $H$-minor-freeness).
Appendix~\ref{sec:lower-bounds} presents the lower bound for \textsc{($2$Conn) Component Order Connectivity}.

%% file: preliminaries.tex
\section{Preliminaries}
\label{sec:preliminaries}

\subsection{Basics}

We use $\NN = \{0,1,2,3,\dots\}$ to denote the natural numbers starting from $0$.
For a positive integer $q$, we denote the set $\{1, 2, \dots, q\}$ by $[q]$.

In this article, we consider simple graphs with a finite number of vertices.
For an undirected graph $G$, sets $V(G)$ and $E(G)$ denote its set of vertices and edges, respectively.
Unless otherwise specified, we use $n$ to denote the number of vertices of the input graph $G$.
We denote an edge $e$ with two endpoints $u, v$ as $e = uv$.
Two vertices $u, v$ in $V(G)$ are \emph{adjacent} to each other if $uv \in E(G)$. 
The open neighborhood of a vertex $v$, denoted by $N_G(v)$, is the set of vertices adjacent to $v$ and its degree $\deg_G(v)$ is $|N_G(v)|$.
The closed neighborhood of a vertex $v$, denoted by $N_G[v]$, is the set $N(v) \cup \{v\}$.
We omit the subscript in the notation for neighborhood and degree if the graph under consideration is clear.
For $S \subseteq V(G)$, we define $N[S] = \bigcup_{v \in S} N[v]$ and $N(S) = N[S] \setminus S$.

For $S \subseteq V(G)$, we denote the graph obtained by deleting $S$ from $G$ by $G - S$ and the subgraph of $G$ induced on the set $S$ by $G[S]$. 
For two subsets $S_1,S_2 \subseteq V(G)$, the set $E(S_1,S_2)$ denotes the edges with one endpoint in $S_1$ and another one in $S_2$. 
We say $S_1,S_2$ are adjacent if $E(S_1,S_2) \neq \emptyset$. 
For a subset $F \subseteq E(G)$ of edges we write $V(F)$ to denote the collection of endpoints of edges in $F$.
The subgraph of $G$ with $V(F)$ as its set of vertices and $F$ as its set of edges is denoted by $G[F]$.

A set of vertices $S \subseteq V(G)$ is said to be an \emph{independent set} in $G$ if no two vertices in $S$ are adjacent to each other.
A set of vertices $S$ is a \emph{vertex cover} of $G$ if $V(G) \setminus S$ is an independent set in $G$.
A {\em tree-decomposition} of a graph $G$ is a pair $(T,\beta)$ where $T$ is a tree and $\beta\colon V(T) \rightarrow 2^{V(G)}$ such that (i) $\bigcup_{t \in V(T)}{\beta(t)}=V(G)$, (ii) for every edge $uv \in E(G)$ there is a $t \in V(T)$ such that $\{u, v\}\subseteq \beta(t)$, and (iii) for every  vertex $v \in V(G)$ the subgraph of $T$ induced by the set $\{t\mid v\in \beta(t)\}$ is connected.
The {\em width} of a tree decomposition is $\max_{t\in V(T)}\{|\beta(t)|-1 \}$ and the {\em treewidth} of $G$, denoted by $\tw(G)$, is the minimum width over all tree decompositions of $G$.

For a tree decomposition $(T, \beta)$ we distinguish one vertex $r$ of $T$ which will be the root of $T$.
This introduces natural parent-child and ancestor-descendant relations in the tree $T$.
We say that such a rooted tree decomposition is \emph{nice} if the following conditions are satisfied:
$(i)$ $\beta(r) = \emptyset$ and $\beta(t) = \emptyset$ for every leaf $t$ of $T$.
$(ii)$  Every non-leaf node of $T$ is of one of the following three types:
$(a)$ \emph{Introduce node}: a node $t$ with exactly one child $t'$ such that $\beta(t) = \beta(t') \cup \{v\}$ for some $v \notin \beta(t')$; we say that $v$ is introduced at $t$.
$(b)$ \emph{Forget node}: a node $t$ with exactly one child $t$ such that $\beta(t) = \beta(t') \setminus \{w\}$ for some vertex $w \in \beta(t')$; we say that $w$ is forgotten at $t$.
$(c)$ \emph{Join node}: a node $t$ with two children $t_1$, $t_2$ such that $\beta(t) = \beta(t_1) = \beta(t_2)$.

A {\em path} $P = (v_1,\ldots, v_\ell)$ is a sequence of distinct vertices such that any pair of consecutive vertices are adjacent with each other.
The vertex set of $P$, denoted by $V(P)$, is the set $\{v_1,\ldots, v_\ell\}$.
The vertices $v_1$ and $v_\ell$ are called \emph{endpoints} of the path whereas the other vertices in $V(P)$ are called \emph{internal vertices}.
For two vertices $u,v \in V(G)$, we use $\dist_G(u,v)$ to denote the length of the shortest path with endpoints $u,v$.
A {\em cycle} $C = (v_1,\ldots, v_\ell)$ is a sequence of distinct vertices such that any pair of consecutive vertices and $v_1,v_\ell$ are adjacent with each other.
A graph is {\em connected} if there is a path between every pair of its vertices and it is {\em disconnected} otherwise. 
A subset $S$ of $V(G)$ is a \emph{connected set of vertices} if $G[S]$ is connected.
A {\em connected component} of a graph $G$ is a maximal connected set of vertices.
A {\em cut vertex} of a graph $G$ is a vertex $v$ such that the number of connected components of $G-\{v\}$ is strictly larger than that of $G$.
A connected graph that has no cut vertex is called {\em 2-connected}.
A \emph{$2$-connected component} of $G$ is a maximal subset $C \subseteq VG(G)$ such that $G[C]$ is $2$-connected.

\subsection{Planar Graphs and Minors}

A planar graph is a graph that can be embedded in the Euclidean plane, that is, there exists a mapping from every vertex to a point on a plane, and from every edge to a plane curve on that plane, such that the extreme points of each curve are the points mapped to the endpoints of the corresponding edge, and all curves are disjoint except on their extreme points.
A plane  graph $G$ is a planar graph with a fixed embedding.
Its faces are the regions bounded by the edges, including the outer infinitely large region. 

The {\em contraction} of an edge $uv$ in $G$ deletes vertices $u$ and $v$ from $G$, and adds a new vertex which is adjacent to vertices that were adjacent to $u$ or $v$.
This process does not introduce self-loops or parallel edges.
The resulting graph is denoted by $G/e$.
For $F \subseteq E(G)$, the graph $G/F$ denotes the graph obtained from $G$ by contracting each connected component in the subgraph $G[F]$ to a vertex.
Similarly, for $S \subseteq V(G)$, the graph $G/S$ denotes the graph obtained from $G$ by contracting each connected component in the subgraph $G[S]$ to a vertex.

A graph $H$ obtained by a sequence of such edge contractions starting from $G$ is said to be a contraction of $G$.
A graph $H$ is a minor of $G$ if $H$ is a subgraph of some contraction of $G$.
For a graph $H$, a graph $G$ is said to be $H$-minor-free if $G$ can not be contracted to $H$. 
We use the following result about $H$-minor-free graphs.

\begin{theorem}[Kostochka \cite{Kostochka84}, Thomason \cite{Thomason84}]
 \label{thm:average-degree-H-minor-free}
 Let $G$ be an $H$-minor-free graph.
 Then there is a constant $d_H$ such that
 \[\sum_{v \in V(G)} \deg_G(v) \leq d_H |V(G)|.\]
 Moreover, one can choose $d_H = \CO(h \sqrt{\log h})$ where $h$ denotes the number of vertices of $H$.
\end{theorem}

\subsection{Parameterized Algorithms and Kernelization}

An instance of a \emph{parameterized problem} is of the form
$(I;k)$ where $I$ is an instance of a (classical) decision problem and $k\in\mathbb{N}$
is the \emph{parameter}.
A parameterized problem is said to be \emph{fixed parameter tractable} (\FPT) if there exists an algorithm $\mathcal{A}$, a computable function $f$, and a constant $c$ such
that, given any instance $(I; k)$ of the parameterized problem, the algorithm $\mathcal{A}$ correctly decides whether $(I; k)$ is a \yes-instance in time $\Oh(f(k)|I|^c)$. 
An accompanying theory of hardness can be used to identify parameterized problems that are unlikely to admit \FPT\ algorithms. 
For the purpose of this article, we call the class of such problems as \WH.

Instances $(I;k)$ and $(I';k')$ of a parameterized problem are \emph{equivalent} if $(I;k)$ is a \yes-instance if and only if $(I';k')$ is a \yes-instance.
A \emph{compression} for a parameterized problem is an algorithm $\mathcal{B}$ that, given an instance $(I,k)$ of the problem, works in time $\Oh((|I| + k)^c)$ (for some constant $c$) and returns an equivalent instance $(I', k')$ of some problem.
A compression is said to be a \emph{kernel} if $(I', k')$ is an instance of the same problem.
If there exists a computable function $g$ such that size of an output obtained by algorithm $\mathcal{B}$ for $(I, k)$ is at most $g(k)$, we say that problem admits a compression of size $g(k)$.
If $g(k)$ is a polynomial function, then we say the problem admits a polynomial compression.
If $g(k) = k^{\textup{polylog}(k)}$, then we say the problem admits a quasipolynomial compression.
Consider the parameterized problems whose input contains a graph.
We say a compression is \emph{minor-preserving} if the graph in the reduced instance $(I', k')$ is $H$-minor-free whenever the graph in the input instance $(I, k)$ is $H$-minor-free.
We define similar notion for kernels.

We refer readers to \cite{CyganFKLMPPS15} for a detailed exposition on the subject.

\subsection{Constraint Satisfaction Problems}

A \emph{CSP-instance} is a tuple $\Gamma = (X,D,\CC)$ where $X$ is a finite set of variables, $D$ is a finite domain, and $\CC$ is a set of constraints $c = ((x_1,\dots,x_{a(c)}),R)$ where $x_1,\dots,x_{a(c)} \in X$ and $R \subseteq D^{a(c)}$.
We refer to $a(c)$ as the \emph{arity} of constraint $c$.
An assignment $\alpha\colon X \rightarrow D$ \emph{satisfies} a constraint $c = ((x_1,\dots,x_{a(c)}),R) \in \CC$ if $(\alpha(x_1),\dots,\alpha(x_{a(c)})) \in R$.
If $\alpha$ does not satisfy $c$, then we say that $\alpha$ \emph{violates} $c$.
An assignment $\alpha\colon X \rightarrow D$ \emph{satisfies} $\Gamma$ if it satisfies every constraint $c \in \CC$.
The instance $\Gamma$ is \emph{satisfiable} if there is a satisfying assignment.

We say that a CSP-instance $\Gamma = (X,D,\CC)$ is \emph{binary} if $a(c) \leq 2$ for all $c \in \CC$.
In the remainder of this work, we restrict ourselves to binary CSP-instances.
Let $c \in \CC$ be a constraint.
We call $c$ a \emph{unary} constraint if $a(c) = 1$ and a \emph{binary} constraint if $a(c) = 2$.
A binary constraint $c = ((x_1,x_2),R) \in \CC$ is called a \emph{permutation constraint} if for every $a \in D$ it holds that $|\{b \in D \mid (a,b) \in R\}| \leq 1$ and $|\{b \in D \mid (b,a) \in R\}| \leq 1$.
A binary CSP-instance is called is \emph{Permutation-CSP-instance} if every binary constraint is a permutation constraint.

Let $\Gamma = (X,D,\CC)$ be a binary CSP-instance.
The \emph{constraint graph of $\Gamma$} is defined as the graph $G(\Gamma)$ with vertex set $V(G(\Gamma)) \coloneqq X$ and edge set
\[E(G(\Gamma)) \coloneqq \{xy \mid ((x,y),R) \in \CC\}.\]
Let $Y \subseteq X$ be a subset of the variables.
We define $\Gamma[Y] \coloneqq (Y,D,\CC[Y])$ to be the \emph{induced subinstance of $\Gamma$} where $\CC[Y]$ contains all constraints $c = ((x_1,\dots,x_{a(c)}),R) \in \CC$ such that $x_i \in Y$ for all $i \in [a(c)]$.
Observe that $G(\Gamma[Y]) = (G(\Gamma))[Y]$.
Also, for $F \subseteq E(G(\Gamma))$, we define $\CC[F]$ to be the set of all binary constraints $c = ((x_1,x_2),R) \in \CC$ such that $x_1x_2 \in F$.
Moreover, we define $\CC - F \coloneqq \CC \setminus \CC[F]$ and $\Gamma - F \coloneqq (X,D,\CC - F)$.

%% file: perm-csp.tex
\section{Permutation CSPs with Size Constraints}
\label{sec:perm-csp-results}

\subsection{Problem Definitions}

In this work, we shall be interested in vertex- and edge-deletion problems on Permutation-CSP-instances.
In the following, we formally define the problems on CSP-instances and state the main algorithmic results on planar input instances (all proofs are given in Section \ref{sec:alg-perm-csp}).
Then, we demonstrate the wide applicability of these problems by reducing various well-known cycle hitting problems to the defined problems on Permutation-CSP-instances.
Finally, we list several kernelization results which, in combination with our algorithmic results for vertex- and edge deletion problems on Permutation-CSP-instances, lead to various subexponential parameterized algorithms for certain cycle hitting problems on planar graphs.

\paragraph{The Basic Edge-Deletion Problem.}
We start by considering the basic edge deletion problem for Permutation-CSPs.

\medskip 
\defparproblem{\textsc{Perm CSP Edge Deletion}}{A Permutation-CSP-instance $\Gamma = (X,D,\CC)$, a set of undeletable edges $U \subseteq E(G(\Gamma))$, and an integer $k$}{$k$}
 {Is there a set $Z \subseteq E(G(\Gamma)) \setminus U$ such that $|Z| \leq k$ and $\Gamma - Z \coloneqq (X,D,\CC - Z)$ is satisfiable?}
\medskip

In this work, we are only interested in CSP-instances $\Gamma$ where $G(\Gamma)$ satisfies a certain property.
We call a binary CSP-instance $\Gamma$ \emph{planar} if the graph $G(\Gamma)$ is planar.
Similarly, a binary CSP-instance $\Gamma$ is \emph{$H$-minor-free} if $G(\Gamma)$ is $H$-minor-free.
This allows to define the problems \textsc{Planar Perm CSP Edge Deletion} and \textsc{$H$-Minor-Free Perm CSP Edge Deletion} where we restrict ourselves to input CSP-instances $\Gamma$ that are planar or $H$-minor-free.
We remark here that we follow similar naming conventions for other vertex- and edge deletion problems without explicitly defining them.

The next theorem is essentially a simple consequence of the contraction decompositions for $H$-minor-free graphs due to Demaine et al.\ \cite{DemaineHK11}.
Still, in combination with recent kernelization results due to Wahlstr{\"{o}}m \cite{Wahlstrom20}, it already leads to improved FPT algorithms for certain edge deletion problems on $H$-minor-free graphs.

\begin{theorem}
 \label{thm:perm-csp-edge-deletion-result}
 There is an algorithm solving \textsc{$H$-Minor-Free Perm CSP Edge Deletion} in time $(|X| + |D|)^{\CO(c_H \sqrt{k})}$ where $c_H$ is a constant depending only on $H$.
\end{theorem}

\paragraph{Size Constraints for Connected Components.}

Next, we turn our attention to vertex-deletion problems.
To increase the applicability of our results, we also extend the CSP-instances with certain size constraints.
In this section, we focus on size constraints that bound the total weight of a connected component (after removing the solution).

Let $\Gamma = (X,D,\CC)$ be a binary CSP instance.
A \emph{1cc-size constraint} is a triple $(w,q,\op)$ where $w \colon X \times D \rightarrow \NN$ is a weight function, $q \in \NN$, and $\op \in \{\leq, \geq\}$.
An assignment $\alpha\colon X \rightarrow D$ satisfies $(w,q,\op)$ on $\Gamma$ if
\[\Big(\sum_{v \in C} w(v,\alpha(v)), q\Big) \in \op\]
for every connected component $C$ of $G(\Gamma)$.
(Here, $(p,q) \in \,\leq$ if $q \leq p$ and $(p,q) \in \,\geq$ if $q \geq p$.)

A \emph{Permutation-CSP-instance with $1$cc-size constraints} is a pair $(\Gamma,\CS)$ where $\Gamma = (X,D,\CC)$ is a Permutation-CSP-instance and $\CS$ is a set of 1cc-size constraints.
We say that $(\Gamma,\CS)$ is \emph{satisfiable} if there is an assignment $\alpha\colon X \rightarrow D$ which satisfies $\Gamma$ as well as every constraint $(w,q,\op) \in \CS$ on $\Gamma$.
Also, we define
\[\|\CS\| \coloneqq \prod_{(w,q,\op) \in \CS} (2 + q).\]

\medskip 
\defparproblem{\textsc{Perm CSP Vertex Deletion with Size Constraints}}{A Permutation-CSP-instance with $1$cc-size constraints $(\Gamma,\CS)$ with $\Gamma = (X,D,\CC)$, a set of undeletable variables $U \subseteq X$, and an integer $k$.}{$k$}
{Is there a set $Z \subseteq X \setminus U$ such that $|Z| \leq k$ and $(\Gamma[X\setminus Z],\CS)$ is satisfiable?}
\medskip

\begin{theorem}
 \label{thm:perm-csp-deletion-result}
 There is an algorithm solving \textsc{Planar Perm CSP Vertex Deletion with Size Constraints} in time $(|X| + |D| + \|\CS\|)^{\CO(\sqrt{k})}$.
\end{theorem}

\paragraph{Satisfiable $2$-Connected Components.}

Next, we state the corresponding problems and results for $2$-connected components.

Let $\Gamma = (X,D,\CC)$ be a binary CSP instance.
A \emph{2cc-size constraint} is a pair $(w,q)$ where $w \colon X \times D \rightarrow \NN$ is a weight function, and $q \in \NN$.
Let $W \subseteq X$.
An assignment $\alpha\colon W \rightarrow D$ satisfies $(w,q)$ on $W$ if
\[\sum_{v \in W} w(v,\alpha(v)) \leq q.\]
Observe that, in comparison to 1cc-size constraints, we only allow to check for upper bounds on the weighted size of a set $W \subseteq X$.

A \emph{Permutation-CSP-instance with $2$cc-size constraints} is a pair $(\Gamma,\CS)$ where $\Gamma = (X,D,\CC)$ is a Permutation-CSP-instance and $\CS$ is a set of 2cc-size constraints.
For $W \subseteq X$, we say that $(\Gamma,\CS)$ is \emph{satisfiable on $W$} if there is an assignment $\alpha\colon W \rightarrow D$ that satisfies $\Gamma[W]$ as well as all 2cc-size constraints $(w,q) \in \CS$ on $W$.
As before, we define
\[\|\CS\| \coloneqq \prod_{(w,q,\op) \in \CS} (2 + q).\]

\medskip 
\defparproblem{\textsc{$2$Conn Perm CSP Vertex Deletion with Size Constraints}}{A Permutation-CSP-instance with $2$cc-size constraints $(\Gamma,\CS)$ with $\Gamma = (X,D,\CC)$, a set of undeletable variables $U \subseteq X$, and an integer $k$.}{$k$}
{Is there a set $Z \subseteq X \setminus U$ such that $|Z| \leq k$ and $(\Gamma,\CS)$ is satisfiable on $W$ for every $2$-connected component $W$ of the graph $G(\Gamma[X \setminus Z])$.}
\medskip

\begin{theorem}
 \label{thm:2cc-perm-csp-deletion-result}
 There is an algorithm solving \textsc{Planar $2$Conn Perm CSP Vertex Deletion with Size Constraints} in time $(|X| + |D| + \|\CS\|)^{\CO(\sqrt{k})}$.
\end{theorem}

\subsection{Reductions}

Next, we argue that several standard vertex and edge deletion problems can be interpreted as special cases of the problems discussed above.
In all the cases, the constraint graph of the CSP-instance we construct is the same as the input graph for the problem in question.
In particular, restrictions of the input graph immediately translate to corresponding restrictions of the constructed CSP instance.
Here, we only focus on vertex deletion problems.
However, the edge deletion versions of all problems naturally translate to the edge deletion versions of the CSP problems.

We start by covering some problems that can be reduced to \textsc{Perm CSP Vertex Deletion with Size Constraints}.
If the size constraints are not needed to build the reduction, we simply omit them.
Also, in the reductions described here, unless explicitly stated otherwise, we define $U \coloneqq \emptyset$ as the set of undeletable variables.

We start with the problem of eliminating all odd cycles in a graph.

\defparproblem{\textsc{Odd Cycle Transversal (OCT)}}{A graph $G$, and an integer $k$}{$k$}
{Does there exist a set $Z \subseteq V(G)$ of size at most $k$ that hits all odd cycles, i.e., $G - Z$ is bipartite?}

This problem can be translated into a Permutation-CSP-instance $\Gamma = (X,D,\CC)$ with
\begin{itemize}
 \item $X \coloneqq V(G)$,
 \item $D \coloneqq \{0,1\}$, and
 \item binary constraints $((v,w),R_{\neq})$ for all $vw \in E(G)$ where $R_{\neq} \coloneqq \{(0,1),(1,0)\}$.
\end{itemize}
It is easy to verify that $Z \subseteq V(G)$ is a solution for $(G,k)$ if and only if $Z$ is a solution for $(\Gamma,k)$.

\medskip

More generally, the same approach can be used to reduce \textsc{Group Feedback Vertex Set} to \textsc{Perm CSP Vertex Deletion with Size Constraints}.
Let $\Sigma$ be a group.
A \emph{$\Sigma$-labeled graph} is a pair $(G,\lambda)$ where $G$ is a graph and $\lambda\colon \{(v,w),(w,v) \mid vw \in E(G)\} \rightarrow \Sigma$ is a mapping satisfying $\lambda(v,w) = (\lambda(w,v))^{-1}$ for all $vw \in E(G)$.
A cycle $C = (v_1,\dots,v_\ell)$ is a \emph{non-null} cycle if $\lambda(C) \coloneqq \lambda(v_1,v_2)\lambda(v_2,v_3)\dots\lambda(v_{\ell-1},v_\ell)\lambda(v_\ell,v_1) \neq 1_\Sigma$ where $1_\Sigma$ denotes the identity element of $\Sigma$.
We remark that this notion if well-defined by the next lemma.

\begin{lemma}[{\cite[Lemma 7]{Guillemot11a}}]
 \label{la:non-null-rotate}
 Let $(v_1,\dots,v_\ell)$ be a cycle.
 Then $\lambda(v_1,\dots,v_\ell) \neq 1_\Sigma$ if and only if $\lambda(v_2,\dots,v_\ell,v_1) \neq 1_\Sigma$.
\end{lemma}

\defparproblem{\textsc{Group Feedback Vertex Set (Group FVS)}}{A $\Sigma$-labeled graph $G$, and an integer $k$}{$k$}
{Does there exist a set $Z \subseteq V(G)$ of size at most $k$ that hits all non-null cycles?}

In order to translate an instance of \textsc{Group FVS} to a Permutation-CSP-instance, we use the following lemma that reformulates \textsc{Group FVS} as a labeling problem.
Let $G$ be a $\Sigma$-labeled graph.
A labeling $\Lambda\colon V(G) \rightarrow \Sigma$ is \emph{consistent} if
\[\Lambda(v) \lambda(v,w) = \Lambda(w)\]
for all $vw \in E(G)$

\begin{lemma}[{\cite[Lemma 8]{Guillemot11a}}]
 \label{la:consistent-group-labeling}
 Let $G$ be a $\Sigma$-labeled graph.
 Then $G$ has no non-null cycle if and only if there is a consistent labeling $\Lambda\colon V(G) \rightarrow \Sigma$.
\end{lemma}

Now, consider the CSP-instance $\Gamma = (X,D,C)$ with
\begin{itemize}
 \item $X \coloneqq V(G)$,
 \item $D \coloneqq \Sigma$,
 \item binary constraints $((v,w),R_{\lambda(v,w)})$ for all $vw \in E(G)$ where $R_\sigma \coloneqq \{(\delta,\delta\sigma) \mid \delta \in \Sigma\}$ for $\sigma \in \Sigma$.
\end{itemize}

The last lemma implies that $Z$ is a \textsc{Group FVS} solution for $(G,k)$ if and only if $Z$ is a solution for $(\Gamma,k)$.

\medskip

Next, we consider two problems where the reductions use the size constraints.
 
\defparproblem{\textsc{Vertex Multiway Cut}}{A graph $G$, a set $T \subseteq V(G)$, and an integer $k$}{$k$}
{Does there exist a set $Z \subseteq V(G) \setminus T$ of size at most $k$ such that all vertices from $T$ lie in different connected components of $G - Z$?}

Consider the instance $(\Gamma,\CS,U)$ with $\Gamma = (X,D,C)$ and
\begin{itemize}
 \item $X \coloneqq V(G)$,
 \item $D \coloneqq \{\circledast\}$,
 \item binary constraints $((v,w),\{(\circledast,\circledast)\})$ for all $vw \in E(G)$,
 \item $U \coloneqq T$ is the set of undeletable variables, and
 \item $\CS = \{(w,1,\leq)\}$ where $w\colon X \times D \rightarrow \NN$ is defined via $w(v,\circledast) = 1$ for all $v \in T$, and $w(v,\circledast) = 0$ for all $v \in V(G) \setminus T$.
\end{itemize}
Again, it is easy to verify that $Z$ is a \textsc{Vertex Multiway Cut} solution for $(G,T,k)$ if and only if $Z$ is a solution for $(\Gamma,\CS,U,k)$.

We remark that \textsc{Vertex Multiway Cut} can also be reduced to \textsc{Perm CSP Vertex Deletion with Size Constraints} without using the size constraints as follows.
Consider the instance $\Gamma = (X,D,C)$ with
\begin{itemize}
 \item $X \coloneqq V(G)$,
 \item $D \coloneqq T$,
 \item unary constraints $(v,\{v\})$ for every $v \in T$ 
 \item binary constraints $((v,w),R_=)$ for all $vw \in E(G)$ where $R_= \coloneqq \{(v,v) \mid v \in T\}$, and
 \item $U \coloneqq T$ is the set of undeletable variables
\end{itemize}

\medskip

For the second example, we consider \textsc{Component Order Connectivity}.

\defparproblem{\textsc{Component Order Connectivity}}{A graph $G$ and integers $k,t$}{$k$}
{Does there exist a set $Z \subseteq V(G)$ of size at most $k$ such that $|C| \leq t$ for every connected component $C$ of $G - Z$?}

Consider the instance $(\Gamma,\CS)$ with $\Gamma = (X,D,\CC)$ and
\begin{itemize}
 \item $X \coloneqq V(G)$,
 \item $D \coloneqq \{\circledast\}$,
 \item binary constraints $((v,w),\{(\circledast,\circledast)\})$ for all $vw \in E(G)$, and
 \item $\CS = \{(w,t,\leq)\}$ where $w\colon X \times D \rightarrow \NN$ is defined via $w(v,\circledast) = 1$ for all $v \in V(G)$.
\end{itemize}
As usual, it is straight-forward to verify that $Z$ is a \textsc{Component Order Connectivity} solution for $(G,k,t)$ if and only if $Z$ is a solution for $(\Gamma,\CS,k)$.

\medskip

Next, we turn to problems that can be reduced to \textsc{$2$Conn Perm CSP Vertex Deletion with Size Constraints}.

The first example is the \textsc{Subset Feedback Vertex Set} problem which asks to eliminate all cycles visiting a terminal vertex $t \in T$ in a graph $G$.

\defparproblem{\textsc{Subset Feedback Vertex Set (Subset FVS)}}{A graph $G$, a set $T \subseteq V(G)$, and an integer $k$}{$k$}
{Does there exist a set $Z \subseteq V(G)$ of size at most $k$ that hits all $T$-cycles?}

A set $Z \subseteq V(G)$ is a solution for \textsc{Subset FVS} if and only if $|W| \leq 2$ or $W \cap T = \emptyset$ for every $2$-connected component $W$ of $G - Z$.
Consider the CSP-instance $\Gamma = (X,D,\CC)$ with
\begin{itemize}
 \item $X \coloneqq V(G)$,
 \item $D \coloneqq \{\circledast\} \uplus V(G) \uplus E(G)$,
 \item unary constraints $(v,\{e \in E(G) \mid v \in e\} \cup \{v\})$ for all $v \in T$,
 \item unary constraints $(v,\{\circledast\} \cup \{e \in E(G) \mid v \in e\})$ for all $v \notin T$, and
 \item binary constraints $((v,w),R_=)$ for all $vw \in E(G)$ where $R_= \coloneqq \{(a,a) \mid a \in D\}$.
\end{itemize}
We show that $Z$ is a \textsc{Subset FVS} solution for $(G,T,k)$ if and only if $Z$ is a solution for $(\Gamma,k)$.

Suppose $Z$ is a \textsc{Subset FVS} solution for $G$ and let $W$ be a $2$-connected component of $G - Z$.
We define a satisfying assignment $\alpha\colon W \rightarrow D$ for $\Gamma[W]$.
If $W \cap T = \emptyset$ then define $\alpha(w) \coloneqq \circledast$ for all $w \in W$.
Otherwise $|W| \leq 2$.
If $|W| = 1$ then set $\alpha(w) \coloneqq w$ for the unique $w \in W$ (note that $w \in T$).
Otherwise $|W| = 2$ and $G[W]$ contains a single edge $e$.
We set $\alpha(w) \coloneqq e$ for all $w \in W$.

In the other direction, assume that $Z$ is a solution for $\Gamma$ and let $W$ be a $2$-connected component of $G - Z$.
Let $\alpha\colon W \rightarrow D$ be a satisfying assignment for $\Gamma[W]$.
By the binary constraints of $\Gamma$ we get that $\alpha(w) = \alpha(w')$ for all $w,w' \in W$.
If $\alpha(w) = \circledast$ for all $w \in W$ then $W \cap T = \emptyset$ using the unary constraints.

If there is some $e \in E(G)$ such that $\alpha(w) = e$ for all $w \in W$, then $W \subseteq e$ using the unary constraints.
Hence, $|W| \leq 2$.

Otherwise there is some $v \in V(G)$ such that $\alpha(w) = v$ for all $w \in W$.
Then $W = \{v\}$.

\medskip

Using the framework of Permutation-CSPs, we can also cover more specialized problems such as the following variant of \textsc{Subset FVS}.

\defparproblem{\textsc{Two Subset Feedback Vertex set (Two Subset FVS)}}{A graph $G$, a set $T \subseteq V(G)$, and an integer $k$}{$k$}{Does there exist a set $Z \subseteq V(G)$ of size at most $k$ that hits all cycles $C$ such that $|C \cap T| \geq 2$?}
\medskip

A set $Z \subseteq V(G)$ is a solution for \textsc{Two Subset FVS} if and only if $|W| \leq 2$ or $|W \cap T| \leq 1$ for every two-connected component $W$ of $G - Z$.
Without loss of generality assume that $T \neq \emptyset$.
Consider the CSP-instance $\Gamma = (X,D,\CC)$ with
\begin{itemize}
 \item $X \coloneqq V(G)$,
 \item $D \coloneqq T \uplus E(G)$,
 \item unary constraints $(v,\{v\} \cup \{e \in E(G) \mid v \in e\})$ for all $v \in T$,
 \item unary constraints $(v,T \cup \{e \in E(G) \mid v \in e\})$ for all $v \notin T$, and
 \item binary constraints $((v,w),R_=)$ for all $vw \in E(G)$ where $R_= \coloneqq \{(a,a) \mid a \in D\}$.
\end{itemize}
We show that $Z$ is a \textsc{Two Subset FVS} solution for $(G,T,k)$ if and only if $Z$ is a solution for $(\Gamma,k)$.

Suppose $Z$ is a \textsc{Subset FVS} solution for $(G,T,k)$ and let $W$ be a $2$-connected component of $G - Z$.
We define a satisfying assignment $\alpha\colon W \rightarrow D$ for $\Gamma[W]$.
If $|W \cap T| = 0$ then define $\alpha(w) \coloneqq t$ for all $w \in W$ where $t \in T$ is arbitrary.
If $|W \cap T| = 1$ then define $\alpha(w) \coloneqq t$ for all $w \in W$ where $t \in W \cap T$.
Otherwise $|W| \leq 2$ and $G[W]$ contains a single edge $e$.
We set $\alpha(w) \coloneqq e$ for all $w \in W$.

In the other direction, assume that $Z$ is a solution for $(\Gamma,k)$ and let $W$ be a $2$-connected component of $G - Z$.
Let $\alpha\colon W \rightarrow D$ be a satisfying assignment for $\Gamma[W]$.
By the binary constraints of $\Gamma$ we get that $\alpha(w) = \alpha(w')$ for all $w,w' \in W$.
If $\alpha(w) \in T$ for some $w \in W$ then $|W \cap T| \leq 1$ using the unary constraints.
Otherwise there is some $e \in E(G)$ such that $\alpha(w) = e$ for all $w \in W$.
Moreover, $W \subseteq e$ using the unary constraints.
Hence, $|W| \leq 2$.

\medskip

As the next example, we consider the \textsc{Subset Group Feedback Vertex Set} problem which in particular contains \textsc{Subset OCT} as a special case.

\defparproblem{\textsc{Subset Group Feedback Vertex Set (Subset Group FVS)}}{A $\Sigma$-labeled graph $G$, a set $T \subseteq V(G)$, and an integer $k$}{$k$}{Does there exist a set $Z \subseteq V(G)$ of size at most $k$ that hits all non-null $T$-cycles?}

As before, we reformulate the problem as a labeling problem.

\begin{lemma}
 Let $(G,\lambda)$ be a $\Sigma$-labeled graph and let $T \subseteq V(G)$.
 Then $G$ has no non-null $T$-cycle if and only if $G[W]$ has a consistent labeling or $W \cap T = \emptyset$ for all $2$-connected components $W$ of $G$.
\end{lemma}

\begin{proof}
 First suppose $G$ has a non-null $T$-cycle $C$.
 Then there is some $2$-connected component $W$ of $G$ such that $C \subseteq W$.
 In particular, $W \cap T \neq \emptyset$ and $G[W]$ does not have a consistent labeling by Lemma \ref{la:consistent-group-labeling}.
 
 In the other direction, let $W$ be a $2$-connected component of $G$ such that $W \cap T \neq \emptyset$ and there is no consistent labeling for $G[W]$.
 Let $t \in W \cap T$.
 By Lemma \ref{la:consistent-group-labeling} there is a non-null cycle $C = (v_1,\dots,v_\ell)$ in $G[W]$.
 If $t \in C$ then $C$ forms a non-null $T$-cycle and we are done.
 Assume otherwise.
 Since $G[W]$ is $2$-connected there are distinct $i,j \in [\ell]$ and paths $P = (u_1,\dots,u_p)$ from $t = u_1$ to $v_i = u_p$ and $Q = (w_1,\dots,w_q)$ from $t = w_1$ to $v_j = w_q$ such that
 $\{v_1,\dots,v_\ell\}$, $\{u_2,\dots,u_{p-1}\}$ and $\{w_2,\dots,w_{q-1}\}$ are pairwise disjoint.
 Without loss of generality suppose that $i < j$.
 Define the two cycles
 \[C_1 \coloneqq (t = u_1,u_2,\dots,u_p = v_i,v_{i+1},\dots,v_j=w_q,w_{q-1},\dots,w_1 = t)\]
 and
 \[C_2 \coloneqq (t = w_1,w_2,\dots,w_q = v_j,v_{j+1},\dots,v_k,v_1,v_{2},\dots,v_i=u_p,u_{p-1},\dots,u_1 = t)\]
 Then
 \[\lambda(C_1)\lambda(C_2) = \left(\prod_{s \in [p-1]} \lambda(u_s,u_{s+1})\right) \lambda(C') \left(\prod_{s \in [p-1]} \lambda(u_s,u_{s+1})\right)^{-1}\]
 where $C' = (v_i,v_{i+1},\dots,v_k,v_1,v_2,\dots,v_i)$.
 This means that $\lambda(C_1)\lambda(C_2) \neq 1_\Sigma$ since $\lambda(C') \neq 1_\Sigma$ by Lemma \ref{la:non-null-rotate}.
 Hence, $C_1$ or $C_2$ forms a non-null $T$-cycle in $G$.
\end{proof}

Now consider the CSP-instance $\Gamma = (X,D,\CC)$ with
\begin{itemize}
 \item $X \coloneqq V(G)$,
 \item $D \coloneqq \{\circledast\} \uplus \Sigma$,
 \item unary constraints $(v,\Sigma)$ for $v \in T$, and
 \item binary constraints $((v,w),R_{\lambda(v,w)}')$ for all $vw \in E(G)$ where $R_\sigma' \coloneqq \{(\circledast,\circledast)\} \cup \{(\delta,\delta\sigma) \mid \delta \in \Sigma\}$ for $\gamma \in \Sigma$.
\end{itemize}

The last lemma implies that $Z$ is a \textsc{Subset Group FVS} solution for $(G,T,k)$ if and only if $Z$ is a solution for $(\Gamma,k)$.

\medskip

As the last example, we consider \textsc{$2$Conn Component Order Connectivity} which again relies on the size constraints.

\defparproblem{\textsc{$2$Conn Component Order Connectivity}}{A graph $G$ and integers $k,t$}{$k$}
{Does there exist a set $Z \subseteq V(G)$ of size at most $k$ such that $|W| \leq t$ for every $2$-connected component $W$ of $G - Z$?}

This problem can be formulated as an instance $(\Gamma,\CS)$ with $\Gamma = (X,D,\CC)$ and
\begin{itemize}
 \item $X \coloneqq V(G)$,
 \item $D \coloneqq \{\circledast\}$,
 \item binary constraints $((v,w),\{(\circledast,\circledast)\})$ for all $vw \in E(G)$, and
 \item $\CS = \{(w,t)\}$ where $w\colon X \times D \rightarrow \NN$ is defined via $w(v,\circledast) = 1$ for all $v \in V(G)$.
\end{itemize}

\medskip

We complete this part by explicitly listing some algorithmic consequences of the reductions described above and Theorems \ref{thm:perm-csp-deletion-result} and \ref{thm:2cc-perm-csp-deletion-result}.

\begin{corollary}
 There are algorithms for the following problems:
 \begin{enumerate}
  \item \textsc{Planar Component Order Connectivity} running in time $(n + t)^{\CO(\sqrt{k})}$,
  \item \textsc{Planar Two Subset FVS} running in time $n^{\CO(\sqrt{k})}$,
  \item \textsc{Planar Subset Group FVS} running in time $(n + |\Sigma|)^{\CO(\sqrt{k})}$, and
  \item \textsc{Planar $2$Conn Component Order Connectivity} running in time $(n + t)^{\CO(\sqrt{k})}$.
 \end{enumerate}
\end{corollary}

For the other problems, we combine Theorems \ref{thm:perm-csp-deletion-result} and \ref{thm:2cc-perm-csp-deletion-result} with kernelization results to obtain subexponential FPT algorithms.

\subsection{Kernels and Subexponential FPT Algorithms}

We can combine the subexponential parameterized algorithms for the problems discussed above with existing kernelization results to also obtain subexponential FPT algorithms.
For this to work, we need two additional properties of the kernelization.
First, we require that the parameter is only increased linearly.
Second, if the input graph is planar (resp.\ $H$-minor-free), then the graph of the reduced instance also has to be planar (resp.\ $H$-minor-free).
In many cases, existing kernelizations do not satisfy the second property.
However, it is usually possible to modify the kernels appropriately.
In the following, we list the kernelization results relevant for this work, and briefly comment on where modifications to existing kernels are required.
All details for these results are given in Section \ref{sec:kernels}.

We start by considering edge deletion versions of the problems above (to refer to the edge deletion version of a problem, we simply replace the word \textsc{Vertex} by \textsc{Edge} in the problem name, or add the word \textsc{Edge} if the word \textsc{Vertex} does not appear). 
In a recent work, Wahlstr{\"{o}}m \cite{Wahlstrom20} obtained quasipolynomial kernels for the following problems.
Also, it can be checked that all these kernels preserve minors of the input graph and do not increase the parameter in question.

\begin{theorem}[Wahlstr{\"{o}}m \cite{Wahlstrom20}]
 \label{thm:edge-kernels-result}
 The following problems have randomized quasipolynomial kernels with failure probability $\CO(2^{-n})$:
 \begin{enumerate}
  \item \textsc{Edge Multiway Cut} parameterized by solution size,
  \item \textsc{Group Feedback Edge Set} parameterized by solution size, for any group, such that the group remains the same in the reduced instance,
  \item \textsc{Subset Feedback Edge Set} with undeletable edges, parameterized by solution size.
 \end{enumerate}
 Moreover, if the graph in the input instance is $H$-minor-free then the graph in the reduced instance is also $H$-minor-free, and the parameter does not increase in the reduced instance.
\end{theorem}

Since all the reductions presented above also go through for the edge deletion versions of the problems, we obtain the following corollary by combining Theorem \ref{thm:perm-csp-edge-deletion-result} and the last theorem.

\begin{corollary}
 There are randomized algorithms solving the following problems with failure probability $\CO(2^{-n})$: 
 \begin{enumerate}
  \item \textsc{$H$-Minor-Free Edge Multiway Cut} in time $2^{\CO(c_H\cdot\sqrt{k}\cdot\textup{polylog}(k))}n^{\CO(1)}$, and
  \item \textsc{$H$-Minor-Free Group Feedback Edge Set} in time $|\Sigma|^{\CO(c_H\cdot\sqrt{k}\cdot\textup{polylog}(k))}n^{\CO(1)}$.
 \end{enumerate}
 Here, $c_H$ denotes a constant that only depends on $H$.
\end{corollary}

We remark that a subexponential fpt algorithm for \textsc{Planar Edge Multiway Cut} was already given in \cite{PilipczukPSL18}.

Observe that \textsc{Subset Feedback Edge Set} is not covered by the corollary since restrictions need to be enforced on $2$-connected components (after removing the solution).
However, we can still obtain a subexponential FPT algorithm on planar graphs by reducing the edge deletion version to the vertex deletion version with undeletable vertices.
Observe that we can easily introduce undeletable vertices into the above reduction by using the set of undeletable variables in the definition of \textsc{Perm CSP Vertex Deletion with Size Constraints} and \textsc{$2$Conn Perm CSP Vertex Deletion with Size Constraints}.
So we obtain the following corollary from combining Theorems \ref{thm:2cc-perm-csp-deletion-result} and \ref{thm:edge-kernels-result}.

\begin{corollary}
 There is a randomized algorithm with failure probability $\CO(2^{-n})$ solving \textsc{Planar Subset Feedback Edge Set with Undeletable Edges} in time $2^{\CO(\sqrt{k}\cdot\textup{polylog}(k))}n^{\CO(1)}$.
\end{corollary}

Next, let us come back to vertex deletion problems.
Here, the situation is generally more complicated since existing kernels for the problems we are interested in usually do not preserve $H$-minor-freeness.
In Section \ref{sec:kernels}, we present modifications to several kernels from \cite{HolsK18,KratschW20}.
These modifications require us to introduce undeletable vertices into the problem definition.
Strictly speaking, this means that our results are not kernelizations, but only compressions.
However, for the purpose of designing subexponential FPT algorithms, this does not pose any additional problems since all reductions presented above can trivially be extended to take undeletable vertices into account.
For any problem $\Pi$ discussed above, we use \textsc{$\Pi$ with Undeletable vertices} to refer to the variant of the problem where undeletable vertices are added.
Also, \textsc{Vertex Multiway Cut with Deletable Terminal} refers to the variant of \textsc{Vertex Multiway Cut} where terminal vertices may be deleted in the solution.

\begin{theorem}
\label{thm:H-minor-graph-Pi-compression-vertex-SFVS}
Let $\Pi$ be \textsc{Vertex Multiway Cut with Deletable Terminal} or \textsc{Subset Feedback Vertex Set}.
There is an algorithm that given an instance $(G,T,k)$ of \textsc{$H$-Minor Free $\Pi$} constructs an equivalent instance $(G',T,F,k')$ of \textsc{$H$-Minor Free $\Pi$ with Undeletable vertices} in randomized polynomial time and with failure probability $\calO(2^{-|V(G)|})$ such that $|V(G')| \in (c_H \cdot k)^{\calO(1)}$ and $k' \leq k$.
\end{theorem}

\begin{theorem}
\label{thm:H-minor-graph-Pi-compression-vertex-GFVS}
 There is an algorithm that given an instance $(G,\lambda, \Sigma, T,k)$ of \textsc{$H$-Minor-Free Group Feedback Vertex Set} constructs an equivalent instance $(G', \lambda', \Sigma, T,F,k')$ of \textsc{$H$-Minor-Free Group Feedback Vertex Set with Undeletable vertices} in randomized polynomial time and with failure probability $\calO(2^{-|V(G)|})$ such that $|V(G')| \in k^{\calO(|\Sigma| \cdot |V(H)|)}$ and $k' \leq k$.
\end{theorem}

We get the following result by combining Theorems~\ref{thm:2cc-perm-csp-deletion-result} and \ref{thm:H-minor-graph-Pi-compression-vertex-SFVS}.

\begin{corollary}
 There is a randomized algorithm solving \textsc{Planar Subset Feedback Vertex Set} with failure probability $\CO(2^{-n})$ in time $2^{\CO(\sqrt{k}\cdot\log k)}n^{\CO(1)}$ where $k$ denotes the solution size.
\end{corollary}

Similarly, we can get the following result by combining Theorems~\ref{thm:perm-csp-deletion-result} and \ref{thm:H-minor-graph-Pi-compression-vertex-GFVS}. 

\begin{corollary}
 There is a randomized algorithm solving \textsc{Planar Group Feedback Vertex Set} with failure probability $\CO(2^{-n})$ in time $2^{\CO(\sqrt{k}\cdot\log k \cdot |\Sigma|)}n^{\CO(1)}$ where $k$ denotes the solution size and $\Sigma$ the input group.
\end{corollary}

For \textsc{Vertex Multiway Cut} on planar graphs, we can actually get a stronger result by exploiting existing kernels on planar graphs. 

\begin{theorem}[Jansen et al.\ \cite{JansenPL19}]
 \textsc{Planar Vertex Multiway Cut} parameterized by solution size admits a deterministic polynomial kernel where the parameter in the reduced instance is not increased.
\end{theorem}

\begin{corollary}
 There is a (deterministic) algorithm solving \textsc{Planar Vertex Multiway Cut} in time $2^{\CO(\sqrt{k} \cdot \log k)}n^{\CO(1)}$ where $k$ denotes the solution size.
\end{corollary}

%% file: structure.tex
\section{Structural Analysis}
\label{sec:structure}

The next two sections are devoted to the proofs of Theorems \ref{thm:perm-csp-edge-deletion-result}, \ref{thm:perm-csp-deletion-result} and \ref{thm:2cc-perm-csp-deletion-result}.
In this section, we start by providing all necessary structural tools.
In particular, we prove Theorem \ref{thm:planarcontr-vert2} and Lemma \ref{lem:introbodyguess}.

\subsection{Contraction Decompositions}

\subsubsection{Edge Partitions}

We first concern ourselves with the contraction decompositions for planar and $H$-minor-free graphs.
For the edge deletion problem \textsc{Perm CSP Edge Deletion} our algorithm relies on a contraction decomposition due to Demaine et al.\ \cite{DemaineHK11} (see Theorem \ref{thm:Hcontr}).

\begin{theorem}[Demaine, Hajiaghayi, Kawarabayashi \cite{DemaineHK11}]
 \label{thm:contraction-decomposition-minor-edge-simple}
 Let $G$ be an $H$-minor free graph and $k \geq 1$.
 Then there is a partition of the edge set $E(G) = E_1 \uplus E_2 \uplus \dots \uplus E_k$ such that, for every $i \in [k]$, it holds that
 \[\tw\big(G/E_i\big) \leq c_H k\]
 for some constant $c_H$ depending only on $H$.
 Moreover, given the graph $G$ and $k \geq 1$, the partition $E(G) = E_1 \uplus E_2 \uplus \dots \uplus E_k$ can be computed in polynomial time.
\end{theorem}

To be more precise, the algorithm from Theorem \ref{thm:perm-csp-edge-deletion-result} relies on the following corollary which allows us to declare certain edges to be uncontractible.

\begin{corollary}
 \label{cor:contraction-decomposition-minor-edge}
 Let $G$ be an $H$-minor free graph and $k \geq 1$.
 Then there is a partition of the edge set $E(G) = E_1 \uplus E_2 \uplus \dots \uplus E_k$ such that, for every $i \in [k]$ and every $W \subseteq E_i$, it holds that
 \[\tw\big(G/(E_i \setminus W)\big) \leq c_H k + |W|\]
 for some constant $c_H$ depending only on $H$.
 Moreover, given the graph $G$ and $k \geq 1$, the partition $E(G) = E_1 \uplus E_2 \uplus \dots \uplus E_k$ can be computed in polynomial time.
\end{corollary}

\begin{proof}
 The statement follows from Theorem \ref{thm:contraction-decomposition-minor-edge-simple} by observing that the contraction of a single edge decreases the treewidth of a graph by at most one.
\end{proof}

\subsubsection{Vertex Partitions}

Next, we turn to the vertex version of the contraction decomposition.
In this work, we prove such a statement only for planar graphs.
More precisely, we obtain the following theorem which implies Theorem \ref{thm:planarcontr-vert2} (here, we give an explicit upper bound on the treewidth of the contracted graph and avoid using $\CO$-notation).

\begin{theorem}
 \label{thm:contraction-decomposition-planar-vertex}
 Let $G$ be a planar graph and $k \geq 1$.
 Then there is a partition of the vertex set $V(G) = V_1 \uplus V_2 \uplus \dots \uplus V_k$ such that, for every $i \in [k]$ and every $W \subseteq V_i$, it holds that
 \[\tw\big(G/E(G[V_i \setminus W])\big) \leq 3k + 14|W| + 2.\]
 Moreover, given the graph $G$ and $k \geq 1$, the partition $V(G) = V_1 \uplus V_2 \uplus \dots \uplus V_k$ can be computed in polynomial time.
\end{theorem}

The remainder of this subsection is devoted to proving the theorem.
First consider the case $k=1$.
Then $W$ forms a vertex cover of the graph $G/E(G[V_1 \setminus W])$ and thus, $\tw(G/E(G[V_1 \setminus W])) \leq |W|$.
So suppose $k \geq 2$.

Let $G$ be a planar graph and let us fix an embedding of $G$ in the plane.
Let $F(G)$ denote the faces of the embedding.
We define the \emph{radial graph} $R \coloneqq R(G)$ with vertex set $V(R) \coloneqq V(G) \cup F(G)$ and edge set
\[E(R) \coloneqq \{vf \mid v \in V(G), f \in F(G), v \text{ is incident to } f\}.\]
Note that $R$ is a connected, bipartite graph with bipartition $(V(G),F(G))$.
Also pick $f_0$ to be the exterior face of the embedding.
For $j \geq 1$ define
\[L_j \coloneqq \{v \in V(R) \mid \dist_R(f_0,v) = 2j - 1\}\]
to be the $j$-th \emph{vertex layer}.
Note that $L_j \subseteq V(G)$ for all $j \geq 1$.
For a vertex $v \in V(G)$ define $L(v) \coloneqq j$ for the unique $j \geq 1$ such that $v \in L_j$.

\begin{observation}
 \label{obs:edge-endpoint-layers}
 Let $vw \in E(G)$.
 Then $|L(v) - L(w)| \leq 1$.
\end{observation}

\begin{proof}
 There is a face $f \in F(G)$ such that $v$ and $w$ are incident to $f$.
 Hence, $\dist_R(v,w) \leq 2$.
 This implies the observation.
\end{proof}

For later reference we give a second description of the layers (with respect to the fixed embedding of $G$).
For $j \geq 1$ define $L_j'$ to be the set of vertices incident to the exterior face of the embedding after deleting all vertices from $L_1' \cup \dots \cup L_{j-1}'$.
The following observation follows from a simple induction on $j \geq 1$.

\begin{observation}
 \label{obs:edge-layers-by-exterior-face}
 $L_j = L_j'$ for all $j \geq 1$.
\end{observation} 

\begin{definition}
 An embedding of a graph $G$ is \emph{$1$-outerplanar} if it is planar and all vertices lie on the exterior face.
 For $k \geq 2$ an embedding of a graph $G$ is \emph{$k$-outerplanar} if it is planar and, after deleting all vertices on the exterior face, the resulting embedding is $(k-1)$-outerplanar.
 A graph is \emph{$k$-outerplanar} if it has a $k$-outerplanar embedding.
\end{definition}

We first record a property of $1$-outplanar graphs that plays an important role later on.

\begin{lemma}
 \label{la:outerplanar-bound-covered-face-cycles}
 Let $G$ be a $1$-outerplanar graph and let $W \subseteq V(G)$.
 Also fix some $1$-outerplanar embedding of $G$ and let $f_0$ denote the exterior face.
 Then
 \[\sum_{f \in F(G) \setminus \{f_0\}} \max\{|N_{R(G)}(f) \cap W| - 1,0\} \leq 13|W|.\]
\end{lemma}

\begin{proof}
 Consider the radial graph $R \coloneqq R(G)$ which is planar and bipartite.
 Let $H \coloneqq R[W \cup F(G)]$ be the induced subgraph on $W$ and $F(G)$.
 For $i \geq 0$ let $F_i \coloneqq \{f \in F(G) \mid |N_{R(G)}(f) \cap W| = i\}$.
 Then
 \[\sum_{i \geq 3} i \cdot |F_i| \leq |E(H)| \leq 2\left(\sum_{i \geq 3} |F_i| + |W|\right) - 4\]
 which implies that
 \[\sum_{i \geq 3} (i-2) \cdot |F_i| \leq 2|W|.\]
 In particular, 
 \[\sum_{i \geq 3} (i-1) \cdot |F_i| \leq 4|W|.\]
 
 It remains to bound the size of $F_2$.
 Consider the graph $H_2 \coloneqq R[W \cup F_2]$ in which every vertex in $F_2$ has exactly two neighbors.
 By outerplanarity, for each $f \in F_2$, it holds that $|\{f' \in F_2 \mid N_R(f) = N_R(f')\}| \leq 3$.
 Moreover, there are at most $3|W|$ elements from $F_2$ with pairwise distinct neighborhoods since $H_2$ is planar (every planar graph with $n$ vertices has at most $3n-6$ edges).
 Hence, $|F_2| \leq 9|W|$.
 
 So in total
 \[\sum_{f \in F(G) \setminus \{f_0\}} \max\{|N_{R(G)}(f) \cap W| - 1,0\} = |F_2| + \sum_{i \geq 3} (i-1) \cdot |F_i| \leq 13|W|.\]
\end{proof}

\begin{observation}
 The graph $G[L_j \cup L_{j+1} \cup \dots \cup L_{j + k}]$ is $(k+1)$-outerplanar.
\end{observation}

\begin{proof}
 This follows from Observation \ref{obs:edge-layers-by-exterior-face} and the definition of the sets $L_j'$, $j \geq 1$.
\end{proof}

An important feature of $k$-outerplanar graphs is that their treewidth is bounded by a linear function in $k$.

\begin{theorem}[Bodlaender \cite{Bodlaender98}]
 Let $G$ be a $k$-outerplanar graph.
 Then $\tw(G) \leq 3k-1$.
\end{theorem}

This implies that
\begin{equation}
 \label{eq:treewidth-of-layers}
 \tw(G[L_j \cup L_{j+1} \cup \dots \cup L_{j + k}]) \leq 3k+2
\end{equation}
for all $j \geq 1$.
Next, we define
\[V_i \coloneqq \bigcup_{j \equiv i \mod k} L_j\]
for all $i \in [k]$.
Note that this partition can be computed in polynomial time since an embedding for a planar graph can be computed in polynomial time.
Fix $i \in [k]$ and $W \subseteq V_i$.

Let $C_1,\dots,C_\ell$ be the connected components of $G[V_i \setminus W]$.
We refer to these components as the \emph{articulation points}.
Note that $C_r \subseteq L_j$ for some $j \equiv i \mod k$ by Observation \ref{obs:edge-endpoint-layers} (recall that $k \geq 2$).
Consistent with previous notation, let $L(C_r) \coloneqq j$ for the unique $j \geq 1$ such that $C_r \subseteq L_j$.

Now let $G'$ be the graph with vertex set
\[V(G') \coloneqq (V(G) \setminus V_i) \cup \{C_r,C_r' \mid r \in [\ell]\}\]
and edge set
\begin{align*}
 E(G') \coloneqq\;\;\; &E(G[V(G) \setminus V_i]\\
                \cup\; &\{vC_r \mid L(v) + 1 = L(C_r) \wedge\exists w \in C_r \colon vw \in E(G)\}\\
                \cup\; &\{vC_r' \mid L(v) = L(C_r) + 1 \wedge\exists w \in C_r \colon vw \in E(G)\}.
\end{align*}

\begin{lemma}
 $\big(G/E(G[V_i \setminus W])\big) - W = G'/\{C_rC_r' \mid r \in [\ell]\}$.
\end{lemma}

\begin{proof}
 Clearly, $\big(G/E(G[V_i \setminus W])\big) - W$ is obtained from $G$ by contracting $C_r$ to a single vertex for all $r \in [\ell]$, and deleting all vertices from $W$.
 Similarly, $G'/\{C_rC_r' \mid r \in [\ell]\}$ is obtained from $G$ by contracting $C_r$ to a single vertex for all $r \in [\ell]$, and deleting all vertices from $W$.
 Note that all edges from the graph $\big(G/E(G[V_i \setminus W])\big) - W$ are present by Observation \ref{obs:edge-endpoint-layers}.
\end{proof}

In order to bound the treewidth of $G/E(G[V_i \setminus W])$ we proceed in three steps.
For the first step we provide an upper bound on the treewidth of $G'$.
Then, in the second step, we prove that contraction of pairs of articulation points does not increase the treewidth significantly.
Finally, all vertices from $W$ are added to all bags which increases the treewidth only by $|W|$.

Let $D_1,\dots,D_t$ be the connected components of $G'$ and let $p \in [t]$.
For $j \equiv i \mod k$ define\[U_j \coloneqq \{C_r' \mid L(C_r) = j\} \cup L_{j+1} \cup \dots \cup L_{j + k - 1} \cup \{C_r \mid L(C_r) = j + k\}.\]
Clearly, there is some $j \equiv i \mod k$ such that $D_p \cap U_j \neq \emptyset$.
Moreover, there are no outgoing edges from the set $U_j$ in the graph $G'$ by Observation \ref{obs:edge-endpoint-layers}.
Hence, $D_p \subseteq U_j$.

This implies that $G'[D_p]$ is a minor of the graph $G[L_j \cup L_{j+1} \cup \dots \cup L_{j + k}]$.
So $\tw(G'[D_p]) \leq 3k+2$ by Equation \eqref{eq:treewidth-of-layers}.

For the second step consider the following DAG $H$ with vertex $V(H) \coloneqq \{D_1,\dots,D_t\} \cup \{C_1,\dots,C_\ell\}$ and edges
\[E(H) \coloneqq \{(C_r,D_p) \mid C_r \in D_p\} \cup \{(D_p,C_r) \mid C_r' \in D_p\}.\]
Note that, for edges $(D_p,C_r)$ and $(C_r,D_{p'})$, it holds that $D_p \subseteq U_j$ and $D_{p'} \subseteq U_{j + k}$ for some $j \equiv i \mod k$.
Moreover, each articulation point has exactly one incoming and one outgoing edge.
Together this proves that $H$ is a DAG.

For $p \in [t]$ we denote by $\deg_H^-(D_p)$ the number of incoming edges for the vertex $D_p$ in the graph $H$.

\begin{lemma}
 $\sum_{p \in [t]\colon \deg_H^-(D_p) \geq 2} \deg_H^-(D_p) - 1 \leq 13|W|$.
\end{lemma}

\begin{proof}
 Fix $j \equiv i \mod k$ and let $S_j \coloneqq L_j \cup L_{j+1} \cup \dots \cup L_{j + k}$.
 Also fix $A$ to be the vertex set of a connected component of $G[S_j]$.
 We define $U_{j,A}$ to be the ``intersection'' of $U_j$ and $A$, i.e.,
 \[U_{j,A} \coloneqq \{C_r' \mid L(C_r) = j, C_r \subseteq A\} \cup (L_{j+1} \cap A) \cup \dots \cup (L_{j + k - 1} \cap A) \cup \{C_r \mid L(C_r) = j + k, C_r \subseteq A\}.\]
 We prove that
 \begin{equation}
  \label{eq:bound:large-indegree-blobs}
  \sum_{p \in [t]\colon D_p \subseteq U_{j,A} \wedge \deg_H^-(D_p) \geq 2} \deg_H^-(D_p) - 1 \leq 13|W \cap A \cap L_j|.
 \end{equation}
 First observe that this implies the lemma since, for each $p \in [t]$, there is some $j \equiv i \mod k$ and some component $A$ of the graph $G[S_j]$ such that $D_p \subseteq U_{j,A}$.
 
 By Observation \ref{obs:edge-layers-by-exterior-face} it holds that $L_j \cap A$ are exactly those vertices incident to the exterior face of $G[A]$.
 In particular, $G[L_j \cap A]$ is connected and $1$-outerplanar.
 Let $B$ be the vetex set of a connected component of $G[A \setminus L_j]$.
 We define
 \[N_j(B) \coloneqq \{v \in L_j \mid \exists w \in B \colon \dist_R(v,w) \leq 2\}.\]
 It is easy to verify that $G[N_j[B]]$ is a cycle.
 More precisely, $N_j(B)$ forms a face cycle of $G[L_j \cap A]$.
 
 For a connected component $B$ of $G[A \setminus L_j]$ define the \emph{articulation degree} $\adeg(B)$ to be the number of articulation points $C_r$ such that $N_j[B] \cap C_r \neq \emptyset$.
 We argue that
 \[\sum_B \max\{\adeg(B) - 1,0\} \leq 13|W \cap A \cap L_j|\]
 where $B$ ranges over connected components of $G[A \setminus L_j]$.
 Note that this implies Equation \eqref{eq:bound:large-indegree-blobs}.
 For a component $B$ it holds that $\adeg(B) \leq \max\{|N_j(B) \cap W|,1\}$. 
 Hence, it suffices to prove that
 \[\sum_B \max\{|N_j(B) \cap W| - 1,0\} \leq 13|W \cap A \cap L_j|.\]
 But this follows from Lemma \ref{la:outerplanar-bound-covered-face-cycles} since each $N_j(B)$ forms a face cycle of $G[L_j \cap A]$.
\end{proof}

Now let $\CC$ be a set of articulation points such that $|\CC| \leq 13|W|$ and $\deg_{H - \CC}^-(D_p) \leq 1$ for all $p \in [t]$.
Then $H - \CC$ is a tree.
This means there a tree decomposition of $\big(G/E(G[V_i \setminus W])\big) - (\CC \cup W)$ of width $3k+2$ by stitching the tree decompositions for the components $D_p$, $p \in [t]$, together along the tree structure of $H - \CC$.
Adding the vertices from $\CC \cup W$ to all bags than gives a tree decomposition of $G/E(G[V_i \setminus W])$ of width $3k + 14|W| + 2$.
This completes the proof of Theorem \ref{thm:contraction-decomposition-planar-vertex}.

\subsection{Guessing Bodies of Segments}

Next, we concern ourselves with the guessing of the bodies which is related to problems defined over $2$-connected components.
In particular, we prove Lemma \ref{lem:introbodyguess}.
More precisely, in Lemma \ref{lem:introbodyguess}, the sets $V_i$ are obtained from the Contraction Decomposition Theorem (Theorem \ref{thm:contraction-decomposition-planar-vertex}) proved in the last subsection.
Here, we concern ourselves with the problem of finding the body sets $\CB_i$.
Towards this end, we assume that we are already given a suitable partition of the vertex set (later, this partition is replaced with the partition computed in Theorem \ref{thm:contraction-decomposition-planar-vertex}).

So let $G$ be a graph and fix $Z \subseteq V(G)$ to be a \emph{solution} of size $|Z| \leq k$.
Also fix a partition $V(G) = V_1 \uplus \dots \uplus V_\ell$.

We refer to the connected components of the graph $G[V_i \setminus Z]$ as the \emph{$i$-segments}.
A \emph{segment} is an $i$-segment for some $i \in [\ell]$.
Let $\CI$ denote the set of all segments.

In the following, we are interested in the $2$-connected components of $G - Z$.
We shall represent the $2$-connected components of $G - Z$ as a pair $(\FT,\gamma)$ where $\FT$ is a forest and $\gamma \colon V(\FT) \rightarrow 2^{V(G)}$ is a function as follows.
Let us first consider the block-cut tree $\FT'$ which contains a node $t_C$ for every $2$-connected component $C$ of $G - Z$, and a node $s_v$ for every cut vertex $v$ of $G - Z$.
Also, there is an edge between $t_C$ and $s_v$ if and only if $v \in C$.
It is well knwon that $\FT'$ is a forest where each connected component of $\FT'$ corresponds to some connected component of $G - Z$.
For each connected component of $\FT'$ we fix a unique root node $t_C$ that corresponds to some $2$-connected component of $G - Z$. 
For simplicity, let us fix a ``canonical'' rooting as follows.
We assume that the vertices of the input graph $G$ are ordered (the order is arbitrary).
Now, we can order the $2$-connected components of $G - Z$ lexicographically.
For each connected component $T$ of $\FT'$, we pick the lexicographically smallest $2$-connected component as the root of $T$.

Now, consider a segment $I \in \CI$.
We need to define several objects based on the decomposition $(\FT,\gamma)$ into the $2$-connected components of $G - Z$.
First, let $r_Z(I)$ denote the unique node $t \in V(\FT)$ which is closest to a root node $t_0$ and for which $I \cap \gamma(t) \neq \emptyset$.
Also, let $B_Z(I) \coloneqq \{t \in V(\FT) \setminus \{r_Z(I)\} \mid I \cap \gamma(t) \neq \emptyset\}$ be the \emph{body of $I$}.
Observe that $r_Z(I) \notin B_Z(I)$.
Moreover, let $W_Z(I) \coloneqq \desc(r_Z(I)) \setminus \{r_Z(I)\}$ denote the set of all nodes that are descendants of $r_Z(I)$, excluding $r_Z(I)$.
Observe that $r_Z(I) \notin W_Z(I)$, but $B_Z(I) \subseteq W_Z(I)$.

Finally, we define $\widehat{r}_Z(I) \coloneqq \gamma(r_Z(I))$, $\widehat{B}_Z(I) \coloneqq \bigcup_{t \in R_Z(I)} \gamma(t)$ and $\widehat{W}_Z(I) \coloneqq \bigcup_{t \in W_Z(I)} \gamma(t)$ to be the corresponding sets of vertices in the graph $G$.
To simplify notation, we will regularly omit the index $Z$ if it is clear from context.

The next theorem allows us to compute the body sets.
Actually, for technical reasons, the objects we guess are slightly more complicated and also need to contain information about the sets $\widehat{r}_Z(I) \cap I$.

\begin{theorem}
 \label{thm:2-cc-from-solution-subset}
 Let $G$ be an $H$-minor-free graph and let $V_1,\dots,V_{\ell}$ be a partition of the vertex set $V(G)$.
 Also let $Z \subseteq V(G)$ be an (unknown) set of vertices of size $|Z| = k$ and let $(\FT,\gamma)$ be the (rooted) decomposition into $2$-connected components of $G - Z$.
 
 Then there is a constant $c_H$, an index $i \in [\ell]$, and a sequence of vertices $\bar s = (s_1,\dots,s_r) \in (V(G))^r$ of length $r = \CO(c_H \frac{k}{\ell})$ such that the following conditions are satisfied:
 \begin{enumerate}[label=(\Roman*)]
  \item\label{item:2-cc-from-solution-subset-1} $V_i \cap Z = \{s_1,\dots,s_j\}$ for some $j \in \{0,\dots,r\}$, and
  \item\label{item:2-cc-from-solution-subset-2} given $G,V_1,\dots,V_\ell,i,\bar s,j$ one can compute in time $n^{\CO(c_H)}$, for each $i$-segment $I$,
   a family of sets $\CF(I)\subseteq 2^{V(G)} \times 2^{V(G)}$ of size $|\CF(I)| \leq n^{c_H}$ such that $(\widehat{B}_Z(I),\widehat{r}_Z(I) \cap I) \in \CF(I)$.
 \end{enumerate}
\end{theorem}

Observe that this is non-trivial since both $\widehat{B}_Z(I)$ and $\widehat{r}_Z(I) \cap I$ depend on $Z$ which is only partially given to the algorithm.
The remainder of this subsection is devoted to proving the theorem.

Let $d_H = \CO(h \sqrt{\log h})$ denote the constant from Theorem \ref{thm:average-degree-H-minor-free}, i.e., $\sum_{v \in V(G)} \deg_G(v) \leq d_H |V(G)|$ holds for every $H$-minor-free graph $G$.

\begin{lemma}
 \label{la:high-degree-tree-size}
 Let $G$ be an $H$-minor-free graph and suppose $V(G) = U \uplus Z$ such that $\deg_G(u) > 2d_H$ for all $u \in U$.
 Then $|U| \leq |Z|$.
\end{lemma}

\begin{proof}
 Suppose towards a contradiction that $|U| > |Z|$.
 Then $|E(G)| \geq \frac{1}{2}|U|2d_H = d_H|U|$.
 Hence,
 \[\sum_{v \in V(G)} \deg_G(v) = 2|E(G)| \geq 2d_H|U| > d_HV(G)\]
 which is a contradiction to Theorem \ref{thm:average-degree-H-minor-free}.
\end{proof}

Now, let us fix an $H$-minor-free graph $G$, a partition $V_1,\dots,V_{\ell}$ of the vertex set $V(G)$, and an unknown solution set $Z \subseteq V(G)$ of size $|Z| = k$.
Also, let $(\FT,\gamma)$ denote the (rooted) decomposition into $2$-connected components of $G - Z$. 

\begin{lemma}
 \label{la:segment-intersection}
 Let $I_1,I_2$ be two distinct segments.
 Then $B(I_1) \cap B(I_2) = \emptyset$.
\end{lemma}

\begin{proof}
 Suppose $t \in B(I_1)$ and let $s$ be the parent node of $t$.
 Then $s \in B(I_1) \cup \{r(I_1)\}$.
 Let $v \in \gamma(t) \cap \gamma(s)$.
 Since $G[I_1]$ is connected it holds that $v \in I_1$.
 So $v \notin I_2$ since $I_1 \cap I_2 = \emptyset$.
 Hence, $t \notin B(I_2)$.
\end{proof}

Let $I$ be a segment.
We define $\val(I)$ to be the maximum number of paths from $\widehat{B}(I)$ to $Z$ in the graph $G[\widehat{W}(I) \cup Z]$ that are pairwise disjoint except on $\widehat{B}(I)$.
Observe that $\val(I) = 0$ whenever $B(I) = \emptyset$.

\begin{lemma}
 Suppose $G$ is $H$-minor-free.
 Then there is a constant $c_H$ such that
 \[\sum_{I \in \CI} \max(\val(I) - c_H,0) \leq c_H \cdot |Z|.\]
\end{lemma}

\begin{proof}
 We define $c_H \coloneqq 6d_H$.
 Let $\CI^* \coloneqq \{I \in \CI \mid \val(I) > c_H\}$.
 We first partition $\CI^*$ into three sets $\CI_0$, $\CI_1$ and $\CI_2$ such that, for all $\mu \in \{0,1,2\}$ and all $I,I' \in \CI_\mu$ it holds that
 \[r(I) \notin B(I')\]
 and
 \[r(I) \neq r(I') \;\;\Rightarrow\;\; E_{\FT}(B(I) \cup \{r(I)\},B(I') \cup \{r(I')\}) = \emptyset.\]
 Actually, let us remark that the first condition is implied by the second one, but we still state it explicitly for later reference.
 
 The partition can be computed inductively as follows.
 Throughout the induction we maintain the property that segments $I$ with the same root bag $r(I)$ are assigned to the same partition class.
 For each $I \in \CI$ define the \emph{height} of $I$ to be the distance from $r(I)$ to the root of the corresponding connected component of $\FT$.
 First, we assign all segments of height $0$ to the class $\CI_0$.
 Clearly, this meets the above requirements.
 Now suppose all segments up to height $h$ have been partitioned and let $I$ be a segment of height $h+1$.
 Let $t \coloneqq r(I)$ and let $s$ be the parent of $t$.
 First suppose $t \in B(I')$ for some segment $I'$ of height at most $h$.
 By Lemma \ref{la:segment-intersection} there is a unique $I'$ of height at most $h$ such that $t \in B(I')$.
 Note that $s \in B(I') \cup \{r(I')\}$.
 Suppose $I' \in \CI_{\mu'}$.
 Moreover, by Lemma \ref{la:segment-intersection}, there is at most one $I''$ such that $s \in B(I'')$.
 Assume $I'' \in \CB_{\mu''}$ (if it exists).
 We assign $I$ to a partition class $\CI_\mu$ for some $\mu \in \{0,1,2\} \setminus \{\mu',\mu''\}$.
 Moreover, we assign all segments with the same root bag $t$ to the same partition class. 
 It can be easily checked that the requirements are satisfied.
 
 In the following we argue that $\sum_{I \in \CI_\mu} (\val(I) - c_H) \leq \frac{c_H}{3} \cdot |Z|$ for all $\mu \in \{0,1,2\}$ which implies the lemma.
 
 Without loss of generality assume $\mu = 0$.
 For each $I \in \CI_0$ we associate a subset $U(I)$ of $V(G) \setminus Z$ as follows.
 In the first step, all elements from $I \cup \widehat{B}(I)$ are added to $U(I)$.
 Note that all sets $U(I)$ are connected and pairwise disjoint.
 Let $X$ be those vertices that are not assigned so far.
 Next, all vertices from $\widehat{W}(I) \cap X$ that are reachable from $\widehat{B}(I)$ in the graph $G[X \cup \widehat{B}(I)]$ are added to $U(I)$.
 Observe that all sets $U(I)$ remain connected and pairwise disjoint.
 
 Now consider the following graph $F$ with vertex set $V(F) \coloneqq \CI_0 \uplus Z$ and edge set
 \[E(F) \coloneqq \{II' \mid I \neq I' \in \CI_0, E_G(U(I),U(I')) \neq \emptyset\} \cup \{Iz \mid I \in \CI_0, z \in Z, E_G(U(I),\{z\}) \neq \emptyset\}\]
 ($F$ is not a multigraph, i.e., there is for only a single edge between $I$ and $I'$ (resp.\ $I$ and $z$) even if $E(U(I),U(I'))$ (resp.\ $E(U(I),\{z\})$) contains more than one element).
 Note that $F$ is a minor of $G$.
 
 We claim that $\val(I) \leq \deg_F(I)$ for all $I \in \CI_0$.
 Let us first complete the proof assuming the claim holds true.
 Then $|\CI_0| \leq |Z|$ by Lemma \ref{la:high-degree-tree-size}.
 This means that
 \[\sum_{I \in \CI_0} (\val(I) - c_H) \leq \sum_{v \in V(F)} \deg_F(v) \leq d_H |V(F)| \leq 2d_H |Z| \leq \frac{c_H}{3} |Z|.\]
 
 To complete the proof it remains show that $\val(I) \leq \deg_F(I)$ for all $I \in \CI_0$.
 Let us fix some $I \in \CI_0$ and a set of paths $P_1^I,\dots,P_{\val(I)}^I$ from $\widehat{B}(I)$ to $Z$ in the graph $G[\widehat{W}(I) \cup Z]$ that are pairwise disjoint except on $\widehat{B}(I)$.
 Without loss of generality assume that no internal vertex of a path $P_i^I$ is contained in $Z$.
 Also define $Z(I) \subseteq Z$ to be set of all endpoints of the paths $P_1^I,\dots,P_{\val(I)}^I$.
 For $z \in Z(I)$ we also write $P_z^I$ for the unique path $P_i^I$ which ends in $z$.
 
 We say that $z \in Z(I)$ is \emph{directly reachable from $I$} if $P_z^I$ does not visit any of the sets $U(I')$ for $I \neq I' \in \CI_0$.
 Let $Z^*(I) \subseteq Z(I)$ be those vertices directly reachable from $I$.
 Let $z \in Z^*(I)$.
 By definition, all internal vertices from $P_z^I$ are contained in the set $U(I)$.
 Hence, $Iz \in E(F)$.
 On the other hand, for $z \in Z(I) \setminus Z^*(I)$, consider the first vertex $w_z$ on $P_z^I$ that is not contained in $U(I)$.
 Then $w_z \in I_z$ for some $I \neq I_z \in \CI_0$.
 Hence, $w_z \in U(I_z)$ and $II_z \in E(F)$.
 It remains to argue that, for a second $z \neq z' \in Z(I) \setminus Z^*(I)$, we have that $I_z \neq I_{z'}$.
 Suppose otherwise.
 Clearly, $r(I) \neq r(I_z)$.
 By the properties of the set $\CI_0$, it holds that $r(I_z) \notin B(I)$.
 Since $P_z^I$ and $P_{z'}^I$ are vertex-disjoint except on $\widehat{B}(I)$ we have that $w_{z} \neq w_{z'}$.
 Moreover, $w_z \in I_z \cap \widehat{r}(I_z)$ and $w_{z'} \in I_{z'} \cap \widehat{r}(I_{z'})$.
 Since $P_z^I$ and $P_{z'}^I$ are vertex-disjoint except on $\widehat{B}(I)$ this is only possible if $r(I_z)$ is a child of some node $t \in B(I)$.
 Hence, $E_T(\{r(I_z)\},B(I)) \neq \emptyset$ which contradicts the second property of the set $\CI_0$.
 Hence, $I_z \neq I_{z'}$ for all distinct $z,z' \in Z(I) \setminus Z^*(I)$.
 
 Together, this means that $\deg_F(I) \geq |Z(I)|$ for all $I \in \CI_0$ which completes the proof of the lemma.
\end{proof}

For $i \in [\ell]$ we define
\[\val^*(V_i) \coloneqq \sum_{I \text{ is $i$-segment}} \max(\val(I) - c_H,0).\]
We claim there exists some $i \in [\ell]$ such that
\begin{enumerate}
 \item $|V_i \cap Z| \leq 2\frac{k}{\ell}$, and
 \item $\val^*(V_i) \leq 2c_H \frac{k}{\ell}$.
\end{enumerate}
Let $A_1 \coloneqq \{i \in [\ell] \mid |V_i \cap Z| > 2\frac{k}{\ell}\}$ be the set of indices that violate the first condition.
Then $k = |Z| > |A_1| \cdot 2\frac{k}{\ell}$ which means that $|A_1| < \frac{\ell}{2}$.

Similarly, let $A_2 \coloneqq \{i \in [\ell] \mid\val^*(V_i) > 2c_H \frac{k}{\ell}\}$ be the set of indices that violate the second condition.
By the lemma
\[c_Hk = c_H|Z| \geq \sum_{i \in A_2}\val^*(V_i) > |A_2|2c_H \frac{k}{\ell}\]
which implies that $|A_2| < \frac{\ell}{2}$.
Hence, there is some $i \in [\ell]$ such that neither $i \in A_1$ nor $i \in A_2$.
Let $j \coloneqq |V_i \cap Z|$ and define $\{s_1,\dots,s_j\} = V_i \cap Z$ (the vertices are enumerated arbitrarily).

\medskip

Now let $I$ be an $i$-segment and consider the graph $G[\widehat{W}(I) \cup Z]$.
In this graph, there are at most $\val(I)$ many paths from $\widehat{B}(I)$ to $Z$ that are pairwise disjoint except on $\widehat{B}(I)$.
Hence, by Menger's theorem, there is a set $S_I \subseteq \widehat{W}(I) \cup Z$ of size $|S_I| = \val(I)$ such that $S_I \cap \widehat{B}(I) = \emptyset$ and every path from $\widehat{B}(I)$ to $Z$ in the graph $G[\widehat{W}(I) \cup Z]$ visits some vertex from $S_I$.
Suppose that $S_I = \{w_1^I,\dots,w_{\val(I)}^I\}$.

We define
\[\{s_{j+1}.\dots,s_r\} = \{w_{c_H+1}^I,\dots,w_{\val(I)}^I \mid I \text{ is $i$ segment with } \val(I) > c_H\}.\]
Note that $r \leq |V_i \cap Z| + \val^*(V_i) \leq 3(c_H+1)\frac{k}{\ell}$.
This completes the description of the sequence $\bar s = (s_1,\dots,s_r)$.

Clearly, Property \ref{item:2-cc-from-solution-subset-1} of Theorem \ref{thm:2-cc-from-solution-subset} is satisfied.
Hence, in order to complete the proof of Theorem \ref{thm:2-cc-from-solution-subset} it remains to verify Property \ref{item:2-cc-from-solution-subset-2}, i.e., we have to argue how to compute the families $\CF(I)$ for all $i$-segments $I$.
Clearly, given access to $G,V_1,\dots,V_\ell,i,\bar s,j$ we can compute the set of all $i$-segments in polynomial time.
So fix $I$ to be an $i$-segment.
The algorithm first guesses $\val(I)$ by iterating through all possible values.

If $\val(I) \leq c_H$ the algorithm iterates through all subsets $S \subseteq V(G) \setminus I$ of size $\val(I)$ and computes the decomposition $(\FT_S,\gamma_S)$ into $2$-connected components of the graph $G - S$.
Also, given a subset $S \subseteq V(G) \setminus I$ of size $\val(I)$, the algorithm guesses a root node $t_0^S$ for the connected component of $(\FT_S,\gamma_S)$ that contains $I$.
Finally, let $r_S(I)$ denote the unique node $t \in V(\FT_S)$ which is closest to the root node $t_0^S$ and for which $I \cap \gamma_S(t) \neq \emptyset$.
Also, let $B_S(I) \coloneqq \{t \in V(\FT_S) \setminus \{r_S(I)\} \mid I \cap \gamma_S(t) \neq \emptyset\}$.
Finally, we add the pair
\[\Big(\bigcup_{t \in B_S(I)} \gamma_S(t),\gamma_S(r_S(I)) \cap I\Big)\]
to the family $\CF(I)$.
For $S = S_I$ together with a suitable choice of $t_0^S$ we get that $\widehat{B}(I) = \bigcup_{t \in B_S(I)} \gamma_S(t)$.

For $\val(I) > c_H$ the algorithm first guesses $L_I = \{w_{c_H+1}^I,\dots,w_{\val(I)}^I\}$.
By appropriately ordering the sequence $\bar s$ and marking certain elements (for example, by repeating the element), this adds a multiplicative factor of at most $n$.
Having guessed $L_I$ the algorithm again iterates through all subsets of $S \subseteq V(G) \setminus I$ of size $c_H$ and computes the decomposition $(\FT_S,\gamma_S)$ into $2$-connected components of the graph $G - (S \cup L_I)$.
Now the algorithm proceeds as in the previous case by guessing a root of $(\FT_S,\gamma_S)$ and adding the pair
\[\Big(\bigcup_{t \in B_S(I)} \gamma_S(t),\gamma_S(r_S(I)) \cap I\Big)\]
to the family $\CF(I)$.
Again, for $S = S_I \setminus L_I$ together with a suitable choice of $t_0^S$ we get that $\widehat{B}(I) = \bigcup_{t \in B_S(I)} \beta(t)$.

This completes the proof of Theorem \ref{thm:2-cc-from-solution-subset}.

\subsection{Obtaining Segmented Graphs with Body Sets}

Having proved Theorems \ref{thm:contraction-decomposition-planar-vertex} and \ref{thm:2-cc-from-solution-subset}, we can now combine both results to prove Lemma \ref{lem:introbodyguess}.
Actually, we shall use a different formulation of Lemma \ref{lem:introbodyguess} which is slightly more convenient due to technical reasons.
We leave it as an exercise for the reader to derive Lemma \ref{lem:introbodyguess} from Corollary \ref{cor:segments-and-bodies-for-unknown-solution} below (all algorithms making use of this result do so via Corollary \ref{cor:segments-and-bodies-for-unknown-solution}).
Towards this end, we first need to introduce some additional notation on segments that is also heavily used in the next section.

A \emph{segmented graph} is a pair $(G,\CI)$ where $G$ is a graph and $\CI = \{I_1,\dots,I_\ell\}$ is a set of pairwise disjoint, connected subsets of $V(G)$.
Moreover, we generally assume that each segment $I \in \CI$ is equipped with a linear order $<_I$.
We refer to the sets $I_1,\dots,I_\ell$ as the \emph{segments} of $(G,\CI)$.
Also, define $V(\CI) \coloneqq \bigcup_{I \in \CI}I$ as the set of vertices appearing in some segment.
We define $G/\CI$ to be the graph obtained from $G$ by contracting each segment to a single vertex.
Note that $G/\CI$ is a minor of $G$ since all sets $I \in \CI$ are connected.
We always use $v_I$ to denote the vertex of $G/\CI$ that corresponds to segment $I$.
Also, $V_\CI \coloneqq \{v_I \mid I \in \CI\}$.
For $U \subseteq V(G)$ we use $\shr(U)$ to denote the set of all vertices $v \in V(G/\CI)$ that correspond to some $u \in U$.
In particular, $v_I \in \shr(U)$ if $U \cap I \neq \emptyset$.
In the other direction, for $U \subseteq V(G/\CI)$, we use $\ext(U)$ to denote the set of all vertices $v \in V(G)$ that correspond to some $u \in U$.
In particular, if $v_I \in U$ then $I \subseteq \ext(U)$.
If $U = \{u\}$ consists of a single vertex, then we also write $\shr(u)$ and $\ext(u)$ instead of $\shr(\{u\})$ and $\ext(\{u\})$. 

\begin{corollary}
 \label{cor:segments-and-bodies-for-unknown-solution}
 There is a polynomial-time algorithm that, given a planar graph $G$, an integer $k$, and a sequence of vertices $(v_1,\dots,v_r) \in (V(G))^r$, computes a set of pairwise vertex-disjoint segments $\CI$ and a function $\CF$ mapping every $I \in \CI$ to a set $\CF(I) \subseteq 2^{V(G)} \times 2^{V(G)}$ of size $|\CF(I)| = n^{\CO(1)}$ such that the following property is satisfied:
 
 For every $Z \subseteq V(G)$ of size $|Z| \leq k$ there is some sequence $(v_1,\dots,v_r) \in (V(G))^r$ of length $r \in \CO(\sqrt{k})$ such that, if $(\CI,\CF)$ is the output of the algorithm on input $(G,k,(v_1,\dots,v_r))$, then
 \begin{enumerate}
  \item $V(\CI) \cap Z = \emptyset$,
  \item $\tw(G/\CI) = \CO(\sqrt{k})$, and 
  \item $(\widehat{B}_Z(I),\widehat{r}_Z(I) \cap I) \in \CF(I)$ for every $I \in \CI$.
 \end{enumerate}
\end{corollary}

\begin{proof}
 Let $(G,k,(v_1,\dots,v_r))$ be the input to the algorithm.
 The algorithm first fixes an arbitrary numbering of the vertices of $G$, i.e., $V(G) = \{u_1,\dots,u_n\}$.
 Let $\ell \coloneqq \lceil\sqrt{k}\rceil$.
 The algorithm applies Theorem \ref{thm:contraction-decomposition-planar-vertex} and otains a partition $V_1,\dots,V_\ell$ of $V(G)$.
 Let $i \in [\ell]$ be the unique index such that $v_1 \in V_i$.
 Also, let $j \in \{0,\dots,n-1\}$ such that $v_2 = u_{j+1}$.
 We define $Z_i \coloneqq \{v_3,\dots,v_{j+2}\}$ and compute $\CI$ as the set of connected components of $G[V_i \setminus Z_i]$.
 Also, we define $\bar s \coloneqq (v_3,\dots,v_r)$.
 Finally, the algorithm computes the function $\CF$ using the algorithm from Theorem \ref{thm:2-cc-from-solution-subset}, Item \ref{item:2-cc-from-solution-subset-2} on input $(G,V_1,\dots,V_\ell,i,\bar s,j)$.
 Clearly, this algorithm runs in polynomial time.
 
 So let $Z \subseteq V(G)$ of size $|Z| \leq k$.
 We need to argue about the existence of a sequence $(v_1,\dots,v_r) \in (V(G))^r$ of length $r \in \CO(\sqrt{k})$ satisfying the properties stated in the corollary.
 By Theorem \ref{thm:2-cc-from-solution-subset} there is some $i \in [\ell]$, a sequence $\bar s = (s_1,\dots,s_r) \in (V(G))^r$ for $r \in \CO(\sqrt{k})$, and an integer $j \in [r]$ such that $V_i \cap Z = \{s_1,\dots,s_j\}$.
 We pick $v_1 \in V_i$ arbitrarily, $v_2 \coloneqq u_{j+1}$, and $v_{p+2} \coloneqq s_p$ for all $p \in [r]$.
 Let $(\CI,\CF)$ denote the output of the above algorithm on input $(G,k,(v_1,\dots,v_r))$.
 Clearly, $V(\CI) \cap Z = \emptyset$ by definition.
 Also, $\tw(G/\CI) = \CO(\ell + |V_i \cap Z|) = \CO(\sqrt{k})$.
 Finally, $(\widehat{B}_Z(I),\widehat{r}_Z(I) \cap I) \in \CF(I)$ for every $I \in \CI$ by Theorem \ref{thm:2-cc-from-solution-subset}, Property \ref{item:2-cc-from-solution-subset-2}.
\end{proof}

For later reference, we also formulate a version of the last corollary for $H$-minor-free graphs assuming Conjecture \ref{conj:contraction-decomposition-minor-vertex} holds.

\begin{corollary}
 \label{cor:segments-and-bodies-for-unknown-solution-minors}
 Assuming Conjecture \ref{conj:contraction-decomposition-minor-vertex}, there is an algorithm that, given an $H$-minor-free graph $G$, an integer $k$, and a sequence of vertices $(v_1,\dots,v_r) \in (V(G))^r$, computes a set of pairwise vertex-disjoint segments $\CI$ and a function $\CF$ mapping every $I \in \CI$ to a set $\CF(I) \subseteq 2^{V(G)} \times 2^{V(G)}$ of size $|\CF(I)| = n^{c_H}$ for some constant $c_H$ depending only on $H$ such that the following property is satisfied:
 
 For every $Z \subseteq V(G)$ of size $|Z| \leq k$ there is some sequence $(v_1,\dots,v_r) \in (V(G))^r$ of length $r \in \CO(c_H\sqrt{k})$ such that, if $(\CI,\CF)$ is the output of the algorithm on input $(G,k,(v_1,\dots,v_r))$, then
 \begin{enumerate}
  \item $V(\CI) \cap Z = \emptyset$,
  \item $\tw(G/\CI) = \CO(\sqrt{k})$, and 
  \item $(\widehat{B}_Z(I),\widehat{r}_Z(I) \cap I) \in \CF(I)$ for every $I \in \CI$.
 \end{enumerate}
 Moreover, the algorithm runs in time $n^{\CO(c_H)}$.
\end{corollary}

\begin{proof}
 The proof is analogous to the proof of Corollary \ref{cor:segments-and-bodies-for-unknown-solution} replacing the application of Theorem \ref{thm:contraction-decomposition-planar-vertex} with Conjecture \ref{conj:contraction-decomposition-minor-vertex}.
\end{proof}

%% file: alg-perm-csp.tex
\section{Algorithms for CSP Deletion Problems on Planar Graphs}
\label{sec:alg-perm-csp}

In this section, we prove Theorems \ref{thm:perm-csp-edge-deletion-result}, \ref{thm:perm-csp-deletion-result} and \ref{thm:2cc-perm-csp-deletion-result}.
All proofs follow essentially the same high-level strategy.
First, we use the Contraction Decomposition Theorem to partition the vertex or edge set of the input constraint graph (depending on the problem in question).
Then, we guess some partition class that has small intersection with the solution as well as the intersection of the solution with said partition class.
This gives rise to a segmented graph $(G,\CI)$ where the segments are the connected components of the chosen partition class after removing all solution vertices.
Now, we translate the initial Permutation-CSP-instance (with constraint graph $G$) to an equivalent binary CSP-instance with constraint graph $G/\CI$.
For \textsc{$2$Conn Perm CSP Vertex Deletion with Size Constraints}, this translation also builds on Theorem \ref{thm:2-cc-from-solution-subset} to guess the body sets of the segments.
Finally, we use dynamic programming to solve the constructed binary CSP-instance over constraint graph $G/\CI$ exploiting that $G/\CI$ has small treewidth.

We remark that the algorithms presented in Sections \ref{subsec:tw-dp} - \ref{subsec:perm-csp-vertex-deletion} are standard algorithms, and the corresponding results should not be surprising to a reader familiar with the concept of contraction decompositions as well as algorithmic approaches to CSPs.
Indeed, the main algorithmic contribution of this section is the subexponential algorithm for \textsc{$2$Conn Perm CSP Vertex Deletion with Size Constraints} which is presented in Section \ref{subsec:2cc-perm-csp-vertex-deletion}.

\subsection{CSPs with Size Constraints on Graphs of Bounded Treewidth}
\label{subsec:tw-dp}

We start by implementing the last step, i.e., we provide an algorithm deciding satisfiability of binary CSPs $\Gamma = (X,D,\CC)$ in time $|D|^{\CO(\tw(G(\Gamma)))}|X|^{\CO(1)}$.
It is well-known that such an algorithm can be obtained via a simple dynamic programming strategy along a tree decomposition of $G$.
However, for Theorems \ref{thm:perm-csp-deletion-result} and \ref{thm:2cc-perm-csp-deletion-result}, we also need to able to take size constraints into account.
Here, our strategy is similar, i.e., we translate the size constraints on $G$ to suitable size constraints over $G/\CI$.
To able to use the same subroutine for both Theorem \ref{thm:perm-csp-deletion-result} and \ref{thm:2cc-perm-csp-deletion-result}, the size constraints that we allow on the contracted graph need to be fairly general.
In the following, we define suitable size constraints and provide a dynamic programming algorithm for input instances of small treewidth.
We remark that the size constraints we introduce may look somewhat unnatural at first glance, but they are designed primarily to fit our needs when proving Theorems \ref{thm:perm-csp-deletion-result} and \ref{thm:2cc-perm-csp-deletion-result}.

Let $\Gamma = (X,D,\CC)$ be a binary CSP instance.
A \emph{global size constraint} is a triple $(w,q,\op)$ where $w \colon X \times D \rightarrow \NN$ is a weight function, $q \in \NN$ and $\op \in \{\leq, \geq\}$.
An assignment $\alpha\colon X \rightarrow D$ satisfies $(w,q)$ if
\[\Big(\sum_{v \in X} w(v,\alpha(v)), q \Big) \in \op.\]

Let $F \subseteq D^2$ such that $F$ is symmetric (i.e., $(a,b) \in F$ if and only if $(b,a) \in F$ for all $a,b \in D$) and let $\CD$ be a partition of $D$.
An \emph{$(F,\CD)$-local size constraint} is a tuple $(w,w_\CD,q,\op)$ where $w \colon X \times D \rightarrow \NN$ and $w_\CD \colon \CD \rightarrow \NN$ are weight functions, $q \in \NN$ and $\op \in \{\leq, \geq\}$.
Let $\alpha\colon X \rightarrow D$ be an assignment.
Let $G \coloneqq G(\Gamma)$ be the constraint graph.
We define $H_\alpha$ to be the subgraph of $G$ with vertex set $V(H_\alpha) \coloneqq V(G)$ and edge set
\[E(H_\alpha) \coloneqq \{vw \in E(G) \mid \alpha(v)\alpha(w) \in F\}.\]
Let $C_1,\dots,C_\ell$ denote the connected components of $H$.
We say $\alpha\colon X \rightarrow D$ satisfies $(w,w_\CD,q,\op)$ if, for every $i \in [\ell]$, there is some $D_i \in \CD$ such that
\[\alpha(v) \in D_i\]
for every $v \in C_i$ and
\[\Big(w_\CD(D_i) + \sum_{v \in C_i} w(v,\alpha(v)), q\Big) \in \op.\]

Let $\CS_\glo$ be a set of global size constraints and $\CS_\loc$ be a set of $(F,\CD)$-local size constraints.
We say that $(\Gamma,\CS_\glo,\CS_\loc)$ is \emph{satisfiable} if there is an assignment $\alpha\colon X \rightarrow D$ that satisfies $\Gamma$,
all global constraints $(w,q,\op) \in \CS_\glo$ and all $(F,\CD)$-local size constraints $(w,q,\op) \in \CS_\loc$.
We define
\[\|\CS_\glo\| \coloneqq \prod_{(w,q,\op) \in \CS_\glo} (2 + q)\]
and similarly,
\[\|\CS_\loc\| \coloneqq \prod_{(w,w_\CD,q,\op) \in \CS_\loc} (2 + q).\]

\begin{theorem}
 \label{thm:tw-binary-csp-size-constraints}
 Let $\Gamma = (X,D,\CC)$ be a binary CSP instance and let $F \subseteq \binom{D}{2}$ and $\CD$ be a partition of $D$.
 Also, let $\CS_\glo$ be a set of global size constraints and $\CS_\loc$ be a set of $(F,\CD)$-local size constraints.
 Moreover, let $k \coloneqq \tw(G(\Gamma))$.
 Then there is an algorithm which decides whether $(\Gamma,\CS_\glo,\CS_\loc)$ is \emph{satisfiable} and runs in time
 \[(|D| + k + \|\CS_\glo\| + \|\CS_\loc\|)^{\CO(k)}|X|^{\CO(1)}.\]
\end{theorem}

\begin{proof}
 Let $G \coloneqq G(\Gamma)$ and let $(T,\beta)$ be a nice tree decomposition of $G$.
 Let $t \in V(T)$.
 We denote $X_t \coloneqq \beta(t)$ and $V_t$ the set of vertices contained in bags below $t$ (including $t$ itself).
 For a mapping $\alpha \colon V_t \rightarrow D$, we define the graph $H_{t,\alpha}$ with vertex set $V(H_{t,\alpha}) \coloneqq V_t$ and
 \[E(H_{t,\alpha}) \coloneqq \{vw \in E(G) \mid \alpha(v)\alpha(w) \in F\}.\]
 
 For every function $f \colon X_t \rightarrow D$, every partitions $\CP$ of $X_t$, all numbers $q_\Fc \in [0,q+1]$ for $\Fc = (w,q,\op) \in \CS_\glo$, and all functions $g_\Fc \colon \CP \rightarrow [0,q+1]$ for $\Fc = (w,w_\CD,q,\op) \in \CS_\loc$,
 we define
 \[c_t[f,\CP,(q_\Fc)_{\Fc \in \CS_\glo},(g_\Fc)_{\Fc \in \CS_\loc}] = 1\]
 if there is an assignment $\alpha \colon V_t \rightarrow D$ such that
 \begin{enumerate}[label = (\Alph*)]
  \item\label{item:tw-dp-1} $\alpha$ satisfies $\Gamma[V_t]$,
  \item\label{item:tw-dp-2} $f(v) = \alpha(v)$ for all $v \in X_t$,
  \item\label{item:tw-dp-3} $C_i \cap X_t \in \CP$ for every $i \in [\ell]$, where $C_1,\dots,C_\ell$ denote the connected components of $H_{t,\alpha}$,
  \item\label{item:tw-dp-4} $\min\Big(\sum_{v \in V_t} w(v,\alpha(v)),q+1\Big) = q_\Fc$ for every $\Fc = (w,q,\op) \in \CS_\glo$,
  \item\label{item:tw-dp-5} for every $i \in [\ell]$ there is some $D_i \in \CD$ such that $\alpha(v) \in D_i$ for every $v \in C_i$,
  \item\label{item:tw-dp-6} $\Big(w_\CD(D_i) + \sum_{v \in C_i} w(v,\alpha(v)), q\Big) \in \op$ for every $i \in [\ell]$ such that $C_i \cap X_t = \emptyset$ and every $\Fc = (w,w_\CD,q,\op) \in \CS_\loc$,
  \item\label{item:tw-dp-7} $\min\Big(\sum_{v \in C_i} w(v,\alpha(v)),q+1\Big) = g_\Fc(C_i \cap X_t)$ for every $i \in [\ell]$ such that $C_i \cap X_t \neq \emptyset$ and every $\Fc = (w,w_\CD,q,\op) \in \CS_\loc$.
 \end{enumerate}
 Otherwise, $c_t[f,\CP,(q_\Fc)_{\Fc \in \CS_\glo},(g_\Fc)_{\Fc \in \CS_\loc}] = 0$.
 Also, for ease of notation, we define $A_t$ as the set of all tuples $(f,\CP,(q_\Fc)_{\Fc \in \CS_\glo},(g_\Fc)_{\Fc \in \CS_\loc})$.
 Observe that
 \[|A_t| \leq |D|^{|X_t|} \cdot |X_t|^{|X_t|} \cdot \|\CS_\glo\| \cdot \|\CS_\loc\|^{|X_t|} \leq (|D| + k + \|\CS_\glo\| + \|\CS_\loc\|)^{\CO(k)}.\]
 
 We compute the values $c_t[f,\CP,\bar q,\bar g]$ in a bottom-up fashion along the tree decomposition $(T,\beta)$.
 We first set
 \[c_t[c_\emptyset,\emptyset,(0)_{\Fc \in \CS_\glo},(c_\emptyset)_{\Fc \in \CS_\loc}] = 1\]
 for all leaf nodes $t$ where $c_\emptyset$ denotes the empty function (i.e., the unique function with an empty domain).
 All other entries for leaf nodes $t$ are set to $0$.
 
 So let $t$ be an internal node and suppose the algorithm already computed all values $c_s[f,\CP,\bar q,\bar g]$ for all children $s$ of $t$.
 
 \begin{description}
  \item[Forget:] Let $s$ be the unique child of $t$ and suppose $X_s = X_t \cup \{v\}$.
   Let $(f,\CP,(q_\Fc)_{\Fc \in \CS_\glo},(g_\Fc)_{\Fc \in \CS_\loc}) \in A_t$ and $(f',\CP',(q_\Fc')_{\Fc \in \CS_\loc},(g_\Fc')_{\Fc \in \CS_\glo}) \in A_s$.
   
   We say \emph{$(f',\CP',(q_\Fc')_{\Fc \in \CS_\glo},(g_\Fc')_{\Fc \in \CS_\loc})$ is consistent with $(f,\CP,(q_\Fc)_{\Fc \in \CS_\glo},(g_\Fc)_{\Fc \in \CS_\loc})$} if the following conditions are satisfied:
   \begin{itemize}
    \item $f(u) = f'(u)$ for all $u \in X_t$,
    \item $(q_\Fc')_{\Fc \in \CS_\glo} = (q_\Fc)_{\Fc \in \CS_\glo}$,
    \item if $\{v\} \in \CP'$ then
     \begin{itemize}
      \item $\Big(g_\Fc'(\{v\}) + w_\CD(D_v),q\Big) \in \op$ for every $\Fc = (w,w_\CD,q,\op) \in \CS_\loc$ where $D_v \in \CD$ is the unique element such that $f(v) \in D_v$,
      \item $\CP = \CP' \setminus \{\{v\}\}$,
      \item $g_\Fc(P) = g_\Fc'(P)$ for every $P \in \CP$ and every $\Fc \in \CS_\loc$,
     \end{itemize}
    \item if $\{v\} \notin \CP'$ and $P \in \CP'$ is the unique element such that $v \in P$ then
     \begin{itemize}
      \item $\CP = (\CP' \setminus \{P\}) \cup \{P \setminus \{v\}\}$,
      \item $g_\Fc(P') = g_\Fc'(P')$ for every $P' \in \CP'$ such that $P' \neq P$ and every $\Fc \in \CS_\loc$, and
      \item $g_\Fc(P \setminus \{v\}) = g_\Fc'(P)$ for every $\Fc \in \CS_\loc$.
     \end{itemize}
   \end{itemize}
   
   We set
   \[c_t[f,\CP,(q_\Fc)_{\Fc \in \CS_\glo},(g_\Fc)_{\Fc \in \CS_\loc}] = 1\]
   if there is some consistent tuple $(f',\CP',(q_\Fc')_{\Fc \in \CS_\loc},(g_\Fc')_{\Fc \in \CS_\glo}) \in A_s$ for which
   \[c_s[f',\CP',(q_\Fc')_{\Fc \in \CS_\glo},(g_\Fc')_{\Fc \in \CS_\loc}] = 1.\]
    
  \item[Introduce:] Let $s$ be the unique child of $t$ and suppose $v$ is introduced at $t$, i.e., $X_t = X_s \cup \{v\}$.
    Let $(f,\CP,(q_\Fc)_{\Fc \in \CS_\glo},(g_\Fc)_{\Fc \in \CS_\loc}) \in A_t$ and $(f',\CP',(q_\Fc')_{\Fc \in \CS_\loc},(g_\Fc')_{\Fc \in \CS_\glo}) \in A_s$.
   
   We say \emph{$(f',\CP',(q_\Fc')_{\Fc \in \CS_\glo},(g_\Fc')_{\Fc \in \CS_\loc})$ is consistent with $(f,\CP,(q_\Fc)_{\Fc \in \CS_\glo},(g_\Fc)_{\Fc \in \CS_\loc})$} if the following conditions are satisfied:
    \begin{itemize}
     \item $f$ satisfies $\Gamma[X_t]$,
     \item $f(u) = f'(u)$ for all $u \in X_s$,
     \item $q_\Fc = \min\Big(q+1,q_\Fc' + w(v,f(v))\Big)$ for all $\Fc = (w,q,\op) \in \CS_\glo$,
     \item if $P \in \CP$ is the unique element for which $v \in P$ then there are $P_1',\dots,P_\ell' \in \CP'$ such that
      \begin{itemize}
       \item $\CP = (\CP' \setminus \{P_1',\dots,P_\ell'\}) \cup \{P\}$,
       \item $P = \{v\} \cup \bigcup_{i \in [\ell]}P_i'$,
       \item $\{P_1',\dots,P_\ell'\} = \{P' \in \CP' \mid \exists u \in P'\colon vu \in E(G) \wedge f(v)f(u) \in F\}$,
       \item $f(u) \in D_v$ for every $u \in P$ where $D_v \in \CD$ is the unique set such that $f(v) \in D_v$,
       \item $g_\Fc(P') = g_\Fc'(P')$ for every $P' \in \CP$ such that $P' \neq P$ and every $\Fc \in \CS_\loc$, and
       \item
       \[g_\Fc(P) = \min\Big(q+1,w(v,f(v)) + \sum_{i = 1}^\ell g_\Fc'(P_i')\Big)\] for every $\Fc = (w,w_\CD,q,\op) \in \CS_\loc$.
      \end{itemize}
    \end{itemize}
    As before, we set
   \[c_t[f,\CP,(q_\Fc)_{\Fc \in \CS_\glo},(g_\Fc)_{\Fc \in \CS_\loc}] = 1\]
   if there is some consistent tuple $(f',\CP',(q_\Fc')_{\Fc \in \CS_\loc},(g_\Fc')_{\Fc \in \CS_\glo}) \in A_s$ for which
   \[c_s[f',\CP',(q_\Fc')_{\Fc \in \CS_\glo},(g_\Fc')_{\Fc \in \CS_\loc}] = 1.\]
   
   \item[Join:] Let $s_1,s_2$ be the children of $t$.
    Note that $X_t = X_{s_1} = X_{s_2}$.
    Let $(f,\CP,(q_\Fc)_{\Fc \in \CS_\glo},(g_\Fc)_{\Fc \in \CS_\loc}) \in A_t$ and $(f_1,\CP_1,(q_\Fc^1)_{\Fc \in \CS_\loc},(g_\Fc^1)_{\Fc \in \CS_\glo}) \in A_{s_1}$ and $(f_2,\CP_2,(q_\Fc^2)_{\Fc \in \CS_\loc},(g_\Fc^2)_{\Fc \in \CS_\glo}) \in A_{s_2}$.
   
    We say \emph{$(f_1,\CP_1,(q_\Fc^1)_{\Fc \in \CS_\glo},(g_\Fc^1)_{\Fc \in \CS_\loc})$ and $(f_2,\CP_2,(q_\Fc^2)_{\Fc \in \CS_\glo},(g_\Fc^2)_{\Fc \in \CS_\loc})$ are consistent with $(f,\CP,(q_\Fc)_{\Fc \in \CS_\glo},(g_\Fc)_{\Fc \in \CS_\loc})$} if the following conditions are satisfied:
    \begin{itemize}
     \item $f(u) = f_1(u) = f_2(u)$ for all $u \in X_t$,
     \item $\CP = \CP_1 \vee \CP_2$, i.e., $\CP$ is the unique finest partition of $X_t$ such that $\CP_1$ and $\CP_2$ refine $\CP$,
     \item for every $P \in \CP$ there is some $D_P \in \CD$ such that $f(u) \in D_P$ for every $u \in P$,
     \item $q_\Fc = \min\Big(q+1,q_\Fc^1 + q_\Fc^2 - \sum_{u \in X_t} w(u,f(u))\Big)$ for all $\Fc = (w,q,\op) \in \CS_\glo$, and
     \item
      \[g_\Fc(P) = \min\Big(q+1, \sum_{P_1 \in \CP_1\colon P_1 \subseteq P} g_\Fc^1(P_1) + \sum_{P_2 \in \CP_2\colon P_2 \subseteq P} g_\Fc^2(P_2) - \sum_{u \in P}w(u,f(u))\Big)\]
      for all $\Fc = (w,w_\CD,q,\op) \in \CS_\loc$ and all $P \in \CP$.
    \end{itemize}
    We set
    \[c_t[f,\CP,(q_\Fc)_{\Fc \in \CS_\glo},(g_\Fc)_{\Fc \in \CS_\loc}] = 1\]
   if there are tuples $(f_i,\CP_i,(q_\Fc^i)_{\Fc \in \CS_\loc},(g_\Fc^i)_{\Fc \in \CS_\glo}) \in A_{s_i}$, $i \in \{1,2\}$, that are consistent with $(f,\CP,(q_\Fc)_{\Fc \in \CS_\glo},(g_\Fc)_{\Fc \in \CS_\loc})$ and for which
   \[c_s[f_i,\CP_i,(q_\Fc^i)_{\Fc \in \CS_\glo},(g_\Fc^i)_{\Fc \in \CS_\loc}] = 1.\]
 \end{description}
  
 To complete the description, the algorithm outputs \yes\ if there is a tuple
 \[(c_\emptyset,\emptyset,(q_\Fc)_{\Fc \in \CS_\glo},(c_\emptyset)_{\Fc \in \CS_\loc}) \in A_r\]
 for which
 \[c_r[c_\emptyset,\emptyset,(q_\Fc)_{\Fc \in \CS_\glo},(c_\emptyset)_{\Fc \in \CS_\loc}] = 1\]
 where $r$ denotes the root node of $T$, such that
 \[(q_\Fc,q) \in \op\]
 for all $\Fc = (w,q,\op) \in \CS_\glo$.
 
 Let us first analyze the running time of the algorithm.
 First, we can compute a nice tree decomposition of $G$ of width $k = \CO(\tw(G))$ in time $2^{\CO(k)}|X|^{\CO(1)}$.
 For every node $t$, we need to compute $|A_t|$ many values.
 If $t$ is an introduce-node or a forget-node, this can be done naively in time $\CO(|A_t| \cdot |A_s| \cdot k^{\CO(1)})$ where $s$ denotes the unique child of $t$.
 Similarly, if $t$ is a join-node, all entries can be computed in time $\CO(|A_t| \cdot |A_{s_1}| \cdot |A_{s_2}| \cdot k^{\CO(1)})$ where $s_1,s_2$ denote the children of $t$.
 Overall, this means that the algorithm runs in time $(|D| + k + \|\CS_\glo\| + \|\CS_\loc\|)^{\CO(k)}|X|^{\CO(1)}$ as desired.
 
 Hence, it remains to prove the correctness.
 We need to argue that $c_t[f,\CP,(q_\Fc)_{\Fc \in \CS_\glo},(g_\Fc)_{\Fc \in \CS_\loc}] = 1$ if and only if there is an assignment $\alpha \colon V_t \rightarrow D$ satisfying Conditions \ref{item:tw-dp-1} - \ref{item:tw-dp-7}.
 This is proved by induction on the structure of the tree $T$.
 It is easy to check that this condition is satisfied for all leaves $t$ of $T$.
 So fix an internal node $t \in V(T)$.
 
 \begin{description}
  \item[Forget:] First assume $t$ is a forget-node.
   Let $s$ be the unique child of $t$, and suppose that $X_s = X_t \cup \{v\}$.
   Observe that $V_t = V_s$.
   Let $(f,\CP,(q_\Fc)_{\Fc \in \CS_\glo},(g_\Fc)_{\Fc \in \CS_\loc}) \in A_t$.
   
   First suppose there is an assignment $\alpha \colon V_t \rightarrow D$ satisfying Conditions \ref{item:tw-dp-1} - \ref{item:tw-dp-7} with respect to node $t$ and the tuple $(f,\CP,(q_\Fc)_{\Fc \in \CS_\glo},(g_\Fc)_{\Fc \in \CS_\loc})$.
   Let $C_1,\dots,C_\ell$ denote the connected components of $H_{t,\alpha}$.
   Let $f' \coloneqq \alpha|_{X_s}$ and $\CP' \coloneqq \{C_i \cap X_s \mid i \in [\ell], C_i \cap X_s \neq \emptyset\}$.
   Also, define $q_\Fc' \coloneqq \min\Big(\sum_{u \in V_s} w(u,\alpha(u)),q+1\Big)$ for all $\Fc = (w,q,\op) \in \CS_\glo$.
   And finally, for every $\Fc = (w,w_\CD,q,\op) \in \CS_\loc$, we define $g_\Fc'\colon \CP' \rightarrow [0,q+1]$ via
   \[g_\Fc'(C_i \cap X_s) \coloneqq \min\Big(q+1,\sum_{u \in C_i} w(u,\alpha(u))\Big)\]
   for every $i \in [\ell]$ such that $C_i \cap X_s \neq \emptyset$.
  
   It is easy to see that $\alpha$ satisfies Conditions \ref{item:tw-dp-1} - \ref{item:tw-dp-7} with respect to node $s$ and the tuple $(f',\CP',(q_\Fc')_{\Fc \in \CS_\glo},(g_\Fc')_{\Fc \in \CS_\loc})$.
   Hence, by the induction hypothesis,
   the algorithm sets
   \[c_s[f',\CP',(q_\Fc')_{\Fc \in \CS_\glo},(g_\Fc')_{\Fc \in \CS_\loc}] = 1.\]
   Moreover, using Conditions \ref{item:tw-dp-1} - \ref{item:tw-dp-7}, it is easy to see that $(f',\CP',(q_\Fc')_{\Fc \in \CS_\glo},(g_\Fc')_{\Fc \in \CS_\loc})$ is consistent with $(f,\CP,(q_\Fc)_{\Fc \in \CS_\glo},(g_\Fc)_{\Fc \in \CS_\loc})$.
   Indeed, the only non-trivial case occurs if $\{v\} \in \CP'$.
   Pick $i \in [\ell]$ such that $v \in C_i$.
   Then
   \[g_\Fc'(\{v\}) + w_\CD(D_v) = \min\Big(q+1,\sum_{u \in C_i} w(u,\alpha(u))\Big) + w_\CD(D_v))\]
   and
   \[\Big(\sum_{u \in C_i} w(u,\alpha(u))\big) + w_\CD(D_v),q\Big) \in \op\]
   for every $(w,w_\CD,q,\op) \in \CD_\loc$.
   But then
   $\Big(g_\Fc'(\{v\}) + w_\CD(D_v),q\Big) \in \op$
   which ensures consistency.
   So we set
   \[c_t[f,\CP,(q_\Fc)_{\Fc \in \CS_\glo},(g_\Fc)_{\Fc \in \CS_\loc}] = 1\]
   as desired.
   
   \medskip
   
   For the other direction, suppose we set $c_t[f,\CP,(q_\Fc)_{\Fc \in \CS_\glo},(g_\Fc)_{\Fc \in \CS_\loc}] = 1$.
   This means there is a tuple $(f',\CP',(q_\Fc')_{\Fc \in \CS_\loc},(g_\Fc')_{\Fc \in \CS_\glo}) \in A_s$ that is consistent with $(f,\CP,(q_\Fc)_{\Fc \in \CS_\glo},(g_\Fc)_{\Fc \in \CS_\loc})$ and for which
   \[c_s[f',\CP',(q_\Fc')_{\Fc \in \CS_\glo},(g_\Fc')_{\Fc \in \CS_\loc}] = 1.\]
   Hence, by the induction hypothesis there is an assignment $\alpha\colon V_s \rightarrow D$ that satisfies Conditions \ref{item:tw-dp-1} - \ref{item:tw-dp-7} with respect to node $s$ and the tuple $(f',\CP',(q_\Fc')_{\Fc \in \CS_\glo},(g_\Fc')_{\Fc \in \CS_\loc})$.
   Using consistency, it is easy to verify that $\alpha$ also satisfies Conditions \ref{item:tw-dp-1} - \ref{item:tw-dp-7} with respect to node $t$ and the tuple $(f,\CP,(q_\Fc)_{\Fc \in \CS_\glo},(g_\Fc)_{\Fc \in \CS_\loc})$.
   Indeed, if $C_1,\dots,C_\ell$ are the connected components $H_{t,\alpha}$, the only non-trivial part is verifying Condition \ref{item:tw-dp-6} for the component $C_i$ such that $v \in C_i$ assuming that $\{v\} \in \CP'$.
   But this follows by the same calculation as above.
  \item[Introduce:] Next, assume $t$ is an introduce-node.
   Let $s$ be the unique child of $t$, and suppose that $X_t = X_s \cup \{v\}$.
   Let $(f,\CP,(q_\Fc)_{\Fc \in \CS_\glo},(g_\Fc)_{\Fc \in \CS_\loc}) \in A_t$.
   
   First suppose there is an assignment $\alpha \colon V_t \rightarrow D$ satisfying Conditions \ref{item:tw-dp-1} - \ref{item:tw-dp-7} with respect to node $t$ and the tuple $(f,\CP,(q_\Fc)_{\Fc \in \CS_\glo},(g_\Fc)_{\Fc \in \CS_\loc})$.
   Let $C_1,\dots,C_\ell$ denote the connected components of $H_{t,\alpha}$.
   Let $\alpha' \coloneqq \alpha|_{V_s}$ denote the restriction of $\alpha$ to $V_s$.
   Let $C_1',\dots,C_{\ell'}'$ denote the connected components of $H_{s,\alpha'}$.
   Observe that $H_{s,\alpha'} = H_{t,\alpha} - \{v\}$.
   Without loss of generality suppose that $v \in C_1$.
   Then
   \[C_1 = \{v\} \cup \bigcup_{i \in [\ell']\colon E_{H_{s,\alpha'}}(\{v\},C_i') \neq \emptyset} C_i'.\]
   Now, let $f' \coloneqq \alpha'|_{X_s} = \alpha|_{X_s}$ and $\CP' \coloneqq \{C_i' \cap X_s \mid i \in [\ell'], C_i' \cap X_s \neq \emptyset\}$.
   Also, define $q_\Fc' \coloneqq \min\Big(\sum_{u \in V_s} w(u,\alpha'(u)),q+1\Big)$ for all $\Fc = (w,q,\op) \in \CS_\glo$.
   And finally, for every $\Fc = (w,w_\CD,q,\op) \in \CS_\loc$, we define $g_\Fc'\colon \CP' \rightarrow [0,q+1]$ via
   \[g_\Fc'(C_i' \cap X_s) \coloneqq \min\Big(q+1,\sum_{u \in C_i'} w(u,\alpha'(u))\Big)\]
   for every $i \in [\ell']$ such that $C_i' \cap X_s \neq \emptyset$.
   
   It is easy to see that $\alpha'$ satisfies Conditions \ref{item:tw-dp-1} - \ref{item:tw-dp-7} with respect to node $s$ and the tuple $(f',\CP',(q_\Fc')_{\Fc \in \CS_\glo},(g_\Fc')_{\Fc \in \CS_\loc})$.
   Hence, by the induction hypothesis,
   the algorithm sets
   \[c_s[f',\CP',(q_\Fc')_{\Fc \in \CS_\glo},(g_\Fc')_{\Fc \in \CS_\loc}] = 1.\]
   Moreover, using Conditions \ref{item:tw-dp-1} - \ref{item:tw-dp-7} as well as the above comments, it is easy to verify that $(f',\CP',(q_\Fc')_{\Fc \in \CS_\glo},(g_\Fc')_{\Fc \in \CS_\loc})$ is consistent with $(f,\CP,(q_\Fc)_{\Fc \in \CS_\glo},(g_\Fc)_{\Fc \in \CS_\loc})$.
   So we set
   \[c_t[f,\CP,(q_\Fc)_{\Fc \in \CS_\glo},(g_\Fc)_{\Fc \in \CS_\loc}] = 1\]
   as desired.
   
   \medskip
   
   For the other direction, suppose we set $c_t[f,\CP,(q_\Fc)_{\Fc \in \CS_\glo},(g_\Fc)_{\Fc \in \CS_\loc}] = 1$.
   This means there is a tuple $(f',\CP',(q_\Fc')_{\Fc \in \CS_\loc},(g_\Fc')_{\Fc \in \CS_\glo}) \in A_s$ that is consistent with $(f,\CP,(q_\Fc)_{\Fc \in \CS_\glo},(g_\Fc)_{\Fc \in \CS_\loc})$ and for which
   \[c_s[f',\CP',(q_\Fc')_{\Fc \in \CS_\glo},(g_\Fc')_{\Fc \in \CS_\loc}] = 1.\]
   Hence, by the induction hypothesis there is an assignment $\alpha'\colon V_s \rightarrow D$ that satisfies Conditions \ref{item:tw-dp-1} - \ref{item:tw-dp-7} with respect to node $s$ and the tuple $(f',\CP',(q_\Fc')_{\Fc \in \CS_\glo},(g_\Fc')_{\Fc \in \CS_\loc})$.
   We define $\alpha\colon V_t \rightarrow D$ via $\alpha(v) \coloneqq f(v)$ and $\alpha(u) \coloneqq \alpha'(u)$ for all $u \in X_s$.
   
   Using consistency and similar arguments as above, it is easy to verify that $\alpha$ satisfies Conditions \ref{item:tw-dp-1} - \ref{item:tw-dp-7} with respect to node $t$ and the tuple $(f,\CP,(q_\Fc)_{\Fc \in \CS_\glo},(g_\Fc)_{\Fc \in \CS_\loc})$.
  \item[Join:] Finally, suppose $t$ is a join node and let $s_1,s_2$ be the two children of $t$.
   Note that $X_t = X_{s_1} = X_{s_2}$.
   
   Let $(f,\CP,(q_\Fc)_{\Fc \in \CS_\glo},(g_\Fc)_{\Fc \in \CS_\loc}) \in A_t$.
   
   First suppose there is an assignment $\alpha \colon V_t \rightarrow D$ satisfying Conditions \ref{item:tw-dp-1} - \ref{item:tw-dp-7} with respect to node $t$ and the tuple $(f,\CP,(q_\Fc)_{\Fc \in \CS_\glo},(g_\Fc)_{\Fc \in \CS_\loc})$.
   Let $\alpha_1 \coloneqq \alpha|_{V_{s_1}}$ and $\alpha_2 \coloneqq \alpha|_{V_{s_2}}$ denote the restrictions of $\alpha$ to $V_{s_1}$ and $V_{s_2}$.
   Also, let $f_i \coloneqq \alpha|_{X_{s_i}}$ for both $i \in \{1,2\}$.
   Next, let $C_1^i,\dots,C_{\ell_i}^i$ denote the connected components of $H_{s_i,\alpha_i}$ for both $i \in \{1,2\}$.
   We define $\CP_i \coloneqq \{X_{s_i} \cap C_j^i \mid j \in [\ell_i], X_{s_i} \cap C_j^i \neq \emptyset\}$.
   Also, let $C_1,\dots,C_\ell$ denote the connected components of $H_{t,\alpha}$.
   Recall that $\CP = \{C_i \cap X_t \mid i \in [\ell], C_i \cap X_t \neq \emptyset\}$ by Property \ref{item:tw-dp-3}.
   Since $V_{s_1} \cap V_{s_2} = X_t$ we get that $\CP = \CP_1 \vee \CP_2$.
   
   Next, let $q_\Fc^i \coloneqq \min\Big(\sum_{u \in V_{s_i}} w(u,\alpha(u)),q+1\Big)$ for all $\Fc = (w,q,\op) \in \CS_\glo$.
   We have that
   \[\sum_{u \in V_{s_2}} w(u,\alpha(u)) + \sum_{u \in V_{s_2}} w(u,\alpha(u)) = \sum_{u \in V_t} w(u,\alpha(u)) + \sum_{u \in X_t} w(u,\alpha(u))\]
   which implies that
   \[q_\Fc = \min\Big(q+1,q_\Fc^1 + q_\Fc^2 - \sum_{u \in X_t} w(u,f(u))\Big)\]
   using Condition \ref{item:tw-dp-4}.
   
   Finally, for every $\Fc = (w,w_\CD,q,\op) \in \CS_\loc$, we define $g_\Fc^i\colon \CP_i \rightarrow [0,q+1]$ via
   \[g_\Fc^i(C_j^i \cap X_{s_i}) \coloneqq \min\Big(q+1,\sum_{u \in C_j^i} w(u,\alpha_i(u))\Big)\]
   for every $j \in [\ell_i]$ such that $C_j^i \cap X_{s_i} \neq \emptyset$.
   We have that
   \[g_\Fc(P) = \min\Big(q+1, \sum_{P_1 \in \CP_1\colon P_1 \subseteq P} g_\Fc^1(P_1) + \sum_{P_2 \in \CP_2\colon P_2 \subseteq P} g_\Fc^2(P_2) - \sum_{u \in P}w(u,f(u))\Big)\]
   for all $P \in \CP$.
   
   Overall, we get that $(f_1,\CP_1,(q_\Fc^1)_{\Fc \in \CS_\glo},(g_\Fc^1)_{\Fc \in \CS_\loc})$ and $(f_2,\CP_2,(q_\Fc^2)_{\Fc \in \CS_\glo},(g_\Fc^2)_{\Fc \in \CS_\loc})$ are consistent with $(f,\CP,(q_\Fc)_{\Fc \in \CS_\glo},(g_\Fc)_{\Fc \in \CS_\loc})$.
   Also, $\alpha_i$ satisfies Conditions \ref{item:tw-dp-1} - \ref{item:tw-dp-7} with respect to node $s_i$ and the tuple $(f_i,\CP_i,(q_\Fc^i)_{\Fc \in \CS_\glo},(g_\Fc^i)_{\Fc \in \CS_\loc})$.
   By the induction hypothesis, this implies that
   \[c_{s_i}[f_i,\CP_i,(q_\Fc^i)_{\Fc \in \CS_\glo},(g_\Fc^i)_{\Fc \in \CS_\loc}] = 1\]
   for both $i \in \{1,2\}$.
   So the algorithms sets
   \[c_t[f,\CP,(q_\Fc)_{\Fc \in \CS_\glo},(g_\Fc)_{\Fc \in \CS_\loc}] = 1\]
   as desired.
   
   \medskip
   
   For the other direction, suppose we set $c_t[f,\CP,(q_\Fc)_{\Fc \in \CS_\glo},(g_\Fc)_{\Fc \in \CS_\loc}] = 1$.
   This means there are tuples $(f_1,\CP_1,(q_\Fc^1)_{\Fc \in \CS_\glo},(g_\Fc^1)_{\Fc \in \CS_\loc}) \in A_{s_1}$ and $(f_2,\CP_2,(q_\Fc^2)_{\Fc \in \CS_\glo},(g_\Fc^2)_{\Fc \in \CS_\loc}) \in A_{s_2}$ that are consistent with $(f,\CP,(q_\Fc)_{\Fc \in \CS_\glo},(g_\Fc)_{\Fc \in \CS_\loc})$ and such that
   \[c_{s_i}[f_i,\CP_i,(q_\Fc^i)_{\Fc \in \CS_\glo},(g_\Fc^i)_{\Fc \in \CS_\loc}] = 1\]
   for both $i \in \{1,2\}$.
   By the induction hypothesis, there are assignments $\alpha_i\colon V_{s_i} \rightarrow D$, $i \in \{1,2\}$, that satisfy Conditions \ref{item:tw-dp-1} - \ref{item:tw-dp-7} with respect to node $s_i$ and $(f_i,\CP_i,(q_\Fc^i)_{\Fc \in \CS_\glo},(g_\Fc^i)_{\Fc \in \CS_\loc})$.
   We define $\alpha\colon V_t \rightarrow D$ via $\alpha(v) \coloneqq \alpha_1(v)$ for all $v \in V_{s_1}$ and $\alpha(v) \coloneqq \alpha_2(v)$ for all $v \in V_{s_2}$.
   Observe that this is well-defined since $\alpha_1(v) = \alpha_2(v)$ for all $v \in V_{s_1} \cap V_{s_2} = X_t$ by consistency and Property \ref{item:tw-dp-2}.
   Using similar arguments as for the other direction, it can be verified that $\alpha$ satisfies Conditions \ref{item:tw-dp-1} - \ref{item:tw-dp-7} with respect to node $t$ and the tuple $(f,\CP,(q_\Fc)_{\Fc \in \CS_\glo},(g_\Fc)_{\Fc \in \CS_\loc})$.
 \end{description}
 To complete the proofs we observe that $(\Gamma,\CS_\glo,\CS_\loc)$ is a \yes-instance if and only if there is a tuple
 \[(c_\emptyset,\emptyset,(q_\Fc)_{\Fc \in \CS_\glo},(c_\emptyset)_{\Fc \in \CS_\loc}) \in A_r\]
 for which
 \[c_r[c_\emptyset,\emptyset,(q_\Fc)_{\Fc \in \CS_\glo},(c_\emptyset)_{\Fc \in \CS_\loc}] = 1\]
 where $r$ denotes the root node of $T$, such that
 \[(q_\Fc,q) \in \op\]
 for all $\Fc = (w,q,\op) \in \CS_\glo$.
\end{proof}

We also need a second result for binary CSPs on graphs of bounded treewidth.
Let $\Gamma = (X,D,\CC)$ be a binary CSP and let $\alpha\colon X \rightarrow D$ be an assignment.
Also, let $w \colon X^2 \times D^2 \rightarrow \NN$ such that $w(x,y,a,b) = w(y,x,b,a)$ for all $x,y \in X$ and $a,b \in D$.
We define the \emph{cost of $\alpha$} as
\[\cost_\Gamma(\alpha) \coloneqq \sum_{x,y \in X \colon \alpha \text{ violates some constraint } ((x,y),R) \in \CC} w(x,y,\alpha(x),\alpha(y)).\]
We remark that the sum also covers unary constraints by setting $x = y$.
Implementing a simple dynamic program along the structure of a nice tree decomposition (similar to the previous theorem), we obtain the following result.

\begin{theorem}
 \label{thm:tw-binary-csp-cost-bound}
 Let $\Gamma = (X,D,\CC)$ be a binary CSP instance and $w \colon X^2 \times D^2 \rightarrow \NN$ a weight function such that $w(x,y,a,b) = w(y,x,b,a)$ for all $x,y \in X$ and $a,b \in D$, and $m \geq 0$.
 Also define $k \coloneqq \tw(G(\Gamma))$.
 Then there is an algorithm which decides whether there is an assignment $\alpha\colon X \rightarrow D$ of cost $\cost_\Gamma(\alpha) \leq m$ in time
 \[|D|^{\CO(k)}|X|^{\CO(1)}.\]
\end{theorem}

\subsection{Permutation CSP Edge Deletion}
\label{subsec:perm-csp-edge-deletion}

Now, we are ready to prove the main algorithmic results from Section \ref{sec:perm-csp-results}.
In this subsection, we start by proving Theorem \ref{thm:perm-csp-edge-deletion-result}.
We restate the problem as well as the main result.

\medskip 
\defparproblem{\textsc{Perm CSP Edge Deletion}}{A Permutation-CSP-instance $\Gamma = (X,D,\CC)$, a set of undeletable edges $U \subseteq E(G(\Gamma))$, and an integer $k$}{$k$}
 {Is there a set $Z \subseteq E(G(\Gamma)) \setminus U$ such that $|Z| \leq k$ and $\Gamma - Z \coloneqq (X,D,\CC - Z)$ is satisfiable?}
\medskip

\begin{theorem}
 There is an algorithm solving \textsc{$H$-Minor-Free Perm CSP Edge Deletion} in time $(|X| + |D|)^{\CO(c_H \sqrt{k})}$ where $c_H$ is a constant depending only on $H$.
\end{theorem}

\begin{proof}
 Let $\Gamma = (X,D,\CC)$ be a Permutation-CSP-instance and let $G \coloneqq G(\Gamma)$.
 By combining constraints over the same variables, we may assume without loss of generality that, for every $uv \in E(G)$, there is at most one constraint over variables $u$ and $v$.
 Also suppose $Z \subseteq E(G) \setminus U$ is a solution, i.e., $|Z| \leq k$ and $\Gamma - Z = (X,D,\CC - Z)$ is satisfiable.
 Let $\ell \coloneqq \left\lceil\sqrt{k}\right\rceil$.
 Let $E_1,\dots,E_\ell$ be the partition of the edge $E(G)$ computed by Corollary \ref{cor:contraction-decomposition-minor-edge}.
 Then there is some $i \in [\ell]$ such that $|E_i \cap Z| \leq \sqrt{k}$.
 Hence, at an additional multiplicative cost of $\ell |E(G)|^{\sqrt{k}} = |X|^{\CO(\sqrt{k})}$, the algorithm can guess $i$ as well as the set $E_i \cap Z$.
 We set $\Gamma' \coloneqq (X,D,\CC - (E_i \cap Z))$ and $k' \coloneqq k - |E_i \cap Z|$.
 Hence, it suffices to determine whether $(\Gamma',k')$ is a \yes-instance.
 
 Let $I_1,\dots,I_m$ denote the connected components of $G[E_i \setminus Z]$ and let $\CI \coloneqq \{I_1,\dots,I_m\}$.
 (Recall that $G[E_i \setminus Z]$ denotes the graph with vertex set $\bigcup_{e \in E_i \setminus Z} e$ and edge set $E_i \setminus Z$.)
 Also, for each $I \in \CI$, we fix an arbitrary linear order $<_I$ on the set $I$.
 This gives rise to a segmented graph $(G,\CI)$.
 We have that $\tw(G/\CI) \leq c_H \sqrt{k}$ by Corollary \ref{cor:contraction-decomposition-minor-edge} for some constant $c_H$ that only depends on $H$.
 Now, the central idea is to translate $\Gamma'$ to a binary CSP instance over the graph $G/\CI$.
 
 Let $\Gamma^* \coloneqq (X^*,D,\CC^*)$ where $X^* \coloneqq V(G/\CI)$, and a set of constraints $\CC^*$ defined as follows.
 For $a \in D$ and $I \in \CI$, we define $\alpha_{I,a} \colon I \rightarrow D$ to be a satisfying assignment such that
 \begin{itemize}
  \item $\alpha_{I,a}(v) = a$ where $v$ denotes the minimal element of $I$ with respect to $<_I$, and
  \item $\alpha_{I,a}$ satisfies all constraints $((u,v),R) \in \CC$ for which $u,v \in I$ and $u = v$ or $uv \in (E_i \setminus Z) \cup U$.
 \end{itemize}
 If such an assignment exists, we say that \emph{$a$ is valid for $I$}.
 Observe that, if $a$ is valid for $I$, then the assignment $\alpha_{I,a}$ is unique since $\Gamma$ is a Permutation-CSP-instance. 
 To bound the number of violated constraints, we also define a weight function $w \colon X^2 \times D^2 \rightarrow \NN$.
 
 For every $I \in \CI$, we introduce a unary constraint $(v_I,R_I)$ to $\CC^*$ where
 \[R_I = \{a \in D \mid a \text{ is valid for } I\}.\]
 For every $a \in R_I$ we define $w(v_I,v_I,a,a)$ to be the number of edges which correspond to some constraint that is violated by $\alpha_{a,I}$ on $\Gamma[I]$.
 
 For every $v \in X \setminus \left(\bigcup_{I \in \CI}I \right)$ and unary constraint $(v,R) \in \CC$, the constraint $(v,R)$ is added to $\CC^*$.
 Also, we set $w(v,v,a,a) \coloneqq k' + 1$ (indicating that this constraint can not be violated).
 Now, let $((u,v),R) \in \CC$ be a binary constraint.
 If $u,v \in X \setminus \left(\bigcup_{I \in \CI}I \right)$ then we add the constraint $((u,v),R)$ to $\CC^*$.
 If $uv \in U$ then we set $w(u,v,a,b) \coloneqq k' + 1$ for all $a,b \in D$.
 Otherwise, $w(u,v,a,b) \coloneqq 1$.
 If $u \in X \setminus \left(\bigcup_{I \in \CI}I \right)$ and $v \in I$ for some $I \in \CI$ then we add the constraint
 \[\Big((u,v_I),\{(a,a') \mid a \in D, a' \in R_I, (a,\alpha_{a',I}(v)) \in R\}\Big).\]
 Also, $w(u,v_I,a,b)$ denotes the number of such constraints over variables $u$ and $v_I$ that are violated when assigning $a$ with $u$ and $b$ to $v_I$.
 If some undeletable constraint is violated, then we set $w(u,v_I,a,b) \coloneqq k' + 1$.
 If $v \in X \setminus \left(\bigcup_{I \in \CI}I \right)$ and $u \in I$ we proceed analogously.
 Finally, if $u \in I$ and $v \in I'$ for some distinct $I,I' \in \CI$ then we add the constraint
 \[\Big((v_I,v_{I'}),\{(a,a') \mid a \in R_I, a' \in R_{I'}, (\alpha_{a,I}(v),\alpha_{a',I'}(w)) \in R\}\Big).\]
 Again, $w(u,v_I,a,b)$ denotes the number of such constraints over variables $u$ and $v_I$ that are violated when assigning $a$ with $u$ and $b$ to $v_I$.
 If some undeletable constraint is violated, then we set $w(u,v_I,a,b) \coloneqq k' + 1$.
 
 Now, it is easy to see that $(\Gamma',k')$ is a \yes-instance if and only if there is an assignment $\alpha^*\colon X^* \rightarrow D$ of cost $\cost_{\Gamma^*}(\alpha^*) \leq k'$.
 By Theorem \ref{thm:tw-binary-csp-cost-bound}, the latter can be checked in time $|D|^{\CO(\tw(G/\CI))}|X|^{\CO(1)} = |D|^{\CO(c_H \sqrt{k})}|X|^{\CO(1)}$.
 In total, this gives a running time of 
 \[|X|^{\CO(\sqrt{k})} \cdot |D|^{\CO(c_H \sqrt{k})}|X|^{\CO(1)} = (|X| + |D|)^{\CO(c_H \sqrt{k})}.\]
\end{proof}

\subsection{Permutation CSP Deletion with Size Constraints}
\label{subsec:perm-csp-vertex-deletion}

Next, we turn to the proof of Theorem \ref{thm:perm-csp-deletion-result}.
Again, we start by restating the problem and the main result.
Also, we provide a variant of Theorem \ref{thm:perm-csp-deletion-result} for $H$-minor-free graphs assuming Conjecture \ref{conj:contraction-decomposition-minor-vertex} holds.

Let $\Gamma = (X,D,\CC)$ be a binary CSP instance.
A \emph{1cc-size constraint} is a triple $(w,q,\op)$ where $w \colon X \times D \rightarrow \NN$ is a weight function, $q \in \NN$, and $\op \in \{\leq, \geq\}$.
An assignment $\alpha\colon X \rightarrow D$ satisfies $(w,q,\op)$ on $\Gamma$ if
\[\Big(\sum_{v \in C} w(v,\alpha(v)), q\Big) \in \op\]
for every connected component $C$ of $G(\Gamma)$.

A \emph{Permutation-CSP-instance with $1$cc-size constraints} is a pair $(\Gamma,\CS)$ where $\Gamma = (X,D,\CC)$ is a Permutation-CSP-instance and $\CS$ is a set of 1cc-size constraints.
We say that $(\Gamma,\CS)$ is \emph{satisfiable} if there is an assignment $\alpha\colon X \rightarrow D$ which satisfies $\Gamma$ as well as every constraint $(w,q,\op) \in \CS$ on $\Gamma$.
Also, we define
\[\|\CS\| \coloneqq \prod_{(w,q,\op) \in \CS} (2 + q).\]

\medskip 
\defparproblem{\textsc{Perm CSP Vertex Deletion with Size Constraints}}{A Permutation-CSP-instance with $1$cc-size constraints $(\Gamma,\CS)$ with $\Gamma = (X,D,\CC)$, a set of undeletable variables $U \subseteq X$, and an integer $k$.}{$k$}
{Is there a set $Z \subseteq X \setminus U$ such that $|Z| \leq k$ and $(\Gamma[X\setminus Z],\CS)$ is satisfiable?}
\medskip

\begin{theorem}
 There is an algorithm solving \textsc{Planar Perm CSP Vertex Deletion with Size Constraints} in time $(|X| + |D| + \|\CS\|)^{\CO(\sqrt{k})}$.
 
 Moreover, assuming Conjecture \ref{conj:contraction-decomposition-minor-vertex}, there is an algorithm solving \textsc{$H$-Minor-Free Perm CSP Vertex Deletion with Size Constraints} in time $(|X| + |D| + \|\CS\|)^{\CO(c_H\sqrt{k})}$ where $c_H$ is a constant that only depends on $H$.
\end{theorem}

\begin{proof}
 We focus on the first statement.
 Let $G = G(\Gamma)$ be the constraint graph and let $\ell \coloneqq \left\lceil\sqrt{k}\right\rceil$.
 Also suppose $Z \subseteq V(G) \setminus U = X \setminus U$ is a solution, i.e., $|Z| \leq k$ and  $(\Gamma[X\setminus Z],\CS)$ is satisfiable.
 Let $V_1,\dots,V_\ell$ be the partition of the vertex set $V(G)$ computed by Theorem \ref{thm:contraction-decomposition-planar-vertex}.
 Then there is some $i \in [\ell]$ such that
 \[|Z \cap V_i| \leq \frac{k}{\ell} \leq \sqrt{k}.\]
 Let $Z_i \coloneqq Z \cap V_i$.
 The algorithm first guesses $i \in [\ell]$ as well as the set $Z_i$ by iterating over all possible choices.
 Observe that the guessing only increases the running time by a factor of $|X|^{(k/\ell) + 1} = |X|^{\CO(\sqrt{k})}$ which does not pose any problems.
 So for the remainder of the proof, suppose that $i$ and $Z_i$ are fixed.
 
 Let $I_1,\dots,I_m$ be the connected components of $G[V_i \setminus Z_i]$ and define $\CI \coloneqq \{I_1,\dots,I_m\}$.
 Also, for each $I \in \CI$, we fix an arbitrary linear order $<_I$ on the set $I$.
 Overall, this gives rise to the segmented graph $(G,\CI)$.
 Now, the central idea is to translate the instance $(\Gamma,\CS)$ to a binary CSP instance with global and local size constraints over the graph $G/\CI$.
 
 Let $\Gamma^* \coloneqq (X^*,D^*,\CC^*)$ where $X^* \coloneqq V(G/\CI)$, $D^* \coloneqq D \uplus \{{\sf sol}\}$, and a set of constraints $\CC^*$ defined as follows.
 For $I \in \CI$ and $a \in D$, we define $\alpha_{I,a} \colon I \rightarrow D$ to be an assignment such that
 \begin{itemize}
  \item $\alpha_{I,a}(v) = a$ where $v$ denotes the minimal element of $I$ with respect to $<_I$,
  \item $\alpha_{I,a}$ satisfies all unary constraints $(v,R) \in \CC$ for which $v \in I$, and
  \item $\alpha_{I,a}$ satisfies all binary constraints $((v,w),R) \in \CC$ for which $v,w \in I$ and $vw \in E(G)$.
 \end{itemize}
 If such an assignment exists, we say that \emph{$a$ is valid for $I$}.
 Observe that, if $a$ is valid for $I$, then the assignment $\alpha_{I,a}$ is unique since $\Gamma$ is a Permutation-CSP-instance and $G[I]$ is connected.
 
 For every $I \in \CI$, we introduce a unary constraint $(v_I,R_I)$ to $\CC^*$ where
 \[R_I = \{a \in D \mid a \text{ is valid for } I\}.\]
 For every $v \in X \setminus \left(Z_i \cup \bigcup_{I \in \CI}I \right)$ and unary constraint $(v,R) \in \CC$, the constraint $(v,R \cup \{{\sf sol}\})$ is added to $\CC^*$.
 For every $v \in Z_i$, we introduce the unary constraint $(v,\{{\sf sol}\})$ to $\CC^*$.
 Also, for every $v \in U \setminus \left(\bigcup_{I \in \CI}I \right)$, we introduce the unary constraint $(v,D)$ to $\CC^*$.
 Now, let $((v,w),R) \in \CC$ be a binary constraint.
 If $v,w \in X \setminus \left(\bigcup_{I \in \CI}I \right)$ then we add the constraint
 \[\Big((v,w), R \cup (D \times \{{\sf sol}\}) \cup (\{{\sf sol}\} \times D) \cup (\{{\sf sol}\} \times \{{\sf sol}\})\Big)\]
 to $\CC^*$.
 If $v \in X \setminus \left(\bigcup_{I \in \CI}I \right)$ and $w \in I$ for some $I \in \CI$ then we add the constraint
 \[\Big((v,v_I),\{(a,a') \mid a \in D, a' \in R_I, (a,\alpha_{a',I}(w)) \in R\} \cup (\{{\sf sol}\} \times R_I)\Big)\]
 If $w \in X \setminus \left(\bigcup_{I \in \CI}I \right)$ and $v \in I$ we proceed analogously.
 Finally, if $v \in I$ and $w \in I'$ for some distinct $I,I' \in \CI$ then we add the constraint
 \[\Big((v_I,v_{I'}),\{(a,a') \mid a \in R_I, a' \in R_{I'}, (\alpha_{a,I}(v),\alpha_{a',I'}(w)) \in R\}\Big).\]
 
 Next, we define the size constraints for $\Gamma^*$.
 First, we add a global size constraint to force that at most $k$ variables $x \in X^*$ are assigned the value ${\sf sol}$.
 More precisely, we define $\CS_\glo \coloneqq \{(w_{{\sf sol}},k,\leq)\}$ where $w_{{\sf sol}} \colon X^* \times D^* \rightarrow \NN$ is defined via $w_{{\sf sol}}(v,{\sf sol}) = 1$ and $w_{{\sf sol}}(v,a) = 0$ for every $a \in D$ and $v \in X^*$.
 For the local constraints, we define $F \coloneqq \{dd' \mid d,d' \in D\} \subseteq (D^*)^2$ and $\CD \coloneqq \{D^*\}$.
 Also, we define $0_\CD$ to denote the function with domain $\CD = \{D^*\}$ that maps $D^*$ to value $0$.
 
 We translate every $1$cc-size constraint $(w,q,\op) \in \CS$ into an $F$-local constraint $(w^*,0_\CD,q,\op) \in \CS_\loc$ as follows.
 We set
 \begin{itemize}
  \item $w^*(v,{\sf sol}) \coloneqq \begin{cases}
                                     q &\text{if } \op = \geq\\
                                     0 &\text{if } \op = \leq
                                    \end{cases}$\;\;\; for all $v \in X^*$,
  \item $w^*(v,a) \coloneqq w(v,a)$ for all $a \in D$ and $v \in X^* \setminus V_\CI$, and
  \item $w^*(v_I,a) \coloneqq \begin{cases}
                               \sum_{u \in I} w(u,\alpha_{a,I}(u)) &\text{if $a$ is valid for $I$}\\
                               0                                   &\text{otherwise}
                              \end{cases}$ \;\;\;for all $a \in D$ and $v_I \in V_\CI$.
 \end{itemize}

 \begin{claim}
  $(\Gamma,\CS,U,k)$ is a \yes-instance if and only if $(\Gamma^*,\CS_\glo,\CS_\loc)$ is satisfiable.
  \proof
  First suppose $(\Gamma,\CS,U,k)$ is a \yes-instance, i.e.,
  there is a set $Z \subseteq X \setminus U$ of size $|Z| \leq k$ and an assignment $\alpha\colon X \setminus Z \rightarrow D$ that satisfies $\Gamma[X \setminus Z]$ as well as every $1$cc-size constraint $(w,q,\op) \in \CS$ on $\Gamma[X \setminus Z]$.
  Moreover, by the above assumptions, $Z \cap V_i = Z_i$.
  We construct an assignment $\alpha^*\colon X^* \rightarrow D^*$ defined via
  \begin{itemize}
   \item $\alpha^*(v_I) \coloneqq \alpha(v)$ where $v$ is the minimal element of $I$ with respect to $<_I$,
   \item $\alpha^*(v) \coloneqq {\sf sol}$ for every $v \in Z$, and
   \item $\alpha^*(v) \coloneqq \alpha(v)$ for all every $v \in X \setminus (Z \cup \bigcup_{I \in \CI}I)$.
  \end{itemize}
  We claim that $\alpha^*$ satisfies $(\Gamma^*,\CS_\glo,\CS_\loc)$.
  First of all, for every $I \in \CI$, $\alpha(v)$ is valid for $I$ where $v$ is the minimal element of $I$ with respect to $<_I$.
  In particular, the unary constraint $(v_I,R_I)$ is satisfied.
  Moreover,
  \[\alpha_{\alpha(v),I}(w) = \alpha(w)\]
  for every $w \in I$.
  This implies that $\alpha^*$ satisfies all constraints in the set $\CC^*$.
  Hence, it remains to consider the size constraints.
  Clearly, the global size constraint is satisfied $(w_{{\sf sol}},k,\leq)$ is satisfied since $|Z| \leq k$.
  So consider some constraint $(w,q,\op) \in \CS$ and its corresponding $(F,\CD)$-local constraint $(w^*,0_\CD,q,\op)$.
  Let $G^* \coloneqq G(\Gamma^*)$ be the constraint graph of $\Gamma^*$ and define $H^*$ to be the graph with $V(H^*) \coloneqq V(G^*) = X^*$ and
  \[E(H^*) \coloneqq \{vw \mid \alpha^*(v)\alpha^*(w) \in F\}\]
  Let $C_1^*,\dots,C_\ell^*$ denote the connected components of $H^*$.
  Fix $i \in [\ell]$.
  Since $\CD = \{D^*\}$ there is clearly some $D_i \in \CD$ such that $\alpha^*(v) \in D_i$ for every $v \in C_i^*$.
  So we need to show that
  \begin{equation}
   \label{eq:local-constraints-in-factor-graph-1cc}
   \Big(\sum_{v \in C_i^*} w^*(v,\alpha^*(v)),q\Big) \in \op.
  \end{equation}
  If there is some $v \in C_i$ such that $\alpha^*(v) = {\sf sol}$ then $C_i = \{v\}$ by definition of the set $F$.
  In particular, Equation \eqref{eq:local-constraints-in-factor-graph-1cc} is satisfied by definition of $w^*$.
  Otherwise, let $C_i \coloneqq \ext(C_i^*)$ be the corresponding set of vertices in $G$.
  Observe that $C_i$ is a connected component of $G - Z$.
  Also, for $I \in \CI$, define $a_I \in D$ such that $\alpha^*(v_I) = a_I$.
  Hence,
  \begin{align*}
   \sum_{v \in C_i^*} w^*(v,\alpha^*(v)) &= \Big(\sum_{v_I \in C_i^* \cap V_\CI} \sum_{u \in I} w(u,\alpha_{a_I,I}(u)) \Big) + \sum_{v \in C_i^* \setminus V_\CI} w(v,\alpha(v))\\
                                         &= \Big(\sum_{v_I \in C_i^* \cap V_\CI} \sum_{u \in I} w(u,\alpha(u)) \Big) + \sum_{v \in C_i^* \setminus V_\CI} w(v,\alpha(v))\\
                                         &= \sum_{v \in C_i} w(v,\alpha(v)).
  \end{align*}
  So Equation \eqref{eq:local-constraints-in-factor-graph-1cc} is satisfied since $\alpha$ satisfies the $1$cc-size constraint $(w,q,\op)$ on $\Gamma[X \setminus Z]$.
  
  \medskip
  
  For the other direction, let $\alpha^*\colon X^* \rightarrow D^*$ be a satisfying assignment for $(\Gamma^*,\CS_\glo,\CS_\loc)$.
  Let $Z \coloneqq \{v \in X^* \mid \alpha^*(v) = {\sf sol}\}$.
  First observe that $Z \subseteq X$ since segment vertices can not be assigned the value ${\sf sol}$ due to the unary constraints.
  Also, $Z \cap V_i = Z_i$ and $Z \cap U = \emptyset$ again by the unary constraints.
  Moreover, $|Z| \leq k$ by the global size constraint $(w_{{\sf sol}},k,\leq)$.
  
  We construct an assignment $\alpha\colon X \setminus Z \rightarrow D$ via
  \begin{itemize}
   \item $\alpha(v) \coloneqq \alpha^*(v)$ for every $v \in X \setminus \left(Z \cup \bigcup_{I \in \CI}I \right)$, and
   \item $\alpha(v) \coloneqq \alpha_{a_I,I}(v)$ for every $I \in \CI$ and $v \in I$ where $a_I \coloneqq \alpha^*(v_I)$.
  \end{itemize}
  It is not difficult to check that $\alpha$ satisfies $\Gamma[X \setminus Z]$.
  So let $(w,q,\op) \in \CS$ and $C$ a connected component of $G - Z$.
  Also, let $C^* \coloneqq \shr(C)$ be the corresponding set in $G/\CI$.
  Observe that $Z \cap I = \emptyset$ for every $I \in \CI$ and hence, $I \subseteq C$ for every $I \in \CI$ for which $C \cap I \neq \emptyset$.
  As before, let $G^* \coloneqq G(\Gamma^*)$ be the constraint graph of $\Gamma^*$ and define $H^*$ to be the graph with $V(H^*) \coloneqq V(G^*) = X^*$ and
  \[E(H^*) \coloneqq \{vw \mid \alpha^*(v)\alpha^*(w) \in F\}\]
  Then $C^*$ is a connected component of $H^*$.
  Using the same calculations as above, it follows that 
  \[\Big(\sum_{v \in C} w(v,\alpha(v)), q \Big) \in \op\]
  because
  \[\Big(\sum_{v \in C^*} w^*(v,\alpha^*(v)), q \Big) \in \op\]
  due to the fact that $\alpha^*$ satisfies the $(F,\CD)$-local constraint $(w^*,0_\CD,q,\op)$.
  \uend
 \end{claim}
 
 By the previous claim, we can now use the algorithm from Theorem \ref{thm:tw-binary-csp-size-constraints} to decide whether $(\Gamma,\CS,U,k)$ is a \yes-instance.
 This completes the description of the algorithm.
 It only remains to analyze its running time.
 Clearly, the instance $(\Gamma^*,\CS_\glo,\CS_\loc)$ can be computed in polynomial time given $(\Gamma,\CS,U,k,i,Z_i)$.
 Also, $\tw(G/\CI) = \CO(\ell + |Z_i|) = \CO(\sqrt{k})$ by Theorem \ref{thm:contraction-decomposition-planar-vertex}.
 So the algorithm from Theorem \ref{thm:tw-binary-csp-size-constraints} runs in time
 \[(|D^*| + \CO(\sqrt{k}) + k + \|\CS\|)^{\CO(\sqrt{k})}|X^*|^{\CO(1)} = (|D| + k + \|\CS\|)^{\CO(\sqrt{k})}|X|^{\CO(1)}.\]
 Overall, this result in a running time of
 \[|X|^{\CO(\sqrt{k})} \cdot (|D| + k + \|\CS\|)^{\CO(\sqrt{k})}|X|^{\CO(1)} = (|X| + |D| + \|\CS\|)^{\CO(\sqrt{k})}.\]
 
 For the second statement of the theorem, the algorithm is essentially identical, but it relies on Conjecture \ref{conj:contraction-decomposition-minor-vertex} instead of Theorem \ref{thm:contraction-decomposition-planar-vertex}.
 In this case $\tw(G/\CI) = \CO(c_H(\ell + |Z_i|)) = \CO(c_H\sqrt{k})$.
 Hence, the algorithm from Theorem \ref{thm:tw-binary-csp-size-constraints} runs in time
 \[(|D| + c_Hk + \|\CS\|)^{\CO(c_H\sqrt{k})}|X|^{\CO(1)} = (|D| + k + \|\CS\|)^{\CO(c_H'\sqrt{k})}|X|^{\CO(1)}\]
 for some suitable constant $c_H'$ that only depends on $c_H$.
\end{proof}

The next theorem provides a variant of the last theorem where we partition the vertex set into $k+1$ classes.
This result may also be of interest.

\begin{theorem}
 There is an algorithm solving \textsc{Planar Perm CSP Vertex Deletion with Size Constraints} in time $(|D| + k + \|\CS\|)^{\CO(k)}|X|^{\CO(1)}$.
 
 Moreover, assuming Conjecture \ref{conj:contraction-decomposition-minor-vertex}, there is an algorithm solving \textsc{$H$-Minor-Free Perm CSP Vertex Deletion with Size Constraints} in time $(|D| + k + \|\CS\|)^{\CO(c_H k)}|X|^{\CO(1)}$ where $c_H$ is a constant that only depends on $H$.
\end{theorem}

\begin{proof}
 The algorithm is exactly the same as in the previous theorem, but it computes a partition $V_1,\dots,V_\ell$ into $\ell \coloneqq k + 1$ blocks.
 This implies there is some $i \in [\ell]$ such that $Z \cap V_i = \emptyset$, and the entire algorithm runs in time $(|D| + k + \|\CS\|)^{\CO(k)}|X|^{\CO(1)}$.
 
 Again, the same analysis provides the statement for $H$-minor-free graphs assuming Conjecture \ref{conj:contraction-decomposition-minor-vertex} holds.
\end{proof}

\subsection{Permutation CSP on Two-Connected Components}
\label{subsec:2cc-perm-csp-vertex-deletion}

Finally, we turn to the proof of Theorem \ref{thm:2cc-perm-csp-deletion-result}.
The proof of this theorem is the most complicated one of this section since we need to incorporate the guessing of the bodies of the segments and we need to verify that all guesses are consistent.
As before, we start by recalling the problem at hand.

Let $\Gamma = (X,D,\CC)$ be a binary CSP instance.
A \emph{2cc-size constraint} is a pair $(w,q)$ where $w \colon X \times D \rightarrow \NN$ is a weight function, and $q \in \NN$.
Let $W \subseteq X$.
An assignment $\alpha\colon W \rightarrow D$ satisfies $(w,q)$ on $W$ if
\[\sum_{v \in W} w(v,\alpha(v)) \leq q.\]
Observe that, in comparison to 1cc-size constraints, we only allow to check for upper bounds on the weighted size of a set $W \subseteq X$.

A \emph{Permutation-CSP-instance with $2$cc-size constraints} is a pair $(\Gamma,\CS)$ where $\Gamma = (X,D,\CC)$ is a Permutation-CSP-instance and $\CS$ is a set of 2cc-size constraints.
For $W \subseteq X$, we say that $(\Gamma,\CS)$ is \emph{satisfiable on $W$} if there is an assignment $\alpha\colon W \rightarrow D$ that satisfies $\Gamma[W]$ as well as all 2cc-size constraints $(w,q) \in \CS$ on $W$.
As before, we define
\[\|\CS\| \coloneqq \prod_{(w,q,\op) \in \CS} (2 + q).\]

\medskip 
\defparproblem{\textsc{$2$Conn Perm CSP Vertex Deletion with Size Constraints}}{A Permutation-CSP-instance with $2$cc-size constraints $(\Gamma,\CS)$ with $\Gamma = (X,D,\CC)$, a set of undeletable variables $U \subseteq X$, and an integer $k$.}{$k$}
{Is there a set $Z \subseteq X \setminus U$ such that $|Z| \leq k$ and $(\Gamma,\CS)$ is satisfiable on $W$ for every $2$-connected component $W$ of the graph $G(\Gamma[X \setminus Z])$.}
\medskip

To prove Theorem \ref{thm:2cc-perm-csp-deletion-result}, we need to build an algorithm solving \textsc{Planar $2$Conn Perm CSP Vertex Deletion with Size Constraints} and which runs in time $(|X| + |D| + \|\CS\|)^{\CO(\sqrt{k})}$.
For this, we build on similar ides as in the previous subsection, but we need to incorporate information about the bodies of segments $I \in \CI$ as obtained in Theorem \ref{thm:2-cc-from-solution-subset}.
Let us start by recalling some notation.

Let $(G,\CI)$ be a segmented graph.
We define $V(\CI) \coloneqq \bigcup_{I \in \CI}I$ as the set of vertices appearing in some segment.
We always use $v_I$ to denote the vertex of $G/\CI$ that corresponds to segment $I$.
Also, $V_\CI \coloneqq \{v_I \mid I \in \CI\}$.
For $U \subseteq V(G)$ we use $\shr(U)$ to denote the set of all vertices $v \in V(G/\CI)$ that correspond to some $u \in U$.
In particular, $v_I \in \shr(U)$ if $U \cap I \neq \emptyset$.
In the other direction, for $U \subseteq V(G/\CI)$, we use $\ext(U)$ to denote the set of all vertices $v \in V(G)$ that correspond to some $u \in U$.
In particular, if $v_I \in U$ then $I \subseteq \ext(U)$.
If $U = \{u\}$ consists of a single vertex, then we also write $\shr(u)$ and $\ext(u)$ instead of $\shr(\{u\})$ and $\ext(\{u\})$. 

Next, let $Z \subseteq V(G) \setminus V(\CI)$ denote a ``solution set''.
Let $(\FT,\gamma)$ be the decomposition into the $2$-connected components of $G - Z$.
Recall that $\FT$ is a forest where each connected component is equipped with a root node.
For $t \in V(\FT)$ we dente by $\desc(t)$ the set of descendants of $t$, including $t$ itself.

We recall the following objects defined for every segment $I \in \CI$.
We have that $r_Z(I)$ denote the unique node $t \in V(\FT)$ which is closest to a root node $t_0$ and for which $I \cap \gamma(t) \neq \emptyset$.
Also, $B_Z(I) \coloneqq \{t \in V(\FT) \setminus \{r(I)\} \mid I \cap \gamma(t) \neq \emptyset\}$ denotes the \emph{body of $I$}.
Observe that $r_Z(I) \notin B_Z(I)$.
Moreover,  $\widehat{r}_Z(I) \coloneqq \gamma(r_Z(I))$ and $\widehat{B}_Z(I) \coloneqq \bigcup_{t \in R_Z(I)} \gamma(t)$ denote the corresponding sets of vertices in the graph $G$.
To simplify notation, we usually omit the index $Z$ if it is clear from context.

To obtain the information on the bodies of segments $I \in \CI$ we build on Corollary \ref{cor:segments-and-bodies-for-unknown-solution} which provides suitable segments and their bodies after guessing $\CO(\sqrt{k})$ many vertices.
Hence, we may assume that this information is additionally given to algorithm.
This is formulated by the following intermediate problem.

\medskip
\defpargarproblem{\textsc{$2$Conn Perm CSP Vertex Deletion with Bodies}}{A Permutation-CSP-instance with $2$cc-size constraints $(\Gamma,\CS)$ with $\Gamma = (X,D,\CC)$, a set of undeletable variables $U \subseteq X$, a set of segments $\CI$ of the graph $G(\Gamma)$, a function $\CF$ on domain $\CI$, and an integer $k$.}{$\tw(G/\CI) + \|\CF\|$}
{Is there a set $Z \subseteq X \setminus (V(\CI) \cup U)$ such that $|Z| \leq k$ and, for $(\FT,\gamma)$ being a rooted decomposition into $2$-connected components of $G(\Gamma) - Z$, $(\Gamma,\CS)$ is satisfiable on $\gamma(t)$ for all $t \in V(\FT)$.}
{If $(\Gamma,\CS,U,\CI,\CF,k)$ is a \yes-instance, then there is a solution such that additionally $(\widehat{B}_Z(I),\widehat{r}_Z(I) \cap I) \in \CF(I)$ for all $I \in \CI$.}
\medskip

Here, we define $\|\CF\| \coloneqq \sum_{I \in \CI}|\CF(I)|$.
To be able to use an algorithm for this problem as a subroutine later on, let us clarify the behavior in case the guarantee is not satisfied.
If $(\Gamma,\CS,U,k)$ is a \yes-instance, an algorithm may output either answer if the guarantee does not hold.
However, if $(\Gamma,\CS,U,k)$ is a \no-instance, an algorithm may never output \yes\ independent of whether the guarantee is satisfied or not. 

\begin{theorem}
 \label{thm:2cc-perm-csp-deletion-with-bodies}
 There is an algorithm solving \textsc{$2$Conn Perm CSP Vertex Deletion with Bodies} in time
 \[(|X| + |D| + \|\CS\| + \|\CF\|)^{\CO(\tw(G/\CI))}.\]
\end{theorem}

\begin{proof}
 We give a polynomial-time algorithm that translates an instance of the problem above into an equivalent CSP-instance $\Gamma^* = (X^*,D^*,\CC^*)$ with global size constraints $\CS_\glo$ and $(F^*,\CD^*)$-local size constraints $\CS_\loc$ (for some suitable choice of $(F^*,\CD^*)$) such that $G(\Gamma^*) = (G(\Gamma))/\CI$ and
 \[|D^*| = (|X| + |D| + \|\CF\|)^{\CO(1)}.\]
 Let $G \coloneqq G(\Gamma)$.
 Towards this end, we describe a solution to $(\Gamma,\CS,U,k)$ by a series of functions defined on $X^* \coloneqq V(G/\CI)$.

 Let $I \in \CI$ be a segment and let $(B,J) \in \CF(I)$.
 We say that $(B,J)$ is \emph{valid} if
 \begin{enumerate}[label = (\Roman*)]
  \item\label{item:valid-region-1} $(\Gamma,\CS)$ is satisfiable on $C$ for every $2$-connected component $C$ of $G[B]$,
  \item\label{item:valid-region-2} $I \subseteq B \cup J$,
  \item\label{item:valid-region-3} every vertex $v \in B \cap J$ is a cut vertex in $G[B \cup J]$,
  \item\label{item:valid-region-4} $G[J]$ is connected, and
  \item\label{item:valid-region-5} $B \cap I \neq \emptyset$ and $G[B \cup I]$ is connected.
 \end{enumerate}
 We start by removing all elements $(B,J)$ from $\CF(I)$ that are not valid.
 Observe that this can be done in polynomial time.

 Now let $\widetilde{\CF} \coloneqq \bigcup_{I \in \CI} \{I\} \times \CF(I)$. 
 We describe a solution to the input instance by the following functions:
 \begin{itemize}
  \item $\alpha \colon X^* \rightarrow D \uplus \{{\sf sol}\} \uplus \widetilde{F}$ such that $\alpha(v_I) \neq {\sf sol}$ for every $I \in \CI$, and $\alpha(v) \neq {\sf sol}$ for every $v \in U \setminus V(\CI)$,
  \item $\alpha^+ \colon X^* \rightarrow D$,
  \item $h \colon X^* \rightarrow [0,|X|]$,
  \item $\cut \colon X^* \rightarrow X \cup \{\top\}$,
  \item $\rho \colon \CI \rightarrow \bigcup_{I \in \CI}\CF(I)$ such that $\rho(I) \in \CF(I)$ for all $I \in \CI$,
 \end{itemize}

 \begin{figure}[t]
  \begin{center}
   \begin{tikzpicture}
    \node at (0,0) {\svg{0.65\linewidth}{bodies-dp2}};
    
    \node at (7.2,2.5) {
     \begin{tabular}{l|lllll}
       vertex & $\alpha$ & $\alpha^+$ & $h$ & $\cut$ & $\rho$\\
       \hline
       $v_1$ & $\textsf{4}$ & $\textsf{1}$ & 4 & $v_3$ & -- \\
       $v_2$ & $\textsf{1}$ & $\textsf{2}$ & 4 & $v_3$ & -- \\
       $v_3$ & $\textsf{4}$ & $\textsf{1}$ & 3 & $v_4$ & -- \\
       $v_4$ & $(I,B,J)$ &  -- & 2 & $v_6$ & -- \\
       $v_7 $ & $\textsf{1}$ &  -- & 0 & $\top$ & -- \\
       $v_I$ & $\textsf{3}$ &  $\textsf{1}$ & 1 & $v_7$ & $(B,J)$ \\
       $v_9$ & $(I,B,J)$ &  -- & 2 & $v_8$ & -- \\
       $v_{11}$ & $(I,B,J)$ & -- & 2 & $v_{10}$ & -- \\
       $v_{12}$ & $\textsf{0}$ & $\textsf{2}$ & 3 & $v_{11}$ & -- \\
     \end{tabular}
    };
   \end{tikzpicture}
   \end{center}
   \caption{A permutation CSP instance remaining after removing variables $Z$, where each 2-connected component is satisfiable. The domain is $D=\{
\textsf{0},\textsf{1},\textsf{2},\textsf{3},\textsf{4}\}$. Each directed edge $x\to y$ represents the relation $y=x+1 \pmod{5}$. Each yellow vertex appears in a unary constraint forcing it to $\textsf{0}$. Thus each 2-connected component has a unique satisfying assignment (note that these assignments do not agree on the cut vertices). Set $B$ is shown in gray, segment $I$ is red, and set $J$ is highlighted in light red. We assume that $v_6$ is the minimal element of $J$ on $I$. The table show the values of certain vertices in a consistent tuple $(\alpha,\alpha^+,h,\cut,\rho)$. 
     }\label{fig:bodies-dp}
 \end{figure}
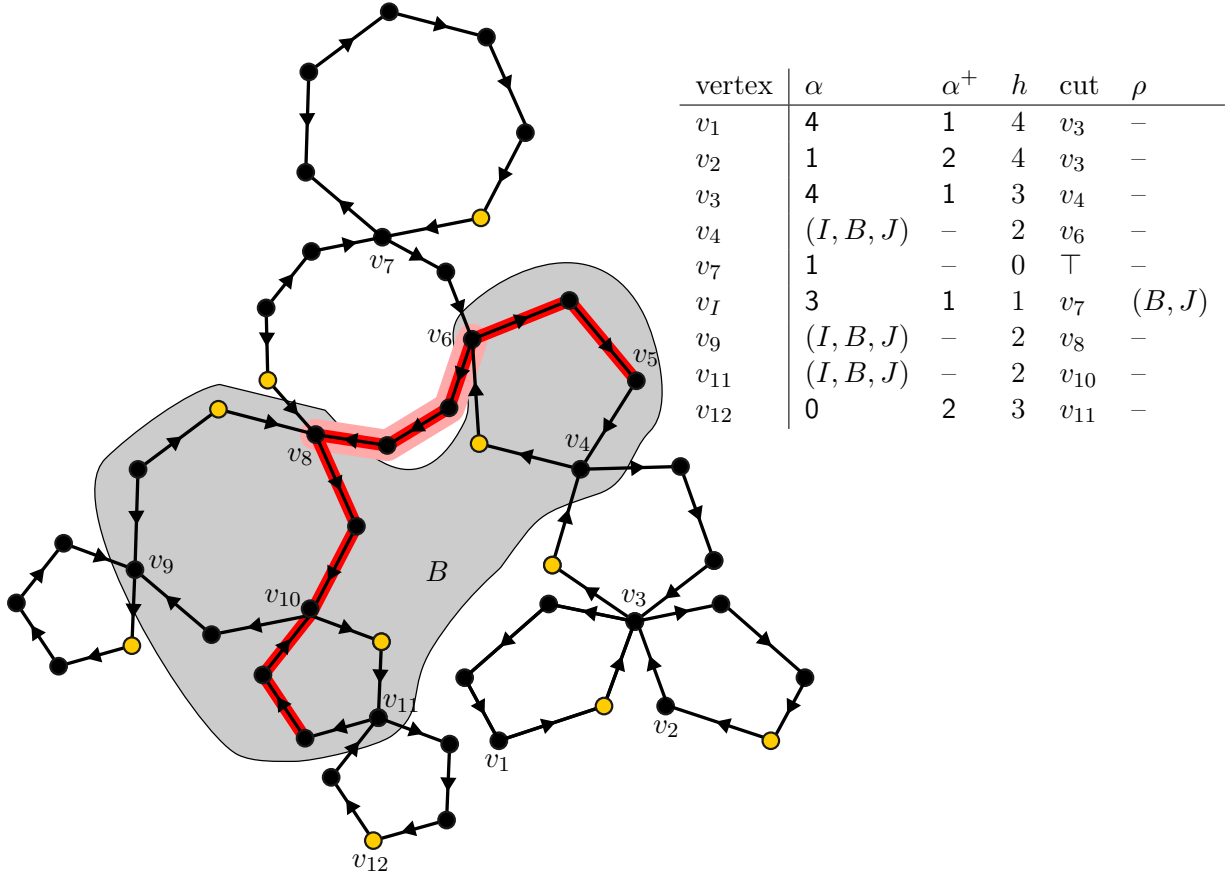
 
 These functions correspond to a satisfying assignment for $\Gamma^*$.
 Consequently, we set
 \[|D^*| \coloneqq (D \uplus \{{\sf sol}\} \cup \widetilde{\CF}) \times D \times [0,|X|] \times (X \cup \{\top\}) \times \Big(\bigcup_{I \in \CI} \CF(I)\Big).\]
 This allows us to associate every tuple of functions $(\alpha,\alpha^+,h,\cut,\rho)$ with some assignment $\alpha^*\colon X^* \rightarrow D^*$, and vice versa.
 Here, we can interpret $\CI$ (the domain of the function $\rho$) as a subset of $X^*$ in the natural way, and view $\rho$ as a function with domain $X^*$ where vertices $v \in V(G/\CI) \setminus V_\CI$ are assigned arbitrary values.
 
 Now, let us first describe how to translate a solution $Z \subseteq V(G) \setminus (V(\CI) \cup U)$ for $(\Gamma,\CS,U,k)$ into a satisfying assignment $(\alpha,\alpha^+,h,\cut,\rho)$ for $(\Gamma^*,\CS_\glo,\CS_\loc)$.
 (The construction of the instance $(\Gamma^*,\CS_\glo,\CS_\loc)$ is completed afterwards in such a way that $(\alpha,\alpha^+,h,\cut,\rho)$ indeed provides a satisfying assignment.)
 
 Suppose $Z \subseteq V(G) \setminus (V(\CI) \cup U)$ is a solution for $(\Gamma,\CS,U,k)$ such that $(\widehat{B}_Z(I),\widehat{r}_Z(I) \cap I) \in \CF(I)$ for all $I \in \CI$.
 Also, let $(\FT,\gamma)$ denote the corresponding rooted decomposition into $2$-connected components of $G - Z$.
 First, we fix $\alpha(v) \coloneqq {\sf sol}$ for every $v \in Z$ and define $\alpha^+(v)$ arbitrarily, $h(v) = 0$ and $\cut(v) = \top$.
 For all other vertices, the functions $(\alpha,\alpha^+,h,\cut,\rho)$ are defined as follows (see Figure \ref{fig:bodies-dp} for an example).
 
 \begin{description}
  \item[$\rho$: ] We first define $\rho(I) \coloneqq (\widehat{B}_Z(I),\widehat{r}_Z(I) \cap I)$ for every $I \in \CI$.
   It is easy to see that $(\widehat{B}_Z(I),\widehat{r}_Z(I) \cap I)$ is valid for every $I \in \CI$. 
  
  \item[$h$: ] For every $v \in V(G/\CI) \setminus Z$ we define $r(v) \coloneqq t$ for the unique highest node $t \in V(\FT)$ such that $\gamma(t) \cap \ext(v) \neq \emptyset$.
   Let $t_0,\dots,t_\ell = r(v)$ denote the unique path from a root of $\FT$ to $r(v)$.
   For every $I \in \CI$ it holds that $\{t_0,\dots,t_\ell\} \cap B_Z(I) = \{t_i,t_{i+1},\dots,t_j\}$ for some $i,j$, i.e., the intersection forms a subpath of $t_0,\dots,t_\ell$.
   Let $t_0',\dots,t_h'$ denote the path obtained from $t_0,\dots,t_\ell$ by contradicting all these subpaths to single vertices for every segment $I \in \CI$.
   We define $h(v) \coloneqq h$.
   
   (Intuitively speaking, $h(v)$ is the height of $t = r(v)$ in the tree $\FT$.
   However, for technical reasons, $h(v)$ is defined slightly differently where single body sets are viewed as a single $2$-connected component to determine the height; see Figure \ref{fig:bodies-dp}).
  \item[$\cut$: ] If $h(v) = 0$ we define $\cut(v) \coloneqq \top$.
   Otherwise, $\cut(v)$ is set to be the cut vertex to the parent of $r(v)$, i.e., $\cut(v) \coloneqq x$ where $x \in \gamma(t) \cap \gamma(s)$ is the unique element for $t = r(v)$ and $s$ being the parent of $t$ in $\FT$.
  \item[$\alpha$ and $\alpha^+$: ]
   Fix some $t \in V(\FT)$.
   We distinguish two cases.
   First suppose that $t \in B_Z(I)$ for some segment $I \in \CI$.
   For every $v \in V(G/\CI)$ such that $r(v) = t$ we set $\alpha(v) \coloneqq (I,\widehat{B}_Z(I),\widehat{r}_Z(I) \cap I)$.
   Also, we define $\alpha^+(v)$ arbitrarily.
 
   In the other case, there is no $I \in \CI$ such that $t \in B_Z(I)$.
   Let $\alpha_t \colon \gamma(t) \rightarrow D$ denote a satisfying assignment for $(\Gamma,\CS)$ on $\gamma(t)$.
   Let $v \in V(G/\CI)$ such that $r(v) = t$.
   If $v \notin V_\CI$ then we define $\alpha(v) \coloneqq \alpha_t(v)$.
   If $v = v_I \in V_\CI$ then $I \cap \gamma(t) = \widehat{r}_Z(I) \cap I \eqqcolon J$ by definition.
   Let $v'$ denote the minimal element contained in $J$ according to the linear order on $I$.
   We define $\alpha(v_I) \coloneqq \alpha_t(v')$.
   Observe that $\cut(v) \notin I$ for every $I \in \CI$ (as otherwise $t \in B_Z(I)$).
   We define $\alpha^+(v) \coloneqq \alpha_t(\cut(v))$ (if $\cut(v) = \top$ then $\alpha^+(v)$ is defined arbitrarily).
 \end{description}
 
 This completes the description of the tuple $(\alpha,\alpha^+,h,\cut,\rho)$.
 
 \medskip
 
 Now, let us complete the construction of $(\Gamma^*,\CS_\glo,\CS_\loc)$.
 Here, the basic idea is that every solution $Z$ for the input problem should provide a satisfying assignment $(\alpha,\alpha^+,h,\cut,\rho)$ as described above, and every satisfying assignment $(\alpha,\alpha^+,h,\cut,\rho)$ gives raise to a solution $Z \coloneqq \alpha^{-1}({\sf sol})$.
 We need to introduce the following additional notation.
 Let $I \in \CI$.
 Observe that $\rho(I) = (B,J)$ for some $B \subseteq V(G)$ and $J \subseteq I$.
 We denote by $\rho_1(I) = B$ the first component of $\rho(I)$, and $\rho_2(I) = J$ denotes the second component of $\rho(I)$. 

 Next, let $J \subseteq I$ such that $G[J]$ is connected and let $a \in D$.
 We define $\alpha_{J,a} \colon J \rightarrow D$ to the unique assignment for which $\alpha_{J,a}(v) = a$ where $v$ denotes the minimal element of $J$ (where elements of $J$ are ordered according to the linear order on $I$),
 and such that $\alpha_{J,a}$ is a satisfying assignment of $\Gamma[J]$.
 If there is no such assignment, we say that $a$ is \emph{invalid} for $J$ and $\alpha_{J,a}$ may be defined in an arbitrary manner.

 Now, we describe a series of conditions that enforces any tuple $(\alpha,\alpha^+,h,\cut,\rho)$ of function to have the intended meaning described above.
 In particular, the tuple of functions constructed from a solution set $Z$ as above satisfies all these conditions.
 We say that \emph{$(\alpha,\alpha^+,h,\cut,\rho)$ is consistent} if the following conditions are satisfied.
 For every $v \in V(G/\CI)$ such that $\alpha(v) \neq {\sf sol}$ it holds that
 \begin{enumerate}
  \item\label{item:csp-1} $h(v) = 0$ if and only if $\cut(v) = \top$,
  \item\label{item:csp-2a} if $\alpha(v) \in D$ and $v \notin V_\CI$ then $\alpha(v) \in R$ for unary every constraint $(v,R) \in \CC$,
  \item\label{item:csp-2b} if $\alpha(v) \in D$ and $v = v_I \in V_\CI$ then $\alpha(v)$ is valid for $\rho_2(I)$,
  \item\label{item:csp-2c} $\alpha^+(v) \in R$ for every unary constraint $(\cut(v),R) \in \CC$ (if $\cut(v) = \top$, this condition is ignored),
  \item\label{item:csp-4} if $\alpha(v) = (I,B,J)$ then $\cut(v) \in B \cap J$,
  \item\label{item:csp-5} if $\cut(v) \notin I$ for every $I \in \CI$, then $\alpha(v) \in D$,
  \item\label{item:csp-6} if $\cut(v) \in I$, then $\alpha(v) = (I,B,J)$ for some $(B,J) \in \CF(I)$,
  \item\label{item:csp-7} if $\alpha(v) = (I,B,J)$ and $v \notin V_\CI$ then $v \in B$,
  \item\label{item:csp-8} if $\alpha(v) = (I,B,J)$ and $v = v_{I'} \in V_\CI$ then $I \neq I'$ and $\rho_2(I') = B \cap I'$,
  \item\label{item:csp-9} $\cut(v) \notin \ext(v)$.
 \end{enumerate}
 Also, for every $vw \in E(G/\CI)$ such that $\alpha(v) \neq {\sf sol}$ it holds that
 \begin{enumerate}[resume]
  \item\label{item:csp-0a} if $v = v_I \in V_\CI$, $w \notin V_\CI$, $\rho(I) = (B,J)$ and $w \in B$ then $\alpha(w) = (I,B,J)$ and $h(w) = h(v) + 1$,
  \item\label{item:csp-0b} if $v = v_I \in V_\CI$, $w = v_{I'} \in V_\CI$, $\rho(I) = (B,J)$ and $I' \cap B \neq \emptyset$ then $\alpha(w) = (I,B,J)$ and $h(w) = h(v) + 1$,
  \item\label{item:csp-0c} if $\alpha(v) = (I,B,J)$, $w \notin V_\CI$ and $w \in B$ then $\alpha(w) = (I,B,J)$ and $h(w) = h(v)$,
  \item\label{item:csp-0d} if $\alpha(v) = (I,B,J)$, $w = v_{I'} \in V_\CI$, $I \neq I'$ and $I' \cap B \neq \emptyset$ then $\alpha(w) = (I,B,J)$ and $h(w) = h(v)$, 
  \item\label{item:csp-0e} if $\alpha(w) = (I,B,J)$ and $v = v_I$ then $\rho(I) = (B,J)$.
 \end{enumerate}
 (Intuitively speaking, these constraints enforce that, if $\rho(I) = (B,J)$, then the set of vertices $v$ that set $\alpha(v) \coloneqq (I,B',J')$ corresponds to $B$, and $\alpha(v) \coloneqq (I,B,J)$ for all such vertices.
 In other words, the constraints enforce that bodies are assigned consistently.)
 And finally, for every $vw \in E(G/\CI)$ such that $\alpha(v) \neq {\sf sol}$ and $\alpha(w) \neq {\sf sol}$ it holds that
 \begin{enumerate}[resume]
  \item\label{item:csp-10} $h(v) \in \{h(w) + 1, h(w), h(w) - 1\}$,
  \item\label{item:csp-11} if $h(v) = h(w)$ then
   \begin{enumerate}
    \item\label{item:csp-11a} $\cut(v) = \cut(w)$ and $\alpha^+(v) = \alpha^+(w)$,
    \item\label{item:csp-11a2} if $v = v_I \in V_\CI$ and $\rho(I) = (B,J)$ then $N(\ext(w)) \cap I \subseteq J$,
    \item\label{item:csp-11b} if $\alpha(v) \in \widetilde{F}$ then $\alpha(w) = \alpha(v)$
    \item\label{item:csp-11c} if $\alpha(v) \in D$ and $v,w \notin V_\CI$ then $(\alpha(v),\alpha(w))$ satisfies all constraints $((v,w),R) \in \CC$,
    \item\label{item:csp-11d} if $\alpha(v) \in D$ and $v = v_I \in V_\CI$, $w \notin V_\CI$ then $a \coloneqq \alpha(v)$ is valid for $\rho_2(I)$,
     and $(\alpha_{a,\rho_2(I)}(v'),\alpha(w))$ satisfies all constraints $((v',w),R) \in \CC$ where $v' \in \rho_2(I)$,
    \item\label{item:csp-11e} if $\alpha(v) \in D$ and $v = v_I \in V_\CI$, $w = w_{I'} \in V_\CI$ then $a \coloneqq \alpha(v)$ is valid for $\rho_2(I)$, $a' \coloneqq \alpha(w)$ is valid for $\rho_2(I')$,
     and $(\alpha_{a,\rho_2(I)}(v'),\alpha_{a',\rho_2(I')}(w'))$ satisfies all constraints $((v',w'),R) \in \CC$ where $v' \in \rho_2(I)$ and $w' \in \rho_2(I')$,
   \end{enumerate}
  \item\label{item:csp-12} if $h(w) = h(v) + 1$ then 
   \begin{enumerate}
    \item\label{item:csp-12a} $\cut(w) \in \ext(v)$,
    \item\label{item:csp-12b2} if $w = v_I \in V_\CI$ and $\rho(I) = (B,J)$ then $N(\ext(v)) \cap I \subseteq J$,
    \item\label{item:csp-12d} if $v,w \notin V_\CI$, then $(\alpha^+(w),\alpha(w))$ satisfies all constraints $((v,w),R) \in \CC$,
    \item\label{item:csp-12e} if $v \notin V_\CI$ and $w = v_I \in V_\CI$, then $a \coloneqq \alpha(w)$ is valid for $\rho_2(I)$,
     and $(\alpha^+(w),\alpha_{a,\rho_2(I)}(w'))$ satisfies all constraints $((v,w'),R) \in \CC$ where $w' \in \rho_2(I)$.
   \end{enumerate}
 \end{enumerate}
 (Here, we mostly ensure that all binary constraints are satisfied, and the root part $J$ for $\rho(I) = (B,J)$ interacts with the rest of graph in a consistent way, see \ref{item:csp-11a2} and \ref{item:csp-12b2}.)
 
 Now, the existence of a consistent tuple can be formulated as a binary CSP-instance $\Gamma^* = (X^*,D^*,C^*)$ such that $G(\Gamma^*) = (G(\Gamma))/\CI$ and
 \[|D^*| = (D \uplus \{{\sf sol}\} \cup \widetilde{\CF}) \times D \times [0,|X|] \times (X \cup \{\top\}) \times \Big(\bigcup_{I \in \CI} \CF(I)\Big).\]
 
 It remains to define the global and local size constraints.
 Since we need to define weight function on $X^* \times D^*$, we interpret elements $\bar a \in D^*$ as vectors of dimension $5$, and use $\bar a_i$ to denote the $i$-th entry.
 In this sense, $\bar a_1$ corresponds to the function $\alpha$, $\bar a_2$ corresponds to the function $\alpha^+$, $\bar a_3$ corresponds to the function $h$, $\bar a_4$ corresponds to the function $\cut$, and $\bar a_5$ corresponds to the function $\rho$.
 First, we add a global size constraint to force that at most $k$ variables $x \in X^*$ are assigned the value ${\sf sol}$ (under the function $\alpha$).
 More precisely, we define $\CS_\glo \coloneqq \{(w_{{\sf sol}},k,\leq)\}$ where $w_{{\sf sol}} \colon X^* \times D^* \rightarrow \NN$ is defined via
 \[w_{{\sf sol}}(v,\bar a) = \begin{cases}
                              1 &\text{if } \bar a_1 = {\sf sol}\\
                              0 &\text{otherwise}
                             \end{cases}\]
 for every $\bar a \in D^*$ and $v \in X^*$, where $\bar a_i$ denotes the $i$-th component of $\bar a$.
 
 Next, we turn to the local size constraints.
 We define
 \[F^* = \Big\{(\bar a,\bar a') \in (D^*)^2 \;\Big|\; \bar a_1 \in D, \bar a_1' \in D, \bar a_3 = \bar a_3'\Big\}.\]
 Also, we define the partition $\CD^*$ of $D^*$ in such a way that two elements $\bar a,\bar a' \in D^*$ end up in the same partition class of $\CD^*$ if and only if
 \[\bar a_2 = \bar a_2' \;\;\wedge\;\; \bar a_3 = \bar a_3' \;\;\wedge\;\; \bar a_4 = \bar a_4'.\]
 
 Now, consider a $2$cc-size constraint $(w,q) \in \CS$.
 We translate $(w,q)$ into a $(F^*,\CD^*)$-local size constraint $(w^*,w_\CD^*,q,\leq) \in \CS_\loc$ as follows.
 For $D' \in \CD^*$ there are $a \in D$, $h \in [0,|X|]$, and $v \in (X \cup \{\top\})$ such that $\bar a_2 = a$, $\bar a_3 = h$, and $\bar a_4 = v$ for all $\bar a \in D'$.
 We define
 \[w_\CD^*(D') \coloneqq \begin{cases}
                           0      &\text{if } v = \top\\
                           w(v,a) &\text{otherwise}
                          \end{cases}.\]
 Next, we define the function $w^*$.
 Let $v \in X^* = V(G/\CI)$ and $\bar a \in D^*$.
 If $v \notin V_\CI$ we define
 \[w^*(v,\bar a) \coloneqq \begin{cases}
                            w(v,\bar a_1) &\text{if } \bar a_1 \in D\\
                            0             &\text{otherwise}
                           \end{cases}.\]
 Also, if $v = v_I \in V_\CI$ we define
 \[w^*(v,\bar a) \coloneqq \begin{cases}
                            \sum_{u \in J}w(u,\alpha_{J,\bar a_1}(u)) &\text{if } \bar a_1 \in D, \bar a_5 = (B,J), \bar a_1 \text{ is valid for } J\\
                            0                                         &\text{otherwise}
                           \end{cases}.\]
 This completes the description of the instance $(\Gamma^*,\CS_\glo,\CS_\loc)$.
 
 \medskip\medskip
 
 First observe that it can be computed in time $(|D| + |X| + |\CS| + \|\CF\|)^{\CO(1)}$.
 Also, $|X^*| \leq |X|$, $\|\CS_\glo\| = k$, $\|\CS_\loc\| = \|\CS\|$, and
 \[|D^*| = (|X| + |D| + \|\CF\|)^{\CO(1)}.\]
 
 Now, we claim that the instance $(\Gamma^*,\CS_\glo,\CS_\loc)$ is equivalent to the input instance.
 More precisely, we claim that the following statements hold:
 \begin{enumerate}[label = (E.\arabic*)]
  \item\label{item:backward} If $(\Gamma^*,\CS_\glo,\CS_\loc)$ is satisfiable with a consistent tuple $(\alpha,\alpha^+,h,\cut,\rho)$ then $Z \coloneqq \alpha^{-1}({\sf sol})$ is a solution for the instance $(\Gamma,\CS,U,k)$.
  \item\label{item:forward} If $Z \subseteq V(G) \setminus (V(\CI) \cup U)$ is a solution for $(\Gamma,\CS,U,k)$ such that $(\widehat{B}_Z(I),\widehat{r}_Z(I) \cap I) \in \CF(I)$ for all $I \in \CI$, then $(\Gamma^*,\CS_\glo,\CS_\loc)$ is satisfiable.
 \end{enumerate}
 Assuming both statement are true, it suffices to check whether $(\Gamma^*,\CS_\glo,\CS_\loc)$ is satisfiable.
 By Theorem \ref{thm:tw-binary-csp-size-constraints} this can be done in time
 \[(|D^*| + \tw(G/\CI) + \|\CS_\glo\| + \|\CS_\loc\|)^{\CO(\tw(G/\CI))}|X^*|^{\CO(1)} = (|X| + |D| + \|\CS\| + \|\CF\|)^{\CO(\tw(G/\CI))}.\]
 Hence, it only remains to prove statements \ref{item:backward} and \ref{item:forward}.
 
 \medskip
 
 We start with statement \ref{item:forward}.
 Suppose $Z \subseteq V(G) \setminus (V(\CI) \cup U)$ is a solution for $(\Gamma,\CS,U,k)$ such that $(\widehat{B}_Z(I),\widehat{r}_Z(I) \cap I) \in \CF(I)$ for all $I \in \CI$.
 Following the translation described above, we obtain a tuple $(\alpha,\alpha^+,h,\cut,\rho)$.
 It can be easily checked that it is consistent.
 So let us consider the size constraints.
 First, it is easy to see that the global size constraint is satisfied since $|Z| \leq k$.
 So consider some $2$cc-size constraint $(w,q)$ and its counterpart $(w^*,w^*_\CD,q,\leq) \in \CS_\loc$.
 For $v \in V(G/\CI)$ let us denote $\alpha^*(v) \coloneqq (\alpha(v),\alpha^+(v),h(v),\cut(v),\rho(v))$ where $\rho(v) \in \bigcup_{I \in \CI}\CF(I)$ is arbitrary if $v \notin V_\CI$, and $\rho(v) \coloneqq \rho(I)$ if $v = v_I \in V_\CI$.
 Let $H$ be the graph with vertex set $V(H) \coloneqq V(G/\CI) = X^*$ and edge set
 \[E(H) \coloneqq \{vw \in E(G/\CI) \mid \alpha^*(v)\alpha^*(w) \in F\}.\]
 Let $C$ be a connected component of $H$ and let $v \in C$.
 If $\alpha(v) \notin D$ then $C = \{v\}$ by definition of the set $F$.
 Let $D' \in \CD^*$ such that $\alpha^*(v) \in D'$.
 We need to ensure that $w_\CD^*(D') + w^*(v,\alpha^*(v)) \leq q$.
 Since $\alpha(v) \notin D$ it holds that $w^*(v,\alpha^*(v)) = 0$ by definition.
 If $\cut(v) = \top$ then $w_\CD^*(D') = 0$.
 Otherwise, $w_\CD^*(D') = w(\cut(v),\alpha^+(v))$.
 In this case, $\alpha(v) \neq {\sf sol}$ (since $\cut(v) = \top$ for all $v \in Z$).
 Let $t \coloneqq r(v)$.
 Then $\cut(v) \in \gamma(t)$ and $\alpha^+(v) = \alpha_t(\cut(v))$ which implies that $w(\cut(v),\alpha^+(v)) \leq q$ since $\alpha_t$ satisfies $(\Gamma,\CS)$ on $\gamma(t)$. 
 
 So suppose that $\alpha(v) \in D$.
 By definition of $F$ and the tuple $(\alpha,\alpha^+,h,\cut,\rho)$, there is some $t \in V(\FT)$ such that $C \subseteq \shr(\gamma(t))$.
 First, this implies that there is some $D' \in \CD^*$ such that $\alpha^*(v) \in D'$ for every $v \in C$.
 Now, we distinguish two cases.
 First, suppose that $t$ is a root of $\FT$.
 Then $w^*_\CD(D') = 0$ by definition.
 Also,
 \begin{align*}
    &\sum_{v \in C} w^*(v,\alpha^*(v))\\
  = &\sum_{v \in C \setminus V_\CI} w(v,\alpha(v)) + \sum_{v = v_I \in C \cap V_\CI}\sum_{u \in \widehat{r}_Z(I) \cap I} w(u,\alpha_{\widehat{r}_Z(I) \cap I,\alpha(v)}(u))\\
  = &\sum_{v \in \gamma(t)} w(v,\alpha_t(v)) \leq q.
 \end{align*}
 
 Otherwise, $t$ is not a root of $\FT$.
 Let $s$ denote its parent.
 Now, $C \subseteq \shr(\gamma(t)) \setminus \shr(\gamma(s))$.
 Let $x$ denote the unique vertex in $\gamma(t) \cap \gamma(s)$.
 Then $w^*_\CD(D') = w(\cut(v),\alpha^+(v)) = w(x,\alpha_t(x))$ for some (and thus every) $v \in C$.
 So
 \begin{align*}
    &w^*(D') + \sum_{v \in C} w^*(v,\alpha^*(v))\\
  = &w^*(v,\alpha^*(v)) + \sum_{v \in C \setminus V_\CI} w(v,\alpha(v)) + \sum_{v = v_I \in C \cap V_\CI}\sum_{u \in \widehat{r}_Z(I) \cap I} w(u,\alpha_{\widehat{r}_Z(I) \cap I,\alpha(v)}(u))\\
  = &w^*(v,\alpha^*(v)) + \sum_{v \in \gamma(t) \setminus \gamma(s)} w(v,\alpha_t(v))\\
  = &\sum_{v \in \gamma(t)} w(v,\alpha_t(v)) \leq q.
 \end{align*}
 So $\alpha^*$ satisfies the $(F,\CD^*)$-local size constraint $(w^*,w^*_\CD,q,\leq) \in \CS_\loc$.
 Overall, this completes the proof of statement \ref{item:forward}.
 
 \medskip\medskip
 
 So let us turn to statement \ref{item:backward}.
 Consider a consistent tuple $(\alpha,\alpha^+,h,\cut,\rho)$ that satisfies $(\Gamma^*,\CS_\glo,\CS_\loc)$ and let $Z = \alpha^{-1}({\sf sol})$.
 First, observe that $Z \cap (V_\CI \cup U) = \emptyset$ which implies that $Z \subseteq V(G) \setminus (V(\CI) \cup U)$ as desired.
 Also, $|Z| \leq k$ by the global size constraints.
 For $v \in V(G/\CI)$ let us denote $\alpha^*(v) \coloneqq (\alpha(v),\alpha^+(v),h(v),\cut(v),\rho(v))$ where $\rho(v) \in \bigcup_{I \in \CI}\CF(I)$ is arbitrary if $v \notin V_\CI$, and $\rho(v) \coloneqq \rho(I)$ if $v = v_I \in V_\CI$.
 Let $H$ be the graph with vertex set $V(H) \coloneqq V(G/\CI) = X^*$ and edge set
 \[E(H) \coloneqq \{vw \in E(G/\CI) \mid \alpha^*(v)\alpha^*(w) \in F\}.\]
 
 Let $C$ be a $2$-connected component of $G - Z$.
 We define $h_C \coloneqq \max_{w \in \shr(C)} h(w)$ and let $C' \coloneqq \{w \in \shr(C) \mid h(w) = h_C\}$.
 Finally, let $A$ be a connected component of $(G/\CI)[C']$.
 By Property \ref{item:csp-11a} it holds that $\cut(v) = \cut(w)$ and $\alpha^+(v) = \alpha^+(w)$ for all $v,w \in A$.
 Let $\cut(A) \coloneqq \cut(v)$ and $\alpha^+(A) \coloneqq \alpha^+(v)$ for some (and thus every) $v \in A$.
 Recall that $\cut(A) \in V(G)$.
 By definition of the set $F$, we it follows that $A$ forms a connected set in the graph $H$.
 To argue that $(\Gamma,\CS)$ is satisfiable on $C$, we distinguish two cases.
 
 \medskip
 
 First, suppose that $\cut(A) \notin I$ for every $I \in \CI$.
 We first claim that $\shr(C) \subseteq A \cup \{\cut(A)\}$.
 If $\shr(C) \not\subseteq A$ then there is some $v \in N_{G/\CI}(A) \cap \shr(C)$ (since $(G/\CI)[\shr(C)]$ is connected).
 Pick $w \in A$ such that $vw \in E(G/\CI)$.
 By Property \ref{item:csp-10} and the definition of the set $A$, it holds that $h(v) + 1 = h(w)$.
 But now, Property \ref{item:csp-12a} implies that $\cut(w) = \cut(A) = v$.
 In other words, $N_{G/\CI}(A) \cap \shr(C) \subseteq \{\cut(A)\}$.
 But this implies that $\shr(C) \subseteq A \cup \{\cut(A)\}$ since otherwise $\cut(A)$ would form a cut vertex in $G[C]$.
 
 Next, let $v_I \in A \cap V_\CI$ be a segment vertex and suppose $\rho(I) = (B,J)$.
 We claim that $C \cap I \subseteq J$ or $C \subseteq I$.
 Suppose $C \nsubseteq I$.
 Let $w' \in N_G(I) \cap C$ and let $w$ be its corresponding vertex in $G/\CI$.
 Observe that $w \in \shr(C)$.
 If $w \in A$ then $h(w) = h(v_I)$ by definition of the set $A$.
 Hence, $N_G(w') \cap I \subseteq J$ by Condition \ref{item:csp-11a2}.
 In the other case, $w = w' = \cut(A)$ and $h(w) = h(\cut(A)) + 1$.
 So $N_G(w') \cap I \subseteq J$ by Condition \ref{item:csp-12b2}.
 In total, this means that $N_G(C \setminus I) \cap C \subseteq J$.
 Since $G[C]$ is $2$-connected, it follows that $C \cap I \subseteq J$ by Properties \ref{item:valid-region-2}-\ref{item:valid-region-4}.
 
 If there is some $I \in \CI$ such that $C \subseteq I$, then $(\Gamma,\CS)$ is satisfiable on $C$ by Properties \ref{item:valid-region-1} and \ref{item:valid-region-2}.
 So suppose that $C \nsubseteq I$ for every $I \in \CI$.
 Then
 \[C \subseteq \{\cut(A)\} \cup A \setminus V_\CI \cup \bigcup_{v_I \in A \cap V_\CI} \rho_2(I) \eqqcolon \widetilde{A}.\]
 Now, we construct a satisfying assignment $\widetilde{\alpha}\colon \widetilde{A} \rightarrow D$ for $\Gamma[\widetilde{A}]$ as follows.
 We set $\widetilde{\alpha}(\cut(A)) \coloneqq \alpha^+(A)$.
 For the other elements, first observe that $\alpha(v) \in D$ for all $v \in A$ by Condition \ref{item:csp-5}.
 For $v \in A \setminus V_\CI$ we define $\widetilde{\alpha}(v) \coloneqq \alpha(v)$.
 Finally, for $v_I \in V_\CI \cap A$ and $w \in \rho_2(I)$, we define
 \[\widetilde{\alpha}(w) \coloneqq \alpha_{\rho_2(I),\alpha(v_I)}(w).\]
 Observe that $\alpha(v_I)$ is valid for $\rho_2(I)$ by Condition \ref{item:csp-2b}.
 
 First consider a unary constraint $(v,R) \in \CC$ where $v \in \widetilde{A}$.
 Then $\widetilde{\alpha}(v) \in R$ by Properties \ref{item:csp-2a} - \ref{item:csp-2c}.
  
 So let $v,w \in \widetilde{A}$ such that $vw \in E(G)$.
 Consider a constraint $((v,w),R) \in \CC$.
 We distinguish several cases depending on which parts of $\widetilde{A}$ contain $v$ and $w$.
 If $v,w \in A \setminus V_\CI$ then $(\widetilde{\alpha}(v),\widetilde{\alpha}(w)) = (\alpha(v),\alpha(w)) \in R$ by Property \ref{item:csp-11c}.
 If $v \in \rho_2(I)$ for some $I \in \CI$ and $w \in A \setminus V_\CI$ (or the other way around) then the constraint $((v,w),R)$ is satisfied by Property \ref{item:csp-11d}.
 If $v,w \in \rho_2(I)$ for some $I \in \CI$ then the constraint $((v,w),R)$ is satisfied by Condition \ref{item:csp-2b}.
 If $v \in \rho_2(I)$ and $w \in \rho_2(I')$ for distinct $I,I' \in \CI$ then the constraint $((v,w),R)$ is satisfied by Condition \ref{item:csp-11e}.
 Finally, suppose one of the two elements $v,w$ is equal to $\cut(A)$.
 Without loss of generality assume that $v = \cut(A)$.
 Recall that $v \notin I$ for every $I \in \CI$.
 Also recall that $h(w') = h(v) + 1$ where $w' \in V(G/\CI)$ is the vertex corresponding to $w$.
 So $\widetilde{\alpha}$ satisfies the constraint $((v,w),R) \in \CC$ by Properties \ref{item:csp-12d} and \ref{item:csp-12e}.
 Overall, this shows that $\widetilde{\alpha}$ satisfies $\Gamma[\widetilde{A}]$.
 
 Next, consider some $2$cc size constraint $(w,q) \in \CS$ and its counterpart $(w^*,w_\CD^*,q,\leq) \in \CS_\loc$.
 Clearly,
 \[\sum_{v \in C} w(v,\widetilde{\alpha}(v)) \leq \sum_{v \in \widetilde{A}} w(v,\widetilde{\alpha}(v)).\]
 and it suffices to prove an upper bound for the second term.
 Recall that $A$ forms a connected set in $H$.
 Let $D' \in \CD^*$ denote the set such that $\alpha^*(v) \in D'$ for all $v \in A$.
 Then
 \[w_\CD^*(D') + \sum_{v \in A} w^*(v,\alpha^*(v)) \leq q.\]
 Since
 \[w_\CD^*(D') + \sum_{v \in A} w^*(v,\alpha^*(v)) = \sum_{v \in \widetilde{A}} w(v,\widetilde{\alpha}(v))\]
 by definition of $(w^*,w_\CD^*)$, it follows that $\widetilde{\alpha}$ satisfies the size constraint $(w,q) \in \CS$ on $C$.
 This completes the first case.
 
 \medskip
 
 For the second case, suppose that $\cut(A) \in I$ for some $I \in \CI$.
 We start by proving the following useful claim.
 
 \begin{claim}
  \label{claim:region-vertices}
  For every $I \in \CI$ and $(B,J) \coloneqq \rho(I)$ it holds that
  \[B = \{v \in V(G) \setminus V(\CI) \mid \alpha(v) = (I,B,J)\} \cup \bigcup_{I' \in \CI \colon \alpha(v_{I'}) = (I,B,J)} \rho_2(I').\]
  Moreover, $h(w) = h(v_I) + 1$ for every $w \in \shr(B) \setminus \{v_I\}$.
  \proof
  The inclusion ``$\supseteq$'' follows directly from Conditions \ref{item:csp-7} and \ref{item:csp-8}.
  For the other direction, first observe that $G[B \cup I]$ is connected by Condition \ref{item:valid-region-5}.
  Hence, $(G/\CI)[\shr(B)]$ is connected as well.
  Note that $v_I \in \shr(B)$ by Property \ref{item:valid-region-5}.
  But now,  the inclusion ``$\subseteq$'' follows from Conditions \ref{item:csp-0a} - \ref{item:csp-0d}.
  The same arguments also imply that $h(w) = h(v_I) + 1$ for every $w \in \shr(B) \setminus \{v_I\}$.
  \uend
 \end{claim}
 
 Now, fix $I \in \CI$ such that $\cut(A) \in I$ and let $(B,J) \coloneqq \rho(I)$.
 We argue that $C \subseteq B$.
 First observe that this implies $(\Gamma,\CS)$ is satisfiable on $C$ by Property \ref{item:valid-region-1}.
 Let
 \[A' \coloneqq \{v \in \shr(C) \mid \alpha(v) = (I,B',J') \text{ for some } (B',J') \in \CF(I)\}.\]
 First note that $A \subseteq A'$ by Condition \ref{item:csp-6}.
 Let $v \in A'$ and suppose $\alpha(v) = (I,B',J')$.
 We first argue that $(B,J) = (B',J')$.
 By Condition \ref{item:valid-region-5} it holds that $G[B' \cup I]$ is connected and $I \cap B' \neq \emptyset$.
 This means that $\alpha(v) = (I,B',J')$ for all $v \in \shr(B') \setminus \{v_I\}$ by Conditions \ref{item:csp-0c} and \ref{item:csp-0d}.
 But then $\rho(I) = (B',J')$ by Condition \ref{item:csp-0e}.
 So $(B,J) = (B',J')$.
 In other words, $\alpha(v) = (I,B,J)$ for every $v \in A'$.
 Also, Claim \ref{claim:region-vertices} implies that $A' \subseteq \shr(B)$ and $h(v) = h(v_I) + 1$ for every $v \in A'$.
 In particular, $v_I \notin A'$ (see also Condition \ref{item:csp-8}).
 
 Now, let $w \in N_{G/\CI}(A') \cap \shr(C)$ and let $v \in A'$ such that $vw \in E(G/\CI)$.
 Then $h(v) = h(w) + 1$ by the definition of the set $A$ and Condition \ref{item:csp-10}.
 Also, $\cut(v) \in I$ by Condition \ref{item:csp-4}.
 So $w = v_I$ by Condition \ref{item:csp-12a}.
 
 Next, let $v_{I'} \in V_\CI \cap A'$ be a segment vertex, and suppose that $\rho(I') = (B',J')$.
 We claim that $C \cap I' \subseteq J'$.
 Let $N_G(I') \cap C$ and let $w \in N_{G/\CI}(v_{I'}) \cap \shr(C)$ be the corresponding vertex in $G/\CI$.
 Then $h(w) = h(v_{I'})$ or $h(w) = h(v_{I'}) + 1$ using Condition \ref{item:csp-10} and the fact that $h(v_{I'}) = h_C$.
 So $N(w) \cap I' \subseteq J'$ by Conditions \ref{item:csp-11a2} and \ref{item:csp-12b2}.
 This means that $N_G(C \setminus I') \cap C \subseteq J'$.
 Since $G[C]$ is $2$-connected, it follows that $C \cap I' \subseteq J'$ by Properties \ref{item:valid-region-2} - \ref{item:valid-region-4}.
 
 Now, let $C' \coloneqq \{v \in C \mid \shr(v) \in A'\} \cup (C \cap I)$.
 Then $C' \subseteq B \cup I$ using Claim \ref{claim:region-vertices}.
 Now, let $w' \in C \setminus C'$ be some vertex such that $N_G(w') \cap C' \neq \emptyset$.
 The arguments above imply that $N_G(w') \cap C' \subseteq I$.
 Let $w \in V(G/\CI)$ be the vertex corresponding to $w'$.
 Observe that $wv_I \in E(G/\CI)$.
 If $h(w) = h(v_I) + 1$ then $\alpha(w) = (I,B',J')$ for some $(B',J') \in \CF(I)$ by Conditions \ref{item:csp-12a} and \ref{item:csp-6}.
 So $w \in A'$ and $w \in C'$, contradicting the definition of $w'$.
 So $h(w) = h(v_I)$ or $h(w) + 1 = h(v_I)$.
 But now, using the same arguments as before, we conclude that $N_G(w') \subseteq J$.
 This implies that $C \subseteq B$ using Properties \ref{item:valid-region-2} - \ref{item:valid-region-4} and the fact that $G[C]$ is $2$-connected.
 Hence, $(\Gamma,\CS)$ is satisfiable on $C$ by Property \ref{item:valid-region-1}.
\end{proof}

Now, we can combine the previous theorem with Corollary \ref{cor:segments-and-bodies-for-unknown-solution} to obtain an algorithm for \textsc{Planar $2$Conn Perm CSP Vertex Deletion with Size Constraints}.

\begin{theorem}
 \label{thm:2cc-perm-csp-deletion}
 There is an algorithm solving \textsc{Planar $2$Conn Perm CSP Vertex Deletion with Size Constraints} in time $(|X| + |D| + \|\CS\|)^{\CO(\sqrt{k})}$.
 
 Moreover, assuming Conjecture \ref{conj:contraction-decomposition-minor-vertex}, there is an algorithm solving \textsc{$H$-Minor-Free $2$Conn Perm CSP Vertex Deletion with Size Constraints} in time $(|X| + |D| + \|\CS\|)^{\CO(c_H\sqrt{k})}$ for some constant $c_H$ that only depends on $H$.
\end{theorem}

\begin{proof}
 Consider the first part of the theorem.
 Let $(\Gamma,\CS,U,k)$ be the input instance.
 Also, let $G \coloneqq G(\Gamma)$ be the constraint graph.
 We choose $r = \CO(\sqrt{k})$ according to Lemma \ref{cor:segments-and-bodies-for-unknown-solution} and iterate over all tuples $(v_1,\dots,v_r) \in (V(G))^r$.
 For $(v_1,\dots,v_r) \in (V(G))^r$, we compute a set of segments $\CI$ and a function $\CF$ using the algorithm from \ref{cor:segments-and-bodies-for-unknown-solution}.
 If $\tw(G/\CI) = \omega(\sqrt{k})$, the algorithm moves to the next tuple.
 Otherwise, we apply Theorem \ref{thm:2cc-perm-csp-deletion-with-bodies} to the input $(\Gamma,\CS,U,\CI,\CF,k)$.
 If the algorithm returns \yes\, then we also output \yes.
 Otherwise, we move to the next tuple $(v_1,\dots,v_r) \in (V(G))^r$.
 If we never outputs \yes\ for any tuple $(v_1,\dots,v_r) \in (V(G))^r$, then we return \no.
 
 Let us first analyse the running time of the algorithm.
 Clearly, in each iteration, $|\CI| \leq |X|$ and $|\CF| = |X|^{\CO(1)}$ by Corollary \ref{cor:segments-and-bodies-for-unknown-solution}.
 Also, the computation of $(\CI,\CF)$ requires time $|X|^{\CO(1)}$.
 Also, we can determine whether $\tw(G/\CI) = \CO(\sqrt{k})$ in time $2^{\CO(\sqrt{k})}|X|^{\CO(1)}$.
 Next, the algorithm from Theorem \ref{thm:2cc-perm-csp-deletion-with-bodies} runs in time
 \[(|X| + |D| + \|\CS\| + \|\CF\|)^{\CO(\tw(G/\CI))} = (|X| + |D| + \|\CS\|)^{\CO(\sqrt{k})}.\]
 Overall, this gives a running time of
 \[|X|^r \Big(|X|^{\CO(1)} + 2^{\CO(\sqrt{k})}|X|^{\CO(1)} + (|X| + |D| + \|\CS\|)^{\CO(\sqrt{k})}\Big) = (|X| + |D| + \|\CS\|)^{\CO(\sqrt{k})}.\]
 
 For the correctness of the algorithm, first observe that we only output \yes\ if $(\Gamma,\CS,U,k)$ is indeed a \yes-instance.
 So suppose that $(\Gamma,\CS,U,k)$ is a \yes-instance and let $Z \subseteq X = V(G)$ be a solution.
 Observe that $|Z| \leq k$.
 By Corollary \ref{cor:segments-and-bodies-for-unknown-solution}, there is some tuple $(v_1,\dots,v_r) \in (V(G))^r$ such that
 \begin{enumerate}
  \item $V(\CI) \cap Z = \emptyset$,
  \item $\tw(G/\CI) = \CO(\sqrt{k})$, and 
  \item $(\widehat{B}_Z(I),\widehat{r}_Z(I) \cap I) \in \CF(I)$ for every $I \in \CI$.
 \end{enumerate}
 where $(\CI,\CF)$ denotes the output of the algorithm from Corollary \ref{cor:segments-and-bodies-for-unknown-solution} on input $(G,(v_1,\dots,v_r))$.
 For this tuple, the algorithm from Theorem \ref{thm:2cc-perm-csp-deletion-with-bodies} outputs \yes.
 
 \medskip
 
 For the second part of the theorem, we proceed in the same way, but we replace the application of Corollary \ref{cor:segments-and-bodies-for-unknown-solution} with Corollary \ref{cor:segments-and-bodies-for-unknown-solution-minors}.
\end{proof}

In particular, the last theorem implies Theorem \ref{thm:2cc-perm-csp-deletion-result}.
We complete this section by considering the edge deletion version of the problem above.

\medskip 
\defparproblem{\textsc{$2$Conn Perm CSP Edge Deletion with Size Constraints}}{A Permutation-CSP-instance with $2$cc-size constraints $(\Gamma,\CS)$ with $\Gamma = (X,D,\CC)$, a set of undeletable edges $U \subseteq E(G(\Gamma))$, and an integer $k$.}{$k$}
{Is there a set $Z \subseteq E(G(\Gamma)) \setminus U$ such that $|Z| \leq k$ and $((X,D,\CC - Z),\CS)$ is satisfiable on $W$ for every $2$-connected component $W$ of the graph $G(\Gamma) - Z$.}
\medskip

\begin{theorem}
 \label{thm:2cc-perm-csp-edge-deletion}
 There is an algorithm solving \textsc{Planar $2$Conn Perm CSP Edge Deletion with Size Constraints} in time $(|X| + |D| + \|\CS\|)^{\CO(\sqrt{k})}$.
 
 Moreover, assuming Conjecture \ref{conj:contraction-decomposition-minor-vertex}, there is an algorithm solving \textsc{$H$-Minor-Free $2$Conn Perm CSP Edge Deletion with Size Constraints} in time $(|X| + |D| + \|\CS\|)^{\CO(c_H\sqrt{k})}$ for some constant $c_H$ that only depends on $H$.
\end{theorem}

\begin{proof}
 We prove the theorem by giving a reduction to the vertex deletion version of the problem.
 As usual, we focus on the first statement of the theorem.
 Let $(\Gamma,\CS,U,k)$ be the input instance where $\Gamma = (X,D,\CC)$.
 We construct an equivalent instance $(\Gamma',\CS',U',k')$, where $\Gamma' = (X',D',\CC')$, of \textsc{Planar $2$Conn Perm CSP Vertex Deletion with Size Constraints} as follows.
 Let $G \coloneqq G(\Gamma)$.
 We define
 \[X' \coloneqq X \uplus E(G)\]
 and
 \[D' \coloneqq D \uplus D^2.\]
 For every $v \in X$ we add the unary constraint $(v,D)$ to $\CC'$
 Also, for every unary constraint $(v,R) \in \CC$, the constraint $(v,R)$ is also added to $\CC'$.
 So consider some binary constraint $((u,v),R) \in \CC$.
 By combining binary constraints over the same set of variables, we may assume without loss of generality that $((u,v),R)$ is the only constraint with variables $u$ and $v$.
 Observe that $uv \in E(G)$.
 We define
 \begin{align*}
  R_{u,v}' \coloneqq \Big\{(a,a') \in D^2 \;\Big|\; &(a,a') \in R \text{ and }\\
                                                    &a \in R_u \text{ for every unary constraint } (u,R_u) \in \CC \text{ and }\\
                                                    &a' \in R_v \text{ for every unary constraint } (v,R_v) \in \CC \text{ and }\\
                                                    &w(u,a) + w(v,a') \leq q \text{ for every } (w,q) \in \CS\Big\}.
 \end{align*}
 Now, we add the binary constraints
 \[((u,uv),\{(a,(a,a')) \in D' \times D' \mid a,a' \in D, (a,a') \in R_{u,v}'\})\]
 and
 \[((uv,v),\{((a,a'),a') \in D' \times D' \mid a,a' \in D, (a,a') \in R_{u,v}'\})\]
 to $\CC'$.
 Observe that both constraints are permutation constraints.
 This completes the description of $\Gamma'$.
 Observe that $\Gamma'$ is Permutation-CSP-instance.
 We set
 \[U' \coloneqq X \cup U\]
 and
 \[k' \coloneqq k.\]
 It remains to define $\CS'$.
 For every $(w,q) \in \CS$ we add the size constraint $(w',q)$ to $\CS'$ where $w' \colon X' \times D \rightarrow \NN$ is defined via
 \begin{itemize}
  \item $w'(v,a) \coloneqq w(v,a)$ for all $v \in X$ and $a \in D$,
  \item $w'(v,(a,a')) \coloneqq 0$ for all $v \in X$ and $(a,a') \in D^2$, and
  \item $w'(e,a) \coloneqq 0$ for all $e \in E(G)$ and $a \in D'$.
 \end{itemize}

 Clearly, the instance $(\Gamma',\CS',U',k')$ can be computed in polynomial time.
 Let $G' \coloneqq G(\Gamma')$.
 Observe that $G'$ is obtained from $G$ by replacing each edge by a path of length $2$.
 In particular, if $G$ is planar, then $G'$ is also planar.
 
 We claim that $(\Gamma,\CS,U,k)$ has a solution if and only if $(\Gamma',\CS',U',k')$ has a solution.
 First, suppose there is a set $Z \subseteq E(G) \setminus U$ such that $|Z| \leq k$ and $((X,D,\CC - Z),\CS)$ is satisfiable on $W$ for every $2$-connected component $W$ of the graph $G(\Gamma) - Z$.
 We define $Z' \coloneqq Z \subseteq X'$.
 Note that $|Z'| = |Z| \leq k = k'$ and $Z' \cap U' = \emptyset$.
 Let $W'$ be a $2$-connected component of $G' - Z'$.
 
 We distinguish two cases.
 First, suppose that $|W'| > 2$.
 Then $W \coloneqq W' \cap X$ is a $2$-connected component of $G - Z$.
 This means there is some assignment $\alpha\colon W \rightarrow D$ that satisfies $(\Gamma,\CS)$ on $W$.
 We construct an assignment $\alpha'\colon W' \rightarrow D$ as follows.
 First, we set $\alpha'(v) \coloneqq \alpha(v)$ for all $v \in W$.
 So suppose $uv \in W' \cap E(G)$.
 Then $u,v \in W' \cap X$ since $|W'| > 2$ and $G'[W']$ is $2$-connected.
 Let $((u,v),R) \in \CC$ be the unique constraint over variables $u$ and $v$.
 We set $\alpha'(uv) \coloneqq (\alpha(u),\alpha(v))$.
 It is easy to see that $\alpha'$ satisfies $\Gamma'[W']$ as well as all size constraints $(w',q) \in \CS'$ on $W'$.
 
 Otherwise, $|W'| \leq 2$.
 If $|W'| = 1$ then $W' = \{v\}$ for some $v \in X$.
 Then $v$ forms an isolated vertex in $G - Z$.
 So $(\Gamma',\CS')$ is satisfiable on $W'$ since $(\Gamma,\CS)$ is satisfiable on $W$.
 Finally, suppose that $|W'| = 2$ which means that $W'$ induces an edge in $G'$.
 This means that $W' = \{u,uv\}$ where $u \in X$ and $uv \in E(G)$.
 But then $W \coloneqq \{u,v\}$ forms a $2$-connected component of $G - Z$.
 So again, $(\Gamma',\CS')$ is satisfiable on $W'$ since $(\Gamma,\CS)$ is satisfiable on $W$.
 
 In the other direction, suppose there is a set $Z' \subseteq X' \setminus U'$ such that $(\Gamma',\CS')$ is satisfiable on $W'$ for every $2$-connected component $W'$ of $G' - Z'$.
 We set $Z \coloneqq Z'$.
 Observe that $Z \subseteq E(G) \setminus U$ and $|Z| \leq k$.
 Let $W$ be a $2$-connected component of $G - Z$ and let $W' \coloneqq W \cup E(G[W])$.
 If $|W| > 2$ then $W'$ forms a $2$-connected component of $G' - Z'$.
 Let $\alpha'\colon W' \rightarrow D'$ be an assignment that satisfies $(\Gamma',\CS')$ on $W'$.
 By the unary constraints, it holds that $\alpha'(v) \in D$ for all $v \in W$.
 This allows us to define an assignment $\alpha \colon W \rightarrow D$ via $\alpha(v) \coloneqq \alpha'(v)$ for all $v \in W$.
 It is easy to see that $\alpha$ satisfies $(\Gamma,\CS)$ on $W$.
 
 If $|W| = 1$ then $W' \coloneqq W$ forms an isolated vertex in $G' - Z'$ and we can argue analogously that $(\Gamma,\CS)$ is satisfiable on $W$.
 So suppose that $W = \{u,v\}$.
 Then $uv \in E(G)$ and there is a $2$-connected component $W'$ of $G' - Z'$ such that $uv \in W'$.
 Let $((u,v),R)$ be the unique binary constraint in $\CC$ over variables $u$ and $v$.
 Since $(\Gamma',\CS')$ is satisfiable on $W'$, there are $a,a' \in D$ such that $(a,a') \in R_{u,v}'$.
 Now define $\alpha \colon W \rightarrow D$ via $\alpha(u) \coloneqq a$ and $\alpha(u) \coloneqq a'$.
 Then $\alpha $satisfies $(\Gamma,\CS)$ on $W$ by definition of the relation $R_{u,v}'$.
 
 Overall, this means we can use the algorithm from Theorem \ref{thm:2cc-perm-csp-deletion} to decide whether $(\Gamma,\CS,U,k)$ is a \yes-instance.
 This completes the proof of the first part of the theorem.
 
 \medskip
 
 For the second part, we proceed analogously.
 Suppose that $G$ is $H$-minor-free and let $h \coloneqq \max(4,|V(H)|)$.
 Recall that $G'$ is obtained from $G$ by replacing each edge by a path of length $2$.
 We claim that $G'$ is $K_h$-minor-free (where $K_h$ denotes the complete graph on $h$ vertices).
 Indeed, suppose towards a contraction that $U_1',\dots,U_h' \subseteq V(G')$ are vertex-disjoint, connected subsets such that $E_{G'}(U_i',U_j') \neq \emptyset$ for all $i \neq j \in [h]$.
 Then we can set $U_i \coloneqq U_i' \cap V(G)$ for $i \in [h]$.
 It is not difficult to see that $U_1,\dots,U_h$ are vertex-disjoint, connected subsets such that $E_G(U_i,U_j) \neq \emptyset$ for all $i \neq j \in [h]$.
 (Here, we use that if $uv \in U_i' \cap E(G)$ then $\{u,v\} \cap U_i' \neq \emptyset$ since $h \geq 4$ and $\deg_{G'}(uv) = 2$.)
 
 So, assuming Conjecture \ref{conj:contraction-decomposition-minor-vertex}, there is an algorithm deciding whether $(\Gamma',\CS',U',k')$ is a \yes-instance in time $(|X'| + |D'| + \|\CS'\|)^{\CO(c_H\sqrt{k'})} = (|X| + |D| + \|\CS\|)^{\CO(c_h\sqrt{k})}$ for some constant $c_h$ that only depends on $h$ by Theorem \ref{thm:2cc-perm-csp-deletion}.
 Since $h$ only depends on $H$, we obtain the second part of the theorem.
\end{proof}

%% file: kernel.tex
\section{Minor Preserving Kernels}
\label{sec:kernels}

In this section, we discuss the kernelization results presented in Section \ref{sec:perm-csp-results} and how they relate to the existing literature.
In particular, explain how to obtain Theorems \ref{thm:edge-kernels-result}, \ref{thm:H-minor-graph-Pi-compression-vertex-SFVS} and \ref{thm:H-minor-graph-Pi-compression-vertex-GFVS}.
For edge deletion problems, the corresponding kernels (of quasi-polynomial size) can be obtained from existing results. 
This is not true for vertex deletion problems.
In the cases of \textsc{Vertex Multiway Cut with Deletable Terminals} and \textsc{Group FVS}, we are able to adopt the existing results with some simple yet non-trivial changes.  
In the case of \textsc{Subset FVS} we need significant changes in comparison to the existing kernel from \cite{HolsK18}.

\subsection{Edge Deletion Problems}

Consider a graph $G$ with a set terminals $T \subseteq V(G)$. 
The pair $(G,T)$ is referred to as a terminal network. 
Let $(G,T)$ be a terminal network, and let $\calT = T_1 \cup \dots \cup T_s$ be a partition of $T$. 
A \emph{multiway cut} for $\calT$ is a set of edges $X \subseteq E(G)$ such that $G - X$ contains no path between any pair of terminals $t \in T_i$ and $t' \in T_j$ for $i \neq  j$. 
Let us define a \emph{multicut-mimicking network} for $(G,T)$ as a terminal network $(G', T)$ where $T \subseteq V(G')$ and for every partition $\calT  = T_1 \cup \dots \cup T_s$ of $T$,  the size of a minimum multiway cut for $T$ is identical in $G$ and $G'$.
Wahlstr{\"{o}}m proved that existence of quasipolynomial multicut-mimicking networks for terminal network \cite{Wahlstrom20}.
In Section~$3$ in  \cite{Wahlstrom20}, the author states: 
``The process [of finding a multicut-mimicking network] repeatedly finds a single edge $e \in E(G)$ with a guarantee that for every set of cut requests $R \subseteq \binom{T}{2}$ there is a minimum multicut $X$ for $R$ in $G$ such that $e \not\in X$. 
We may then contract the edge $e$ and repeat the process. 
Thus the end product is a multicut-mimicking network, and the edges that survive until the end of the process form a multicut-covering set."
The following theorem is a direct consequence of this and Corollary~$3$ in \cite{Wahlstrom20}.

\begin{theorem}[Wahlstr{\"{o}}m \cite{Wahlstrom20}]
 The following problems have randomized quasipolynomial kernels with failure probability $\CO(2^{-n})$:
 \begin{enumerate}
  \item \textsc{Edge Multiway Cut} parameterized by solution size,
  \item \textsc{Group Feedback Edge Set} parameterized by solution size, for any group, such that the group remains the same in the reduced instance,
  \item \textsc{Subset Feedback Edge Set} with undeletable edges, parameterized by solution size.
 \end{enumerate}
 Moreover, if the graph in the input instance is $H$-minor-free then the graph in the reduced instance is also $H$-minor-free, and the parameter does not increase in the reduced instance.
\end{theorem}

It is known that \textsc{Edge Bipartization} is a special case of \textsc{Group Feedback Edge Set} with the group $\mathbb{Z}_2$.
As both of these problems are \textsf{NP-Complete},  we can conclude that similar result holds for \textsc{Edge Bipartization}.

\subsection{Vertex Deletion Problems}

In this subsection, we consider the following three problems: \textsc{Vertex Multiway Cut with Deletable Terminal (Vertex MwC-DT}),  \textsc{Group Subset Feedback Vertex Set (Group FVS)},  and \textsc{Subset Feedback Vertex Set (Subset FVS)}.
Let $\Pi$ be one of these problems. 
Consider an instance  $(G, T, k)$\footnote{For the sake of consistency, we slightly abuse the notation in the case of \textsc{Group Feedback Vertex Set} and denote an generic instance by $(G, T, k)$ instance of $(G, \lambda, \Sigma, T, k)$.} of $\Pi$ where $G$ is a graph,  a subset $T$ of $V(G)$ and an integer $k$.
Set $T$ is the collection of terminal vertices, an approximate solution\footnote{There is a polynomial time algorithm that constructs such a solution.}, and the collection of terminal vertices for the first, second and third problem, respectively.
Kratsch and Wahlstr{\"{o}}m \cite{KratschW20} proved that the first two problem admit randomized polynomial kernels parameterized by solution size whereas Hols and Kratsch \cite{HolsK18} proved the same result for the third problem.
In this subsection, we modify their kernelization algorithms to obtain Theorems~\ref{thm:H-minor-graph-Pi-compression-vertex-SFVS} and \ref{thm:H-minor-graph-Pi-compression-vertex-GFVS}.
For reader's convenience,  we restate the theorems here.

\begin{theorem}[Theorem \ref{thm:H-minor-graph-Pi-compression-vertex-SFVS} restated]
\label{thm:H-minor-graph-Pi-compression-vertex-SFVS-restated}
Let $\Pi$ be \textsc{Vertex Multiway Cut with Deletable Terminal} or \textsc{Subset Feedback Vertex Set}.
Then there is an algorithm that, given an instance $(G,T,k)$ of \textsc{$H$-Minor Free $\Pi$}, constructs an equivalent instance $(G',T,F,k')$ of \textsc{$H$-Minor Free $\Pi$ with Undeletable vertices} in randomized polynomial time and with failure probability $\calO(2^{-|V(G)|})$ such that $|V(G')| \in (c_H \cdot k)^{\calO(1)}$ and $k' \leq k$.
\end{theorem}

\begin{theorem}[Theorem \ref{thm:H-minor-graph-Pi-compression-vertex-GFVS} restated]
\label{thm:H-minor-graph-Pi-compression-vertex-GFVS-restated}
 There is an algorithm that, given an instance $(G,\lambda, \Sigma, T,k)$ of \textsc{$H$-Minor-Free Group FVS}, constructs an equivalent instance $(G', \lambda', \Sigma, T,F,k')$ of \textsc{$H$-Minor-Free Group FVS with Undeletable vertices} in randomized polynomial time and with failure probability $\calO(2^{-|V(G)|})$ such that $|V(G')| \in k^{\calO(|\Sigma| \cdot |V(H)|)}$ and $k' \leq k$.
\end{theorem}

As mentioned in the previous section, \textsc{Odd Cycle Transversal} problem can be encoded as an instance of \textsc{Group FVS} with group $\mathbb{Z}_2$.
As both these problems are \textsf{NP-Complete}, we can conclude that the similar result hold for \textsc{$H$-Minor Free Odd Cycle Transversal} problem.

If we consider \textsc{Planar $\Pi$} problems, then we can improve the above compression algorithms to kernelization algorithms. 

\begin{theorem}
\label{thm:planar-Pi-kernel-vertex-SFVS}
 Let $\Pi$ be \textsc{Vertex Multiway Cut with Deletable Terminal} or \textsc{Subset Feedback Vertex Set}.
There is an algorithm that given an instance $(G,T,k)$ of \textsc{Planar $\Pi$} constructs an equivalent instance $(G',T,k')$ of \textsc{Planar $\Pi$} in randomized polynomial time and with failure probability $\calO(2^{-|V(G)|})$ such that $|V(G')| \in k^{\calO(1)}$ and $k' \leq k$.
\end{theorem}

\begin{theorem}
\label{thm:planar-Pi-kernel-vertex-GFVS}
There is an algorithm that given an instance $(G,\lambda, \Sigma,T,k)$ of \textsc{Planar Group Feedback Vertex Set} constructs an equivalent instance $(G',\lambda', \Sigma, T,k')$ of \textsc{Planar Group Feedback Vertex Set} in randomized polynomial time and with failure probability $\calO(2^{-|V(G)|})$ such that $|V(G')| \in k^{\calO(|\Sigma|)}$ and $k' \leq k$.
\end{theorem}

\paragraph{Overview of the subsection:}
The kernelization algorithms of Kratsch and Wahlstr{\"{o}}m \cite{KratschW20}, and  Hols and Kratsch \cite{HolsK18} can be seen broadly as \emph{Mark and Torso} scheme.
We focus of `Marking' scheme and `Torso-ing' scheme separately.
We need a marking scheme that works while keeping the input graph $H$-minor free.
In the cases of \textsc{Vertex MwC-DT} and \textsc{Group FVS}, algorithms by Kratsch and Wahlstr{\"{o}}m \cite{KratschW20} satisfy this criteria. 
However, in the case of \textsc{Subset FVS}, the marking scheme of Hols and Kratsch \cite{HolsK18} introduces some edges.
We present a different marking scheme that does not introduce any new edge.
The `Torso-ing' scheme may also introduce some new edges in the graph.
We present \emph{Contract to Undeletable} scheme that ensures that the reduced graph is a minor of the graph in the input instance.
Moreover, if the graph in the input instance is $H$-minor free then we are able to bound the size of the graph in the reduced instance. 
In the cases of \textsc{Vertex McW-DT} and \textsc{Subset FVS}, it is easy to replace \emph{Torso} operation by \emph{Contract to Undeletable} scheme.
We present unified arguments for these two cases in Lemma~\ref{lemma:compression-H-minor-free-Pi-vertex-SFVS}.
This leads to the results of Theorem~\ref{thm:H-minor-graph-Pi-compression-vertex-SFVS-restated}.
However, in the case of \textsc{Group FVS}, the labeling of edges makes it difficult to contract them.
In this case, we need different set of arguments to bound the size of reduced instance which we present in Lemma~\ref{lemma:compression-H-minor-free-Pi-vertex-GFVS}.
The proof of Theorem~\ref{thm:H-minor-graph-Pi-compression-vertex-GFVS-restated} follows directly from this.

A drawback of \emph{Contract to Undeletable} scheme is that we are left with some undeletable vertices in the graph.
Hence, in the case of \textsc{$H$-Minor Free $\Pi$}, we can only present a compression to \textsc{$H$-Minor Free $\Pi$ with Undeletable Vertices}.
In the restricted case of \textsc{Planar $\Pi$}, we are able to replace these undeletable vertices by a large grid without affecting the planarity.
See Lemma~\ref{lemma:kernel-planar-Pi-vertex-SFVS} and Lemma~\ref{lemma:kernel-planar-Pi-vertex-GFVS}.
This proves the kernelization results stated in Theorem~\ref{thm:planar-Pi-kernel-vertex-SFVS} and Theorem~\ref{thm:planar-Pi-kernel-vertex-GFVS}.

\subsubsection{Replacing `Torso' by `Contract to Undeletable' }

The randomized kernelization algorithms in \cite{KratschW20} and \cite{HolsK18} first mark a subset $W \subseteq V(G)$ of size $k^{\calO(1)}$ such that $T \subseteq W$ and if there is a solution of $(G, T, k)$, then there is a solution contained in $W$ with very high probability.
It then marks all vertices in $V(G) \setminus W$ as undeletable. 
The algorithm constructs a graph $G'$ from a copy of $G$ by taking a torso operation as follows:
It deletes all vertices in $V(G) \setminus W$.
For two vertices $u, v \in W$, it adds edge $uv$ to $E(G')$ if there exists a $u, v$-path in $G$ with internal vertices from $V(G) \setminus W$. 
The algorithm obtains a kernel of desired size by removing all but two parallel edges between any two vertices.

Note that the torso operation may introduce new edges.
Hence if the input graph is $H$-minor free then the resulting graph after taking the torso operation mentioned above may not be $H$-minor free.
In \emph{Contract to Undeletable} scheme, instead of taking the torso operation,  we construct graph $G'$ from $G$ by contracting each connected components of $G - W$ to a single vertex.
Let $F$ be the set of vertices obtained after the contraction.
For two vertices $u, v \in F$, if $N_G'(u) \subseteq N_{G'}(v)$ then we delete $u$.
As $G'$ is obtained from $G$ by contracting some edges and deleting some vertices, $G'$ is a minor of $G$.
We argue that the number of remaining vertices in $F$ is at most $\calO(c_H \cdot |W|)$ using the fact that $G$ is an $H$-minor free graph.
Here, $c_H$ is a constant that depends only on the size of $H$.
The following lemma,  whose proof we differ to later part of the section, will be useful to prove such bound.

\begin{lemma} 
\label{lemma:bipartite-H-minor}
For a fixed graph $H$, let $G = (X \cup Y, E)$ be a simple $H$-minor free bipartite graph.
If  for  all  distinct $u, v \in Y$,  $N(u) \not\subseteq N(v)$ then $|Y| \le c_H |X|$. 
Here, $c_H$ is a constant that depends only on the size of $H$.
\end{lemma}

In the following lemma, let $\Pi$ be either \textsc{Vertex MwC-DT} or \textsc{Edge-Subset FVS}.
We find it convenient to present the following lemmas in terms of \textsc{Edge-Subset FVS} instead of \textsc{Subset FVS}.
Please see the corresponding section for formal definitions these two variation of the problem. 
For an instance $(G, T, k)$ of \textsc{Edge-Subset FVS} problem,  set $T$ denotes the collection of end points of terminal edges $T_E$.

\begin{lemma}
\label{lemma:compression-H-minor-free-Pi-vertex-SFVS}
Consider an instance $(G, T, k)$ of  \textsc{$H$-Minor-Free $\Pi$} and suppose there is a subset $W \subseteq V(G)$  with the guarantee that $T \subseteq W$ and if there is a solution of $(G, T, k)$, then there is a solution contained in $W$.
Then, one can obtain an equivalent instance $(G^{\circ}, T, F, k)$ of \textsc{$H$-Minor-Free $\Pi$ with Undeleteable Vertices} in polynomial time such that $|V(G')| \in \calO(c_H \cdot |W|)$.
\end{lemma}
\begin{proof}
Consider an algorithm that given an instance $(G, T, k)$ of \textsc{$H$-Minor-Free $\Pi$} applies the following two reduction rules.
\begin{enumerate}
\item It contracts each connected component of $G - W$ to a vertex.
Let $G'$ be the resulting graph and let $F$ be the collection of all the vertices obtained after contraction i.e. $F = V(G') \setminus W$.
\item If there are two vertices $u, v \in F$ such that $N_{G'}(u) \subseteq N_{G'}(v)$, then delete $u$.
Let $G^{\circ}$ be the resulting graph.
The algorithm returns $(G^{\circ}, T, F, k)$ as the reduced instance.
\end{enumerate}

The algorithm obtains $G^{\circ}$ from $G$ by series of edge contractions and vertex deletion operations.
Hence, $G'$ is a minor of $G$ .
It is easy to see that the algorithm terminates in polynomial time.
We first argue that the reduction rules are safe and then bound the number of vertices in $G^{\circ}$.

\noindent \textbf{(1)}
For every vertex $u \in F \subseteq V(G')$, let $M_u$ be the corresponding connected component of $G - W$.
Alternately, vertex $u$ is obtained by contracting $M_u$ to a single vertex.
For any subset $Z$ of $W$, let $A$ be a connected component of $G - Z$.
As $F \cap W = \emptyset$, for any $u \in F$, either $M_u \subseteq A$ or $M_u \cap A = \emptyset$.
By the construction, for any connected component $A'$ of $G' - Z$, there is a connected component $A$ of $G - Z$ such that $A'$ can be obtained from $A$ by contracting all $M_u$ that intersects with $A$ to a single vertex.

Let $Z$ be a solution of $(G, T, k)$ such that $Z \subseteq W$ and hence $Z \cap F = \emptyset$.
Such a solution exists by the property of $W$.
We argue that $Z$ is a solution of $(G, T, F, k)$.
Assume this is not true.
In the case of \textsc{Vertex MwC-DT}, this implies that there exists a connected component $A'$ of $G' - Z$ such that it contains at least two terminals.
But this implies the corresponding connected component $A$ of $G - Z$ also contains at least two terminals.
This contradicts the fact that $Z$ is a solution of $(G, T, k)$.
Using the similar arguments for \textsc{Edge-Subset FVS}, a $T_E$-cycle in $A'$ implies a $T_E$-cycle in $A$, which contradicts the fact that $Z$ is a solution of $(G, T, k)$.
Hence, our assumption is wrong and $Z$ is a solution of $(G, T, F, k)$.

To prove the reverse direction, let $Z'$ be a solution of $(G', T, F, k)$.
Assume $Z'$ is not a solution of $(G, T, k)$.
In the case of \textsc{Vertex MwC-DT}, this implies there is a connected component $A$ of $G - Z'$ that has at least two terminals.
As $T \subseteq W$, this implies that there are at least two terminals in the corresponding connected component $A'$ of $G' - Z'$.
This contradicts the fact that $Z'$ is a solution of $(G', T, F, k)$.
Using the similar arguments in the case of \textsc{Edge-Subset FVS}, any $T_E$-cycle in $A$ imply $T_E$-cycle in $A'$ (as $V(T_E) = T \subseteq W$) which contradicts the fact that $Z'$ is a solution of $(G', T, F, k)$.
Here, $T_E$ is the collection of terminal edges in the graph and hence $T = V(T_E)$.
Hence our assumption is wrong and $Z'$ is a solution of $(G, T, k)$.
This concludes the proof of correctness of the first reduction rule.

\noindent \textbf{(2)} For the second reduction rule, the correctness of the forward direction follows from the fact that $G^{\circ}$ is a subgraph of $G'$.
Let $Z^{\circ}$ be a solution of instance $(G^{\circ}, T, F, k)$ of \textsc{Vertex MwC-DT}. 
If $Z^{\circ}$ is not a solution of $(G',  T, F, k)$,  there there are two terminals that are connected in $G' - Z^{\circ}$.
Moreover, vertex $u$ is present in the path connecting these two terminals.
As $N_{G'}(u) \subseteq N_{G'}(v)$, and $v \in F$ and hence $v \not\in Z^{\circ}$, there is a path connecting $t_1, t_2$ in $G^{\circ} - Z^{\circ}$.
This contradicts the fact that $Z^{\circ}$ is a solution of $(G^{\circ}, T, F, k)$.
Let $Z^{\circ}$ be a solution of instance $(G^{\circ}, T, F, k)$ of \textsc{Edge-Subset FVS}. 
If $Z^{\circ}$ is not a solution of $(G',  T, F, k)$,  there there is a $T_E$-cycle that contains vertex $u$ in $G' - Z^{\circ}$.
As $N_{G'}(u) \subseteq N_{G'}(v)$, and $v \in F$ and hence $v \not\in Z^{\circ}$, there is a $T$-cycle in $G^{\circ} - Z^{\circ}$.
This contradicts the fact that $Z^{\circ}$ is a solution of $(G^{\circ}, T, F, k)$.
This implies that the second rule is safe for any problem $\Pi$.

It remains to argue the bound on the number of vertices in $G^{\circ}$. 
As $G^{\circ}$ is a minor of $G$, which  is an $H$-minor free graph, $G^{\circ}$ is an $H$-minor free graph.
The vertices in $W \subseteq V(G)$ are not affected by either of reduction rules, and hence $W \subseteq V(G')$.
An exhaustive application of the first reduction rule ensures that $V(G^{\circ}) \setminus W$ is an independent set in  $G^{\circ}$.
Consider bipartite graph obtained from $G^{\circ}$ by deleting all edges in $E(G^{\circ})$ whose both endpoints are in $W$.
As the algorithm applies the second reduction rule exhaustively, this bipartite graph satisfies the premise of Lemma~\ref{lemma:bipartite-H-minor}. 
Hence, the number of vertices in $G'$ is at most $\calO(c_H \cdot |W|)$. 
This concludes the proof of the lemma.
\end{proof}

In the following lemma, we argue that in the case of \textsc{Planar $\Pi$} problems, we can encode undeletable vertices by a large grid while maintaining planarity.

\begin{lemma}
\label{lemma:kernel-planar-Pi-vertex-SFVS}
Consider an instance $(G, T, k)$ of a problem  \textsc{Planar $\Pi$} and suppose there is a subset $W \subseteq V(G)$ with the guarantee that $T \subseteq W$ and if there is a solution of $(G, T, k)$, then there is a solution contained in $W$.
Then, one can obtain an equivalent instance $(G', T, k)$ of \textsc{Planar $\Pi$} in polynomial time such that $|V(G')| \in \calO(k^2 \cdot |W|)$.
\end{lemma}
\begin{proof}
Consider an instance $(G, T, k)$ of \textsc{Planar $\Pi$}.
Using Lemma~\ref{lemma:compression-H-minor-free-Pi-vertex-SFVS}, one can compute an equivalent instance $(G', T, F,k)$ of \textsc{Planar $\Pi$ with Undeletable Vertices} such that $|V(G')| \in \calO(|W|)$.
To complete the proof,  it is sufficient to prove that there is an algorithm that given an instance $(G',T, F,k)$ of \textsc{Planar $\Pi$ with Undeletable Vertices},  runs in polynomial time, and returns an equivalent instance $(G^{\circ}, T,  k)$ of \textsc{Planar $\Pi$} such that $|V(G^{\circ})| \in \calO(k^2  \cdot |V(G')|)$.

Recall that the algorithm in Lemma~\ref{lemma:compression-H-minor-free-Pi-vertex-SFVS} constructs set $F$ of undeletable vertices by contracting connected components of $G - W$ for some subset $W \subseteq V(G)$.  
Hence, $F$ is an independent set in $G'$.

Consider an algorithm that given an instance $(G', T, F,k)$ of \textsc{Planar $\Pi$ with Undeletable Vertices} construct the graph $G^{\circ}$ from $G'$ as follows.
Let $u \in F$ be a vertex of degree $d := \deg_G(u)$.
The algorithm replaces the vertex $u$ by a $(2d(k+1)) \times (k+1)$ grid with vertices $u_{i,j}$ for all $i \in [2d(k+1)]$ and $j \in [k+1]$.
Let $M_u$ be the set of all vertices $u_{i,j}$ for $i \in [2d(k+1)]$ and $j \in [k+1]$.
The algorithm adds grid edges $u_{i,j}u_{i+1,j}$ for all  $i \in [2d(k+1)-1]$, $j \in [k+1]$ as well as $u_{i,j}u_{i,j+1}$ for all $i \in [2d(k+1)]$, $j \in [k]$.
Let $w_1,\dots,w_d$ be the neighbors of $u$ (we choose numbering according to the cyclic order of a planar embedding of $G$).
As $F$ is an independent set in $G'$,  for every $r \in [d]$,  $w_r \notin F$.
The algorithm adds edges from $w_r$ to $u_{2(r-1)(k+1)+2j,1}$ for all $j \in [k+1]$ in a way that preserves planarity.
It returns $(G^{\circ}, T, k)$ as the reduced instance.

It is clear that the algorithm terminates in polynomial time and $G^{\circ}$ is a planar graph.
As $G$ is a planar graph,  the sum of degrees of vertices in $F$ is at most $|V(G)|$.
This implies that the number of new vertices added is at most $\calO(k^2 \cdot |V(G)|)$.
It remains to argue that both instances are equivalent.

Let $Z'$ be a solution of $(G', T, F, k)$. 
We claim that $Z'$ is also a solution of $(G^{\circ},T,k)$.
Let $A^{\circ}$ be a connected component of $G^{\circ} - Z'$. 
Since $G^{\circ}[M_u]$ is connected and $M_u \cap Z' = \emptyset$ for every $u \in F$ we get that $A^{\circ}\cap M_u = \emptyset$ or $M_u \subseteq A^{\circ}$ for all $u \in F$.
Hence, there is a corresponding connected component $A'$ of $G' - Z'$ where each set $M_u \subseteq A^{\circ}$ is replaced by the vertex $u$.

As $Z'$ is a solution of $(G',T,F,k')$, in the case of \textsc{Vertex MwC-DT}, this implies that there is at most one terminal in $A'$.
As $u \notin T$, this implies that there is at most one terminal in $A^{\circ}$.
Similarly, in the case of \textsc{Edge-Subset FVS}, there is no $T_E$-cycle completely contained in $A'$.
Here, $T_E$ is the collection of terminal edges in the graph and hence $T = V(T_E)$.
As $u \notin V(T'_E)$, this implies that there is no $T^{\circ}_E$-cycle which is completely contained in $A^{\circ}$.
Since, $A^{\circ}$ is an arbitrary connected component of $G^{\circ} - Z'$, this implies that $Z'$ is a solution of $(G^{\circ}, T, k)$.

In the other direction, suppose that $Z^{\circ}$ is a solution of $(G^{\circ},T,k)$.
We argue that $Z' := Z^{\circ} \setminus \bigcup_{u \in U} M_u$ is a solution of $(G',T,F,k)$.
Towards this end, we first argue that $Z'$ is a solution of $(G^{\circ},T,k)$.
Let $m_u \in Z^{\circ} \cap (\bigcup_{u \in U} M_u)$ and suppose that $m_u \in M_u$ for some $u \in F$.
It suffices to argue that $Z'' := Z^{\circ} \setminus \{m_u\}$ is a solution.
Assume not and consider the graph $G - Z''$.
In the case of \textsc{Vertex MwC-DT}, this implies that there are at least two terminals in $T$ that are connected in $G - Z''$.
Moreover, $m_u$ is present in the the path connecting these two terminals.
We claim that we can find a path connecting these two terminals that does not contain vertex $m_u$.
Indeed, this follows from the fact that, for each vertex which is adjacent to some vertices from $M_{F} := \bigcup_{u \in F}M_u$, there are $(k+1)$ vertex-disjoint segments connecting the same vertices outside of $M_{F}$.
Since all of these paths are not hit by $Z''$ we can reroute parts of the path connecting two terminals to avoid $m_u$.
This will contradict the fact that $Z''$ is a solution.
Hence, our assumption is wrong and $Z'$ is a solution of $(G^{\circ},T,k)$.
Now we can follow the same arguments as in the previous paragraph.
Let $A^{\circ}$ be a connected component of $G^{\circ} - Z'$. 
Hence,  there is at most one terminal inside $A^{\circ}$.
Consider the set $A' := (A \setminus F) \cup \bigcup_{u \in A \cap F} M_u$.
Then $A'$ is a connected component of $G' - Z'$ and there is  at most one terminal inside $A'$. 
As $A^{\circ}$ is an arbitrary connected component of $G^{\circ} - Z'$, this is true for any connected component of $G' - Z'$.
This concludes that $Z'$ is a solution of $(G^{\circ}, T,k)$.
If we replace the path connecting two terminals in the above arguments by $T_E$-cycle, then the similar results hold for \textsc{Edge-Subset FVS}.
This concludes the proof of the lemma.
\end{proof}

\subsubsection{\textsc{Vertex Multiway With Deletable Terminals}}

The kernelization algorithm of Kratsch and Wahlstr{\"{o}}m consists of the following four steps (See Section~$4.2$ in \cite{KratschW20}):
$(i)$ Bound the number of terminals;
$(ii)$ Analysis the solution (i.e.  prove that any minimum solution that contains the maximum number of terminals is `closest');
$(iii)$ Setting up the gammoid and applying the representative set lemma to obtain a collection of undeletable vertices; and 
$(iv)$ Shirking the input graph using torso operation.
We summarize the first three steps in the following lemma.

\begin{lemma}
\label{lemma:vertex-multicut-deletable-terminals-marking}
Suppose $(G, T, k)$ is a \yes\ instance of \textsc{Vertex MwC-DT}.
There exists a set $W \subseteq V(G)$ of size at most $\calO(k^3)$ such that $T \subseteq W$ and there is a solution of $(G, T, k)$ contained in $W$.
Furthermore, there is an algorithm that finds such a set in randomized polynomial time with failure probability $\calO(2^{-|V(G)|})$.
\end{lemma}

The above lemma, along with Lemma~\ref{lemma:compression-H-minor-free-Pi-vertex-SFVS} and Lemma~\ref{lemma:kernel-planar-Pi-vertex-SFVS}, prove the results about \textsc{Vertex MwC-DT} in Theorem~\ref{thm:H-minor-graph-Pi-compression-vertex-SFVS-restated} and Theorem~\ref{thm:planar-Pi-kernel-vertex-SFVS}, respectively.

\subsubsection{Basic Toolkits Used in Compression for Subset FVS}

Unlike in the case of \textsc{Vertex MwC-DT}, we can not use the existing results for \textsc{Subset FVS}.
We modified the known results in the next part using the preliminary results mentioned here.

\paragraph{Expansion Lemma}
Let $t$ be a positive integer and $G$ a bipartite graph with
vertex bipartition $(P, Q)$.
A set of edges $M \subseteq E(G)$ is called a \emph{$t$-expansion of $P$ into $Q$} if 
$(i)$ every vertex of $P$ is incident with exactly $t$ edges of $M$, and 
$(ii)$ the number of vertices in $Q$ which are incident with at least one edge in $M$ is exactly $t|P|$.
We say that $M$ \emph{saturates} the endpoints of its edges.
Note that the set $Q$ may contain vertices which are \emph{not} saturated by $M$.

We generalize this notation of $t$-expansion to \emph{$t$-expansion with priorities}.
Consider a bipartite graph $G$ with vertex partition $(P, Q)$ and whose edges are partitioned into $E_r$  and $E_b$.
In other words, every edge has either red or blue color.
Informally, if available, we prefer including a red colored edge over a blue colored edge in a $t$-expansion.
\begin{definition}
A set of edges $M \subseteq E(G)$ is called a \emph{$t$-expansion with priority of $P$ into $Q$} if the following conditions are true.
\begin{enumerate}
\item Every vertex of $P$ is incident with exactly $t$ edges of $M$.
\item The number of vertices in $Q$ which are incident with at least one edge in $M$ is exactly $t|P|$.
\item Consider an edge $pq \in E(G)$ such that $p \in P$, $q \in Q \setminus V(M)$ (and hence $pq \not\in M$).
If $pq \in E_r$ then all edges in $M$ that are incident on $p$ are in $E_r$.
\end{enumerate}
\end{definition}
We call the third property as \emph{priority property}.
We define a measure $\mu$ for $t$-expansion $M$ of $P$ into $Q$ as follows:
$\mu(M) = \sum_{p \in P} x^p_r \cdot x^p_b$, where $x^p_r = 1$ if set $\{q\ |\ q \in Q\setminus V(M)\ \land\ pq \in E_r  \}$ is non-empty and $0$ otherwise, and $x^p_b = |\{q\ |\ q \in Q \land pq \in E_b \cap M\}|$.
Note that $M$ is a $t$-expansion of $P$ into $Q$ with priority property if and only if $\mu(M) = 0$.
We argue that if there is $t$-expansion of $P$ into $Q$ then there is $t$-expansion  of $P$ into $Q$ with priority property.
Moreover, such an expansion can be easily found in polynomial time using the following modification.
Let $M \subseteq E(G)$ be a $t$-expansion of $P$ into $Q$.
If $\mu(M) = 0$, then it is a desired $t$-expansion.
Otherwise, let $p$ be a vertex in $P$ for which value of $x^p_r \cdot x^p_b$ is non-zero. 
This implies, there exists $q_1 \in Q$ and $q_2 \in Q \setminus V(M)$ such that $pq_1 \in E_b \cap M$,  $pq_2 \in E_r$ (note that $pq_2 \not\in M$).
Consider $M' = (M \cup \{pq_2\}) \setminus \{pq_1\}$.
It is easy to verify that $M'$ is a $t$-expansion of $P$ into $Q$ and $\mu(M') < \mu(M)$.
As this swapping operation produces another $t$-expansion with reduced measure, this process results in a $t$-expansion with measure zero in polynomial time.

The following generalization of Hall's Matching Theorem is known as \emph{expansion lemma}.

\begin{lemma}[{\cite[Lemma 2.18]{CyganFKLMPPS15}}]
 \label{lem:expansion-lemma}
 Let $t$ be a positive integer and $G$ be a bipartite graph with vertex bipartition $(P,Q)$ such that $|Q| > t |P|$ and there are no isolated vertices in $Q$.
 Then there is nonempty vertex sets $X \subseteq P$ and $Y \subseteq Q$ such that 
 $(i)$ $X$ has a $t$-expansion into $Y$, and 
 $(ii)$ no vertex in $Y$ has a neighbour outside $X$. Furthermore two such sets $X$ and $Y$ can be found in time polynomial in the size of $G$.
\end{lemma}

We get the following result by applying the simple modifications stated in the previous paragraph to the proof of Lemma~\ref{lem:expansion-lemma}.

\begin{lemma}\label{lem:expansion-lemma-priority} Let $t$ be a positive integer and $G$ be a bipartite graph with vertex bipartition $(P,Q)$ such that $|Q| > t |P|$ and there are no isolated vertices in $Q$.
Moreover,  $E(G) = E_r \uplus E_b$.
Then there is nonempty vertex sets $X \subseteq P$ and $Y \subseteq Q$ such that 
$(i)$ $X$ has a $t$-expansion with priority property into $Y$, and 
$(ii)$ no vertex in $Y$ has a neighbour outside $X$. Furthermore two such sets $X$ and $Y$ can be found in time polynomial in the size of $G$.
\end{lemma}

We need sets $X, Y$ of Lemma~\ref{lem:expansion-lemma-priority} with an additional property.
The following is an adaptation of the proof of the lemma in \cite[Lemma~$5$]{PhilipRST19} that states a similar result with a different premise.

\begin{lemma}\label{lem:extra-expansion-vertex} If the premises of Lemma~\ref{lem:expansion-lemma-priority} are satisfied then we can find, in polynomial time, sets $X, Y$ of the kind described in Lemma~\ref{lem:expansion-lemma-priority} and a vertex $w \in Y $ such that there exists a $t$-expansion $M$ (with priority property) from $X$ into $Y$ that does not saturate $w$.
\end{lemma}
\begin{proof}
Let $X, Y$ be sets of the kind guaranteed to exist by Lemma~\ref{lem:expansion-lemma-priority}, and let $M$ be a $t$-expansion of $X$ into $Y$ with priority property.
It is enough to show that we can find, in polynomial time, a pair $X,Y$ for which $|Y| > t|X|$ holds.
We give a proof by algorithm.
We start by setting $X= Y = \emptyset$.
\begin{enumerate}
\item Find sets $X'\subseteq P $ and $Y'\subseteq Q $ as guaranteed to exist by Lemma~\ref{lem:expansion-lemma-priority}.
Let $M'$ be a $t$-expansion of $X'$ into $Y'$ with priority property.
If $|Y'| >t|X'|$ then return $((X \cup X)',(Y \cup Y'))$.
Otherwise, if there is a vertex \(w \in Q\setminus Y'\) which has no neighbour in \(P \setminus X'\) then return $((X\cup X'),(Y\cup Y'\cup\{w\}))$.

\item At this point we have $|X'| < |P|$ and $|Y'|=t|X'|$.
From above we get that there is $t$-expansion, say $M$, from $X$ into $Y$.
Since $X \cap X' = \emptyset = Y \cap Y'$ we get that $M \cup M'$ is a $t$-expansion of $X \cup X'$ into $Y \cup Y'$.
Set $\hat{X} \gets X \cup X', \hat{Y} \gets Y\cup Y'$.
Then (i) there is a \(t\)-expansion of \(\hat{X}\) into \(\hat{Y}\), (ii) no vertex in \(\hat{Y}\) has a neighbour outside \(\hat{X}\), and (iii) \(|\hat{Y}|=t|\hat{X}|\).
We apply the modification mentioned earlier to ensure that $M \cup M'$ is a $t$-expansion of $X \cup X'$ into $Y \cup Y'$ with priority property.

\item Let \(\hat{P}=(P-\hat{X}),\hat{Q}=(Q-\hat{Y})\). Consider the subgraph \(\hat{G}=G[\hat{P}\cup{}\hat{Q}]\) and its vertex bipartition \((\hat{P},\hat{Q})\).
Since every vertex in $\hat{Q}$ has at least one neighbour in $\hat{P}$ (otherwise we would have returned in Step~$(1)$) we get that there are no isolated vertices in the set
$\hat{Q}$ in graph $\hat{G}$.
Since $|\hat{Y}| = t|\hat{X}|$ and \(|Q|> t |P|\) we have that $|\hat{Q}|=|Q| - t|\hat{X}| > t|P| - t|\hat{X}|=t(|P| - |\hat{X}|)> t |\hat{P}|$. 
Thus graph $\hat{G}$ and its vertex bipartition $(\hat{P},\hat{Q})$ satisfy the premises of Lemma~\ref{lem:expansion-lemma-priority}.
Set $G\gets\hat{G}$,$P\gets\hat{P}$, $Q\gets\hat{Q}$, $X\gets\hat{X}$, $Y\gets\hat{Y}$ and go to Step~$(1)$.
\end{enumerate}
This concludes the description of the algorithm.

Note that before Step~$(1)$ is executed it is always the case that $(i)$ there is a $t$-expansion with priority property from $X$ into $Y$, $(ii)$ no vertex in $Y$ has a neighbour outside $X$, and $(iii)$ $|Y|=t|X|$.
So we get that \emph{if} the algorithm terminates (which it does only at step (1)) it returns a correct pair of vertex subsets.
Hence the algorithm is correct.

The graph $G$ from the premise of Lemma~\ref{lem:expansion-lemma-priority} has a vertex bipartition $(P,Q)$ with $(|P| > 0,|Q|>0 ,|Q| > t |P|)$, and the sets $\hat{X},\hat{Y}$ in Steps~$(2)$ and $(3)$ satisfy $0<|\hat{X}|<|P|$ and $|\hat{Y}| = t|\hat{X}|$.
So the sets $\hat{P},\hat{Q}$ of Step~$(3)$ satisfy $|\hat{P}| >0, |\hat{Q}|>0, |\hat{Q}|> t|\hat{P}|\).
Thus the graph $\hat{G}$ computed in Step~$(3)$ has strictly fewer vertices than the graph $G$ passed in to the previous Step~$(1)$.
Since we update $G\gets\hat{G}$ before looping back to Step~$ (1)$, we get that the algorithm terminates in polnomially many steps.
This concludes the proof of the lemma.
\end{proof}

\paragraph{Flowers and Blockers}
\label{paragraph:flowers-blockers}
For $A \subseteq V(G)$ a path with endpoints in $A$ and internal vertices not in $A$ is called an \emph{$A$-path}.
\begin{proposition}[Gallai's Theorem~\cite{Gallai64}]
\label{prop:gallai-theorem}
Let $A \subseteq V(G)$ and $k \in \mathbb{N}$.
If the maximum number of vertex-disjoint $A$-paths is strictly less than $k + 1$, then there exists a set $B \subseteq V(G)$ of at most $2k$ vertices that intersect every $A$-path.
\end{proposition}

Consider a graph $G$ and set of terminal edges $T_E$.
A set $\{C^l_1, C^l_2, \dots, C^l_t\}$\footnote{We use $C^l$ to denote a cycle and $C$ to denote a connected component.} of $T_E$-cycles that contains vertex $z$ is called an \emph{$z$-flower} of order $t$ with respect to $T_E$, if the set of vertices $V(C^l_i) \setminus \{z\}$ are pairwise disjoint.

\begin{lemma}[{\cite[Lemma 14]{HolsK18}}]
\label{lemma:find-z-flower}
Consider a graph $G$, a set of terminal edges $T_E \subseteq E(G)$, and a vertex $z \in V(G)$.
There is an algorithm that runs in polynomial time and finds a $z$-flower of maximum order i.e. maximum number of $T_E$-cycles that intersects only in $z$.
\end{lemma}

A set $B \subseteq V(G) \setminus \{z\}$ of size $t$ is called an \emph{$z$-blocker} of size $t$ with respect to $T_E$, if each $T_E$-cycle through $z$ also contains at least one vertex of the set $B$.

\subsubsection{\textsc{Subset Feedback Vertex Set}}

As mentioned before, we can not directly use the existing results for \textsc{Subset FVS}.
We briefly sketch the kernelization algorithm of Hols and Kratsch~\cite{HolsK18} and highlight why is it not possible to get a result similar to that of Lemma~\ref{lemma:vertex-multicut-deletable-terminals-marking} for \textsc{Vertex MwC-DT} as a direct corollary of their work.
Before that we mention the following two variation of the problem.
For a subset $T$ of $V(G)$, a $T$-cycle is a cycle that contains at least one vertex in $T$.
We define $T_E$-cycle for a set $T_E \subseteq E(G)$ in the similar way.

\defproblem{\textsc{Subset Feedback Vertex Set (Subset FVS)}}{Graph $G$, set $T \subseteq V(G)$, and integer $k$}{Does there exist a set $Z \subseteq V(G)$ of size at most $k$ that hits all $T$-cycles?}

\defproblem{\textsc{Edge-Subset Feedback Vertex Set (Edge-Subset FVS)}}{Graph $G$, set $T_E \subseteq E(G)$, and integer $k$}{Does there exist a set $Z \subseteq V(G)$ of size at most $k$ that hits all $T_E$-cycles?}

\medskip
There are simple polynomial-time reductions between these two problems that do not change the solution size $k$~\cite{CyganPPW13}.
More precisely, given an instance $(G, T, k)$ of \textsc{Subset FVS}, one can obtain an equivalent instance $(G', T_E, k)$ of \textsc{Edge-Subset FVS} in polynomial time.
Also, given an instance $(G', T_E, k)$ of \textsc{Edge-Subset FVS}, one can obtain an equivalent instance $(G, T, k)$ of \textsc{Subset FVS} in polynomial time such that $|T| \le |T_E|$.
Hence, these two problems are equivalent with respect to the size of kernel and the running time of an FPT algorithm. 
It is also easy to verify that these reduction insures that $G$ is an $H$-minor free graph if and only if $G'$ is.

As in \cite{HolsK18}, we find it convenient to work with  \textsc{Edge-Subset FVS}.
Hols and Kratsch~\cite{HolsK18} presented a kernelization algorithm for \textsc{Edge-Subset FVS} on general graphs using the matroid-based tools of Kratsch and Wahlström~\cite{KratschW20}.
We present an overview of their algorithm and highlight the steps which may not preserve `$H$-minor-free-ness' of the input graph.

The kernelization algorithm by Hols and Kratsch consists of two parts. 
In the first part, they establish a randomized polynomial kernelization for the problem parameterized by $|T_E| + k$.
For an instance $(G, T_E, k)$, the algorithm marks subset $W \subseteq V(G)$ of size at most $\calO(k \cdot |T_E|^2)$ such that $V(T_E) \subseteq W$ and if there is a solution of $(G, T_E, k)$, then there is a solution contained in $W$.
It marks all vertices in $V(G) \setminus W$ as undeletable. 
The algorithm constructs a graph $G'$ from a copy of $G$ by taking a torso operation.

The second part of Hols and Kratsch algorithm reduce the size of set $T_E$ until it is polynomially bounded in $k$. 
During this reduction process, the algorithm detects certain pairs $\{u, v\}$ of different vertices with the property that each solution must contain at least one of the vertices.
To handle this case, the authors generalized the \textsc{Edge-Subset FVS} problem to \textsc{Pair Constrained Edge-Subset FVS} problem.
The input of this problem consists of an addition set $\calP$ of pair of vertices.
The objective is to find a solution such that for every  $\{u, v\} \in \calP$, either $u$ or $v$ is in it.
The algorithm reduces the number of terminal edges for instance $(G, T_E, k, \calP = \emptyset)$ of \textsc{Pair Constrained Edge-Subset FVS}.
It obtains an equivalent instance of \textsc{Edge-Subset FVS} as follows: for every $\{u, v\} \in \calP$, it adds $uv$ to $T_E$ and adds two parallel edges to $E(G)$ whose endpoints are $u, v$. 

The last operation of adding edges corresponding to pairs in pair-constraints might not preserve $H$-minor-free-ness of the input graph.
Hence, our algorithm can not rely on this generalization of \textsc{Edge-Subset FVS} problem.
We modify Hols and Kratch's algorithm to ensure that it does not add any edges in the existing graph.
This modification includes new reduction rule and a different application of $2$-Expansion Lemma.

We present the kernelization algorithm in the similar fashion as that of Hols and Kratch's algorithm.
First we present a randomized polynomial kernelization for parameter $|T_E| + k$.
Then,  we present a polynomial time algorithm to reduce the number of terminal edges. 

\paragraph{Randomized polynomial kernelization for parameter $|T_E|  + k$}
The first part of the Hol and Kratsch's algorithm consists of the following four steps (See Section~$3$ in~\cite{HolsK18}): $(i)$ Analyzing solutions (to prove that there is a \emph{dominant} solution of any \yes\ instance); $(ii)$ Setting up the gammoid; $(iii)$ Applying the representative set lemma; and $(iv)$ Shirking the input graph.
We summarize the first three steps in the following lemma.

\begin{lemma}
\label{lemma:subset-FVS-marking}
Suppose $(G, T_E, k)$ is a \yes\ instance of \textsc{Edge-Subset FVS}.
There exists a set $W \subseteq V(G)$ of size at most $\calO(k \cdot |T_E|^2)$ such that $V(T_E) \subseteq W$ and there is a solution of $(G, T_E, k)$ contained in $W$.
Furthermore, there is an algorithm that finds such a set in randomized polynomial time with failure probability $\calO(2^{-|V(G)|})$.
\end{lemma}

\paragraph{Reducing the size of $T_E$}
In this section, we present an algorithm that given an instance $(G, T_E, F, k)$ of $H$-\textsc{Minor-Free Edge-Subset FVS with Undeletable Vertices}, runs in polynomial time and returns another equivalent instance $(G', T'_E, F, k')$ such that the number of edges in $T'_E$ are bounded by some polynomial function of $k'$.

As in \cite{HolsK18}, we modify the input instance to ensure that if there is a set $Z \subseteq V(G)$ that hits all $T_E$-cycles then there is a set $Z' \subseteq V(G)$ that hits all $T_E$-cycles, $|Z'| \le |Z|$, and $Z' \cap V(T_E) = \emptyset$. 
We delete vertex $v \in V(G)$ and decrease $k$ by one if $v$ has a self-loop and the corresponding edge is in $T_E$.
After this, we delete all loops.
Consider parallel edges with endpoints $v$ and $w$.
If no edge among these parallel edges is in $T_E$, then we delete all except one edge.
If at least one these edges is in $T_E$, then we delete all except two edges, so that one of remaining edges is in $T_E$.
For every edge $uv$ in $T_E$, we subdivides it twice to get a path $(u, u_1, v_1, v)$.
In $T_E$, we replace $uv$ by $u_1v_1$.
Consider a set $Z \subseteq V(G)$ that hits all $T_E$-cycles in such modified graph.
One can find set $Z' \subseteq V(G)$ that hits all $T_E$-cycles such that $|Z'| \le |Z|$ and $Z' \cap V(T_E) = \emptyset$.

We start with restating some reduction rules.
\begin{reduction rule}[{\cite[Rule 1]{HolsK18}}]
\label{rr:trivial-no}
If $k \le 0$ and there is a $T_E$-cycle, then return a trivial \no\ instance.
\end{reduction rule}

\begin{reduction rule}[{\cite[Rule 2]{HolsK18}}]
\label{rr:delete-bridge} Delete all bridges and connected components not containing edge from $T_E$.
\end{reduction rule}

\begin{reduction rule}[{\cite[Rule 3]{HolsK18}}]\label{rr:terminal-bridge} If there is an edge $e \in T_E$ such that $e$ is a bridge in $G - (T_E \setminus \{e\})$, then remove $e$ from $T_E$ i.e. return instance $(G, T_E \setminus \{e\}, F, k)$.
\end{reduction rule}

Suppose there is a vertex set $Z \subseteq V(G)$ such that $Z$ hits all $T_E$-cycles and $Z \cap V(T_E) = \emptyset$.
We note that $Z$ may intersects with $F$.
 
\begin{reduction rule}[{\cite[Rule 6]{HolsK18}}] \label{rr:flower} Suppose there is a $z$-flower of order $k + 1$ in $G$ for a vertex $z \in Z$.
If $z \in F$ then return a trivial \no\ instance otherwise return instance $(G - \{z\}, T_E, F, k - 1)$.
\end{reduction rule}

The above reduction rule can be applied for any vertex  in $V(G)$.
But, as in \cite{HolsK18}, applying this rule only for vertices in $Z$ suffices for the purpose. 

We need to define some terms and a subroutine before stating our final reduction rule.
Consider a set of vertices $Z \subseteq V(G)$ that hits every $T_E$-cycle in $G$ and $Z \cap V(T_E) = \emptyset$.
Hence, every edge $e \in T_E$ is a bridge in $G - Z$.
Consider the graph obtained from $G - Z$ by removing all edges in $T_E$.
Following Cygan et al.~\cite{CyganPPW13} and Hols and Kratsch~\cite{HolsK18}, we call such a component \emph{bubble}.
By the definition, every bubble is a connected subgraph of $G$.
Also, no two bubbles can be connected by more than one edge of $T_E$.

Consider the graph $\calH^{+}_Z$ obtained from $G$ by contracting all edges in a spanning tree for every bubble in $G$. 
Note that every $T_E$-cycle in $\calH^{+}_Z$ must contain a vertex of the set $Z$.
Let $\calH_Z$ be the graph obtained from $\calH^{+}_Z$ by deleting all vertices in $Z$.
It is easy to see that $\calH_Z$ is a forest and $E(\calH_Z) = T_E$.
Define $\psi:V(\calH_Z) \rightarrow 2^{V(G) \setminus Z}$ as follows:
$\psi(u) = C$ if $C$ is a bubble (i.e. connected component) in $G - Z - T_E$ and $u$ is obtained by contracting all edges in a spanning tree of $G[C]$.
For a set $U \subseteq V(\calH_Z)$, $\psi(U) = \bigcup_{u \in U} \psi(u)$.

A vertex $u$ in $V(\calH_Z)$ is called solitary, leaf, and inner if the degree of $u$ in $\calH_Z$ is zero, one, and at least two, respectively.
As $\calH_Z$ is a forest and $E(\calH_Z) = T_E$, it is sufficient to bound the number of non-solitary vertices in $\calH_Z$ to bound the size of $T_E$.
To bound these, we argue that any simple path in $\calH_Z$, with all internal vertices of degree two, is of length $\calO(k^2)$.
Also, the number of non-solitary leaves in $\calH_Z$ is at most $\calO(k^3)$.
We denote the collection of non-solitary leaves in $\calH_Z$ by $L(\calH_Z)$. 
As $\calH_Z$ is obtained from $\calH^{+}_Z$ by removing vertices in $Z$, we have $L(\calH_Z) \subseteq V(\calH^{+}_Z)$.
For a vertex $z \in Z$, let $L_z$ be the set of non-solitary leaves adjacent with $z$ in $\calH^+_Z$.

We now specify a subroutine used by the reduction rule to delete some edges in graph $G$.
The subroutine identifies a \emph{high degree} vertex $z \in Z \subseteq V(G)$ and deletes some edges incident on it.
In $\calH^+_Z$, these edges correspond to edges incident on $z$ and some leaf in $L_z$.
We describe the subroutine assuming that a blocker and a bipartite graph constructed by the subroutine satisfies desired properties.
We justify these assumptions in Claims~\ref{claim:intersection-inclusion}, \ref{claim:premise-expansion}, and \ref{claim:subroutine-running}.

\subparagraph{Subroutine.}
The subroutine takes as input an instance $(G, T_E, F, k)$ and set of vertices $Z \subseteq V(G)$ of size at most $8k$ that hits all $T_E$-cycles and $Z \cap V(T_E) = \emptyset$, and returns another instance $(G', T_E', F', k')$.
Without loss of generality, we assume that Reduction Rules~\ref{rr:trivial-no}, \ref{rr:delete-bridge}, \ref{rr:terminal-bridge}, and \ref{rr:flower} are not applicable on $(G, T_E, F, k)$.

The subroutine starts by constructing graph $\calH^+_Z$ as described above.
If all vertices in $Z$ are adjacent with at most $10(k + 2)^2$ non-solitary leaves in $L(\calH_Z)$ then the algorithm returns $(G, T_E, F, k)$.
Otherwise, let $z$ be a vertex in $Z$ that is adjacent with at least $10(k + 2)^2 + 1$ non-solitary leaves in $L(\calH_Z)$.

Let $L_z$ be the set of non-solitary leaves adjacent with $z$ in $\calH^+_Z$.
Let $T_E^{\circ} \subseteq T_E$ be the collection of edges that are incident on $L_z$.
The subroutine computes a blocker $B_z$  with respect to $T_E^{\circ}$  in graph $G - (Z \setminus \{z\})$ such that its size is at most $2k$,  and $B_z \cap V(T_E^{\circ}) = \emptyset$ and $B_z \cap \psi(L_z) = \emptyset$.
Let $\calC_1$ be the collection of all connected components of  $G - (\{z\} \cup B_z \cup (Z \setminus \{z\}))$\footnote{We prefer this notation over $G - (B_z \cup Z)$ as $\{z\}$ and $B_z \cup (Z \setminus \{z\})$ play different roles in the subroutine.}.
Let $\calC_2$ be the subset of $\calC_1$ such that for every $C$ in $\calC_2$, $V(C) \cap \psi(L_z) \neq \emptyset$.
In Claim~\ref{claim:intersection-inclusion},  we argue that for any $C$ in $\calC_2$, there is a unique vertex $u$ in $L_z$ such that $V(C) \cap \psi(u) \neq \emptyset $.
Moreover, if $V(C) \cap \psi(u) \neq \emptyset$ then $\psi(u) \subseteq V(C)$. 

The subroutine constructs an auxiliary bipartite graph $G_z(B_z \cup (Z \setminus \{z\}), \calC)$ as follows:
For every connected component $C_i$ in $\calC_2$, it adds a vertex $v_i$ corresponding to it in $\calC$.
It adds an edge $uv_i$ for $u \in B_z \cup (Z \setminus \{z\})$ and $v_i \in \calC$ if and only if there is an edge in $E(G)$ incident on $u$ whose another endpoint is in $V(C_i)$.
It partitions $E(G_z)$ into $E_r$ and $E_b$ as follows:
Consider edge $uv_i$ in $E(G_z)$. 
If there is $T_E^{\circ}$-path from $z$ to $u$ whose all interval vertices are in $V(C_i)$ then the subroutine includes $uv_i$ in $E_r$, otherwise it includes it in $E_b$.

Suppose, $G_z$ and $t = k + 2$ satisfies the premise of Lemma~\ref{lem:expansion-lemma-priority}, and sets $X \subseteq B_z \cup (Z \setminus \{z\}), Y \subseteq \calC$ and a vertex $y \in Y$ are of the kind described in Lemma~\ref{lem:extra-expansion-vertex}. 
The subroutine construct graph $G'$ by deleting all the edges incident on $z$ whose other endpoint is in $V(C_y) \cap \psi(L_z)$. 
Here, $C_y$ is the connected component corresponding to vertex $y$ in $\calC$.
It returns instance $(G', T_E, F, k)$.
This completes the description of the subroutine.
See Figure~\ref{fig:graph-partition}

\begin{figure}
\begin{center}
\includegraphics[scale=0.9]{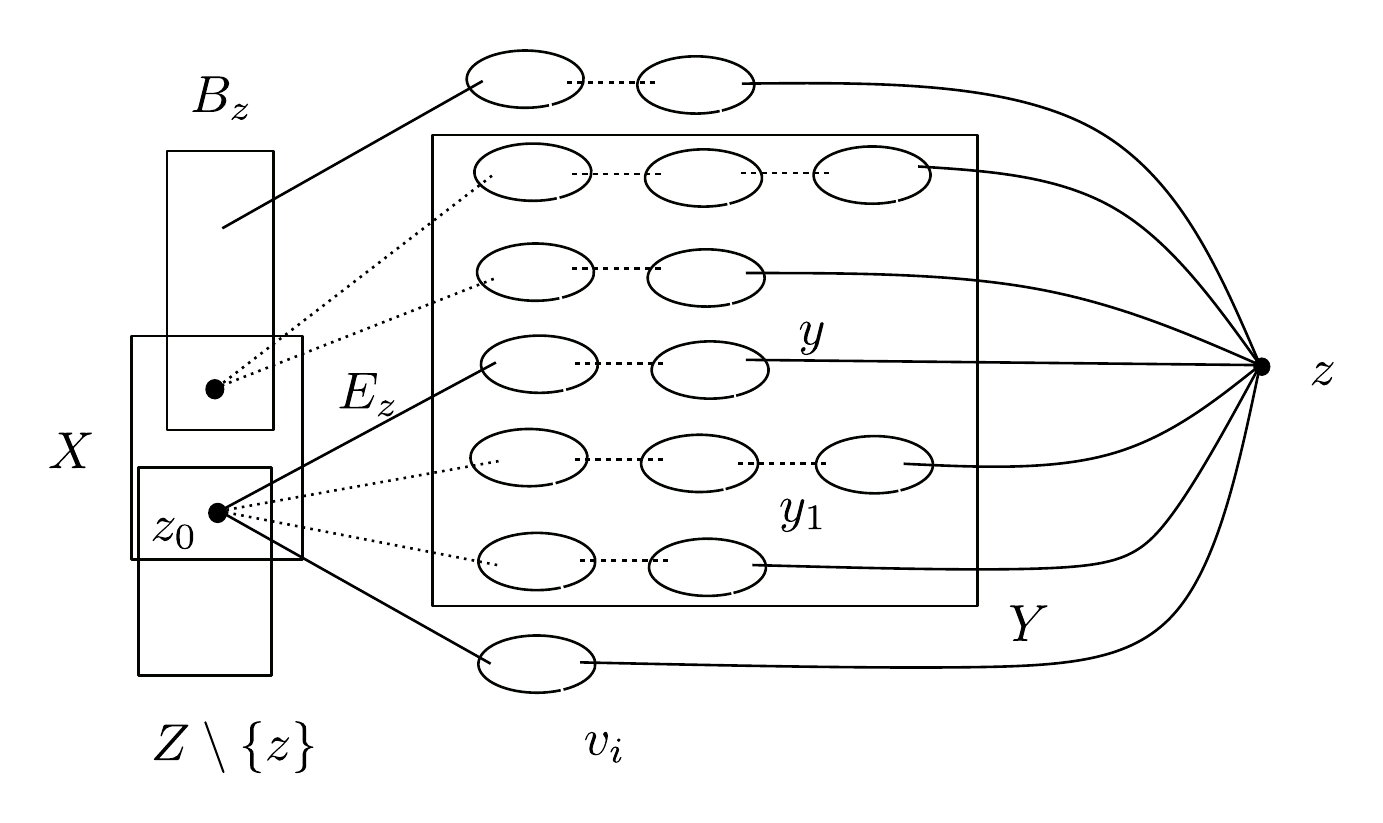}
\end{center}
\caption{Partition of the graph while apply the subroutine. 
Dashed edges represents terminal edges while dotted edges are edges in matching $M$. 
Vertices $y, y_1$ and $v_i$ in $G_z$ corresponds to connected components $C_y, C_1$ and $C_i$ in $G$, respectively in the proof of Lemma~\ref{lemma:rr-expansion-correct}.\label{fig:graph-partition}}
\end{figure}

We argue the running time and some desired properties of intermediate states in the subroutine in the following claims.

\begin{claim}\label{claim:blocker-properties}
There exists a blocker with the desired properties.
Moreover, the subroutine can compute such blocker in time polynomial in the size of input. 
\proof
As Reduction~\ref{rr:flower} is not applicable on $(G, T_E, F, k)$,  there is no $z$-flower of order $k + 1$ in $G$ with respect to $T_E$.
As $T_E^{\circ} \subseteq T_E$, there is no $z$-flower of order $k + 1$ in $G - (Z \setminus \{z\})$ with respect to $T_E^{\circ}$.
The maximum number of vertex disjoint $\psi(L_z)$-paths in $G - (Z \setminus \{z\})$ is at most $k$, as otherwise the $\psi(L_z)$-paths together with vertex $z$ would correspond to a $z$-flower of order $k + 1$ in $G - (Z \setminus \{z\})$.
From Proposition~\ref{prop:gallai-theorem}, it follows that there exists a set $B_z \subseteq V(G) \setminus Z$ of size at most $2k$ intersecting every $L_z$-path.
Since every $T_E^{\circ}$-cycle through $z$ in $G - (Z \setminus \{z\})$ must contain an $L_z$-path in $G - Z$, set $B_z$ is a $z$-blocker of size at most $2k$ in $G - (Z \setminus \{z\})$.

We modify $z$-blocker $B_z$, without increasing its size, to obtain another $z$-blocker that does not intersect with $\psi(L_z)$.
Suppose there exists a vertex $u$ in $L_z$ such that $B_z \cap \psi(u) \neq \emptyset$.
As $u$ is a non-solitary leaf in $\calH_Z$, there exists a unique vertex $u_1$ in $\calH_Z$ which is adjacent with $u$ in $\calH_Z$.
Moreover, the unique edge $ww_1$, where $w \in \psi(u)$ and $w_1 \in \psi(u_1)$, is a terminal edge in graph $G$ (more specifically $ww_1 \in T_E^{\circ}$).
By the definition of $T_E^{\circ}$, no $T_E^{\circ}$-cycle has all its vertices in $\{z\} \cup \psi(u)$.
By the definition of $\psi$, the only vertex adjacent with $\psi(u)$ in $V(G) \setminus Z$ is $w_1$.
Hence, $B'_z = B_z \cup \{w_1\}\setminus \psi(u)$ is another $z$-blocker of size at most $|B_z|$ that intersects fewer vertices in $\psi(L_z)$.
Repeated application of this process results in a $z$-blocker that does not intersect $\psi(L_z)$.

For notational convenience, we assume that $B_z$ is a $z$-blocker of size at most $2k$ that does not intersect $\psi(L_z)$.
It remains to argue that $B_z$ does not intersect $V(T^{\circ}_E)$.
Suppose, there exist edge $ww_1$ in $T_E^{\circ}$ such that $B_z \cap \{w, w_1\} \neq \emptyset $.
By the construction, there are two vertices $u, u_1 \in V(\calH_Z)$ such that $u \in L_z$ and $ww_1$ is the terminal edge across $\psi(u)$ and $\psi(u_1)$.
As $B_z \cap \psi(L_z) = \emptyset$, we have $w_1 \in B_z$.
Recall that the modifications mentioned at the start of this subsection ensure that there exists a unique neighbour say $w'_1$ of $w_1$ in $\psi(u_1)$.
Hence, $B_z \cup \{w'_1\} \setminus \{w_1\}$ is another $z$-blocker of size at most $|B_z|$ that intersects fewer vertices in $V(T_E^{\circ})$.
Repeated application of this process results in a $z$-blocker that does not intersect $V(T_E^{\circ})$.
\uend
\end{claim}

\begin{claim}[resume]\label{claim:intersection-inclusion} For any $C$ in $\calC_2$, there is a unique vertex $u$ in $L_z$ such that $V(C) \cap \psi(u) \neq \emptyset $.
Moreover, if $V(C) \cap \psi(u) \neq \emptyset$ then $\psi(u) \subseteq V(C)$. 
\proof
Let $C$ be any connected component of $G - (\{z\} \cup B_z \cup (Z \setminus \{z\}))$ such that there is a vertex corresponding to it in $\calC$.
We claim that there is at most one vertex $u$ in $L_z \subseteq V(\calH_Z)$ such that $V(C) \cap \psi(u) \neq \emptyset$.
Assume, for the same of contradiction, there are two vertices $u_1, u_2$ in $L_z$ such that $\psi(u_1)$ and $\psi(u_2)$ intersect with $V(C)$.
As $G - ({z} \cup B_z \cup(Z \setminus \{z\}))$ is a subgraph of $G - Z$, there is a connected component $C'$ of $G - Z$ such that $V(C) \subseteq V(C')$.
Consider a path in $G - Z$ that starts at a point in $\psi(u_1)$ and ends at a point in $\psi(u_2)$.
Note that this path does not contain vertex $z$.
By the construction of $\calH_Z$, the corresponding unique path connecting $u_1$ to $u_2$ in $\calH_Z$ contains an edge in $T^{\circ}_E$.
By the definition of $L_z$, there are edges in $G$ incident on $z$ whose other endpoints are in $\psi(u_1)$ and in $\psi(u_2)$.
By the construction of function $\psi$, both $\psi(u_1)$ and $\psi(u_2)$ are connected sets in $G$.
This implies that there is a $T_E^{\circ}$-cycle passing through $z$ that does not intersect $B_z$.
This contradicts the fact that $B_z$ is a $z$-blocker in $G - (Z \setminus \{z\})$.

For the second part,  assume that there is vertex $u \in L_z$ and a connected component $C$ of $G - (\{z\} \cup B_z \cup (Z \setminus \{z\}))$ such that $V(C) \cap \psi(u) \neq \emptyset$ and $V(C) \setminus \psi(u) \neq \emptyset$.
This implies $V(C) \setminus \psi(u) \subseteq (\{z\} \cup B_z \cup (Z \setminus \{z\}))$.
By the construction of $\calH_Z$,  $\psi(u) \cap Z = \emptyset$ and hence $V(C) \setminus \psi(u) \subseteq B_z$.
As $u \in L_z$,  this is a contradiction to the fact that $B_z \cap \psi(L_z) = \emptyset$.
Hence, our assumption is wrong.
\uend
\end{claim}

\begin{claim}[resume]
\label{claim:premise-expansion}
If there is a vertex $z$ in $Z$ that is adjacent with at least $10(k + 2)^2 + 1$ non-solitary leaves in $L(\calH_Z)$, then the bipartite graph $G_z(B_z \cup (Z \setminus \{z\}), \calC)$ and $t = k + 2$ satisfies the premise of Lemma~\ref{lem:expansion-lemma-priority}.
\proof
We argue that $|\calC| > (k + 2) \cdot |B_z \cup (Z \setminus \{z\})|$ and there is no isolated vertex in $\calC$.

By Claim~\ref{claim:intersection-inclusion},  for any $C$ in $\calC_2$, there is a unique vertex $u$ in $L_z$ such that $V(C) \cap \psi(u) \neq \emptyset$.
Hence,  the subroutine adds at least  $|L_z| \ge 10(k+2)^2 + 1$ many vertices to $\calC$.
As $|B_z| \le 2k$ and $|Z| \le 8k$, we have $|\calC| > (k + 2) \cdot |B_z \cup (Z \setminus \{z\})|$.

We now argue that there are no isolated vertices in $\calC$.
Assume, for the sake of contradiction, that there is a connected component $C$ of $G - (\{z\} \cup B_z \cup (Z \setminus \{z\}))$ such that there is no edge with one endpoint in $V(C)$ and another endpoint in $B_z \cup (Z \setminus \{z\})$.
As $C$ is not deleted by the subroutine, there is a vertex $u$ in $L_z$ such that $\psi(u) \subseteq V(C)$. 
Consider any edge, say $e$, incident on $u$ in $\calH_Z$.
By the construction, $e$ is an edge in $T_E$ (more specifically in $T_E^{\circ}$).
If there is a $T_E$-cycle containing edge $e$ and whose all vertices are in $z \cup V(C)$,  then there is $T_E^{\circ}$-cycle in $G - (B_z \cup (Z \setminus \{z\}))$.
Recall that any $T_E$-cycle in $G - (Z \setminus \{z\})$ passes through $z$.
This contradicts the fact that  $B_z$ is a $z$-blocker in $G - (Z \setminus \{z\})$.
Hence,  edge $e$ does not participate in any $T_E$-cycle in $G$.
But all such edges are removed by Reduction Rule~\ref{rr:terminal-bridge}.
This contradicts the fact that Reduction Rule~\ref{rr:terminal-bridge} is not applicable on $(G, T_E, F, k)$. 
Hence, every vertex in $\calC$ is adjacent with at least some vertices in $B_z$.
\uend
\end{claim}

\begin{claim}[resume]
\label{claim:subroutine-running}
The subroutine terminates in time polynomial in the size of input.
\proof
Given an instance $(G, T_E, F, k)$ and a set $Z \subseteq V(G)$, the subroutine can construct $\calH^+_Z$ in polynomial time.
If every vertex in $Z$ is adjacent with at most $10(k + 2)^2$ non-solitary leaves in $H_Z$ than the subroutine terminates and the claim is vacuously true.
Otherwise, let $z$ be a vertex in $Z$ such that $|L_z| \ge 10(k + 2)^2 + 1$.  
By Claim~\ref{claim:blocker-properties}, the subroutine can compute a blocker with desired properties in polynomial time.
By Lemma~\ref{lem:expansion-lemma-priority} and Lemma~\ref{lem:extra-expansion-vertex}, the subroutine can find desired sets and a vertex in polynomial time.
This concludes the proof of the claim.
\uend
\end{claim}

The last reduction rule applies the subroutine.

\begin{reduction rule}
\label{rr:expansion-lemma}
The reduction returns the instance that the subroutine outputs  when the input is instance  $(G, T_E, F, k)$ and set $Z \subseteq V(G)$.
\end{reduction rule}

We say Reduction Rule~\ref{rr:expansion-lemma} is `not applicable' if the input instance and the output instance are the same.
We now prove the correctness of the reduction rule.
Recall that sets $X \subseteq B_z \cup (Z \setminus \{z\}), Y \subseteq \calC$ and vertex $y \in Y$ are of the kind described in Lemma~\ref{lem:extra-expansion-vertex}.
Let $M$ be a $(k + 2)$-expansion of $X$ into $Y$ with priority property.
Set $E_z$ is the collection of (non-terminal) edges that are incident on $z$ and whose other endpoints are in $V(C_y) \cap \psi(L_z) = \psi(u)$.
Here, $C_y$ is the connected component corresponding to $y$ and $u$ is the vertex in $L_z$ such that $V(C_y) \cap \psi(u)  \neq \emptyset$.
Claim~\ref{claim:intersection-inclusion} implies $V(C_y) \cap \psi(L_z) = \psi(u)$.
Graph $G'$ is obtained from $G$ by deleting all edges in $E_z$.
We prove the following claim before proving the safeness of Reduction Rule~\ref{rr:expansion-lemma}.

\begin{claim}[resume]\label{claim:prop-sol-backword}
Consider a solution $Z'$ to $(G', T_E, F, k)$ such that $z \not\in Z'$.
Then, $X \setminus Z \subseteq Z'$.
\proof
Assume, for the sake of contradiction, that there exists a vertex $x \in (X \setminus Z) \setminus Z'$.
Let $Y_x$ be the collection of vertices in $\calC$ that are adjacent to $x$ via edges in $M$.
Formally, $Y_x = \{y_i \in \calC\ |\ xy_i \in M\}$. 
Note that $v_i$ is not in $Y_x$.
As $M$ is a $(k + 2)$-expansion of $X$ into $Y$, $|Y_x| = k + 2$.
Let $C_i$ be the connected component of $G - (\{z\} \cup B_z \cup (Z \setminus \{z\}))$ corresponding to vertex $y_i \in Y_x \subseteq \calC$.
As $|Y_x| = k + 2$ and $|Z'| \le k$, there are at least two vertices in $Y_x$, say $y_1, y_2$, such that $V(C_1) \cap Z' =  V(C_2) \cap Z' = \emptyset$.
Let $u_1, u_2$ be the unique vertices in $L_z$ such that $\psi(u_1) \subseteq V(C_1)$ and $\psi(u_2) \subseteq V(C_2)$.
Claim~\ref{claim:intersection-inclusion} implies that such vertices exists.
Consider a cycle $C^l$ that starts with $z$, passes through $V(C_1)$ to vertex $x$ and then goes back to $z$ through $V(C_2)$.
Such cycle contains vertices in $\psi(u_1)$ and $\psi(u_2)$.

We argue that $C^l$ contains an edge from $T_E$.
As $u_1$ is a non-solitary leaf in $\calH_Z$, there exists an unique edge, say $e_1$, incident on $u_1$ in $\calH_Z$.
By the construction of $\calH_Z$, edge $e_1$ is a terminal edge in $G$.
Recall that $\psi(u_1)$ is a connected component in graph $G - Z - T_E$.
Hence, any path connecting a vertex in $V(G) \setminus Z$ to a vertex in $\psi(u_1)$ contains edge $e_1$.
More specifically, the part of cycle $C^l$ that connects a vertex in $\psi(u_1)$ to vertex $x$ contains a terminal edge.
This implies there exists $T_E^{\circ}$-cycle in $G' - Z'$.
This contradicts the fact that $Z'$ is a solution of $(G, T_E, F, k)$.
Hence our assumption is wrong and $X \setminus Z \subseteq Z'$.
\uend
\end{claim}

In the following lemma, we argue that the reduction rule is safe.

\begin{lemma}
\label{lemma:rr-expansion-correct}
$(G, T_E, F, k)$ is a \yes\ instance if and only if $(G', T_E, F, k)$ is.
\end{lemma}
\begin{proof}
If every vertex in $Z$ is adjacent with at most $10(k + 2)^2$ non-solitary leaves in $L(\calH_Z)$ then the subroutine returns the same instance and the lemma is vacuously true. 
We consider the case where $G'$ is a proper subgraph of $G$.
Recall that sets $X \subseteq B_z \cup (Z \setminus \{z\}), Y \subseteq \calC$ and vertex $y \in Y$ are of the kind described in Lemma~\ref{lem:extra-expansion-vertex} and Set $E_z$ is the collection of (non-terminal) edges that are incident on $z$ and whose other endpoints are in $V(C_y) \cap \psi(L_z)$.

As $G'$ is a proper subgraph of $G$, any solution of $(G, T_E, F, k)$ is a solution of $(G', T_E, F, k)$.
We argue that a solution $Z'$ to $(G', T_E, F, k)$ is a solution of $(G, T_E, F, k)$.
By the construction, $G' = G - E_z$ and hence $G - \{z\} = G' - \{z\}$.
Consider a solution $Z'$ to $(G', T_E, F, k)$.
If $z \in Z'$ then $Z'$ is a solution of $(G, T_E, F, k)$.
In the remaining proof,  we consider the case when $z \not\in Z'$.

Assume, for the sake of contradiction, that $Z'$ is not a solution of $(G, T_E, F, k)$.
This implies that there exists a $T_E$-cycle, say $C^l$, in $G - Z'$.
As $Z'$ hits all $T_E$-cycles in $G' = G - E_z$, cycle $C^l$ contains at least one edge from $E_z$.
By the definition of $B_z$ and $Z$, cycle $C^{l}$ intersects with a connected component, say  $C_i$,  in $G - ({z} \cup B_z \cup(Z \setminus \{z\}))$. 
As $Z$ hits all $T_E$-cycles in $G$, all the vertices in $C^l$ are not included in $V(C_i)$.
By Claim~\ref{claim:intersection-inclusion} and the fact that  all the edges incident on $z \in Z$ are in $E(G') \setminus T_E$, i.e. non-terminal edges, all vertices in $C^l$ are not included in $V(C_i) \cup \{z\}$.
As $N_{B_z}(Y) \subseteq X$ (by Lemma~\ref{lem:extra-expansion-vertex}), cycle $C^{l}$ contains at least one vertex from $X$.
By Claim~\ref{claim:prop-sol-backword}, $X \setminus Z \subseteq Z'$. 
Hence, cycle $C^{l}$ contains at least one vertex from $X \cap Z$.
Let $z_0 \in X \cap Z $ be that vertex.
As $C^l$ is a cycle in $G - Z'$, vertex $z_0$ is not in $Z'$.
See Figure~\ref{fig:graph-partition}.

We argue that there is no $T_E$-path from $z$ to $z_0$ in $G$ whose all internal vertices are in $V(C_y)$.
If there exists such a path then, by the construction of the auxiliary bipartite graph $G_z(B_z \cup (Z \setminus \{z\}), \calC)$, edge $z_0y$ is colored red i.e. it is in $E_r$.
Recall that edges in $E_r$ are preferred over the edges of $E_b$ in the matching.
As $M$ is $(k + 2)$-expansion of $X$ into $Y$ with the priority property and $xy \not\in M$, every edge in $M$ that is incident on $z_0$ is in $E_r$.
By the similar arguments as that in Claim~\ref{claim:prop-sol-backword}, this implies that every solution of $(G', T_E, F, k)$ must contain vertex $z_0$, a contradiction.
Hence, there is no $T_E$-path from $z$ to $z_0$ in $G$ whose all internal vertices are in $V(C_y)$.

As $C^l$ is a $T_E$-cycle that contains vertices $z, z_0$ in $G - Z'$, there exists vertex $v_i \in \calC$, different from $y$, such that there exists $T_E$-path from $z$ to $z_0$ in $G - Z'$ whose internal vertices are in $V(C_i)$.
Here, $C_i$ is the connected component of $G - (\{z\} \cup B_z \cup (Z \setminus \{z\}))$ corresponding to vertex $v_i \in \calC$.
We note that $v_i$ may not be in $Y$.
Let $Y_{z_0}$ be the collection of vertices in $\calC$ that are adjacent to $z_0$ via edges in $M$.
Formally, $Y_{z_0} = \{y_i \in \calC\ |\ z_0y_i \in M\}$. 
As $M$ is a $(k + 2)$-expansion of $X$ into $Y$, $|Y_{z_0}| = k + 2$.
As $|Z'| \le k$, there is at least one vertex, say $y_1$, in $Y_{z_0}$ such that $y_1 \neq v_i$ and $V(C_1) \cap Z' = \emptyset$.
Hence, there exists a path in $G - Z'$ from $z$ to $z_0$ whose all internal vertices are in $V(C_1)$. 
This implies there exists a $T_E$-cycle containing $z$, $z_0$, and vertex disjoint paths whose internal vertices are in $V(C_1)$ and $V(C_i)$.
Note that this cycle does not contain edges in $E_z$.
Hence, such a path exists in $G' - Z'$.
This contradicts the fact that $Z'$ is a solution of $(G', T_E, F, k)$.
Hence, our assumption is wrong and there is no $T_E$-cycle in $G - Z'$.
\end{proof}

\begin{lemma}
\label{lemma:bound-terminal-nrs}
Consider an instance $(G, T_E, F, k)$ of 
$H$-\textsc{Minor-Free Edge-Subset FVS with Undeletable Vertices}.
Then, in polynomial time, one can compute an equivalent instance $(G', T'_E, F', k')$ such that $|T'_E| \in \calO(k^5)$.
\end{lemma}
\begin{proof}
The algorithm computes a set $Z \subseteq V(G)$ such that 
$Z$ intersects all $T_E$-cycles in $G$ using the algorithm in~\cite{EvenNZ00}.
We can assume, without loss of generality, that $Z \cap V(T_E) = \emptyset$. 
If $|Z| > 8k$ then algorithm returns a trivial \no\ instance. 
Otherwise, it applies Reduction Rules~\ref{rr:trivial-no} to \ref{rr:expansion-lemma}.
It exhaustively applies the least indexed reduction rule which is applicable.
To applies Reduction Rule~\ref{rr:flower} and Reduction Rule~\ref{rr:expansion-lemma}, the algorithm uses $Z$ as the desired set.
The algorithm returns the reduced instance, say $(G', T'_E, F', k')$.
This completes the description of the algorithm.

The approximation algorithm in~\cite{EvenNZ00} grantees a solution which is $8$ times the size of an optimum solution.
By the description of the algorithm, correctness and running time of the algorithm in~\cite{EvenNZ00} and of reduction rules, we can conclude that the algorithm runs in polynomial time and returns an equivalent instance.
Also, it is easy to see that $k' \le k$.
We argue the bound on the number of edges in $T'_E$ when $(G', T_E', F', k')$ is not a trivial \no\ instance.

Note that Reduction Rules~\ref{rr:trivial-no} to \ref{rr:expansion-lemma} are not applicable on $(G', T'_E, F', k')$.
Let $Z'$ be the $8$-factor approximate solution of $G'$  returned by the algorithm in~\cite{EvenNZ00}. 
As $(G', T'_E, F', k')$ is not a trivial \no\ instance, $|Z'| \le 8k$.
Let $\calH^{+}_{Z'}$ and $\calH_{Z'}$ be the graphs obtained from $G'$, using $Z'$, as described above the subroutine.
We first argue that the number of non-solitary leaves is $\calO(k^3)$.
We then argue that the length of any simple path in $\calH$ whose internal vertices are of degree two is $\calO(k^2)$. 

We prove that every non-solitary vertex of degree at most two in $\calH_{Z'}$ is adjacent with at least one in ${Z'}$ in $\calH^+_{Z'}$.
Assume, for the sake of contradiction, that there is a non-solitary vertex $v$ in $V(\calH_{Z'})$ that has degree at most two in $\calH^{+}_{Z'}$ and is not adjacent with any vertex in $Z'$. 
Suppose there is only one edge, say $e_1 \in T_E$, that is incident on $v$ in $\calH_{Z'}$ and hence in $\calH^{+}_{Z'}$.
Consider a bubble corresponding to $v$ in $G$.
Formally, consider $\psi(v)$ in $G'$, where $\psi$ is the function defined before the subroutine.
Note that $e_1$ is the only (terminal) edge with one edge point in $\psi(v)$ and another endpoint in $V(G') \setminus \psi(v)$.
Hence, $e_1$ is a bridge edge in graph $G'$.
This contradicts the fact that Reduction Rule~\ref{rr:delete-bridge} is not applicable to $(G', T'_E, F', k')$.
Suppose there are two edges, say $e_1, e_2 \in T'_E$, that are incident on $v$ in $\calH_{Z'}$ and hence in $\calH^{+}_{Z'}$.
Note that $e_1, e_2$ are the only two (terminal) edges with one edge point in $\psi(v)$ and another endpoint in $V(G') \setminus \psi(v)$.
Hence, $e_1 \in T'_E$ is a bridge edge in $G' - (T'_E \setminus \{e_1\})$. 
This contradicts the fact that Reduction Rule~\ref{rr:terminal-bridge} is not applicable to $(G', T'_E, F', k')$.
As $v$ is any arbitrary non-solitary vertex of degree at most two in $\calH^{+}_{Z'}$, every such vertex is adjacent with some vertex in $Z'$.

In graph $\calH^{+}_{Z'}$, every vertex in $Z'$ is adjacent to at most $10(k + 2)^2$ non-solitary leaves.
If this is not the case, Reduction Rule~\ref{rr:expansion-lemma} is applicable on $(G', T'_E, F', k')$.
This, together with the fact that $|Z'| \le 8k$ implies that the number of non-solitary leaves in $\calH_{Z'}$ is $ \calO(k^3)$.

Consider a simple path $P$ in $\calH_{Z'}$ such that all internal vertices in $P$ are of degree two.
Every vertex in $Z'$ is adjacent to at most $2(k + 1)$ vertices in $P$.
If not, as $E(P) \subseteq E(T'_E)$, there is a $z$-flower of order $k + 1$ in $\calH_{Z'}$ and hence in $G'$.
This contradicts the fact that Reduction Rule~\ref{rr:flower} is not applicable on $(G', T'_E, F, k')$.
As $|Z'| \le 8k$, we get desired bound on the number of vertices in $P$.
Since, $P$ was an arbitrary path in $\calH_{Z'}$, we can conclude that any such path in $\calH_{Z'}$ as at most $\calO(k^2)$ vertices.

By the standard arguments, we can bound the number of vertices in the forest using the upper bound on the number of leaves and the size of simple path whose all internal vertices are of degree two.
Hence,  the number of vertices in $\calH_{Z'}$ that are not solitary vertices (i.e. isolated vertices) is at most  $\calO(k^5)$.
This implies $|E(\calH_{Z'})| = |T'_E| \in \calO(k^5)$ which concludes the proof.
\end{proof}

Lemma~\ref{lemma:subset-FVS-marking} and Lemma~\ref{lemma:bound-terminal-nrs} imply the following result.

\begin{lemma}
\label{lemma:subset-FVS-marking-bound-k}
Consider an instance $(G, T_E, F = \emptyset, k)$ of $H$-\textsc{Minor-Free Edge-Subset FVS with Undeletable Vertices}.
Then, in randomized polynomial time and with failure probability $\mathcal{O}(2^{-|V(G)|})$, one can compute an equivalent instance $(G', T_E, F, k')$ such that $|V(G')| \le \calO(c_H \cdot k^{11})$ and $k' \le k$.
\end{lemma}

Lemma~\ref{lemma:subset-FVS-marking-bound-k} and Lemma~\ref{lemma:compression-H-minor-free-Pi-vertex-SFVS} imply that there is an algorithm that given an instance $(G,T,k)$ of \textsc{$H$-Minor Free Edge-Subset FVS} constructs an equivalent instance $(G',T,F,k')$ of \textsc{$H$-Minor Free Edge-Subset FVS with Undeletable vertices} in randomized polynomial time and with failure probability $\calO(2^{-|V(G)|})$ such that $|V(G')| \in k^{\calO(1)}$ and $k' \leq k$.
By the simple reduction, stated in \cite{CyganPPW13} and mentioned at the start of this subsection, this implies the result about \textsc{Subset FVS} in Theorem~\ref{thm:H-minor-graph-Pi-compression-vertex-SFVS-restated}.
Using similar arguments with Lemma~\ref{lemma:kernel-planar-Pi-vertex-SFVS}, imply the result about \textsc{Subset FVS} in Theorem~\ref{thm:planar-Pi-kernel-vertex-SFVS}.

\subsubsection{\textsc{Group Feedback Vertex Set}}

Throughout this subsection, we consider a finite group $\Sigma$ with group operation $\otimes$ and $1_{\Sigma}$ as its identity element.
Consider an instance $(G, \lambda',  \Sigma, k)$ of \textsc{Group FVS}.
Here, $G$ is an undirected graph and function $\lambda'$ is an edge labelling i.e. $\lambda': E(G) \mapsto \Sigma$.
Consider the digraph $D$ on $V(G)$ obtained from $G$ by replacing every edge $uv$ by arcs $uv$ and $vu$.
Let $\lambda: A(D) \mapsto \Sigma$ be an arc labelling such that $\lambda(uv) = \lambda'(uv)$ and $\lambda(vu) = (\lambda'(uv))^{-1}$. 
It is easy to verify that a set $Z \subseteq V(G) = V(D)$ hits all non-null cycles in $G$ if and only if it hits all non-null cycles in $D$.

We present a brief overview of the kernelization algorithm for \textsc{Group FVS} presented in Section $7.3$ in \cite{KratschW20}.
Given an instance $(G, \lambda', \Sigma, k)$, the algorithm starts by finding a group feedback vertex set\footnote{For notational clarity, we state a weaker bound on the size of $X$.} $X$ of $(G, \lambda',  \Sigma)$ of size $\calO(|\Sigma|^2 \cdot k^2)$.
We assume, without loss of generality, that such a set exists as otherwise the algorithm can conclude that the input is a \no\ instance.
It then constructs an equivalent instance $(D, \lambda, \Sigma, k)$ of \textsc{Group FVS} such that $\lambda$ satisfies certain properties with respect to $X$.
The following lemma summaries these steps.
\begin{lemma}[{\cite[Lemma 7.6 \& 7.7]{KratschW20}}]
There is a polynomial time algorithm that given an instance $(G, \lambda', \Sigma, k)$ of \textsc{Group FVS} finds a subset $X \subseteq V(G)$ and constructs an equivalent instance $(D, \lambda, \Sigma, k)$ of \textsc{Group FVS} such that the following conditions are true.
\begin{enumerate}
\item $G$ is an underlying undirected graph of $D$.
\item Set $X$ is an independent set of size $\calO(|\Sigma|^2 \cdot k^2)$ in $G$, and it is a group feedback vertex set of $(D, \lambda, \Sigma)$.
\item Every arc of $D - X$ has label $1_{\Sigma}$ and every arc incident on $X$ is oriented out of $X$.
\end{enumerate}
\end{lemma}

Using the properties of instance $(D, \lambda, \Sigma, k)$,  the algorithm constructs a graph $G$ and set of terminals $T$ such that one can infer a solution of the original instance from a multiway cut of some subset of $T$ in $G^{\circ}$. 
Let $G^{\circ}$ be the graph created by taking the underlying undirected graph of $D - X$,  adding a new vertex $x(\alpha)$ for every vertex $x \in X$ and every label $\alpha \in \Sigma$,  and finally connecting $x(\alpha)$ to $v$ if and only if $x$ is incident with an arc $xv$ in $D$ with $\lambda(xv) = \alpha$.
Let $T = \{x(\alpha)\ |\ x \in X, \alpha \in \Sigma\}$.
For a set $S \subseteq T$, define $X(S) = \{x \in X\ |\ x(\alpha) \in S, \alpha \in \Sigma\}$, and for $S \subseteq X$ define $T(S) = \{x(\alpha)\ |\ x \in S, \alpha \in \Sigma\}$.
A set $S \subseteq T$ is \emph{regular} if $T(X(S)) = S$.

\begin{lemma}[{\cite[Lemma 7.9]{KratschW20}}]
Let $(D,\lambda,\Sigma, k)$ be an instance of \textsc{Group FVS} and let $X \subseteq V(D)$ be a group feedback vertex set for $D$ such that every arc of $D - X$ has label $1_{\Sigma}$, $D[X]$ is an independent set, and every arc incident with $X$ is oriented out of $X$.
Let $G^{\circ}$ and $T$ be defined from $D$ and $X$ as above.
The instance $(D, \lambda,\Sigma, k)$ is \yes\ instance of \textsc{Group FVS} if and only if there is a regular set $T_X \subseteq T$, a partition $\calT = \{T_{\alpha}\ |\ \alpha \in \Sigma\}$ of $T \setminus T_X$ and a multiway cut $Y$ of $T$ in $G^{\circ} - T_X$ such that
\begin{enumerate}
\item if $x(\alpha) \in T_{\alpha'}$ and $x(\beta) \in T_{\beta'}$ then $\alpha \otimes \alpha^{-1} = \beta' \otimes \beta^{-1}$, and
\item $|Y| + |X(T_X)| \le k$.
\end{enumerate}
Moreover, if $Y$ and $T$ fulfill these properties, then $X' = (Y \setminus T ) \cup X (T_X \cup (Y \cap T ))$ is a solution for $(D, \lambda, k)$.
\end{lemma}
The algorithm then computes a cut-covering set $W^{\circ} \subseteq V(G^{\circ})$ of size $k^{\calO(|\Sigma|)}$ for generalized $|\Sigma|$-partitions of $T$ in $G^{\circ}$ using Theorem $5.14$ in \cite{KratschW20}.
As the final step, it applies a torso operation for $D$ down to $(W^{\circ} \setminus T) \cup X$, modified to account for the labelling $\lambda$.

Note that $G - X$ is same graph as $G^{\circ} - T$ and $T \subseteq W^{\circ}$.
Hence, the above lemma implies that if $(D, \lambda, \Sigma, k)$ is a \yes\ instance of \textsc{Group FVS}, then there exists a solution in $W := (W^{\circ} \setminus T) \cup X$.
This implies the following lemma which is analogues to Lemma~\ref{lemma:vertex-multicut-deletable-terminals-marking} and Lemma~\ref{lemma:subset-FVS-marking} in the cases of \textsc{Vertex MwC-DT} and \textsc{Subset FVS}, respectively. 

\begin{lemma}
\label{lemma:group-FVS-marking}
Suppose $(G,\lambda, \Sigma, k)$ is a \yes\ instance of \textsc{Group FVS}.
There exists a set $W \subseteq V(G)$ of size $k^{\calO(|\Sigma|)}$ such that there is a solution of $(D, \lambda, \Sigma, k)$ contained in $W$, and every edge of $G - W$ has label $1_{\Sigma}$.
Furthermore, there is an algorithms that finds such a set in randomized polynomial time with failure probability $\calO(2^{-|V(G)|})$.
\end{lemma}

We remark that if $G$ is an $H$-minor free graph then it does not imply that $G^{\circ}$ is an $H$-minor free graph.
However,  graph $G^{\circ}$ is used only to obtain set $W$ and to identify some vertices in $X$ that are in any solution.
Hence, the overall process of identifying set $W$ does not introduce any new vertex or edge.

Consider two vertices $u, v \in V(G) \setminus W$ such that $uv$ is an edge.
By the property of $W$, the label of $uv$ is $1_{\Sigma}$.
As there is a solution which is completely contained in $W$, it is safe to contract edge $uv$.
Hence, the `contraction' step follows easily.
Now, suppose $N_G(u) \subseteq N_G(v)$.
Unlike in the case of \textsc{Vertex MwC-DT} or \textsc{Subset FVS}, it is \emph{not safe} to delete $u$ as edges incident on $u$ and $v$ might have different labels.
In this case, we can not invoke Lemma~\ref{lemma:bipartite-H-minor} and need the following result whose proof we defer to the later part of the subsection.
\begin{lemma} 
\label{lemma:bipartite-H-minor-GFVS}
For a fixed graph $H$ on $h$ vertices, let $G = (X \cup Y, E)$ be a simple $H$-minor free bipartite graph such that there is no isolated vertex in $Y$.
Define $Y(X') = \{y \in Y\ |\ X' \subseteq N(y)\}$ for every $X' \subseteq X$.
If for every $X' \subseteq X$ of size strictly less than $h$, $|Y(X')| \le q$ then $|Y| \le \calO((q + h) \cdot |X|^{h})$.
\end{lemma}

We now prove a lemma analogous to Lemma~\ref{lemma:compression-H-minor-free-Pi-vertex-SFVS}.

\begin{lemma}
\label{lemma:compression-H-minor-free-Pi-vertex-GFVS}
Consider an instance $(G, \lambda, \Sigma, k)$ of \textsc{$H$-Minor-Free Group FVS} and suppose there is a subset $W \subseteq V(G)$  with the guarantee that if there is a solution of $(G, \lambda, \Sigma, k)$, then there is a solution contained in $W$, and every edge of $G - W$ has label $1_{\Sigma}$.
Then, one can obtain an equivalent instance $(G^{\circ}, \lambda^{\circ}, \Sigma, F, k^{\circ})$ of \textsc{$H$-Minor-Free Group FVS with Undeleteable Vertices} in polynomial time such that $|V(G')| \in (|\Sigma| + |W|)^{\calO(|V(H)|)}$ and $k^{\circ} \le k$.
\end{lemma}

\begin{proof}
Consider an algorithm that given an instance $(G, \lambda, \Sigma, k)$ of  \textsc{$H$-Minor-Free Group FVS}  and set $W$ as specified in the statement of the lemma applies the following three reduction rules.
\begin{enumerate}
\item It contracts each connected component of $G - W$ to a vertex \emph{without deleting any parallel edges}.
Formally, let $M_u$ be a connected component of $G - W$.
The algorithm deletes $M_u$ and adds a new vertex $u$ to $G$.
For every edge $vw$ such that $v \in M_u$ and $w \in W$, it adds edge $uw$ and assigns $\lambda'(uw) = \lambda(vw)$.
For all the remaining edges $e$, $\lambda'(e) = \lambda(e)$. 
Let $G'$ be the graph obtained after contracting each connected component of $G - W$ and let $F$ be the collection of all new vertices added i.e. $F = V(G') \setminus W$.
It returns instance $(G', \lambda', \Sigma, F, k)$.

\item Suppose there are two edges $e_1, e_2$ in $G$ that have same endpoints.
Moreover,  one of the endpoint, say $u$, is in $F$ and hence another endpoint, say $w$, is in $W$.  
If $\lambda(e_1) \otimes \lambda(e_2) \neq 1_{\Sigma}$, then it deletes $w$ and returns $(G - \{w\}, \lambda|_{E(G - {w})}, \Sigma, F, k - 1)$.
Otherwise, if $\lambda(e_1) = \lambda(e_2)$, then it deletes $e_1$ and returns $(G - \{e_1\}, \lambda|_{E(G) \setminus \{e_1\}}, \Sigma,  F, k)$.

Let $G'$ be the resulting graph after exhaustively applying the reduction rule.
It returns $(G', \lambda|_{E(G')}, \Sigma, F, k')$.

\item If there are two vertices $u, v \in F$ such that for every edge $uw$ in $G'$ there is an edge $vw$ in $G'$ such that $\lambda(uw) = \lambda(vw)$, then delete $u$.

Let $G^{\circ}$ be the resulting graph after exhaustively applying the reduction rule.
It returns $(G^{\circ}, \lambda|_{E(G^{\circ})}, \Sigma, F, k)$.
\end{enumerate}

The algorithm obtains $G^{\circ}$ from $G$ by series of edge contractions and vertex deletion operations.
Hence, $G'$ is a minor of $G$.
Moreover, the algorithm terminates in the time polynomial in the size of the input instance.
We first argue that the reduction rules are safe and then bound the number of undeletable vertices in $G^{\circ}$.

\noindent \textbf{(1)}
We argue the correctness of the first reduction rule.
Note that every edge of $G - W$ has label $1_{\Sigma}$.
This implies that every edge which is present in $G$ but not in $G'$ has label $1_{\Sigma}$.
For every vertex $u \in F \subseteq V(G')$, let $M_u$ be the corresponding connected component of $G - W$.
Alternately, vertex $u$ is obtained by contracting $M_u$ to a single vertex without deleting any parallel edges.
For any subset $Z$ of $W$, let $A$ be a connected component of $G - Z$.
As $F \cap W = \emptyset$, for any $u \in F$, either $M_u \subseteq A$ or $M_u \cap A = \emptyset$.
By the construction, for any connected component $A'$ of $G' - Z$, there is a connected component $A$ of $G - Z$ such that $A'$ can be obtained from $A$ by contracting all $M_u$ that intersects with $A$ to a single vertex.

To prove the forward direction, let $Z$ be a solution of $(G, \lambda, \Sigma, k)$ such that $Z \subseteq W$ and hence $Z \cap F = \emptyset$.
Such a solution exists by the property of $W$.
Assume $Z$ is not a solution of $(G', \lambda', \Sigma, F, k)$.
This implies that there exists a connected component $A'$ of $G' - Z$ such that it contains a non-null cycle. 
As every edge in $E(G) \setminus E(G')$ has label $1_{\Sigma}$, this implies the corresponding connected component $A$ of $G - Z$ also contains a non-null cycle. 
This contradicts the fact that $Z$ is a solution of $(G, \lambda, \Sigma, k)$.
Hence, our assumption is wrong and $Z$ is a solution of $(G', \lambda', \Sigma, F, k)$.

To prove the reverse direction, let $Z'$ be a solution of $(G', \lambda', \Sigma, F, k)$.
Assume $Z'$ is not a solution of $(G, \lambda, \Sigma, k)$.
This implies that there is at least one connected component $A$ of $G - Z'$ that has a non-null cycle. 
As every edge in $E(G) \setminus E(G')$ have label $1_{\Sigma}$, this implies that $A'$ contains a non-null cycle which contradicts the fact that $Z'$ is a solution of $(G', \lambda', \Sigma, F, k)$.
Hence our assumption is wrong $Z'$ is a solution of $(G, \lambda, \Sigma, k)$.

This concludes the proof of correctness of the first reduction rule.

\noindent \textbf{(2)} 
In the first case, the safeness of the reduction rule follows from the fact that any solution contains at least one vertex in $\{u, w\}$ and there is a solution which is completely contained inside $W$.
In the second case, the correctness of the forward direction follows from the fact that the resulting graph is a subgraph of the input graph.
In the reverse direction, let $Z'$ be a solution of instance $(G', \lambda|_{E(G')}, \Sigma, F, k)$. 
It is easy to verify that if there is non-null cycle in $G - Z'$ that contains $e_1$, then there is a non-null cycle in $G' - Z'$ that contains $e_2$.
This contradicts the fact that $Z'$ is a solution.
Hence, $Z'$ is a solution of $(G, \lambda, \Sigma, F, k)$.

This concludes the proof of correctness of the second reduction rule.

\noindent \textbf{(3)} The correctness of the forward direction follows from the fact that $G^{\circ}$ is a subgraph of $G'$.
To prove the reverse direction, let $Z^{\circ}$ be a solution of instance $(G^{\circ},  \lambda, \Sigma, F, k)$ of \textsc{Group FVS}. 
If $Z^{\circ}$ is not a solution of $(G',  \lambda, \Sigma, F, k)$,  there there a non-null cycle containing vertex $u$ in $G'$.
As for every edge $uw$ in $G'$ there is an edge $vw$ in $G'$ such that $\lambda(uw) = \lambda(vw)$,  and $v \in F$ and hence $v \not\in Z^{\circ}$, there is non-null cycle $G^{\circ} - Z^{\circ}$.
This contradicts the fact that $Z^{\circ}$ is a solution of $(G^{\circ},  \lambda, \Sigma, F, k)$.
This implies that the third reduction rule is safe.

To complete the lemma, it remains to argue an upper bound on the number of undeletable vertices.
Consider the bipartite graph obtained from $G^{\circ}$ by deleting all the edges whose both endpoint are in $V(G^{\circ}) \setminus F$.
Consider a subset $W^{\circ}$ of $W$ and let $F(W^{\circ})$ be the collection of vertices in $F$ that are adjacent with every vertex in $W^{\circ}$.
Formally, $F(W^{\circ}) = \{u\in F\ |\ N_{G^{\circ}}(u) = W^{\circ}\}$.
We first argue that $|F(W^{\circ})| \le |\Sigma|^{|W^{\circ}|}$.

As the second reduction rule is not applicable, for any vertex $u \in F(W^{\circ})$ and $w \in W^{\circ}$, there are at most two edges $e_1, e_2$ whose endpoints are $u, w$.
Moreover, $\lambda(e_1) \otimes \lambda(e_2) = 1_{\Sigma}$ i.e. $\lambda(e_1) = \lambda(e_2)^{-1}$.
For the sake of counting, we arbitrarily prefer one of the element in $\{\lambda(e_1), \lambda(e_2)\}$ over another.
Formally, define a function $\psi: \Sigma \mapsto \Sigma$ such that for any $\alpha \in \Sigma$, we have $\psi(\alpha) \in \{\alpha, \alpha^{-1}\}$ and $\psi(\alpha) = \psi(\alpha^{-1})$.
Let $[\Sigma]^{|W^{\circ}|}$ denote the collection of tuples of size $|W^{\circ}|$ where each entry in the tuple is in $\Sigma$.
Index the vertices in $W^{\circ}$ arbitrarily from $1$ to $|W^{\circ}|$.
Every vertex in $F(W^{\circ})$ naturally corresponds to an unique tuple $t$ in this collection.
Namely, for a vertex $u \in F(W^{\circ})$, let $t_u$ be the unique tuple in $[\Sigma]^{|W^{\circ}|}$ such that if $w_i$ is the $i^{th}$ vertex in $W^{\circ}$ then the $i^{th}$ entry in tuple $t_u$ is $\psi(\lambda(u))$.
Note that for two vertices $u, v$ in $F(W^{\circ})$, if $t_u$ is same as that of $t_v$, then the second reduction rule would have deleted $u$.
Hence, every vertex in $F(W^{\circ})$ corresponding to a unique tuple $[\Sigma]^{|W^{\circ}|}$.
This implies that there are at most $|\Sigma|^{|W^{\circ}|}$ many vertices in $F(W^{\circ})$.

We assume, without loss of generality, that there are no isolated vertices in $F$ as otherwise one can safely delete them.
Hence, the bipartite graph obtained from $G^{\circ}$ by deleting edges whose both endpoints are in $V(G^{\circ}) \setminus F$, satisfies the premise of Lemma~\ref{lemma:bipartite-H-minor-GFVS} for $q = |\Sigma|^h$. 
Hence, Lemma~\ref{lemma:bipartite-H-minor-GFVS} implies that the number of vertices in $F$ is $\calO((|\Sigma|^h + h) \cdot |W|^h)$. 
This concludes the proof of the lemma.
\end{proof}

We now prove a lemma analogous to that of Lemma~\ref{lemma:kernel-planar-Pi-vertex-SFVS}.

\begin{lemma}
\label{lemma:kernel-planar-Pi-vertex-GFVS}
Consider an instance $(G, \lambda, \Sigma, k)$ of  \textsc{Planar Group FVS} and suppose there is a subset $W \subseteq V(G)$  with the guarantee that if there is a solution of $(G, \lambda, \Sigma, k)$, then there is a solution contained in $W$, and every edge of $G - W$ has label $1_{\Sigma}$.
Then, one can obtain an equivalent instance $(G^{\circ},  \lambda^{\circ}, \Sigma, k^{\circ})$ of \textsc{Planar Group FVS} in polynomial time such that $|V(G')| \in k^2 \cdot (|\Sigma| + |W|)^{\calO(1)}$ and $k^{\circ} \le k$.
\end{lemma}
\begin{proof}
The proof follows closely along the lines of the proof of Lemma~\ref{lemma:kernel-planar-Pi-vertex-SFVS}.
Consider an instance $(G, \lambda, \Sigma, k)$ of \textsc{Planar Group FVS}.
Using Lemma~\ref{lemma:compression-H-minor-free-Pi-vertex-GFVS}, one can compute an equivalent instance $(G',  \lambda', \Sigma, F, k')$ of \textsc{Planar Group FVS with Undeleteable Vertices} such that $|V(G')| \in  (|\Sigma| + |W|)^{\calO(1)}$.
To complete the proof,  it is sufficient to prove that there is an algorithm that given an instance an instance $(G',\lambda', \Sigma, F,k)$ of \textsc{Planar $\Pi$ with Undeletable Vertices},  runs in polynomial time, and returns an equivalent instance $(G^{\circ}, \lambda, \Sigma, k)$ of \textsc{Planar Group FVS} such that $|V(G^{\circ})| \in \calO(k^2  \cdot |V(G')|)$.

Recall that the algorithm in Lemma~\ref{lemma:compression-H-minor-free-Pi-vertex-GFVS} constructs set $F$ of undeletable vertices by contracting connected components of $G - W$ for some subset $W \subseteq V(G)$.  
Hence, $F$ is an independent set in $G'$.

Consider the algorithm that given an instance $(G',  \lambda', \Sigma, F, k')$ of \textsc{Planar Group FVS with Undeleteable Vertices} constructs graph $G^{\circ}$ from $G'$ as follows.
Let $u \in F$ be a vertex of degree $d := \deg_G(u)$.
The algorithm replaces the vertex $u$ by a $(2d(k+1)) \times (k+1)$ grid with vertices $u_{i,j}$ for all $i \in [2d(k+1)]$ and $j \in [k+1]$.
Let $M_u$ be the set of all vertices $u_{i,j}$ for $i \in [2d(k+1)]$ and $j \in [k+1]$.
The algorithm adds grid edges $u_{i,j}u_{i+1,j}$ for all  $i \in [2d(k+1)-1]$, $j \in [k+1]$ as well as $u_{i,j}u_{i,j+1}$ for all $i \in [2d(k+1)]$, $j \in [k]$.
All these edges gets label $1_{\Sigma}$.
Let $w_1,\dots,w_d$ be the neighbors of $u$ (we choose numbering according to the cyclic order of a planar embedding of $G$).
As $F$ is an independent set in $G'$,  for every $r \in [d]$,  $w_r \notin F$.
The algorithm adds edges from $w_r$ to $u_{2(r-1)(k+1)+2j,1}$ for all $j \in [k+1]$ in a way that preserves planarity.
These edges gets the same label as that of $uw_r$.

The argument to prove that $G^{\circ}$ is a planar graph with $\calO(k^2 \cdot |V(G)|)$ vertices, the algorithm terminates in polynomial time and both the instance are equivalent is identical to that arguments used in the proof of Lemma~\ref{lemma:kernel-planar-Pi-vertex-SFVS}. 
One only need to replace $T_E$-cycle in the case of \textsc{Edge-Subset FVS} by non-null cycle to complete the argument.
This concludes the proof of the lemma.
\end{proof}

Lemma~\ref{lemma:group-FVS-marking},  along with Lemma~\ref{lemma:compression-H-minor-free-Pi-vertex-GFVS} and Lemma~\ref{lemma:kernel-planar-Pi-vertex-GFVS}, prove Theorem~\ref{thm:H-minor-graph-Pi-compression-vertex-GFVS-restated} and Theorem~\ref{thm:planar-Pi-kernel-vertex-GFVS}, respectively.

\subsubsection{Proofs of Lemma~\ref{lemma:bipartite-H-minor} and Lemma~\ref{lemma:bipartite-H-minor-GFVS}}

Recall that Lemma~\ref{lemma:bipartite-H-minor} states that for a fixed graph $H$,  if $G = (X \cup Y, E)$ is a simple $H$-minor free bipartite graph such that  for  all  distinct $u, v \in Y$,  $N(u) \not\subseteq N(v)$ then $|Y| \le c_H |X|$. 
Here, $c_H$ is a constant that depends only on the size of $H$.
We need the following result in the proof.

\begin{proposition}[{\cite[Theorem 3]{EppsteinLS10}}]
\label{prop:nr-maximal-cliques-degenerate}
The largest possible number of maximal cliques in an $n$-vertex graph with degeneracy $d$ is $n \cdot 3^{d}$.
\end{proposition}

\begin{proof}[Proof of Lemma~\ref{lemma:bipartite-H-minor}]
We can assume, without loss of generality, that there are no isolated vertices in $G$.
If there are isolated vertices in $X$, then we delete them and argue the lemma for the remaining graph.
If there is an isolated vertex in $Y$, then it is only vertex  in $Y$ and hence the lemma is vacuously true.
By Theorem~\ref{thm:average-degree-H-minor-free}, for a fixed graph $H$, there exists a constant, say $c$, such that every $H$-minor free graph is $c$-degenerate.
Hence, $G$ is a $c$-degenerate graph.

We partition vertices in $Y$ into two sets based on their degrees.
Let $Y_{\ge c + 1}$ be the set of vertices of degree at least $c + 1$ and $Y_{\le c}$ be the set of vertices of degree at most $c$.
We argue the bound on the sizes of these two sets separately.

Consider graph $G' = G[X \cup Y_{\ge c + 1}]$.
As $G'$ is a subgraph of $G$, it is a $c$-degenerate graph.
Hence, the number of edges in $G'$ is at most $c |V(G')| = c|X| + c|Y_{\ge c + 1}|$.
As every vertex in $Y_{\ge c + 1}$ has degree at least $c + 1$, and there are no parallel edges in $G$, the number of edges in $G'$ is at least $(c + 1) |Y_{\ge c + 1}|$.
Hence, $(c + 1)|Y_{\ge c + 1}| \le c|X| + c|Y_{\ge c + 1}|$ which implies $|Y_{\ge c + 1}| \le c|X|$.

In the remaining part, we bound the size of $Y_{\le c}$.
Construct an auxiliary graph $G_0$ by connecting two vertices in $X$ by an edge if they have a common neighbor in $Y_{\le c}$.
Formally, define $V(G_0) := X$, and $E(G_0) := \{x_1x_2|\ x_1, x_2 \in X\ \land\ \exists y \in Y_{\le c} \text{ such that } x_1, x_2 \in N(y)\}$.

\begin{claim}
 $G_0$ is a $((c + 1) \cdot \binom{c}{2})$-degenerate graph.
 \proof
 Assume, for the sake of contradiction, that there is a subgraph, say $G'_0$, of $G_0$ that has at least $((c + 1) \cdot \binom{c}{2}) \cdot |V(G'_0)|$ many edges.
 Using this graph, we construct a minor, say $G'$, of $G$ that has  at least $(c + 1)|V(G')|$ many edges, contradicting the fact that $G$ is a $c$-degenerate graph.
 
 For an edge $x_1x_2$, fix an arbitrary vertex $y$ in $Y_{\le c}$ such that that $x_1, x_2 \in N(y)$.
 By the construction of $G_0$, such a vertex exists for every edge in $E(G_0)$.
 Define function $\lambda: E(G'_0) \rightarrow Y_{\le c}$ as follows: $\lambda(x_1x_2) = y$ if and only if $x_1, x_2 \in N(y)$.
 Note that for any vertex $y$ in $Y_{\le c}$, $|\lambda^{-1}(y)| := |\{x_1x_2\ |\ \lambda(x_1x_2) = y\}| \le \binom{c}{2}$. 
 
 We now construct an one-to-one function $\psi$ from a subset, say $F$, of $E(G'_0)$ to $Y_{\le c}$.
 Initialize $F := \emptyset$.
 List edges in $E(G_0)$ is an arbitrary order.
 For $i \in [|E(G_0)|]$, add $e_i$ to $F$ if for any $1 \le j < i$, $\lambda(e_j) \neq \lambda(e_i)$, and define $\psi(e_i) = \lambda(e_j)$. 
 It is clear that $\psi: F \rightarrow Y_{\le c}$ is an one-to-one function.
 Moreover, for $y \in Y_{\le c}$, $|\psi^{-1}(y)| = 1 \ge |\lambda^{-1}(y)|/\binom{c}{2}$.
 Hence, $|F| = |\bigcup_{y \in Y_{\le c}} \psi^{-1}(y)| \ge |E(G'_0)|/\binom{c}{2} = (c + 1)\cdot |V(G'_0)| $.

 We now construct the minor $G'$ of $G$ such that $V(G') = V(G'_0)$ and $E(G') = F$ using the following procedure.
 Delete all vertices in $X \setminus V(G'_0)$ and $Y_{\ge c + 1}$.
 For every vertex $y \in Y_{\le c}$, if $\psi^{-1}(y) = \emptyset$ then delete $y$.
 Otherwise, let $\{x_1x_2\} = \psi^{-1}(y)$. 
 By the definitions of functions $\lambda$ and $\psi$, $x_1y, x_2y \in E(G)$.
 Delete all edges incident on $y$ except $x_1y$ and $x_2y$.
 In the resulting graph, every vertex in $Y_{\le c}$ is adjacent with exactly two vertices in $X$. 
 For every triple $x_1, x_2, y$ such that $x_1, x_2 \in X$, $y \in Y_{\le c}$, and $x_1y, x_2y \in E(G)$, contract edge $x_1y$ and rename new vertex as $x_1$.
 This completes the construction of $G'$.
 
 It is easy to verify that $V(G') = V(G'_0)$ and $E(G') = F$.
 Hence, $|E(G')| \ge (c + 1) \cdot |V(G')|$.
 This contradicts the fact that $G'$, which is a minor of $G$, is a $c$-degenerate graph.
 Hence, our assumption is wrong and $G_0$ is a $((c + 1)\binom{c}{2})$-degenerate graph.
 \uend
\end{claim}

Consider a mapping that maps every vertex $y \in Y_c$ to a clique on vertices $N(y) \subseteq X = V(G_0)$ in $G_0$.
By the construction of $G_0$,  this mapping maps every vertex in $Y_c$ to a clique (not necessarily a maximal clique) in $G_0$.
As for all  distinct $y_1, y_2 \in Y$,  $N(y_1) \not\subseteq N(y_2)$, each vertex in $Y_{\le c}$ is mapped to a different clique in $G_0$.
Let $d = (c + 1) \binom{c}{2}$. 
Since,  $G_0$ is a $d$-degenerate graph,  the size of any maximal clique is at most $(d + 1)$.
Hence,  for a maximal clique in $G_0$, there are at most $\binom{d + 1}{2}$ many vertices in $Y_{\le c}$ whose mapping is  contained in it.  
By Proposition~\ref{prop:nr-maximal-cliques-degenerate}, the number of maximal cliques in $G_0$ is at most $3^d \cdot |X|$.
Hence,  the number of vertices in $Y_{\le c}$ is at most $3^d  \cdot \binom{d + 1}{2} \cdot |X|$, where $d = (c + 1) \binom{c}{2}$. 
This, together with the fact that $|Y_{\ge c + 1}| \le c|X|$, concludes the proof of the lemma. 
\end{proof}

Recall that Lemma~\ref{lemma:bipartite-H-minor-GFVS} states that for a fixed graph $H$ on $h$ vertices, if $G = (X \cup Y, E)$ is a simple $H$-minor free bipartite graph such that there is no isolated vertex in $Y$ and for every $X' \subseteq X$ of size strictly less than $h$, $|Y(X')| \le q$ then $|Y| \le \calO((q + h) \cdot |X|^{h})$.
Here, $Y(X') := \{y \in Y\ |\ X' \subseteq N(y)\}$ for every $X' \subseteq X$.

\begin{proof}[Proof of Lemma~\ref{lemma:bipartite-H-minor-GFVS}]
Any graph $H$ on $h$ vertices is a minor of a clique $K_h$ which is a minor of complete bipartite graph $K_{h, h}$.
Hence, if $G$ is $H$-minor free then $G$ is also $K_{h, h}$-minor free.
We use $\binom{X}{< h}$ and $\binom{X}{= h}$ to denote the collections of the subset of $X$ of size strictly less than $h$ and of size exactly $h$, respectively.
The size of each of these collection is at most $|X|^h$.

We can assume, without loss of generality, that there are no isolated vertices in $G$.
If there are isolated vertices in $X$, then we delete them and argue the lemma for the remaining graph.
Note that there is no isolated vertex in $Y$.
We partition vertices in $Y$ into two sets based on their degrees.
Let $Y_{\ge h}$ be the set of vertices of degree at least $h$ and $Y_{< h}$ be the set of vertices of degree strictly less than $h$.
We argue the bound on the sizes of these two sets separately.

The number of subset of $X$ that are of size strictly less than $h$ is at most $|X|^{h - 1}$.
Hence, the total number of vertices present in $Y(X')$ for some set $X'\subseteq X$ of size strictly less than $h$ is at most $q \cdot |X|^{h - 1}$.
This implies the total number vertices in $Y_{<h}$ is $\calO(q \cdot |X|^h)$.

To argue the bound on the number of vertices in $Y_{\ge h}$, we arbitrarily index the vertices in $X$ from $1$ to $|X|$.
Consider a map $f:Y_{\ge h} \mapsto \binom{X}{=h}$ such that for every vertex $y$ in $Y_{\ge h}$, $f(y)$ is the collection of $h$ smallest indexed vertices in $N_G(y) \cap X$.
It is easy to verify that $f$ defines a function from $Y_{\ge h}$ to $\binom{X}{=h}$.
Assume that there are $(h - 1) \cdot |X|^h + 1$ many vertices in $Y_{\ge h}$.
By Pigeon-Hole Principle, this implies that there exist a set $X'$ in $\binom{X}{=h}$ such that the cardinality of set $Y' = \{y \in Y_{\ge h} | f(y) = X'\}$ is at least $h$.
Note that, by the definitions, the subgraph of $G$ on induced on $X' \cup Y'$ is a complete bipartite graph $K_{h, h}$.
This contradicts the fact that $G$ is $K_{h, h}$-minor free.
Hence, our assumption is wrong and there are at most $(h - 1) \cdot |X|^h$ many vertices in $Y_{\ge h}$.

This implies that the number of vertices in $Y$ is $\calO((q + h) \cdot |X|^h)$. 
\end{proof}

%% file: lower-bound.tex
\section{Lower Bounds}
\label{sec:lower-bounds}

In this section, we provide lower bound for \textsc{Component Order Connectivity} and related problems on planar graph.
Recall that, given a graph $G$ and integers $k,t$, the problems asks to delete at most $k$ vertices such that every connected component of the resulting graph has size at most $t$.
In the following, we prove that this problem is \WH\ and can not be solved in time $n^{o(\sqrt{k})}$ on planar graphs (assuming $t$ is part of the input) assuming \ETH.
Similar results are obtained for several variants of the problem.

\subsection{Components of Equal Size}

To prove the lower bounds, we build on the \textsc{Grid Tiling} problem which is a standard tool for obtaining lower bounds for problems on planar graphs (see, e.g., \cite[Chapter 14.4.1]{CyganFKLMPPS15}).

\medskip 
\defparproblem{\textsc{Grid Tiling}}{Integers $k,n$, and sets $S_{i,j} \subseteq \{0,\dots,n\} \times \{0,\dots,n\}$ for all $i,j \in [k]$.}{$k$}
{Are there numbers $r_i,c_j \in \{0,\dots,n\}$, $i,j \in [k]$, such that $(r_i,c_j) \in S_{i,j}$ for all $i,j \in [k]$?}
\medskip

\begin{theorem}[see, e.g., {\cite[Theorem 14.28]{CyganFKLMPPS15}}]
 \label{thm:grid-tiling-lower-bound}
 The \textsc{Grid Tiling} problem is \WH\ and, unless \ETH\ fails, it can not be solved in time $f(k)n^{o(k)}$ for any computable function $f$.
\end{theorem}

Now, let us turn to proving lower bounds for (variants of) \textsc{Component Order Connectivity}.
Since edge deletions versions of the problems can usually be reduced to the vertex deletion version, we prove the lower bounds for the edge deletions versions.
Also, to simplify the reductions, we consider variants where vertices of the input graph are weighted, and certain edges may be declared undeletable.
At the end of this section, we argue how to extend the lower bounds to the standard versions of the problems.

We formalize the reduction from \textsc{Grid Tiling} on the variant of \textsc{Component Order Connectivity} where all components need to have size exactly $t$ after removing the solution.

\medskip 
\defparproblem{\textsc{Extended Component Equality Edge Connectivity}}{A graph $G$, a weight function $w\colon V(G) \rightarrow \NN$, a set $U \subseteq E(G)$, and integers $k,t$.}{$k$}
{Is there a set $Z \subseteq E(G) \setminus U$ such that $|Z| \leq k$ and $w(C) = t$ for every connected component $C$ of $G - Z$?}
\medskip

\begin{lemma}
 \label{la:reduction-component-order-connectivity-basic}
 There is a polynomial-time algorithm that, given an instance $(n,k,(S_{i,j})_{i,j \in [k]})$ of \textsc{Grid Tiling}, computes an equivalent instance $(G,w,U,k',t)$ of \textsc{Planar Extended Component Equality Edge Connectivity} such that
 \begin{enumerate}
  \item $|E(G)| = \CO(k^2n^2)$,
  \item $w(v) = \CO(kn^3)$ for all $v \in V(G)$, and
  \item $k' = \CO(k^2)$.
 \end{enumerate}
\end{lemma}

\begin{proof}
 Without loss of generality assume that $k$ is even.
 Indeed, if $k$ is odd, one can create an equivalent instance by increasing $k$ by one and setting $S_{i,k+1} = S_{k+1,i} \coloneqq \{0,\dots,n\} \times \{0,\dots,n\}$ for all $i \in [k+1]$.
 
 Let $(n,k,(S_{i,j})_{i,j \in [k]})$ be an instance of \textsc{Grid Tiling}.
 We describe how to construct an equivalent instance $(G,w,U,k',t)$ of \textsc{Planar Extended Component Equality Edge Connectivity}.
 Let $N \coloneqq n+1$.
 For the weight function $w$ it turns out to be more convenient to specify integers in base $N$.
 Specifically, we denote $(a_3,a_2,a_1,a_0)_N \coloneqq a_3N^3 + a_2N^2 + a_1N + a_0$ for all $a_1,a_2,a_3,a_4 \in [N-1]$.
 We define
 \[t \coloneqq 2\big( (n,n,n,n)_N + k(n,n,0,0)_N + 2(0,0,n,n)_N \big).\]
 We construct the graph $G$ as follows.
 For every $i,j \in [k]$ there is a vertex $w_{i,j}$ (it helps to think of $w_{i,j}$ as a vertex in a grid in row $i$ and column $j$ and $w_{1,1}$ being placed in the top left corner).
 Moreover, for $i \in \{0,\dots,k\}$, $j \in [k]$ and $p \in [n]$ we introduce vertices $v_{i,j,p}$ (\emph{vertical vertices)} and, for $i \in [k]$, $j \in \{0,\dots,k\}$ and $p \in [n]$, we add $h_{i,j,p}$ (\emph{horizontal vertices)}.
 We add edges $v_{i,j,p}v_{i,j,p+1}$ as well as $h_{i,j,p}h_{i,j,p+1}$ for all $p \in [n-1]$.
 For every $i,j \in [k]$ we also add edges $w_{i,j}v_{i-1,j,n}$, $w_{i,j}v_{i,j,1}$, $w_{i,j}h_{i,j-1,n}$, and $w_{i,j}h_{i,j,1}$
 Moreover, we introduce a vertex $u$ which is connected to all vertices $v_{0,j,1}$, $v_{k,j,n}$, $h_{i,0,1}$, and $h_{i,k,n}$ for all $i,j \in [k]$.
 
 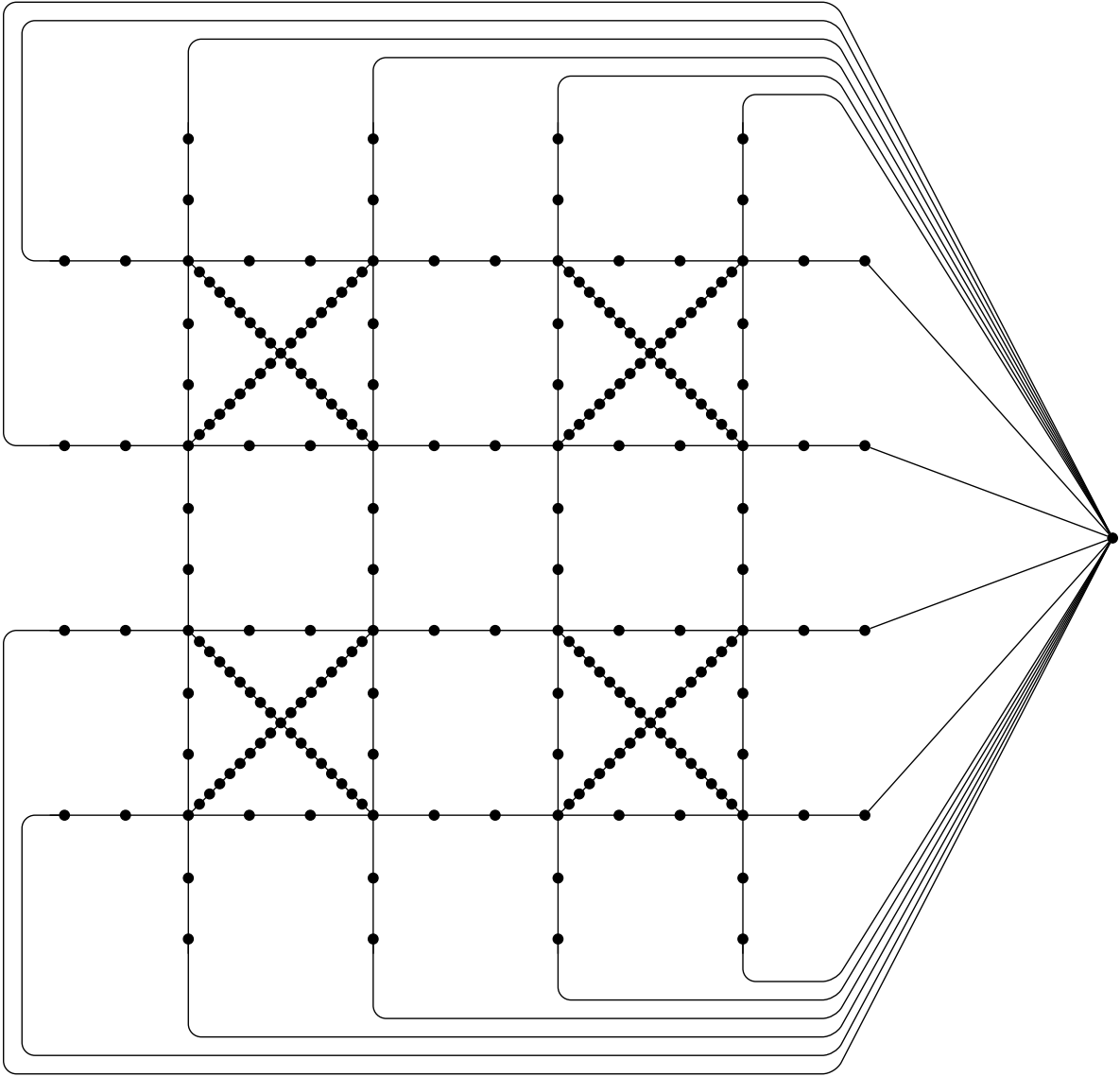
\begin{figure}
  \centering
  \scalebox{1.3}{
  \begin{tikzpicture}
   \draw (0.5,2) -- (9.33,2);
   \draw (0.5,4) -- (9.33,4);
   \draw (0.5,6) -- (9.33,6);
   \draw (0.5,8) -- (9.33,8);
   \draw (2,0.5) -- (2,9.5);
   \draw (4,0.5) -- (4,9.5);
   \draw (6,0.5) -- (6,9.5);
   \draw (8,0.5) -- (8,9.5);
   
   \draw (2,2) -- (4,4);
   \draw (2,4) -- (4,2);
   \draw (2,6) -- (4,8);
   \draw (2,8) -- (4,6);
   \draw (6,2) -- (8,4);
   \draw (6,4) -- (8,2);
   \draw (6,6) -- (8,8);
   \draw (6,8) -- (8,6);
   
   \foreach \i in {1,2,3,4}{
    \foreach \j in {1,2,3,4}{
     \node[smallvertex] (w\i\j) at (2*\j,2*\i) {};
    }
   }
   \foreach \i in {0,1,2,3,4}{
    \foreach \j in {1,2,3,4}{
     \foreach \p in {1,2}{
      \node[smallvertex] (v\i\j\p) at (2*\j,2*\i + 0.66*\p) {};
     }
    }
   }
   \foreach \i in {1,2,3,4}{
    \foreach \j in {0,1,2,3,4}{
     \foreach \p in {1,2}{
      \node[smallvertex] (h\i\j\p) at (2*\j + 0.66*\p,2*\i) {};
     }
    }
   }
   \foreach \i in {1,3}{
    \foreach \j in {1,3}{
     \node[smallvertex] (c\i\j) at (2*\j + 1,2*\i + 1) {};
     \foreach \p in {1,2,3,4,5,6,7,8}{
      \foreach \d/\e in {-1/-1,1/-1,-1/1,1/1}{
       \node[smallvertex] at (2*\j + 1 + \d*\p*0.11,2*\i + 1 + \e*\p*0.11) {};
      }
     }
    }
   }
   \node[smallvertex] (u) at (12,5) {};
   
   \draw (u) edge (h142);
   \draw (u) edge (h242);
   \draw (u) edge (h342);
   \draw (u) edge (h442);
   
   \draw[rounded corners] (u) -- (9,9.8) -- (8,9.8) -- (v442);
   \draw[rounded corners] (u) -- (9,10) -- (6,10) -- (v432);
   \draw[rounded corners] (u) -- (9,10.2) -- (4,10.2) -- (v422);
   \draw[rounded corners] (u) -- (9,10.4) -- (2,10.4) -- (v412);
   \draw[rounded corners] (u) -- (9,10.6) -- (0.2,10.6) -- (0.2,8) -- (h401);
   \draw[rounded corners] (u) -- (9,10.8) -- (0,10.8) -- (0,6) -- (h301);
   
   \draw[rounded corners] (u) -- (9,0.2) -- (8,0.2) -- (v041);
   \draw[rounded corners] (u) -- (9,0) -- (6,0) -- (v031);
   \draw[rounded corners] (u) -- (9,-0.2) -- (4,-0.2) -- (v021);
   \draw[rounded corners] (u) -- (9,-0.4) -- (2,-0.4) -- (v011);
   \draw[rounded corners] (u) -- (9,-0.6) -- (0.2,-0.6) -- (0.2,2) -- (h101);
   \draw[rounded corners] (u) -- (9,-0.8) -- (0,-0.8) -- (0,4) -- (h201);
   
  \end{tikzpicture}
  }
  \caption{Visualization of the graph $G$ for $k=4$ and $n=2$.}
  \label{fig:graph-construction-equal-components}
 \end{figure}
 
 Now fix $i,j \in [k]$ such that $i,j$ are both odd.
 We introduce vertices $c_{i,j}$ and $d_{i,j,p}^{\nearrow}$, $d_{i,j,p}^{\searrow}$, $d_{i,j,p}^{\swarrow}$, $d_{i,j,p}^{\nwarrow}$ for all $p \in [N^2-1]$.
 We add edges $c_{i,j}d_{i,j,1}^{\ar}$ and $d_{i,j,p}^{\ar}d_{i,j,p+1}^{\ar}$ for all $\ar \in \{\nearrow,\searrow,\swarrow,\nwarrow\}$ and $p \in [N^2-2]$.
 Also, there are edges $w_{i,j}d_{i,j,N^2-1}^{\nwarrow}$, $w_{i+1,j}d_{i,j,N^2-1}^{\swarrow}$, $w_{i,j+1}d_{i,j,N^2-1}^{\nearrow}$, and $w_{i+1,j+1}d_{i,j,N^2-1}^{\searrow}$.
 Note that $i+1,j+1 \in [k]$ since $k$ is even and $i,j$ are odd.
 This completes the description of the graph $G$.
 For later reference, define
 \[W \coloneqq \{w_{i,j} \mid i,j \in [k]\},\]
 \[V \coloneqq \{v_{i,j,p} \mid i \in \{0,\dots,k\},j \in [k], p \in [n]\},\]
 \[H \coloneqq \{h_{i,j,p} \mid i \in [k], j \in \{0,\dots,k\}, p \in [n]\},\]
 \[C \coloneqq \{c_{i,j} \mid i,j \in [k], i,j \equiv 1 \mod 2\},\]
 and
 \[D \coloneqq \{d_{i,j,p}^{\ar} \mid i,j \in [k], i,j \equiv 1 \mod 2, p \in [N^2-1], \ar \in \{\nearrow,\searrow,\swarrow,\nwarrow\}\}.\]
 So $V(G) = W \cup V \cup H \cup C \cup D \cup \{u\}$.
 
 Next, let us specify the weights of vertices.
 We define
 \[w(u) \coloneqq t - k(n,n,0,0)_N\]
 and
 \[w(w_{i,j}) \coloneqq t - (n,n,n,n)_N\]
 for all $w_{i,j} \in W$.
 Also,
 \[w(c_{i,j}) \coloneqq t - 2(0,0,n,n)_N\]
 for $c_{i,j} \in C$ and
 \[w(d_{i,j,p}^{\ar}) \coloneqq (0,0,0,1)_N\]
 for $d_{i,j,p}^{\ar} \in D$.
 Moreover, for $i,j \in \{0,\dots,k\}$ and $p \in [n+1]$, we define
 \[w(v_{i,j,p}) \coloneqq \begin{cases}
                           (0,1,0,1) &\text{if } i \equiv 1 \mod 2\\
                           (0,1,0,0) &\text{if } i \equiv 0 \mod 2
                          \end{cases}\]
 and
 \[w(h_{i,j,p}) \coloneqq \begin{cases}
                           (1,0,1,0) &\text{if } j \equiv 1 \mod 2\\
                           (1,0,0,0) &\text{if } j \equiv 0 \mod 2
                          \end{cases}.\]
 
 Finally, it remains to specify the set $U$ of undeletable edges.
 Fix $i,j \in [k]$ such that $i,j$ are both odd.
 For ease of notation we define $d_{i,j,0}^{\ar} \coloneqq c_{i,j}$ for all $\ar \in \{\nearrow,\searrow,\swarrow,\nwarrow\}$.
 Also, $d_{i,j,N^2}^{\nwarrow} \coloneqq w_{i,j}$, $d_{i,j,N^2}^{\nearrow} \coloneqq w_{i,j+1}$, $d_{i,j,N^2}^{\searrow} \coloneqq w_{i+1,j+1}$, and $d_{i,j,N^2}^{\swarrow} \coloneqq w_{i+1,j}$.
 We have that
 \[d_{i,j,p}^{\nwarrow}d_{i,j,p+1}^{\nwarrow} \in U \iff p = (0,0,a,b)_N \wedge (a,b) \notin S_{i,j}\]
 and
 \[d_{i,j,p}^{\nearrow}d_{i,j,p+1}^{\nearrow} \in U \iff p = (0,0,n-a,b)_N \wedge (a,b) \notin S_{i,j+1}\]
 and
 \[d_{i,j,p}^{\searrow}d_{i,j,p+1}^{\searrow} \in U \iff p = (0,0,n-a,n-b)_N \wedge (a,b) \notin S_{i+1,j+1}\]
 and
 \[d_{i,j,p}^{\swarrow}d_{i,j,p+1}^{\swarrow} \in U \iff p = (0,0,a,n-b)_N \wedge (a,b) \notin S_{i+1,j}.\]
 To complete the description, we set 
 \[k' \coloneqq 2k(k+1) + k^2.\]
 
 Clearly, the instance $(G,w,U,k',t)$ can be computed in polynomial time.
 Also, the construction meets the specified size bounds.
 Hence, it remains to prove that $(n,k,(S_{i,j})_{i,j \in [k]})$ has a solution if and only if $(G,w,U,k',t)$ has a solution.
 
 Let us start with some basic observations on solutions $Z$ of the instance $(G,w,U,k',t)$.
 The vertices of $G$ that do not have degree exactly two are those contained in $W \cup C \cup \{u\}$.
 Each of these vertices has weight strictly greater than $t/2$.
 Hence, each connected component of $G - Z$ contains at most one these vertices.
 
 On the other hand, consider $V \cup H \cup D$ the set of vertices of degree exactly two.
 The graph $G[V \cup H \cup D]$ has exactly $k'$ many connected components and each component is incident to two vertices of $W \cup C \cup \{u\}$.
 This means that $Z$ contains exactly one edge incident to a vertex from each connected component of $G[V \cup H \cup D]$.
 In particular, a solution $Z$ assigns each vertex in $(V \cup H \cup D) \setminus Z$ to one of the vertices from $W \cup C \cup \{u\}$.
 
 We now formally prove the correctness of the reduction.
 Let $r_i,c_j \in \{0,\dots,n\}$, $i,j \in [k]$, be a solution for $(n,k,(S_{i,j})_{i,j \in [k]})$.
 Again, for ease of notation, we define $w_{0,j} = w_{k+1,j} \coloneqq u$ and $w_{i,0} = w_{i,k+1} \coloneqq u$.
 Also, $v_{i,j,0} \coloneqq w_{i,j}$ and $v_{i,j,n+1} \coloneqq w_{i+1,j}$.
 Moreover, $h_{i,j,0} \coloneqq w_{i,j}$ and $h_{i,j,n+1} \coloneqq w_{i,j+1}$.
 We define
 \begin{align*}
  Z \coloneqq\;\;\; &\{v_{i,j,c_j}v_{i,j,c_j+1} \mid i \in \{0,\dots,k\}, j \in [k]\}\\
             \cup\; &\{h_{i,j,r_i}h_{i,j,r_i+1} \mid i \in [k], j \in \{0,\dots,k\}\}\\
             \cup\; &\{d_{i,j,p}^{\nwarrow}d_{i,j,p+1}^{\nwarrow} \mid p = (0,0,r_i,c_j)_N, i,j \in [k], i,j \equiv 1 \mod 2\}\\
             \cup\; &\{d_{i,j,p}^{\nearrow}d_{i,j,p+1}^{\nearrow} \mid p = (0,0,n-r_i,c_{j+1})_N, i,j \in [k], i,j \equiv 1 \mod 2\}\\
             \cup\; &\{d_{i,j,p}^{\searrow}d_{i,j,p+1}^{\searrow} \mid p = (0,0,n-r_{i+1},n-c_{j+1})_N, i,j \in [k], i,j \equiv 1 \mod 2\}\\
             \cup\; &\{d_{i,j,p}^{\swarrow}d_{i,j,p+1}^{\swarrow} \mid p = (0,0,r_{i+1},n-c_j)_N, i,j \in [k], i,j \equiv 1 \mod 2\}
 \end{align*}
 Note that $|Z| = 2k(k+1) + 4\left(\frac{k}{2}\right)^2 = k'$.
 We need to verify that all connected components of $G - Z$ have weight exactly $t$.
 As already indicated above, it is easy to verify that there is a one-to-one correspondence between $W \cup C \cup \{u\}$ and the connected components of $G - Z$.
 We calculate their weights.
 
 Let $i,j \in [k]$ such that $i,j$ are both odd.
 Let $C_{i,j}$ be the connected component of $G - Z$ that contains $c_{i,j}$.
 We have that
 \begin{align*}
  w(C_{i,j}) &= w(c_{i,j}) + (0,0,r_i,c_j)_N + (0,0,n-r_i,c_{j+1})_N\\
             &\;\;\;\;\; + (0,0,n-r_{i+1},n-c_{j+1})_N + (0,0,r_{i+1},n-c_j)_N\\
             &= w(c_{i,j}) + 2(0,0,n,n)_N = t.
 \end{align*}
 Next, let $W_{i,j}$ be the connected component that contains $w_{i,j}$.
 Then
 \begin{align*}
  w(W_{i,j}) &= w(w_{i,j}) + (r_i,0,r_i,0)_N + (n-r_i,0,0,0)_N +  (0,c_j,0,c_j)_N + (0,n-c_j,0,0)_N \\
             &\;\;\;\;\; + N^2 - 1 - (0,0,r_i,c_j)_N\\
             &= w(w_{i,j}) + (n,n,n,n)_N = t.
 \end{align*}
 The calculations for the components of $w_{i,j+1}$, $w_{i+1,j}$, and $w_{i+1,j+1}$ are similar.
 We have
 \begin{align*}
  w(W_{i,j+1}) &= w(w_{i,j+1}) + (n-r_i,0,n-r_i,0)_N + (r_i,0,0,0)_N +  (0,c_{j+1},0,c_{j+1})_N\\
             &\;\;\;\;\;  + (0,n-c_{j+1},0,0)_N + N^2 - 1 - (0,0,n-r_i,c_{j+1})_N\\
             &= w(w_{i,j}) + (n,n,n,n)_N = t,
 \end{align*}
 \begin{align*}
  w(W_{i+1,j}) &= w(w_{i,j}) + (r_{i+1},0,r_{i+1},0)_N + (n-r_{i+1},0,0,0)_N +  (0,n - c_j,0,n - c_j)_N\\
             &\;\;\;\;\; + (0,c_j,0,0)_N + N^2 - 1 - (0,0,r_{i+1},n-c_j)_N\\
             &= w(w_{i,j}) + (n,n,n,n)_N = t
 \end{align*}
 and
 \begin{align*}
  w(W_{i+1,j+1}) &= w(w_{i,j}) + (n - r_{i+1},0,n - r_{i+1},0)_N + (r_{i+1},0,0,0)_N\\
             &\;\;\;\;\; + (0,n - c_{j+1},0,n - c_{j+1})_N + (0,c_{j+1},0,0)_N\\
             &\;\;\;\;\; + N^2 - 1 - (0,0,n - r_{i+1},n - c_{j+1})_N\\
             &= w(w_{i,j}) + (n,n,n,n)_N = t.
 \end{align*}
 Finally, consider the connected component $C_u$ that contains vertex $u$.
 We have that
 \begin{align*}
  w(C_u) &= w(u) + \sum_{i \in [k]} ((r_i,0,0,0)_N + (n-r_i,0,0,0)_N) + \sum_{j \in [k]} ((0,c_j,0,0)_N + (0,n-c_j,0,0)_N)\\
         &= w(u) + k(n,n,0,0)_N = t
 \end{align*}
 This proves that $Z$ is a solution.
 
 For the other direction let $Z$ be a solution for $(G,w,U,k',t)$.
 By the above comments, each component of $G - Z$ contains exactly one of the vertices $W \cup C \cup \{u\}$.
 So there are $k^2 + \frac{k^2}{4} + 1$ components in $G - Z$, each of weight $t$.
 
 Consider first the set $H$.
 We have $|Z \cap \{h_{i,j,p}h_{i,j,p+1} \mid i,j \in [k], p \in \{0,\dots,n\}\}| = k(k+1)$ and $k(k+1)n$ vertices need to be assigned to components of $G - Z$.
 Such a vertex can only be assigned to a vertex $w_{i,j}$ or to vertex $u$.
 By the weight constraint, one can assign at most $n$ vertices of $H$ to $w_{i,j}$ and $kn$ vertices to $u$.\
 Hence, one needs to assign exactly $n$ vertices from $H$ to $w_{i,j}$ for every $i,j \in [k]$ as well as assign exactly $kn$ vertices from $H$ to $u$.
 This implies that there are numbers $r_i \in \{0,\dots,n\}$ such that $h_{i,j,p}h_{i,j,p+1} \in Z$ if and only if $p = r_i$.
 
 Now we can repeat the same argument for the vertical vertices in $V$.
 Hence, there are numbers $c_j \in \{0,\dots,n\}$ such that $v_{i,j,p}v_{i,j,p+1} \in Z$ if and only if $p = c_j$.
 
 Now fix some $i,j \in [k]$ such that $i,j$ are both odd.
 We need to verify that $(r_i,c_j) \in S_{i,j}$, $(r_{i+1},c_j) \in S_{i+1,j}$, $(r_i,c_{j+1}) \in S_{i,j+1}$, and $(r_{i+1},c_{j+1}) \in S_{i+1,j+1}$.
 Here, we only perform the calculations for the first case, the other cases are similar.
 Not taking the vertices $d_{i,j,p}^{\nwarrow}$ into account, the weight of the component containing $w_{i,j}$ is exactly
 \[w(w_{i,j}) + (r_i,0,r_i,0)_N + (n-r_i,0,0,0)_N +  (0,c_j,0,c_j)_N + (0,n-c_j,0,0)_N\]
 Hence, we need to add $(0,0,n-r_i,n-c_j)$ to the weight using the vertices $d_{i,j,p}^{\nwarrow}$.
 Suppose $d_{i,j,p}^{\nwarrow}d_{i,j,p+1}^{\nwarrow} \in Z$.
 By the above constraint we conclude that $p = (0,0,r_i,c_j)_N$.
 Moreover, $d_{i,j,p}^{\nwarrow}d_{i,j,p+1}^{\nwarrow} \notin U$.
 Hence, $(r_i,c_j) \in S_{i,j}$ by definition.
\end{proof}

In combination with Theorem \ref{thm:grid-tiling-lower-bound}, we obtain the following corollary.

\begin{corollary}
  \textsc{Planar Extended Component Equality Edge Connectivity} (parameterized by $k$) is \WH\ and, unless \ETH\ fails, it can not be solved in time $f(k)n^{o(\sqrt{k})}$ for any computable function $f$.
\end{corollary}

\subsection{Variants}

Next, we argue how to adapt the above reduction for other variants of \textsc{Component Order Connectivity}.
We do not intend to cover every possible variant of the problem, but rather choose certain variants to exemplify how the above reduction can be modified.
Also, we do spell out the reductions in detail, but rather only describe which changes need to be made to the proof of Lemma \ref{la:reduction-component-order-connectivity-basic}. 

First, we cover the edge deletion version of \textsc{Component Order Connectivity}.

\medskip 
\defparproblem{\textsc{Extended Component Order Edge Connectivity}}{A graph $G$, a weight function $w\colon V(G) \rightarrow \NN$, a set $U \subseteq E(G)$, and integers $k,t$.}{$k$}
{Is there a set $Z \subseteq E(G) \setminus U$ such that $|Z| \leq k$ and $w(C) \leq t$ for every connected component $C$ of $G - Z$?}
\medskip

\begin{lemma}
 There is a polynomial-time algorithm that, given an instance $(n,k,(S_{i,j})_{i,j \in [k]})$ of \textsc{Grid Tiling}, computes an equivalent instance $(G,w,U,k',t)$ of \textsc{Planar Extended Component Order Edge Connectivity} such that
 \begin{enumerate}
  \item $|V(G)| = \CO(k^2n^2)$,
  \item $w(v) = \CO(kn^3)$ for all $v \in V(G)$, and
  \item $k' = \CO(k^2)$.
 \end{enumerate}
\end{lemma}

\begin{proof}
 The proof is completely analogous to the reduction above.
 In particular, the construction is exactly the same.
 Since the number of components of $G - Z$ is always the same for every potential solution, we get that $w(C) \leq t$ for every connected component $C$ of $G - Z$ implies that $w(C) = t$ for every connected component $C$ of $G - Z$.
\end{proof}

Next, we turn to the variant where we restrict the size of every $2$-connected component.
As for the first reduction, we focus on the variant where every $2$-connected component should have size exactly $t$ (after removing the solution)

\medskip 
\defparproblem{\textsc{Extended $2$Conn Component Equality Edge Connectivity}}{A graph $G$, a weight function $w\colon V(G) \rightarrow \NN$, a set $U \subseteq E(G)$, and integers $k,t$.}{$k$}
{Is there a set $Z \subseteq E(G) \setminus U$ such that $|Z| \leq k$ and $w(C) = t$ for every $2$-connected component $C$ of $G - Z$?}
\medskip

\begin{lemma}
 There is a polynomial-time algorithm that, given an instance $(n,k,(S_{i,j})_{i,j \in [k]})$ of \textsc{Grid Tiling}, computes an equivalent instance $(G,w,U,k',t)$ of \textsc{Planar Extended $2$Conn Component Equality Edge Connectivity} such that
 \begin{enumerate}
  \item $|V(G)| = \CO(k^2n^2)$,
  \item $w(v) = \CO(kn^3)$ for all $v \in V(G)$, and
  \item $k' = \CO(k^2)$.
 \end{enumerate}
\end{lemma}

\begin{proof}
 We modify the construction as follows to obtain a graph $G'$.
 Each vertex $v \in V(G)$ of degree $\delta$ is replaced by a cycle $C_\delta$ of length $\delta$ (for $\delta = 2$ we replace $v$ by two vertices connected by an edge).
 All edges of the cycle are undeletable.
 Let $v(1),\dots,v(\delta)$ denote the vertices of the cycle in the cyclic order.
 Also, we number the neighbors of $v$ in an arbitrary order.
 Let $N(v,i)$ denote the $i$-th neighbor of $v$ for $i \in [\delta]$.
 
 For an edge $vu \in E(G)$ such that $u = N(v,i)$ and $v = N(u,j)$ we add edges $v(i)u(j)$ and $v(i+1)u(j+1)$ (where positions are taken modulo $\delta$).
 If $vu$ is undeletable, then both of these edges become undeletable.
 
 For the weight function, we set $w(v(1)) \coloneqq n \cdot w(v) - \deg(v) + 1$ and $w(v(i)) \coloneqq 1$ for all $v \in V(G)$ and $i \geq 2$.
 Finally, we update $k'' \coloneqq 2k'$.
 
 We argue that the new instance $(G',w',U',k'',t \cdot n)$ of \textsc{Planar Extended $2$Conn Component Equality Edge Connectivity} has a solution if and only if the original instance $(G,w,U,k',t)$ of \textsc{Planar Extended Component Equality Edge Connectivity}.
 Let $Z$ be a solution for the original instance $(G,w,U,k',t)$ of \textsc{Planar Extended Component Equality Edge Connectivity}.
 Then we simply delete, for all $uv \in Z$, the edges $v(i)u(j)$ and $v(i+1)u(j+1)$ where $u = N(v,i)$ and $v = N(u,j)$.
 Let $Z'$ be the resulting set of edges.
 Then the $2$-connected components of $G' - Z'$ are exactly the connected components of $G' - Z'$ and it is easy to verify that $Z'$ is a solution.
 
 In the other direction, the main observation is that, if $v(i)u(j)$ is part of a solution $Z'$, then also its partner edge must be contained in a solution, since otherwise there has to be a bridge edge in $G' - Z'$ and its $2$-connected component can not have size exactly $t$.
\end{proof}

Finally, the next example demonstrates how to modify the reductions for vertex deletion problems.

\medskip 
\defparproblem{\textsc{$2$Conn Component Order Connectivity}}{A graph $G$, a weight function $w\colon V(G) \rightarrow \NN$, a set $U \subseteq V(G)$, and integers $k,t$.}{$k$}
{Is there a set $Z \subseteq V(G) \setminus U$ such that $|Z| \leq k$ and $w(C) \leq t$ for every $2$-connected component $C$ of $G - Z$?}
\medskip

\begin{lemma}
 There is a polynomial-time algorithm that, given an instance $(n,k,(S_{i,j})_{i,j \in [k]})$ of \textsc{Grid Tiling}, computes an equivalent instance $(G,w,U,k',t)$ of \textsc{Planar $2$Conn Component Order Connectivity} such that
 \begin{enumerate}
  \item $|V(G)| = \CO(k^2n^2)$,
  \item $w(v) = \CO(kn^3)$ for all $v \in V(G)$, and
  \item $k' = \CO(k^2)$.
 \end{enumerate}
\end{lemma}

\begin{proof}
 We follow again the construction presented in the last lemma, but with one modification.
 We update
 \[k'' \coloneqq n \cdot \left(\frac{19}{4}k^2 + 4k\right) = n \cdot \left(2k' - |W \cup C + \{u\}| - 1\right).\]
 
 We argue that the new instance $(G',w',U',k'',t \cdot n)$ of \textsc{Planar $2$CC Size Deletion} has a solution if and only if the original instance $(G,w,U,k',t)$ of \textsc{Planar Extended Component Equality Edge Connectivity}.
 Let $Z$ be a solution for the original instance $(G,w,U,k',t)$ of \textsc{Planar Extended Component Equality Edge Connectivity}.
 Recall that the number of connected components of $G - Z$ is $|W \cup C + \{u\}|$.
 We first delete, for all $uv \in Z$, the edge $v(i)u(j)$ where $u = N(v,i)$ and $v = N(u,j)$.
 Note that, for each $uv \in Z$, there is a second edge $v(i+1)u(j+1)$.
 We either wish to delete this edge or turn it into a bridge edge in the resulting graph.
 This can be achieved by deleting exactly $k' - |W \cup C + \{u\}| - 1$ of these edges.
 
 In the other direction, we add $uv$ to a solution $Z$ if and only if $v(i)u(j) \in Z'$ or $v(i+1)u(j+1) \in Z'$ for a solution $Z'$ of $(G',w',U',k'',t \cdot n)$.
 It is easy to verify that $|Z| \leq k'$ and $Z$ is a solution of $(G,w,U,k',t)$.
 
 Here, we can repeat the same argument as before.
 We simply use the reduction from the last lemma
 Since the number of $2$-connected components of $G - Z$ is always the same for every potential solution, we get that $w(C) \leq t$ for every $2$-connected connected component $C$ of $G - Z$ implies that $w(C) = t$ for every $2$-connected connected component $C$ of $G - Z$.
\end{proof}

Again, in combination with Theorem \ref{thm:grid-tiling-lower-bound}, we obtain the following lower bounds.

\begin{corollary}
 Let $\Pi$ be one of the problems \textsc{Extended Component Order Edge Connectivity}, \textsc{Extended $2$Conn Component Equality Edge Connectivity} and \textsc{$2$Conn Component Order Connectivity}.
 Then $\Pi$ (parameterized by $k$) is \WH\ and, unless \ETH\ fails, it can not be solved in time $f(k)n^{o(\sqrt{k})}$ for any computable function $f$.
\end{corollary}

\subsection{Removing Weights and Undeletable Vertices}

Next, we show how to realize the reductions without using vertex weights or declaring certain edges/vertices to be undeletable.
We only perform the reduction for the first problem \textsc{Planar Extended Component Equality Edge Connectivity}.
We remark that reductions for the other variants are similar.

\medskip 
\defparproblem{\textsc{Component Equality Edge Connectivity}}{A graph $G$, and integers $k,t$.}{$k$}
{Is there a set $Z \subseteq E(G)$ such that $|Z| \leq k$ and $|C| = t$ for every connected component $C$ of $G - Z$?}
\medskip

\begin{lemma}
 There is a polynomial-time many-one reduction from \textsc{Planar Extended Component Equality Edge Connectivity} to \textsc{Planar Component Equality Edge Connectivity} which leaves the parameter $k$ unchanged.
\end{lemma}

\begin{proof}
 Let $(G,w,U,k,t)$ be an instance of \textsc{Planar Extended Component Equality Edge Connectivity}.
 We first argue how to obtain an equivalent instance $(G',w',\emptyset,k,t')$ of \textsc{Planar Extended Component Equality Edge Connectivity}.
 
 We construct the graph $G'$ from $G$ as follows.
 For each edge $uv \in U$ we add vertices $x(uv,i)$ for all $i \in [k]$.
 Each vertex $x(uv,i)$ is connected to $u$ and $v$.
 Hence, for each edge $uv \in U$, we add $k$ internally vertex-disjoint paths from $u$ to $v$.
 Also, we set
 \[w'(v) \coloneqq nk \cdot w(v) - k\cdot |\{u \in N_G(v) \mid uv \in U\}|\]
 for all $v \in V(G)$ and $w'(x(uv,i)) \coloneqq 2$ for all $uv \in U$ and $i \in [k]$.
 Finally, $t' \coloneqq nkt$.
 
 We prove the correctness of the reduction.
 Let $Z$ be a solution for $(G,w,U,k,t)$.
 Then $Z$ is also a solution for the updated instance $(G',w',\emptyset,k,t')$.
 Indeed, let $C$ be a connected component of $G' - Z$.
 Then
 \begin{align*}
  w'(C) &= \left(\sum_{v \in V(G) \cap C} nk \cdot w(v) - k\cdot |\{u \in N_G(v) \mid uv \in U\}\right) + \sum_{u,v \in V(G) \cap C, uv \in U} 2k\\
        &= \sum_{v \in V(G) \cap C} nk \cdot w(v) = nk \cdot \sum_{v \in V(G) \cap C} w(v) = nkt = t'.
 \end{align*}

 On the other hand, let $Z'$ be a solution for $(G',w',\emptyset,k,t')$.
 For each $uv \in U$ there are are $k+1$ internally vertex-disjoint paths from $u$ to $v$ in $G'$.
 Hence, $u$ and $v$ have to be in the same connected component of $G' - Z'$.
 This component also contains $x(uv,i)$ since such a vertex can not appear in a connected component alone.
 Let $Z \coloneqq (E(G) \cap Z') \setminus U$.
 By the observations above it holds that $Z$ is also a solution for $(G',w',\emptyset,k,t')$.
 Also, repeating the calculations above, it is also a solution for $(G,w,U,k,t)$.
 
 Now, in a second step, we argue how to remove the weights.
 Let $(G,w,\emptyset,k,t)$ be an instance of \textsc{Planar Extended Component Equality Edge Connectivity}.
 We construct an equivalent instance $(G',k,t)$ of \textsc{Planar Component Equality Edge Connectivity}.
 The graph $G'$ is simply obtained from $G$ by adding, for each $v \in V(G)$, fresh copies $v(1),\dots,v(w(v)-1)$ each of which is connected only to $v$.
 Note that it is not possible to remove any of the newly added edges since this would create a component of size $1$ which is not allowed in any solution (if $t=1$ then all weights have to be $1$, otherwise there is no solution).
 Clearly, the instance $(G',k,t)$ is equivalent to $(G,w,\emptyset,k,t)$.
\end{proof}

\begin{corollary}
  \textsc{Planar Component Equality Edge Connectivity} (parameterized by $k$) is \WH\ and, unless \ETH\ fails, it can not be solved in time $f(k)n^{o(\sqrt{k})}$ for any computable function $f$.
\end{corollary}